\documentclass{lmcs}
\usepackage{hyperref}
\usepackage{booktabs}
\usepackage{amssymb}
\usepackage{leftidx}

\newcommand{\skipnoindent}{\noindent}

\newcommand{\nat}{\mathbb{N}}
\newcommand{\arity}{\#}
\newcommand{\arityof}[1]{\arity{#1}}
\newcommand{\rankof}[1]{\rho({#1})}
\newcommand{\frof}[1]{\mathrm{fr}({#1})}

\newcommand{\cardof}[1]{\mathrm{card}({#1})}

\newcommand{\width}[1]{\mathrm{wd}({#1})}

\newcommand{\univ}{\mathsf{U}}
\newcommand{\universe}{\mathbb{U}}

\newcommand{\vars}{\mathbb{V}}

\newcommand{\preds}{\mathbb{P}}
\newcommand{\isdef}{\stackrel{\scalebox{.4}{$\mathsf{def}$}}{=}}
\newcommand{\iffdef}{\stackrel{\hspace*{1.5pt}\scalebox{.4}{$\mathsf{def}$}}{\iff}}
\newcommand{\interv}[2]{[{#1}..{#2}]}
\newcommand{\tuple}[1]{\langle {#1} \rangle}

\newcommand{\set}[1]{\{ {#1} \}}
\newcommand{\mset}[1]{\{\!\!\{{#1}\}\!\!\}}

\newcommand{\pow}[1]{\mathrm{pow}({#1})}
\newcommand{\mpow}[1]{\mathrm{mpow}({#1})}
\newcommand{\dom}[1]{\mathrm{dom}({#1})}

\newcommand{\np}{$\mathsf{NP}$}

\newcommand{\conp}{$\mathsf{co}$-$\mathsf{NP}$}

\newcommand{\relations}{\mathbb{R}}

\newcommand{\alphabet}{\mathbb{A}}

\newcommand{\arel}{\mathsf{r}}
\newcommand{\aarel}{\mathsf{a}}
\newcommand{\abrel}{\mathsf{b}}
\newcommand{\acrel}{\mathsf{c}}

\newcommand{\erel}{\mathsf{e}}

\newcommand{\posfunc}{\mathfrak{P}}
\newcommand{\profile}[1]{\posfunc_{\scriptscriptstyle{#1}}}

\newcommand{\states}{\mathcal{Q}}

\newcommand{\initstates}{\mathcal{I}}
\newcommand{\initstate}{\iota}

\newcommand{\arrow}[2]{\xrightarrow{{\scriptscriptstyle #1}}_{{\scriptstyle #2}}}

\newcommand{\arun}{\theta}
\newcommand{\pre}[1]{\leftidx{^\bullet}{\!{#1}}{}}
\newcommand{\post}[1]{\leftidx{}{{#1}}{^\bullet}}
\newcommand{\prepost}[1]{\leftidx{^\bullet}{\!{#1}}{^\bullet}}
\newcommand{\reach}{\leadsto}

\newcommand{\auto}[2]{\mathcal{A}_{
    {#1}
    \ifthenelse{\equal{#2}{}}{}{,{#2}}
}}

\newcommand{\autsat}[2]{\mathcal{A}^{I}_{
    {#1}
    \ifthenelse{\equal{#2}{}}{}{,{#2}}
}}

\newcommand{\autcut}[2]{\mathcal{A}^{\scriptscriptstyle\mathsf{cut}}_{
    {#1}
    \ifthenelse{\equal{#2}{}}{}{,{#2}}
}}

\newcommand{\autcst}[2]{\mathcal{A}^{\scriptscriptstyle\mathsf{cst}}_{
    {#1}
    \ifthenelse{\equal{#2}{}}{}{,{#2}}
  }
}

\newcommand{\expof}[1]{{#1}^{\scriptscriptstyle\mathsf{exp}}}

\makeatletter
\newcommand*{\da@rightarrow}{\mathchar"0\hexnumber@\symAMSa 4B }
\newcommand*{\da@leftarrow}{\mathchar"0\hexnumber@\symAMSa 4C }
\newcommand*{\xdashrightarrow}[2][]{%
  \mathrel{%
    \mathpalette{\da@xarrow{#1}{#2}{}\da@rightarrow{\,}{}}{}%
  }%
}
\newcommand{\xdashleftarrow}[2][]{%
  \mathrel{%
    \mathpalette{\da@xarrow{#1}{#2}\da@leftarrow{}{}{\,}}{}%
  }%
}
\newcommand*{\da@xarrow}[7]{%
  \sbox0{$\ifx#7\scriptstyle\scriptscriptstyle\else\scriptstyle\fi#5#1#6\m@th$}%
  \sbox2{$\ifx#7\scriptstyle\scriptscriptstyle\else\scriptstyle\fi#5#2#6\m@th$}%
  \sbox4{$#7\dabar@\m@th$}%
  \dimen@=\wd0 %
  \ifdim\wd2 >\dimen@
    \dimen@=\wd2 %
  \fi
  \count@=2 %
  \def\da@bars{\dabar@\dabar@}%
  \@whiledim\count@\wd4<\dimen@\do{%
    \advance\count@\@ne
    \expandafter\def\expandafter\da@bars\expandafter{%
      \da@bars
      \dabar@
    }%
  }%
  \mathrel{#3}%
  \mathrel{%
    \mathop{\da@bars}\limits
    \ifx\\#1\\%
    \else
      _{\copy0}%
    \fi
    \ifx\\#2\\%
    \else
      ^{\copy2}%
    \fi
  }%
  \mathrel{#4}%
}
\makeatother

\newcommand{\store}{\mathfrak{s}}
\newcommand{\struc}{\sigma}

\newcommand{\allstof}[4]{
  \Omega_{#1}^{{#2}/{#3}}
  \ifthenelse{\equal{#4}{}}{}{({#4})}
}
\newcommand{\substruc}{\sqsubseteq}

\newcommand{\astruc}{\mathsf{S}}

\newcommand{\Bigstar}{\mathop{\Asterisk}}

\newcommand{\comp}{\bullet}

\newcommand{\Comp}{\mathop{\raisebox{-3pt}{\scalebox{1.8}{$\bullet$}}}}
\newcommand{\kcomp}{\comp^{\scriptscriptstyle\sharp k}}
\newcommand{\Kcomp}{\mathop{\Comp\nolimits^{\scriptscriptstyle\sharp k}}}

\newcommand{\predname}[1]{\mathsf{#1}}

\newcommand{\apred}{\predname{A}}
\newcommand{\bpred}{\predname{B}}
\newcommand{\cpred}{\predname{C}}

\newcommand{\ppred}{\predname{P}}

\newcommand{\emp}{\predname{emp}}

\let\Asterisk\undefined
\newcommand{\Asterisk}{\mathop{\scalebox{1.9}{\raisebox{-0.2ex}{$\ast$}}}\hspace*{1pt}}%
\renewcommand{\vec}[1]{\mathbf #1}
\newcommand{\fv}[1]{\mathrm{fv}({#1})}
\newcommand{\fvof}[2]{\mathrm{fv}_{\scriptscriptstyle #1}({#2})}
\newcommand{\connof}[1]{\mathsf{conn}({#1})}

\newcommand{\mso}{\textsf{MSO}}

\newcommand{\Models}{\Vdash}

\newcommand{\seplog}{\textsf{SL}}

\newcommand{\slr}{\textsf{SLR}}
\newcommand{\fol}{\textsf{FO}}

\newcommand{\tbsl}{$\mathsf{TWB^{\scriptscriptstyle{\text{\slr}}}}$}
\newcommand{\gl}{\textsf{GL}}

\newcommand{\abstruc}{\struc^\sharp}

\newcommand{\asid}{\Delta}
\newcommand{\csid}{\Gamma}
\newcommand{\maxvarinruleof}[1]{\mathsf{maxVars}({#1})}
\newcommand{\maxrelatominruleof}[1]{\mathsf{maxRelAtoms}(#1)}
\newcommand{\maxrulearityof}[1]{\mathsf{maxPredAtoms}(#1)}
\newcommand{\relationsno}[1]{\mathsf{relNo}(#1)}
\newcommand{\predsno}[1]{\mathsf{predNo}(#1)}
\newcommand{\maxrelarityof}[1]{\mathsf{maxRelArity}(#1)}
\newcommand{\maxpredarityof}[1]{\maxvarinruleof{#1}}

\newcommand{\arule}{\rho}

\newcommand{\gfp}{\mathrm{gfp}}

\newcommand{\rbr}{{\bf ]\!]}}
\newcommand{\lbr}{{\bf [\![}}
\newcommand{\sem}[1]{{\lbr #1 \rbr}}

\newcommand{\sidsem}[2]{\sem{{#1}}_{#2}}

\newcommand{\csem}[2]{\sem{{#1}}^{\scriptscriptstyle\mathsf{c}}_{#2}}
\newcommand{\rcsem}[2]{\sem{{#1}}^{\scriptscriptstyle\mathsf{r}}_{#2}}
\newcommand{\intfusion}[1]{\mathtt{IF}({#1})}

\newcommand{\sintfusion}[2]{\widetilde{\mathtt{IF}}({#1}\ifthenelse{\equal{#2}{}}{}{,{#2}})}

\newcommand{\extfusion}[2]{\mathtt{EF}({#1},{#2})}
\newcommand{\reachfusion}[1]{\mathtt{EF}^*({#1})}
\newcommand{\ireachfusion}[1]{\mathtt{IEF}^*({#1})}

\newcommand{\reachof}[3]{\mathrm{reach}_{#1}^{#2}({#3})}

\newcommand{\step}[1]{\Rightarrow_{\scriptscriptstyle{#1}}}
\newcommand{\diseq}{\mathfrak{d}}

\newcommand{\langof}[2]{\mathcal{L}_{#1}({#2})}
\newcommand{\runsof}[3]{\mathcal{R}^{#1}_{~{#2}}({#3})}

\newcommand{\proj}[2]{{#1}\!\!\downharpoonleft_{\scriptscriptstyle{#2}}}
\newcommand{\projrel}[2]{\pi_{#2}({#1})}
\newcommand{\kproj}[2]{{#1}\!\!\downharpoonleft_{\scriptscriptstyle{#2}}^{\scriptscriptstyle\sharp k}}
\newcommand{\supp}[1]{\mathrm{supp}({#1})}

\newcommand{\satbasetuples}{\mathsf{SatBase}}
\newcommand{\satbasetuplesof}[1]{\satbasetuples({#1})}

\newcommand{\basepairof}[1]{\mathsf{Base}({#1})}

\newcommand{\basecomp}{\otimes}

\newcounter{index}

\newcommand{\graphof}[1]{\mathfrak{G}_{\scriptscriptstyle{#1}}}
\newcommand{\entryof}[1]{\mathsf{entry}({#1})}

\newcommand{\eqof}[1]{\approx_{\scriptscriptstyle{#1}}}

\newcommand{\atpos}[2]{{#1}^{\scriptscriptstyle[{#2}]}}

\newcommand{\graph}{G}

\newcommand{\labels}{\Omega}

\newcommand{\trans}{\delta}

\newcommand{\nodes}{\mathcal{N}}
\newcommand{\edges}{\mathcal{E}}
\newcommand{\alabel}{\lambda}
\newcommand{\tree}{T}
\newcommand{\subtree}[2]{{#1}|_{{#2}}}

\newcommand{\twof}[1]{\mathrm{tw}({#1})}
\newcommand{\charform}[1]{\Theta({#1})}
\newcommand{\exclof}[1]{{#1}^\exists}

\newcommand{\structures}{\mathcal{S}}

\newcommand{\mcsubstruc}{\substruc^{mc}}
\newcommand{\funsplit}[1]{\mathtt{split}({#1})}

\newcommand{\extfusionone}[2]{\mathtt{EF}_1({#1},{#2})}
\newcommand{\reachfusionone}[1]{\mathtt{EF}_1^*({#1})}
\newcommand{\reachfusiontwo}[1]{\mathtt{EF}_2^*({#1})}
\newcommand{\reachfusionk}[1]{\mathtt{EF}_{#1}^*}

\newcommand{\absextfusionone}[2]{\mathtt{ef}_1^{\scriptscriptstyle{\sharp}}({#1},{#2})}
\newcommand{\absreachfusionone}[1]{\mathtt{ef}_1^{\scriptscriptstyle\sharp*}({#1})}

\newcommand{\kabsextfusionone}[3]{\mathtt{ef}_1^{\scriptscriptstyle\sharp{#1}}({#2},{#3})}
\newcommand{\kabsreachfusionone}[2]{\mathtt{ef}_1^{\scriptscriptstyle\sharp{#1}*}({#2})}

\newcommand{\mcolabs}[1]{{#1}^{\sharp}}
\newcommand{\kmcolabs}[2]{{#2}^{\scriptscriptstyle\sharp{#1}}}

\newcommand{\ucolor}{\mathcal{C}}

\newcommand{\funcol}[1]{\mathcal{C}_{#1}}

\newcommand{\bluecols}{\mathbb{C}^{blue}}
\newcommand{\greencols}{\mathbb{C}^{green}}
\newcommand{\redcols}{\mathbb{C}^{red}}
\newcommand{\allcols}{\mathbb{C}}
\newcommand{\colorof}[1]{\gamma({#1})}

\newcommand{\bluetype}{\mathsf{B}}
\newcommand{\greentype}{\mathsf{G}}
\newcommand{\redtype}{\mathsf{R}}

\newcommand{\kabssem}[3]{\langle\!\!\langle #2 \rangle\!\!\rangle_{#3}^{\scriptscriptstyle\sharp{#1}}}

\newcommand{\upf}{\mathtt{up}}
\newcommand{\rightf}{\mathtt{right}}
\newcommand{\south}{\mathtt{S}}
\newcommand{\east}{\mathtt{E}}
\newcommand{\west}{\mathtt{W}}
\newcommand{\north}{\mathtt{N}}
\newcommand{\internal}{\mathtt{I}}
\newcommand{\tile}{\mathtt{T}}


\begin{document}

\title[The TWB Problem for an Inductive Separation Logic of Relations]
{The Treewidth Boundedness Problem for an Inductive Separation Logic
  of Relations}

\author[M.~Bozga]{Marius Bozga\lmcsorcid{0000-0003-4412-5684}}[a]
\author[L.~Bueri]{Lucas Bueri\lmcsorcid{0000-0002-8589-6955}}[a]
\author[R.~Iosif]{Radu Iosif\lmcsorcid{0000-0003-3204-3294}}[a]
\author[F.~Zuleger]{Florian Zuleger\lmcsorcid{0000-0003-1468-8398}}[b]

\address{Univ. Grenoble Alpes, CNRS, Grenoble INP, VERIMAG, 38000, France}
\address{Institute of Logic and Computation, Technische Universit\"{a}t Wien, Austria}

\begin{abstract}
  The treewidth boundedness problem for a logic asks for the existence
  of an upper bound on the treewidth of the models of a given formula
  in that logic. This problem is found to be undecidable for first
  order logic. We consider a generalization of Separation Logic over
  relational signatures, interpreted over standard relational
  structures, and describe an algorithm that decides the treewidth
  boundedness problem for this logic. Furthermore, our algorithm can
  give an estimate of the bound of the models of a given formula, in
  case there is a finite such bound.
\end{abstract}

\maketitle


\section{Introduction}

The treewidth of a graph is a positive integer measuring, informally
speaking, how far a graph is from a tree. For instance, trees have
treewidth one, series-parallel graphs (i.e., circuits with one input
and one output that can be either cascaded or overlaid) have treewidth
two, whereas $k \times k$ square grids have treewidth $k$, for any $k
\geq 1$. The treewidth parameter is a cornerstone of algorithmic
tractability. For instance, many \np-complete graph problems such as
Hamiltonicity and 3-Colorability become polynomial-time, when
restricted to inputs whose treewidth is bounded by a constant (see,
e.g., \cite[Chapter 11]{DBLP:series/txtcs/FlumG06} for a survey of
classical treewidth-parameterized tractable problems).

Structures are interpretations of relation symbols that define the
standard semantics of first and second order logic
\cite{DBLP:books/daglib/0080654}. They provide a unifying framework
for reasoning about a multitude of graph types e.g., graphs with
multiple edges, labeled graphs, colored graphs, hypergraphs, etc. The
notion of treewidth is straightforwardly generalized from graphs to
structures. In this context, bounding the treewidth by a constant sets
the frontier between the decidability and undecidability of monadic
second order (\mso) logical theories. A result of Courcelle
\cite{CourcelleI} proves that \mso\ is decidable over bounded
treewidth structures, by reduction to the emptiness problem of tree
automata. A dual result of Seese \cite{Seese91} proves that each class
of structures with a decidable \mso\ theory necessarily has bounded
treewidth. Since \mso\ is the yardstick of graph specification logics
\cite{courcelle_engelfriet_2012}, these results show that
\emph{treewidth bounded} classes of structures are tantamount to the
existence of decision procedures for important classes of properties,
in those areas of computing where graphs are relevant such as, e.g.,
static analysis \cite{10.1145/582153.582161}, databases
\cite{AbitebouldBunemanSuciu00} and concurrency
\cite{DBLP:conf/birthday/2008montanari}.

This paper considers the \emph{treewidth boundedness problem}, which
asks for the existence of a bound on the treewidths of the models of a
formula given in input. We show that for first-order logic (and
implicitly \mso) the problem is already undecidable. This negative
result for classical logics motivates our focus on \emph{substructural
logics} that have, in addition to boolean conjunction, a
conjunction-like connective, for which Gentzen's natural deduction
rules of weakening and contraction do not hold. We prove the
decidability of this problem for a generalization of Separation Logic
to relational signatures, interpreted over structures.

Separation Logic (\seplog)
\cite{Ishtiaq00bias,Reynolds02,Cardelli2002Spatial} is a first order
substructural logic with a \emph{separating conjunction} $*$ that
decomposes structures. For reasons related to its applications to the
deductive verification of pointer-manipulating programs, the models of
\seplog\ are finite partial functions, called \emph{heaps}. In
\seplog, the separating conjunction stands for the union of heaps with
disjoint domains.

When combined with \emph{inductive definitions}~\cite{ACZEL1977739},
\seplog\ gives concise descriptions of the recursive data structures
(singly- and doubly-linked lists, trees, etc.) used in imperative
programming (e.g., C, C++, Java, etc.). The shape of these structures
can be described using only existentially quantified separating
conjunctions of (dis-)equalities and points-to atoms. This subset of
\seplog\ is referred to as the \emph{symbolic heap} fragment.

\seplog\ is a powerful tool for reasoning about low-level pointer
updates. It allows to describe actions \emph{locally}, i.e., only with
respect to the resources (e.g., memory cells, network nodes) involved,
while framing out the part of the state that is irrelevant for the
action. This principle of describing mutations, known as \emph{local
reasoning} \cite{CalcagnoOHearnYan07}, is at the heart of scalable
compositional proof techniques for pointer programs
\cite{CalcagnoDistefanoOHearnYang11,CHIN20121006,10.1145/1449764.1449782,10.1007/11804192_6}.

The \emph{Separation Logic of Relations} (\slr) is the generalization
of \seplog\ to relational signatures, interpreted over
structures. This logic has been first considered for relational
databases and object-oriented languages
\cite{10.1007/978-3-540-27864-1_26}. Here the separating conjunction
splits the interpretation of each relation symbol from the signature
into disjoint parts. For instance, the formula $\arel(x_1, \ldots,
x_n)$ describes a structure in which all relations are empty and
$\arel$ consists of a single tuple of values $x_1, \ldots, x_n$,
whereas $\arel(x_1, \ldots, x_n) * \arel(y_1, \ldots, y_n)$ says that
$\arel$ consists of two distinct tuples, i.e., the values of $x_i$ and
$y_i$ differ for at least one index $1 \leq i \leq n$.  Moreover, when
encoding graphs by structures, \slr\ allows to specify edges that have
no connected vertices, isolated vertices, or both. The same style of
composition is found in other spatial logics interpreted over graphs,
such as the \gl\ logic of Cardelli et al \cite{Cardelli2002Spatial}.

Our motivation for studying the models of \slr\ arose from recent work
on deductive verification of self-adapting distributed systems, where
Hoare-style local reasoning is applied to write correctness proofs for
systems with dynamically reconfigurable network architectures
\cite{AhrensBozgaIosifKatoen21,DBLP:conf/cade/BozgaBI22,DBLP:conf/concur/BozgaBI22}.
The assertion language of these proofs is \slr, with unary relation
symbols used to model nodes (processes) of the network and relation
symbols of arity two or more used to model links (communication
channels) between nodes. Just as user-defined inductive predicates are
used in \seplog\ to describe data structures (lists, trees, etc.),
\slr\ inductive predicates are used to describe common architectural
styles (e.g., pipelines, rings, stars, etc.) that ensure correct and
optimal behavior of many distributed applications.

The decidability result from this paper defines the class of
\slr\ formul{\ae} whose models are treewidth bounded and provides a
reasonable estimate on the bound, in case one exists. On one hand,
this algorithm answers the question \emph{does the set of structures
  defined by a given system of inductive definitions have a decidable
  \mso\ theory?} If this is the case, problems such as, e.g.,
Hamiltonicity, $k$-Colorability, Planarity, etc. are decidable on this
set of structures. Another application is the decidability of the
\emph{entailment problem}
$\sidsem{\phi}{\asid}\subseteq\sidsem{\psi}{\asid}$ asking if each
model of a formula $\phi$ is also a model of another formula $\psi$,
when the predicate symbols in $\phi$ and $\psi$ are interpreted by a
set of inductive definitions $\asid$. In principle, the decidability
of this problem depends on \begin{enumerate*}[(i)]
\item $\phi$ having only treewidth bounded models, for a computable
  upper bound, and
\item both $\phi$ and $\psi$ being \mso-definable
  \cite{DBLP:conf/cade/IosifRS13}.
\end{enumerate*}
The algorithm described in the paper provides a key ingredient for
defining fragments of \slr\ with a decidable entailment problem, which
is tantamount to automating proof generation in Hoare
logic~\cite{DBLP:conf/cade/BozgaBI22,DBLP:conf/concur/BozgaBI22}.

\subsection{Related work}
One of the first fragments of \seplog\ with a decidable entailment
problem relied on an ad-hoc translation into equivalent
\mso\ formul{\ae}, together with a static guarantee of treewidth
boundedness, called
\emph{establishment}~\cite{DBLP:conf/cade/IosifRS13}. More recently,
the entailment problem in this fragment of \seplog\ has been the focus
of an impressive body of
work~\cite{DBLP:conf/concur/CookHOPW11,DBLP:conf/lpar/KatelaanZ20,%
  EchenimIosifPeltier21b,EchenimIosifPeltier21,DBLP:conf/fossacs/LeL23}. In
particular, the \emph{establishment problem} ``is a given set of
inductive definitions established'' has been found to be
\conp-complete~\cite{DBLP:conf/esop/JansenKMNZ17}. Moreover, lifting
the establishment condition leads to the undecidability of
entailments, as showed in~\cite{DBLP:journals/ipl/EchenimIP22}. The
establishment problem can, in fact, be seen as the precursor of the
treewidth boundedness problem studied in the present paper.

The treewidth parameter showed also in a recent comparison between the
expressivity of \slr\ with inductive definitions and that of
\mso~\cite{Concur23}. When restricting the interpretation of the
logics to treewidth bounded graphs, \slr\ strictly subsumes \mso,
i.e., for each \mso\ formula $\phi$ and integer $k\geq1$, there exists
a formula $\psi$ of \slr\ that defines the models of $\phi$ of
treewidth at most $k$. Moreover, the logics are incomparable for
classes of graphs of unbounded treewidth.

\subsection{Motivating examples}
We introduce the reader to \slr\ and the treewidth boundedness problem
by means of examples. \autoref{fig:examples}~(a) shows a chain,
defined by an unfolding of the inductive predicate
$\apred(x_1,x_2)$. The chain starts at $x_1$ and ends at $x_2$. The
elements of the chain are labeled by a unary relation symbol
$\mathsf{a}$ and the neighbours are linked by a binary relation
$\arel$. Each unfolding of the inductive definition
$\apred(x_1,x_2)\leftarrow \exists y ~.~ \mathsf{a}(x_1) *
\arel(x_1,y) * \apred(y,x_2)$ instantiates the existential quantifier
to an element distinct from the existing ones. This is because every
instantiation of an existential quantifier is placed into a set
labeled by $\mathsf{a}$ and the semantics of the separating
conjunction requires that these sets must be disjoint in each
decomposition of a model of $\mathsf{a}(x_1) * \arel(x_1,y) *
\apred(y,x_2)$ into models of $\mathsf{a}(x_1) * \arel(x_1,y)$ and
$\apred(y,x_2)$. Then, each model of $\exists x_1 \exists x_2 ~.~
\apred(x_1,x_2)$ is a possibly cyclic chain, because nothing is
enforced on $x_2$, which can be mapped back to a previous
instantiation of $y$. Hence each model of this sentence has treewidth
two at most.

\autoref{fig:examples}~(b) shows a family of models for a slightly
modified definition of the chain from \autoref{fig:examples}~(a),
given by the recursive rule $\apred(x_1,x_2)\leftarrow \exists y ~.~
\arel(x_1,y) * \apred(y,x_2)$, where the instantiations of the
existential quantifiers are not placed into any particular set. In
this case, one can fold a sufficiently large chain onto itself and
creating a square grid, by using the same element of the structure
more than once to instantiate a quantifier. Then, the sentence
$\exists x_1 \exists x_2 ~.~ \apred(x_1,x_2)$ has an infinite set of
models containing larger and larger square grid minors, thus having
unbounded treewidth.

Since placing every quantifier instance into the same set guarantees
treewidth boundedness, as in, e.g., \autoref{fig:examples}~(a), a
natural question that arises is what happens when these instances are
placed into two (not necessarily disjoint) sets? The inductive
definition of the predicate $\apred$ in \autoref{fig:examples}~(c)
creates an unbounded number of disconnected $\arel$-edges whose
endpoints are arbitrarily labeled with $\mathsf{a}$ and $\mathsf{b}$,
respectively. In this case, one can instantiate a $\mathsf{a}$-labeled
(resp. $\mathsf{b}$-labeled) variable with a new element or a previous
$\mathsf{b}$ (resp. $\mathsf{a}$) element and build chains (or sets of
disconnected chains), of treewidth two at most two. Again, this is
because a simple cycle with more than two elements has treewidth two.

  \begin{figure}[htbp]
    \begin{center}
    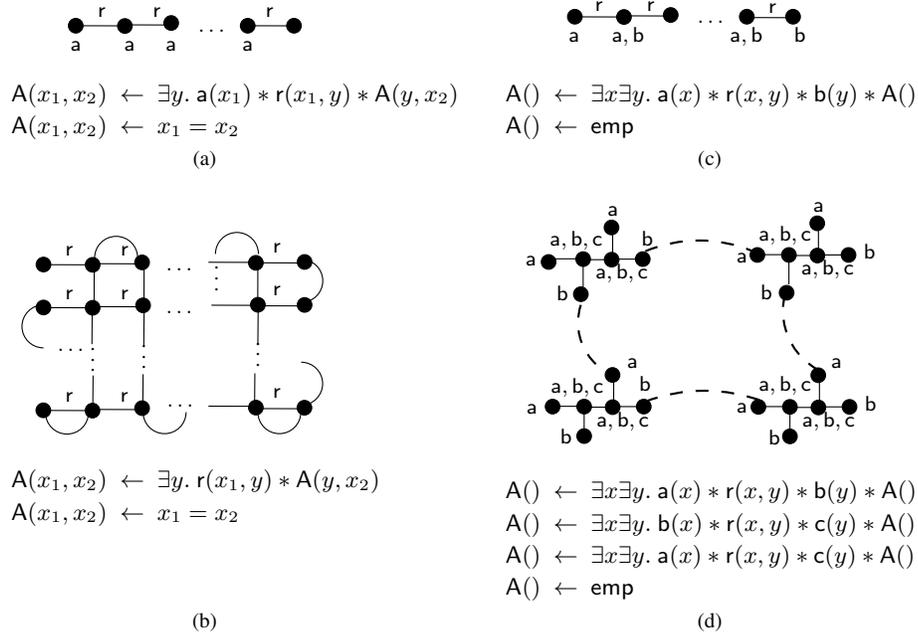
    \caption{\label{fig:examples} Examples of bounded and unbounded
  treewidth models}
    \end{center}
  \end{figure}

Let us now consider three unary relation symbols $\mathsf{a}$,
$\mathsf{b}$ and $\mathsf{c}$ and three types of disconnected
$\arel$-edges (according to the labels of their endpoints) created by
three recursive definitions of \autoref{fig:examples} (d), namely
$\mathsf{a}$-$\mathsf{b}$, $\mathsf{b}$-$\mathsf{c}$ and
$\mathsf{a}$-$\mathsf{c}$ edges. In this case, the sentence
$\apred()$, where $\apred$ is a predicate symbol of zero arity, has
models with unboundedly large square grid minors, obtained by
``glueing'' these edges (i.e., instantiating several quantifiers with
the same element from different sets). The glued pairs are connected
with dotted lines in \autoref{fig:examples} (d). Hence, these
structures form a set of unbounded treewidth.

These examples highlight the main ideas behind an algorithm that
decides the existence of a bound on the treewidths of the models of a
given formula, with predicates interpreted by set of inductive
definitions. First, one needs to identify the definitions that can
iterate any number of times producing building blocks of unboundedly
large grids (modulo edge contractions). Second, these structures must
connect elements from different sets, e.g., $\mathsf{a}$, $\mathsf{b}$
or $\mathsf{c}$ in \autoref{fig:examples}. A complication is that
these sets can be defined not only by monadic relation symbols, but
also by $n$-ary relation atoms where all but one variable have the
same values for any occurrence. For instance, the variable $x_2$ in
\autoref{fig:examples} (a) has the same value in an arbitrarily long
unfolding of $\apred(x_1,x_2)$ and we could have written
$\arel(x_1,x_2)$ instead of $\mathsf{a}(x_1)$ in the first rule, with
the same effect, while avoid using `$\mathsf{a}$' altogether. Last,
the interplay between the connectivity and labeling of the building
blocks is important. For instance, in \autoref{fig:examples}~(d), the
building blocks of the grid are structures consisting of six elements,
that connect three `$\mathsf{a}$' with three `$\mathsf{b}$' elements.
In fact, as we shall prove, a necessary and sufficient condition for
treewidth-boundedness is that any two such ``iterable'' substructures
that connect at least three elements must also place these elements in
at least one common set (e.g., $\mathsf{a}$, $\mathsf{b}$ or
$\mathsf{c}$ in our example).


\section{The Separation Logic of Relations}
\label{sec:preliminaries}

This section defines formally the Separation Logic of Relations (\slr)
and its corresponding treewidth boundedness problem. It also
introduces most of the technical notions used throughout the paper.

Let $\nat$ be the set of positive integers, zero included and $\nat_+
\isdef \nat \setminus\set{0}$. Given integers $i$ and $j$, we write
$\interv{i}{j}$ for the set $\set{i,i+1,\ldots,j}$, assumed to be
empty if $i>j$. For a set $A$, we denote by $\pow{A}$ its
powerset. The cardinality of a finite set $A$ is $\cardof{A}$. By
writing $S = S_1 \uplus S_2$, we mean that $S_1$ and $S_2$ partition
$S$, i.e., that $S = S_1 \cup S_2$ and $S_1 \cap S_2 = \emptyset$.

Multisets are denoted as $\mset{a,b,\ldots}$ and all set operations
(union, intersection, etc.)  are used with multisets as well. In
particular, a binary operation involving a set and a multiset
implicitly lifts the set to a multiset and returns a multiset. The
multi-powerset (i.e., the set of multisets) of $A$ is denoted as
$\mpow{A}$.

For a binary relation $R \subseteq A \times A$, we denote by $R^*$ its
reflexive and transitive closure and by $R^=$ the smallest equivalence
relation that contains $R$, i.e., the closure of $R^*$ by
symmetry. For a set $S \subseteq A$, we denote by $\proj{R}{S}$ the
relation obtained by removing from $R$ all pairs with an element not
in $S$. A binary relation $R \subseteq A \times B$ is an \emph{$A$-$B$
matching} iff $\set{a,b} \cap \set{a',b'} = \emptyset$, for all
distinct pairs $(a,b),(a',b') \in R$.

\subsection{Structures}
Let $\relations$ be a finite and fixed set of \emph{relation symbols},
of arities $\arityof{r}\geq1$, for all $r \in \relations$. A relation
symbol of arity one (resp. two) is called \emph{unary}
(resp. \emph{binary}).

A \emph{structure} is a pair $\astruc=(\univ,\struc)$, where $\univ$
is an \emph{infinite} set called the \emph{universe} and $\struc :
\relations \rightarrow \pow{\univ^+}$ is an \emph{interpretation}
mapping each relation symbol $\arel$ into a \emph{finite} subset of
$\univ^{\arityof{\arel}}$. We consider only structures with finite
interpretations, because the logic under consideration (defined below)
can only describe sets of finite structures. The \emph{support}
$\supp{\struc}\isdef\{u_i \mid \tuple{u_1,\ldots,u_{\arityof{\arel}}}
\in \struc(\arel),~ i \in \interv{1}{\arityof{\arel}}\}$ of an
interpretation is the (necessarily finite) set of elements that occur
in a tuple from the interpretation of a relation symbol. The support
of a structure is the support of its interpretation.

Two structures $(\univ_1,\struc_1)$ and $(\univ_2,\struc_2)$ are
\emph{locally disjoint} iff $\struc_1(\arel) \cap \struc_2(\arel) =
\emptyset$, for all $\arel\in\relations$ and \emph{disjoint} iff
$\supp{\struc_1} \cap \supp{\struc_2} = \emptyset$. Two structures are
\emph{isomorphic} iff they differ only by a renaming of their elements
(see, e.g., \cite[Section A3]{DBLP:books/daglib/0082516} for a formal
definition of isomorphism between structures).

We consider the \emph{composition} as a partial binary operation
between structures, defined as pointwise disjoint union of the
interpretations of relation symbols:

\begin{defi}
  The composition of two locally disjoint structures $(U_1,\struc_1)$
  and $(U_2,\struc_2)$ is $(\univ_1,\struc_1) \comp (\univ_2,\struc_2)
  \isdef (\univ_1\cup\univ_2,\struc_1\uplus\struc_2)$, where
  $(\struc_1\uplus\struc_2)(\arel)\isdef\struc_1(\arel)\uplus\struc_2(\arel)$,
  for all $\arel\in\relations$. The composition is undefined if
  $(U_1,\struc_1)$ and $(U_2,\struc_2)$ are not locally disjoint.
\end{defi}

For example, \autoref{fig:composition} shows the composition of two
structures $\astruc_1$ and $\astruc_2$, whose interpretations are
represented as hyper-graphs with edges denoting tuples from the
interpretation of relation symbols $\mathsf{a}$, $\mathsf{b}$ and
$\mathsf{c}$, of arities $3$, $2$ and $2$, respectively. Note that
$\astruc_1$ and $\astruc_2$ are locally disjoint but not disjoint, for
instance the elements $u_2$ and $u_3$ are present in the support of
both structures.

  \begin{figure}[htbp]
    \begin{center}
    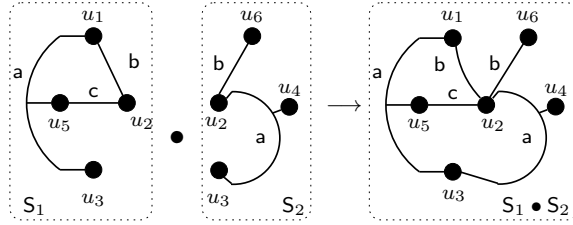
    \caption{\label{fig:composition} Composition of structures}
    \end{center}
  \end{figure}

\subsection{Treewdith}
A graph is a pair $\graph = (\nodes,\edges)$, such that $\nodes$ is a
finite set of \emph{nodes} and $\edges \subseteq \nodes \times \nodes$
is a set of \emph{edges}. A (simple) \emph{path} in $\graph$ is a
sequence of (pairwise distinct) nodes $v_1, \dots, v_n$, such that
$(v_i,v_{i+1}) \in \edges$, for all $i \in \interv{1}{n-1}$. We say
that $v_1, \dots, v_n$ is an \emph{undirected path} if
$\set{(v_i,v_{i+1}),(v_{i+1},v_i)} \cap \edges \neq \emptyset$
instead, for all $i \in \interv{1}{n-1}$. A set of nodes $S \subseteq
\nodes$ is \emph{connected in $\graph$} iff between any two nodes in
$S$ there is an undirected path in $\graph$ that involves only nodes
from $S$. A graph $\graph$ is \emph{connected} iff $\nodes$ is
connected in $\graph$.

Given a set $\labels$ of labels, a \emph{$\labels$-labeled unranked
tree} is a tuple $\tree = (\nodes,\edges,r,\alabel)$, where
$(\nodes,\edges)$ is a graph, $r \in \nodes$ is a designated node
called the \emph{root}, such that there exists a unique simple path
from $r$ to any other node $n \in \nodes\setminus\set{r}$ and no path
from $r$ to $r$ in $(\nodes,\edges)$. The mapping $\alabel : \nodes
\rightarrow \labels$ associates each node of the tree a label from
$\labels$.

\begin{defi}\label{def:treewidth}
  A \emph{tree decomposition} of a structure $\astruc=(\univ,\struc)$
  is a $\pow{\univ}$-labeled unranked tree
  $\tree=(\nodes,\edges,r,\alabel)$, such that the following
  hold: \begin{enumerate}
  \item\label{it1:treewidth} for each relation symbol $\arel \in
    \relations$ and each tuple $\tuple{u_1, \ldots,
      u_{\arityof{\arel}}} \in \struc(\arel)$ there exists a node $n
    \in \nodes$, such that $\set{u_1, \ldots, u_{\arityof{\arel}}}
    \subseteq \alabel(n)$,
  \item\label{it2:treewidth} for each element $u \in \supp{\struc}$,
    the set of nodes $\set{n \in \nodes \mid u \in \alabel(n)}$ is
    nonempty and connected in $(\nodes,\edges)$.
  \end{enumerate}
  The \emph{width} of the tree decomposition is $\width{\tree} \isdef
  \max_{n \in \nodes} \cardof{\alabel(n)}-1$. The \emph{treewidth} of
  the structure $\struc$ is $\twof{\struc} \isdef \min
  \set{\width{\tree} \mid \tree \text{ is a tree decomposition of }
    \struc}$.
\end{defi}

Note that, since we consider only structures with finite support, tree
decompositions are finite trees with finite sets as labels, hence the
treewidth of a structure is a well-defined integer. A set of
structures is \emph{treewidth-bounded} iff the set of corresponding
treewidths is finite and \emph{treewidth-unbounded} otherwise.  We
assume basic acquaintance with the notions of grid and minor. It is
known that a set of structures having infinitely many minors
isomorphic to some $n \times n$ grid is treewidth-unbounded
\cite{DBLP:journals/tcs/Bodlaender98}.

\subsection{Separation Logic of Relations}
The \emph{Separation Logic of Relations} (\slr) uses a set of
\emph{variables} $\vars = \set{x,y,\ldots}$ and a set of
\emph{predicates} $\preds = \set{\apred, \bpred, \ldots}$ with given
arities $\arityof{\apred}\geq0$. A predicate of zero arity is called
\emph{nullary}.

The formul{\ae} of \slr\ are defined by the syntax in
\autoref{fig:slr}~(a). A variable is \emph{free} if it does not occur
within the scope of an existential quantifier and $\fv{\phi}$ denotes
the set of free variables of $\phi$. A \emph{sentence} is a formula
with no free variables. For a formula $\phi$, we denote by
$\exclof{\phi}$ the sentence obtained by existentially quantifying its
free variables.  A formula without quantifiers is called
\emph{quantifier-free}.

Instead of the standard boolean conjunction, \slr\ has a
\emph{separating conjunction} $*$. The formul{\ae} $x\neq y$ and
$\apred(x_1,\ldots,x_{\arityof{\apred}})$ are called
\emph{disequalities} and \emph{predicate atoms}, respectively. To
alleviate notation, we denote by $\apred$ the predicate atom
$\apred()$, whenever $\apred$ is nullary. A formula without predicate
atoms is called \emph{predicate-free}. A \emph{qpf} formula is both
quantifier- and predicate-free.

\begin{figure}[htbp]
  \begin{center}
    \[\phi := \emp \mid x=y \mid x\neq y \mid \arel(x_1, \ldots,
    x_{\arityof{\arel}}) \mid \apred(x_1, \ldots,
    x_{\arityof{\apred}}) \mid \phi * \phi \mid \exists x ~.~ \phi\]
    (a) %
    \[\begin{array}{rclcl}
    (\univ,\struc) & \models^\store_\asid & \emp & \iffdef &
    \struc(\arel) = \emptyset \text{, for all } \arel\in\relations \\
    (\univ,\struc) & \models^\store_\asid & x \sim y & \iffdef &
    (\univ,\struc) \models^\store_\asid \emp \text{ and } \store(x)
    \sim \store(y) \text{, for } \sim \ \in\!\set{=,\neq} \\
    (\univ,\struc) & \models^\store_\asid & \arel(x_1, \ldots, x_{k})
    & \iffdef & \struc(\arel) = \set{\tuple{\store(x_1), \ldots,
        \store(x_{k})}} \text{ and } \struc(\arel') = \emptyset
    \text{, for } \arel' \in \relations \setminus\set{\arel} \\
    (\univ,\struc) & \models^\store_\asid & \apred(y_1, \ldots, y_{n})
    & \iffdef & \struc \models^\store_\asid \phi[x_1/y_1, \ldots,
      x_{n}/y_{n}] \text{, for some } \apred(x_1, \ldots, x_{n})
    \leftarrow \phi \in \asid \\
    (\univ,\struc) & \models^\store_\asid & \phi_1 * \phi_2 & \iffdef
    & \text{exist structures } (\univ_i,\struc_i) \text{, where }
    (\univ,\struc) = (\univ_1,\struc_1) \comp (\univ_2,\struc_2) \\
    &&&&\text{and } (\univ_i,\struc_i) \models^\store_\asid \phi_i
    \text{, for both } i = 1,2 \\
    (\univ,\struc) & \models^\store_\asid & \exists x ~.~ \phi &
    \iffdef & \struc \models^{\store[x\leftarrow u]}_\asid \phi
    \text{, for some } u \in \univ
    \end{array}\] (b)
  \end{center}
  \caption{\label{fig:slr} The syntax (a) and semantics (b) of the
    Separation Logic of Relations}
\end{figure}

\begin{defi}\label{def:sid}
  A \emph{set of inductive definitions (SID)} is a \emph{finite} set $\asid$
  of \emph{rules} of the form $\apred(x_1, \ldots,
  x_{\arityof{\apred}}) \leftarrow \phi$, where $x_1, \ldots,
  x_{\arityof{\apred}}$ are pairwise distinct variables, called
  \emph{parameters}, such that $\fv{\phi} \subseteq \set{x_1, \ldots,
    x_{\arityof{\apred}}}$.  
\end{defi}

The semantics of \slr\ is given by the satisfaction relation
$(\univ,\struc) \models^\store_\asid \phi$ between structures and
formul{\ae}, parameterized by a store $\store$ and a SID $\asid$.  We
write $\store[x \leftarrow u]$ for the store that maps $x$ into $u$
and agrees with $\store$ on all variables other than $x$.  By
$[x_1/y_1, \ldots, x_n/y_n]$ we denote the substitution that replaces
each free variable $x_i$ by $y_i$ in a formula $\phi$. The result of
applying the substitution $[x_1/y_1, \ldots, x_n/y_n]$ to the formula
$\phi$ is denoted as $\phi[x_1/y_1, \ldots, x_n/y_n]$, where, by
convention, the existentially quantified variables from $\phi$ are
renamed to avoid clashes with $y_1, \ldots, y_n$. Then
$\models^\store_\asid$ is the least relation that satisfies the
constraints in \autoref{fig:slr}~(b).

Note that the interpretation of equalities and relation atoms differs
in \slr\ from first-order logic, namely $x = y$ requires that the
structure is empty and $\arel(x_1,\ldots,x_{\arityof{\arel}})$ denotes
the structure in which all relations symbols are interpreted by empty
sets, except for $\arel$, which contains the tuple of store values of
$x_1,\ldots,x_{\arityof{\arel}}$ only. Moreover, every structure
$(\univ,\struc)$, such that $(\univ,\struc) \models^\store_\asid
\phi$, interprets each relation symbol as a finite set of tuples,
defined by a finite least fixpoint iteration over the rules from
$\asid$. The assumption that each structure has an infinite universe
excludes the cases in which a formula becomes unsatisfiable because
there are not enough elements to instantiate the quantifiers
introduced by the unfolding of the rules, thus simplifying the
definitions.

If $\phi$ is a sentence (resp. a predicate-free formula), we omit the
store $\store$ (resp. the SID $\asid$) from $\astruc
\models^\store_\asid \phi$. For a \slr\ sentence $\phi$, let
$\sidsem{\phi}{\asid}\isdef\set{\astruc\mid\astruc\models_\asid\phi}$
be the set of \emph{$\asid$-models} of $\phi$. If $\phi$ is, moreover,
predicate-free we say that $\phi$ is \emph{satisfiable} iff
$\sem{\phi}\neq\emptyset$.

For a qpf formula $\phi$, we write $x \eqof{\phi} y$ (resp. $x
\not\eqof{\phi} y$) iff $x=y$ is (resp. is not) a logical consequence
of $\phi$, i.e., $\store(x)=\store(y)$ for each store $\store$ and
structure $\astruc$, such that $\astruc \models^\store \phi$. Note
that $x \not\eqof{\phi} y$ is the negation of $x \eqof{\phi} y$, which
is different from that $x \neq y$ is implied by $\phi$. We define
several quantitative measures relative to SIDs: 

\begin{defi}\label{def:sid-measures}
Let $\asid$ be a SID. We denote by:
\begin{itemize}[label=$\triangleright$]
\item $\maxvarinruleof{\asid}$ the maximum number of variables that
  occur, either free or existentially quantified, in a rule from
  $\asid$,
\item $\maxrelatominruleof{\asid}$ the maximum number of relation
  atoms that occur in a rule from $\asid$,
\item $\maxrulearityof{\asid}$ the maximum number of predicates that
  occur in a rule from $\asid$,
\item $\maxrelarityof{\asid}$ the maximum arity of relation symbols
  occurring in $\asid$,
\item $\predsno{\asid}$ the number of predicate symbols occurring in
  $\asid$,
\item $\relationsno{\asid}$ the number of relation symbols occurring
  in $\asid$.
\end{itemize}
\end{defi}

\subsection{Simplifying assumptions}
In the rest of this paper, we simplify the technical development by
two assumptions, that lose no generality. The first assumption is that
no equalities occur in the given SID (\autoref{lemma:eq-free}).

\begin{defi}\label{def:eq-free}
  A formula is \emph{equality-free} iff it contains no equalities nor
  predicate atoms in which the same variable occurs twice. A rule
  $\apred(x_1,\ldots,x_n) \leftarrow \phi$ is equality-free iff $\phi$
  is equality-free. A SID is equality-free iff it consists of
  equality-free rules.
\end{defi}

\begin{lem}\label{lemma:eq-free}
  Given a SID $\asid$, one can build an equality-free SID $\asid'$,
  such that $\sidsem{\apred}{\asid} = \sidsem{\apred}{\asid'}$, for
  each nullary predicate $\apred$. Moreover, all quantitative measures
  (\autoref{def:sid-measures}) of $\asid'$ are the same as for
  $\asid$, except for $\predsno{\asid'} \leq \predsno{\asid} \cdot
  {\maxpredarityof{\asid}}^{\maxpredarityof{\asid}}$. 
\end{lem}
\begin{proof}
  See \cite[Lemma 9]{Concur23}. The construction of $\asid'$ considers
  predicates $\apred_{I_1, \ldots, I_n}$, where $\apred$ is a
  predicate symbol that occurs in $\asid$ and $I_1 \uplus \ldots
  \uplus I_n = \interv{1}{\arityof{\apred}}$ is a partition. Since the
  number of partitions of $\interv{1}{\arityof{\apred}}$ is
  asymptotically bounded by $\arityof{\apred}^{\arityof{\apred}} \leq
  \maxpredarityof{\asid}^{\maxpredarityof{\asid}}$, we obtain the bound
  on $\predsno{\asid'}$.
\end{proof}

The following notion of \emph{unfolding} is used to define several
technical notions and state the second simplifying assumption. Let
$\phi$ and $\psi$ be formul{\ae} and $\asid$ be a SID. We denote by
$\phi \step{\asid} \psi$ the fact that $\psi$ is obtained by replacing
a predicate atom $\apred(y_1,\ldots,y_n)$ in $\phi$ by a formula
$\rho[x_1/y_1,\ldots,x_n/y_n]$, where
$\apred(x_1,\ldots,x_n)\leftarrow\rho$ is a rule from $\asid$. A
\emph{$\asid$-unfolding} is a sequence of formul{\ae} $\phi_1
\step{\asid} \ldots \step{\asid} \phi_n$. The $\asid$-unfolding is
\emph{complete} if the last formula is predicate-free. The following
statement is a direct consequence of the semantics of \slr:

\begin{prop}\label{prop:unfolding}
  Let $\phi$ be a sentence, $\asid$ a SID and $\astruc$ a structure.
  Then $\astruc\in\sidsem{\phi}{\asid}$ iff $\astruc \models^\store
  \psi$, for a store $\store$ and complete $\asid$-unfolding
  $\phi\step{\asid}^*\exists x_1 \ldots \exists x_n ~.~ \psi$, where
  $\psi$ is a qpf formula.
\end{prop}
\begin{proof}
  ``$\Leftarrow$'' By induction on the definition of the satisfaction
  relation $\models^\store_\asid$. ``$\Rightarrow$'' By induction on
  the length of the $\asid$-unfolding.
\end{proof}

The second assumption is that any $\asid$-unfolding of a nullary
predicate by the given SID yields a predicate-free formula that is
satisfiable. Again, this assumption loses no generality
(\autoref{lemma:all-sat}).

\begin{defi}\label{def:all-sat}
  A SID $\asid$ is \emph{all-satisfiable} for a nullary predicate
  $\apred$ iff each predicate-free formula $\phi$ which is the outcome
  of a complete $\asid$-unfolding $\apred \step{\asid}^* \phi$ is
  satisfiable.
\end{defi}

\begin{lem}\label{lemma:all-sat}
  Given a SID $\asid$ and a nullary predicate $\apred$, one can build
  a SID $\overline{\asid}$ all-satisfiable for $\apred$, such that
  $\sidsem{\apred}{\asid}=\sidsem{\apred}{\overline{\asid}}$. Moreover,
  all quantitative measures (\autoref{def:sid-measures}) of $\asid'$
  are the same as for $\asid$, except for $\predsno{\overline{\asid}}
  \leq \predsno{\asid} \cdot \relationsno{\asid} \cdot
       {\maxpredarityof{\asid}}^{\maxrelarityof{\asid}}$. 
\end{lem}
\begin{proof}
  For space reasons, this proof is given in \autoref{app:all-sat}.
\end{proof}

\subsection{The treewidth boundedness problem}
We are ready to state the main problem addressed in this
paper. Before, we state two technical lemmas that state several
relations between qpf formul{\ae} and the upper bounds on the
treewidth of their models:

\begin{lem}\label{lemma:qpf-treewidth}
  Let $\phi$, $\psi$ be qpf formul{\ae}, $x_0,x_1,x_2,\ldots,x_k$
  variables and $\arel$ a relation symbol of arity $k$. Then, the
  following hold: \begin{enumerate}
  \item\label{it1:qpf-treewidth} $\twof{\sem{\exclof{(\phi *
        \Bigstar_{i=1}^k x_0 = x_i)}}} \le
    \twof{\sem{\exclof{\phi}}}$,
  \item\label{it2:qpf-treewidth} $\twof{\sem{\exclof{\phi}}} - 1 \le
    \twof{\sem{\exclof{(\phi * \Bigstar_{i=1}^k x_0 \neq x_i)}}} \le
    \twof{\sem{\exclof{\phi}}}$ if $\phi * \Bigstar_{i=1}^k x_0 \neq
    x_i$ satisfiable,
  \item\label{it3:qpf-treewidth} $\twof{\sem{\exclof{\phi}}} - 1 \le
    \twof{\sem{\exclof{(\phi * \arel(x_1,\ldots,x_k))}}} \le
    \twof{\sem{\exclof{\phi}}} + k$ if $\phi * \arel(x_1,\ldots,x_k)$
    satisfiable.
  \item\label{it4:qpf-treewidth} $\twof{\sem{\exclof{(\phi * \psi)}}}
    \le \twof{\sem{\exclof{\phi}}} + \cardof{\fv{\psi}}$ if $\psi$
    contains only relation atoms.
  \end{enumerate}
\end{lem}
\begin{proof}
  For space reasons, this proof is given in
  \autoref{app:qpf-treewidth}.
\end{proof}

\begin{lem}\label{lemma:sep-treewidth}
  Let $\phi$ and $\psi$ be qpf formul{\ae} and $F \isdef \fv{\phi}
  \cap \fv{\psi}$, such that $\phi * \psi$ is satisfiable and $x
  \not\eqof{\phi} y$, for all $x,y \in F$. Let $\eta \isdef
  \Bigstar_{x,y\in F,~ x \eqof{\psi} y} x = y$. Then,
  $\twof{\sem{\exclof{(\phi * \eta)}}} \le \twof{\sem{\exclof{(\phi *
        \psi)}}} + \cardof{F}$.
\end{lem}
\begin{proof}
  For space reasons, this proof is given in
  \autoref{app:sep-treewidth}.
\end{proof}

The main result of this paper is a decidability proof for the
following decision problem:

\begin{defi}\label{def:btw}
  The \tbsl\ problem asks whether the set $\sidsem{\phi}{\asid}$ is
  treewidth-bounded, for an SID $\asid$ and \slr\ sentence $\phi$
  given as input.
\end{defi}

This result is tightened by a proof of the undecidability of the
treewidth-boundedness problem for first-order logic
(\autoref{sec:undecidability}), that further improves our
understanding of the relation between the expressivity of classical
and substructural logics. The decidability proof proceeds in two
steps. First, we show the decidability of the problem for sentences of
the form $\apred$, where $\apred$ is a nullary predicate symbol, and a
class of SIDs having a particular property, called
\emph{expandability}, defined below
(\autoref{sec:expandable-tb}). Second, we show how to reduce the
treewidth-boundedness problem for arbitrary sentences and SIDs to the
problem for nullary predicate atoms and expandable SIDs
(\autoref{sec:general-tb}).


\section{Expandable Sets of Inductive Definitions}
\label{sec:expandable-tb}

This section introduces the formal definitions of canonical models and
expandable SIDs, needed for the first part of the proof of
decidability of the \tbsl\ problem. The main result of this section is
that the treewidth boundedness problem is decidable for the sets of
models of a nullary predicate defined by an expandable SID.

For simplicity, in the rest of this paper we shall represent sentences
$\phi$ by nullary predicate atoms $\apred$. This loses no generality
since $\sidsem{\phi}{\asid} = \sidsem{\apred}{\asid \cup \set{\apred
    \leftarrow \phi}}$ provided that $\apred$ is not defined by any
other rule in $\asid$.  In the rest of this section we fix an
arbitrary SID $\asid$ and nullary predicate $\apred$.

\subsection{Canonical models}

Intuitively, a $\asid$-model of $\apred$ is \emph{canonical} if it can
be defined using a store that matches only those variables that are
equated in the outcome of the complete $\asid$-unfolding of $\apred$
that ``produced'' the model, in the sense of
\autoref{prop:unfolding}. A \emph{rich canonical model}
records, moreover, the disequalities introduced during the unfolding.

\begin{defi}\label{def:canonical-model}
  A store $\store$ is \emph{canonical for $\phi$} iff
  $\store(x)=\store(y)$ only if $x \eqof{\phi} y$, for all
  $x,y\in\fv{\phi}$. A \emph{rich canonical $\asid$-model} of a
  sentence $\phi$ is a pair $(\astruc,\diseq)$, where
  $\astruc=(\univ,\struc)$ is a structure and $\diseq \subseteq
  \univ\times\univ$ is a symmetric relation, such that there exists a
  complete $\asid$-unfolding $\phi \step{\asid}^* \exists x_1 \ldots
  \exists x_n ~.~ \psi$, where $\psi$ is qpf, and a store $\store$
  canonical for $\psi$, such that $\astruc \models^\store \psi$ and
  $\diseq(u,v)$ iff there exist variables $x \in \store^{-1}(u)$, $y
  \in \store^{-1}(v)$ and the disequality $x \neq y$ occurs in
  $\psi$. We denote by $\rcsem{\phi}{\asid}$ the set of rich canonical
  $\asid$-models of $\phi$ and $\csem{\phi}{\asid}\isdef\set{\astruc
    \mid (\astruc,\diseq)\in\rcsem{\phi}{\asid}}$ the set of
  \emph{canonical} $\asid$-models of $\phi$. If $\phi$ is
  predicate-free, we write $\csem{\phi}{}$ (resp. $\rcsem{\phi}{}$)
  instead of $\csem{\phi}{\asid}$ (resp. $\rcsem{\phi}{\asid}$).
\end{defi}

A store $\store$ is \emph{injective} over a set of variables $x_1,
\ldots, x_n$ iff $\store(x_i)=\store(x_j)$ implies $i=j$, for all
$i,j\in\interv{1}{n}$. Note that the canonical $\asid$-models of an
equality free SID $\asid$ can be defined considering injective,
instead of canonical stores. Nevertheless, this more general
definition of canonical models using canonical stores will become
useful later on, when predicate-free formul{\ae} with equalities will
be considered.

Canonical models are important for two reasons. First, their treewidth
is bounded:

\begin{lem}\label{lemma:canonical-btw}
  $\twof{\astruc} \le \maxvarinruleof{\asid}-1$, for any
  $\astruc \in \csem{\apred}{\asid}$.
\end{lem}
\begin{proof}
  Let $\astruc = (\univ,\struc) \in \csem{\apred}{\asid}$ be a
  canonical $\asid$-model of $\apred$. We define a tree decomposition
  $\tree=(\nodes,\edges,r,\alabel)$ of $\astruc$ as follows. The graph
  of $\tree$ is any derivation tree of $\asid$ whose outcome is
  $\astruc$. This is a tree labeled with rules from $\asid$, whose
  parent-child relation is defined as follows: if $n$ is a node
  labeled with a rule $\arule$, for each predicate atom
  $\bpred(z_1,\ldots,z_{\arityof{\bpred}})$ that occurs in $\arule$,
  there is exactly one child $m$ of $n$ whose label is a rule that
  defines $\bpred$. The bag $\alabel(n)$ contains exactly those
  elements that are the store values of the variables occurring free
  or bound in $\arule$. We check that $\tree$ is a tree decomposition
  of $\astruc$ by proving the two points of \autoref{def:treewidth}:
  \begin{itemize}[left=.5\parindent]
  \item[{(\ref{it1:treewidth})}] each tuple
  $\tuple{u_1, \ldots, u_{\arityof{\arel}}} \in \struc(\arel)$ occurs
  in $\astruc$ because of a relation atom
  $\arel(z_1,\ldots,z_{\arityof{\arel}})$ that occurs in the label of
  a node $n$ from the parse tree. Then $u_1, \ldots,
  u_{\arityof{\arel}} \in \alabel(n)$, by the definition of $\tree$.
  \item[{(\ref{it2:treewidth})}] let $n, m
  \in \nodes$ be nodes of $\tree$ and $u \in \alabel(n) \cap
  \alabel(m)$ be an element. Then the label of each node on the path
  between $n$ and $m$ in the parse tree contains a variable whose
  store value is $u$, hence the set $\set{p \in \nodes \mid u \in
    \alabel(p)}$ is non-empty and connected in $\tree$. \qedhere
  \end{itemize}
\end{proof}

Second, any model is obtained via an \emph{internal fusion} of a rich
canonical model. The internal fusion is a unary operation that takes
as input a structure and outputs a set of structures obtained by
joining certain elements from its support. This operation is formally
defined as quotienting with respect to certain equivalence relations:

\begin{defi}\label{def:quotient}
  Let $\astruc=(\univ,\struc)$ be a structure and $\approx \ \subseteq
  \univ \times \univ$ be an equivalence relation, where
  $[u]_{\approx}$ is the equivalence class of $u\in\univ$. The
  \emph{quotient} $\astruc_{/\approx} =
  (\univ_{/\approx},\struc_{/\approx})$ is $\univ_{/\approx} \isdef
  \set{[u]_\approx \mid u \in \univ}$ and $\struc_{/\approx}(\arel)
  \isdef \set{\tuple{[u_1]_{\approx}, \ldots,
      [u_{\arityof{\arel}}]_{\approx}} \mid \tuple{u_1, \ldots,
      u_{\arityof{\arel}}} \in \struc(\arel)}$, for all
  $\arel\in\relations$.
\end{defi}
For example, \autoref{fig:internal-fusion}~(a) shows the outcome of
quotienting a structure with respect to an equivalence relation, whose
equivalence classes are encircled with dashed lines.

A fusion operation glues elements without losing tuples from the
interpretation of a relation symbol. For this reason, we consider only
equivalence relations that are \emph{compatible} with a given
structure and define internal fusion as the following unary operation:

\begin{defi}\label{def:internal-fusion}
  An equivalence relation $\approx \ \subseteq \univ \times \univ$ is
  \emph{compatible} with a structure $\astruc=(\univ,\struc)$ iff for
  all $\arel\in\relations$ and any two tuples $\tuple{u_1, \ldots,
    u_{\arityof{\arel}}}, \tuple{v_1, \ldots, v_{\arityof{\arel}}} \in
  \struc(\arel)$, there exists $i \in \interv{1}{\arityof{\arel}}$
  such that $u_i \not\approx v_i$.  An \emph{internal fusion} of
  $\astruc$ is a structure isomorphic to $\astruc_{/\approx}$, for an
  equivalence relation $\approx$ compatible with $\astruc$. Let
  $\intfusion{\astruc}$ be the set of internal fusions of $\astruc$
  and $\intfusion{\structures} \isdef \bigcup_{\astruc\in\structures}
  \intfusion{\astruc}$, for a set $\structures$ of structures.
\end{defi}
For example, \autoref{fig:internal-fusion}~(b) shows a possible
internal fusion of a structure. Note that the equivalence relation
from \autoref{fig:internal-fusion}~(a) is not compatible with the
structure and cannot be used in a fusion.

  \begin{figure}[htbp]
    \begin{center}
    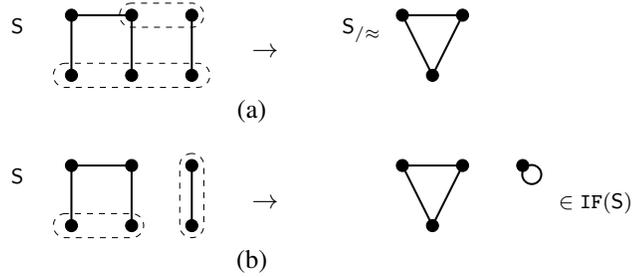
    \caption{\label{fig:internal-fusion} Quotient (a) and internal fusion (b)}
    \end{center}
  \end{figure}

For technical reasons, we introduce also the internal fusion of a rich
canonical model, as quotienting with respect to an equivalence
relation that does not violate the disequality relation:

\begin{defi}\label{def:strong-internal-fusion}
  An equivalence relation $\approx \ \subseteq \univ\times\univ$ is
  \emph{compatible} with a rich canonical model $(\astruc,\diseq)$ iff
  it is compatible with $\astruc=(\univ,\struc)$ and $\diseq(u,v)$
  only if $u \not\approx v$. We denote by
  $\sintfusion{\astruc}{\diseq}$ the set of structures isomorphic to
  $\astruc_{/\approx}$, where $\approx$ is some equivalence relation
  compatible with $(\astruc,\diseq)$.
\end{defi}

The following lemma relates the sets of models, canonical and rich
canonical models of a sentence, via the two types of internal fusion:

\begin{lem}\label{lemma:strong-internal-fusion}
  $\sidsem{\apred}{\asid} =
  \sintfusion{\rcsem{\apred}{\asid}}{} \subseteq
  \intfusion{\csem{\apred}{\asid}}$.
\end{lem}
\begin{proof}
  It is sufficient to prove $\sidsem{\apred}{\asid} =
  \sintfusion{\rcsem{\apred}{\asid}}{}$, since
  $\sintfusion{\rcsem{\apred}{\asid}}{} \subseteq
  \intfusion{\csem{\apred}{\asid}}$ is immediate, by
  \autoref{def:strong-internal-fusion}, because any equivalence
  relation that is compatible with a rich canonical $\asid$-model
  $(\astruc,\diseq)$ is also compatible with the canonical
  $\asid$-model $\astruc$.

  \skipnoindent ``$\subseteq$'' Let $\astruc \in \sidsem{\apred}{\asid}$ be
  a structure. By \autoref{prop:unfolding}, we have $\astruc \models
  \exists y_1 \ldots \exists y_m ~.~ \psi$, where $\psi$ is a qpf
  formula, such that $\fv{\psi} = \set{y_1, \ldots, y_m}$ and $\apred
  \step{\asid}^* \exists y_1 \ldots \exists y_m ~.~ \psi$ is a
  complete $\asid$-unfolding. Then there exists a store $\store$, such
  that $\astruc \models^\store \psi$. Let
  $\overline{\astruc}=(\overline{\univ},\overline{\struc})$ be a
  structure and $\overline{\store}$ be an injective store over
  $y_1,\ldots,y_m$. Since $\asid$ is equality-free, there are no
  equality atoms in $\psi$, hence such a structure and injective store
  exist. We consider $\approx
  \ \subseteq\overline{\univ}\times\overline{\univ}$ to be the least
  equivalence relation such that $\overline{\store}(y_i) \approx
  \overline{\store}(y_j) \iffdef \store(y_i) = \store(y_j)$, for all
  $1 \le i < j \le m$. To prove that $\approx$ is compatible with
  $\overline{\astruc}$, consider two tuples
  $\tuple{\overline{\store}(z_1),\ldots,\overline{\store}(z_{\arityof{\arel}})},
  \tuple{\overline{\store}(z'_1),\ldots,\overline{\store}(z'_{\arityof{\arel}})}
  \in \overline{\struc}(\arel)$, for some $\arel\in\relations$ and
  suppose, for a contradiction, that $\overline{\store}(z_i) \approx
  \overline{\store}(z'_i)$, for all $i \in
  \interv{1}{\arityof{\arel}}$. Then
  $\arel(z_1,\ldots,z_{\arityof{\arel}}) *
  \arel(z'_1,\ldots,z'_{\arityof{\arel}})$ is a subformula of $\psi$,
  modulo a reordering of atoms. By the definition of $\approx$, we
  have $\store(z_i)=\store(z'_i)$, for all $i \in
  \interv{1}{\arityof{\arel}}$, in contradiction with $\astruc
  \models^\store \psi$ and the semantics of the separating
  conjunction. Since $\overline{\astruc} \models^{\overline{\store}}
  \psi$ and $\overline{\store}$ is injective over $y_1, \ldots, y_m$,
  we obtain that $\proj{\overline{\store}}{\set{y_1,\ldots,y_m}}$ is a
  bijection between $\set{y_1,\ldots,y_m}$ and
  $\supp{\overline{\struc}}$ hence $\overline{\store}^{-1}(u)$ is a
  singleton, for each $u \in \supp{\overline{\struc}}$. Let $\diseq
  \subseteq \overline{\univ} \times \overline{\univ}$ be the relation
  defined as $\diseq(u,v)$ iff the disequality
  $\overline{\store}^{-1}(u) \neq \overline{\store}^{-1}(u)$ occurs in
  $\psi$. Then $\approx$ is compatible with
  $(\overline{\astruc},\diseq) \in \rcsem{\apred}{\asid}$ hence
  $\overline{\astruc}_{/\approx} \in
  \sintfusion{\rcsem{\apred}{\asid}}{}$. Finally, the mapping $h :
  \supp{\struc} \rightarrow \supp{\overline{\struc}}$ defined as
  $h(\store(y_i)) \isdef [\overline{\store}(y_i)]_{\approx}$, for all
  $i \in \interv{1}{m}$ is shown to be an isomorphism between
  $\astruc$ and $\overline{\astruc}_{/\approx}$, leading to $\astruc
  \in \sintfusion{\rcsem{\apred}{\asid}}{}$, by the fact that the set
  $\sintfusion{\rcsem{\apred}{\asid}}{}$ is closed under isomorphism
  (\autoref{def:strong-internal-fusion}).

  \skipnoindent ``$\supseteq$'' Let $\astruc \in
  \sintfusion{\rcsem{\apred}{\asid}}{}$ be a structure. Then there
  exists a rich canonical $\asid$-model $(\overline{\astruc},\diseq)
  \in \rcsem{\apred}{\asid}$, where
  $\overline{\astruc}=(\overline{\univ},\overline{\struc})$ and an
  equivalence relation $\approx \ \subseteq \overline{\univ} \times
  \overline{\univ}$ such that $\approx$ is compatible with
  $(\overline{\astruc},\diseq)$ and $\astruc$ is isomorphic to
  $\overline{\astruc}_{/\approx}$. Since $(\overline{\astruc},\diseq)
  \in \rcsem{\apred}{\asid}$, there exists a complete $\asid$-unfolding
  $\apred \step{\asid}^* \exists y_1 \ldots \exists y_m ~.~ \psi$, such
  that $\psi$ is qpf and a store $\store$, injective over $y_1,
  \ldots, y_m$, such that $\overline{\astruc} \models^\store \psi$ and
  $\diseq(\store(z),\store(z'))$ for each disequality $z\neq z'$ from
  $\psi$. Let $\overline{\store}$ be the store defined as
  $\overline{\store}(y_i) = [\store(y_i)]_{\approx}$, for all $i \in
  \interv{1}{m}$. We prove $\overline{\astruc}_{/\approx}
  \models^{\overline{\store}} \psi$ by induction on the structure of
  $\psi$, considering the following cases: \begin{itemize}[label=$\triangleright$]
  \item $\psi = y_i\neq y_j$: because $\approx$ is compatible with
    $(\overline{\astruc}, \diseq)$, we have $[\store(y_i)]_\approx
    \neq [\store(y_j)]_\approx$, hence $\overline{\store}(y_i) \neq
    \overline{\store}(y_j)$.
  \item $\psi = \arel(y_{i_1}, \ldots, y_{i_{\arityof{\arel}}})$:
    because $\overline{\astruc} \models^\store \arel(y_{i_1}, \ldots,
    y_{i_{\arityof{\arel}}})$, we have $\overline{\struc}(\arel) =
    \set{\tuple{\store(y_{i_1}), \ldots,
        \store(y_{i_{\arityof{\arel}}})}}$ and
    $\overline{\struc}_{/\approx}(\arel) =
    \set{\tuple{[\store(y_{i_1})]_\approx, \ldots,
        [\store(y_{i_{\arityof{\arel}}})]_\approx}}$, by
    \autoref{def:quotient}.
  \item $\psi = \psi_1 * \psi_2$: because $\overline{\astruc}
    \models^\store \psi_1 * \psi_2$, there exist locally disjoint
    structures $\overline{\astruc}_1 \comp \overline{\astruc}_2 =
    \overline{\astruc}$, such that $\overline{\astruc}_i
    \models^\store \psi_i$, for $i = 1,2$. Since $\approx$ is
    compatible with $\astruc$, the structures
    ${\overline{\astruc}_1}_{/\approx}$ and
    ${\overline{\astruc}_2}_{/\approx}$ are locally disjoint, by
    \autoref{def:internal-fusion}. Then their composition is defined
    and we have $\overline{\astruc}_{/\approx} =
    {\overline{\astruc}_1}_{/\approx} \comp
    {\overline{\astruc}_2}_{/\approx}$. By the inductive hypothesis,
    we have ${\overline{\astruc}_i}_{/\approx}
    \models^{\overline{\struc}} \psi_i$, for $i=1,2$, thus
    $\overline{\astruc}_{/\approx} \models^{\overline{\struc}} \psi_1 *
    \psi_2$.
  \end{itemize}
  Hence $\overline{\astruc}_{/\approx} \in \sidsem{\apred}{\asid}$ and
  $\astruc \in \sidsem{\apred}{\asid}$ follows, since the set
  $\sidsem{\apred}{\asid}$ is closed under isomorphism, see, e.g.,
  \cite[Proposition 7]{Concur23} for a proof.
\end{proof}

\subsection{Expandable sets of inductive definitions}

We introduce the notion of \emph{expandable} SID, a key ingredient of
our proof of decidability for the \tbsl\ problem. A structure is a
\emph{substructure} of another if the former is obtained from the
latter by removing elements from its support:

\begin{defi}\label{def:substructure}
  Let $\astruc_i=(\univ_i,\struc_i)$ be structures, for $i=1,2$.
  $\astruc_1$ is \emph{included} in $\astruc_2$ iff $\univ_1 \subseteq
  \univ_2$ and $\struc_1(\arel) \subseteq \struc_2(\arel)$, for all
  $\arel\in\relations$. $\astruc_1$ is a \emph{substructure} of
  $\astruc_2$, denoted $\astruc_1 \substruc \astruc_2$, iff $\astruc_1
  \subseteq \astruc_2$ and $\struc_1(\arel) = \set{ \tuple{u_1,
      \ldots, u_{\arityof{\arel}}} \in \struc_2(\arel) \mid u_1,
    \ldots, u_{\arityof{\arel}} \in \supp{\struc_1}}$, for all
  $\arel\in\relations$.
\end{defi}

A SID is \emph{expandable} if any set of canonical models of a
sentence are all substructures of the same canonical model of that
sentence, that can be, moreover, placed ``sufficiently far away'' one
from another.

\begin{defi}\label{def:expandable}
  A SID $\csid$ is \emph{expandable} for a nullary predicate $\apred$
  iff for each sequence of pairwise disjoint canonical models
  $\astruc_1=(\univ_1,\struc_1), \ldots, \astruc_n=(\univ_n,\struc_n)
  \in \csem{\apred}{\csid}$, there exists a rich canonical model
  $(\astruc,\diseq) \in \rcsem{\apred}{\csid}$, where
  $\astruc=(\univ,\struc)$, such that: \begin{enumerate}
  \item\label{it1:def:expandable} $\astruc_1 \comp \ldots \comp
    \astruc_n \substruc \astruc$,
  \item\label{it2:def:expandable} $\diseq(u,v)$ holds for no $u \in
    \supp{\struc_i}$ and $v \in \supp{\struc_j}$, where $1 \le i < j
    \le n$, and

  \item\label{it3:def:expandable} for no relation symbol
    $\arel\in\relations$ and tuples
    $\tuple{u_1,\ldots,u_{\arityof{\arel}}},
    \tuple{v_1,\ldots,v_{\arityof{\arel}}} \in\struc(\arel)$ there
    exist $1 \le i < j \le n$, such that
    $\set{u_1,\ldots,u_{\arityof{\arel}}} \cap \supp{\struc_i} \neq
    \emptyset$, $\set{v_1,\ldots,v_{\arityof{\arel}}} \cap
    \supp{\struc_j} \neq \emptyset$ and
    $\set{u_1,\ldots,u_{\arityof{\arel}}} \cap
    \set{v_1,\ldots,v_{\arityof{\arel}}} \neq \emptyset$.
  \end{enumerate}
\end{defi}

\begin{exa}\label{ex:expandable}
The SID $\asid$ from \autoref{fig:expandable} is expandable
for $\apred$, because any choice of pairwise disjoint canonical models
of $\apred$ (here $\astruc_1$, $\astruc_2$ and $\astruc_3$) can be
embedded in a canonical model of $\apred$, such that there are no
pairs, in the interpretation of the binary relation symbol
$\mathsf{e}$, that stretch from one such model to another (the pairs
not entirely inside the support of either $\astruc_1$, $\astruc_2$ or
$\astruc_3$ are depicted in dashed lines).
\end{exa}

\begin{figure}[htbp]
  \begin{center}
    \begin{minipage}{.42\textwidth}
      {\[ \begin{array}{rcl}
        \apred() & \leftarrow & \exists y.~\bpred(y) \\
        \bpred(x_1) & \leftarrow &
        \exists y.~ \mathsf{a}(x_1) * \mathsf{e}(x_1,y) * \bpred(y) \\
        \bpred(x_1) & \leftarrow & \mathsf{a}(x_1)
      \end{array} \]}
    \end{minipage}
    \begin{minipage}{.57\textwidth}
      \centerline{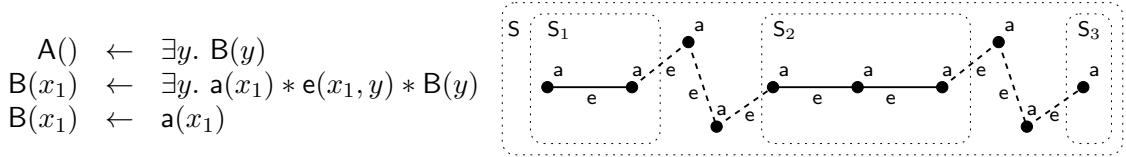}
    \end{minipage}
  \end{center}
  \caption{\label{fig:expandable}An expandable SID}
\end{figure}

\begin{exa}\label{ex:non-expandable}(continued from \autoref{ex:expandable})
  Consider the SID $\asid'$ obtained from $\asid$ by changing its last
  rule into $\bpred(x_1) \leftarrow \emp$. $\asid'$ is not expandable
  for $\apred$, because the canonical models of $\apred$ are acyclic
  chains of elements, whose neighbours are related by $\mathsf{e}$,
  such that all but the last element is labeled by $\mathsf{a}$. Then,
  two such structures $\astruc_1$ and $\astruc_2$ cannot be embedded
  in a third structure $\astruc$ as substructures, because of the last
  non-labeled element of $\astruc_1$ that occurs in the middle of
  $\astruc$ and must be labeled by $\mathsf{a}$. This violates the
  definition of substructures (see \autoref{def:substructure}),
  in which the labeling of an element is the same in a substructure
  and in its enclosing structure.
\end{exa}

The \emph{external fusion} is a binary operation that glues
elements from disjoint structures:

\begin{defi}\label{def:external-fusion}
  An \emph{external fusion} of the structures $\astruc_1 =
  (\univ_1,\struc_1)$ and $\astruc_2 = (\univ_2,\struc_2)$ is a
  structure isomorphic to $(\astruc'_1 \comp \astruc'_2)_{/\approx}$,
  where $\astruc'_i=(\univ'_i,\struc'_i)$ are disjoint isomorphic
  copies of $\astruc_i$ and $\approx \subseteq \univ'_1 \times
  \univ'_2$ is the smallest equivalence relation containing a nonempty
  $\supp{\struc'_1}$-$\supp{\struc'_2}$ matching that is compatible
  with $\astruc'_1 \comp \astruc'_2$. Let
  $\extfusion{\astruc_1}{\astruc_2}$ be the set of external fusions of
  $\astruc_1$ and $\astruc_2$. For a set of structures $\structures$,
  let $\reachfusion{\structures}$ (resp. $\ireachfusion{\structures}$)
  be the closure of $\structures$ under taking external (resp. both
  internal and external) fusions.
\end{defi}
For example, \autoref{fig:external-fusion} shows the external fusion
of two disjoint structures via a matching relation (the equivalence
classes of the matching relation are encircled with dashed lines).
Note that the conditions (\ref{it2:def:expandable}) and
(\ref{it3:def:expandable}) of \autoref{def:expandable} ensure
that the external fusion of these substructures is not hindered by
their position inside the larger structure. For instance, any matching
relation between the supports of the substructures $\astruc_1$,
$\astruc_2$ and $\astruc_3$ from \autoref{fig:expandable} can be used
to define an external fusion of these structures. This is because
there are no pairs, from the interpretation of the $\mathsf{e}$
relation symbol in $\astruc$, that have an element in common and the
other non-common elements in the support of two different
substructures. If such pairs existed, the non-common elements could
not be fused by an equivalence relation compatible with $\astruc_i
\comp \astruc_j$, for any $1 \le i < j \le 3$.

  \begin{figure}[htbp]
    \begin{center}
    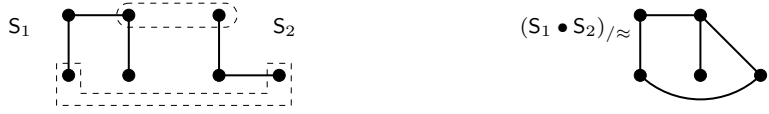
    \caption{\label{fig:external-fusion} External fusion}
    \end{center}
  \end{figure}

The following lemma proves one direction of the equivalence between
the treewidth boundedness of $\sidsem{\apred}{\asid}$ and that of the
set of structures obtained by applying both internal and external
fusion to the canonical models from $\csem{\apred}{\asid}$.

\begin{lem}\label{lemma:external-fusion}
 Let $\asid$ be an expandable SID for a nullary predicate
 $\apred$. Then, \begin{enumerate*}
  \item\label{it1:lemma:external-fusion}
    $\ireachfusion{\csem{\apred}{\asid}}$ is treewidth-bounded only if
  \item\label{it2:lemma:external-fusion} $\sidsem{\apred}{\asid}$ is
    treewidth-bounded only if
  \item\label{it3:lemma:external-fusion}
    $\reachfusion{\csem{\apred}{\asid}}$ is treewidth-bounded.
  \end{enumerate*}
\end{lem}
\begin{proof}
``(\ref{it1:lemma:external-fusion}) $\Rightarrow$
  (\ref{it2:lemma:external-fusion})''
  $\intfusion{\csem{\apred}{\asid}} \subseteq
  \ireachfusion{\csem{\apred}{\asid}}$ holds trivially, by
  \autoref{def:external-fusion}, leading to
  $\sidsem{\apred}{\asid} \subseteq
  \ireachfusion{\csem{\apred}{\asid}}$, by
  \autoref{lemma:strong-internal-fusion}.
  ``(\ref{it2:lemma:external-fusion}) $\Rightarrow$
  (\ref{it3:lemma:external-fusion})'' Let $\astruc=(\univ,\struc) \in
  \reachfusion{\csem{\apred}{\asid}}$ be a structure. It is sufficient
  to prove that $\astruc \substruc \astruc'$ for another structure
  $\astruc' \in \sidsem{\apred}{\asid}$, because $\twof{\astruc} \le
  \twof{\astruc'}$, in this case. Then there exist pairwise disjoint
  structures $\astruc_1=(\univ_1,\struc_1), \ldots,
  \astruc_n=(\univ_n,\struc_n) \in \csem{\apred}{\asid}$ and an
  equivalence relation $\approx \ \subseteq \big(\bigcup_{i=1}^n
  \univ_i\big) \times \big(\bigcup_{i=1}^n \univ_i\big)$, that is
  compatible with $\astruc_1 \comp \ldots \comp \astruc_n$, matches
  only elements from different structures and is not the identity,
  such that $\astruc$ is isomorphic to $(\astruc_1 \comp \ldots \comp
  \astruc_n)_{/\approx}$. By \autoref{def:expandable}, there
  exists a rich canonical model $(\astruc'',\diseq) \in
  \rcsem{\apred}{\asid}$, such that \begin{enumerate*}
\item $\astruc \substruc \astruc''$,
\item $\diseq(u,v)$ holds for no $u \in \supp{\struc_i}$ and $v \in
  \supp{\struc_j}$, where $1 \le i < j \le n$, and
\item for no relation symbol $\arel\in\relations$ and tuples
  $\tuple{u_1,\ldots,u_{\arityof{\arel}}},
  \tuple{v_1,\ldots,v_{\arityof{\arel}}} \in\struc(\arel)$, there
  exist $1 \le i < j \le n$, such that
  $\set{u_1,\ldots,u_{\arityof{\arel}}} \cap \supp{\struc_i} \neq
  \emptyset$, $\set{v_1,\ldots,v_{\arityof{\arel}}} \cap
  \supp{\struc_j} \neq \emptyset$ and
  $\set{u_1,\ldots,u_{\arityof{\arel}}} \cap
  \set{v_1,\ldots,v_{\arityof{\arel}}} \neq \emptyset$.
\end{enumerate*}
By the last two conditions, $\approx$ is compatible with
$(\astruc'',\diseq)$, leading to $\astruc''_{/\approx} \in
\sintfusion{\rcsem{\apred}{\asid}}{} = \sidsem{\apred}{\asid}$ by
\autoref{lemma:strong-internal-fusion}. We conclude by taking
$\astruc' = \astruc''_{/\approx}$.
\end{proof}

The missing direction $\reachfusion{\csem{\apred}{\asid}} \Rightarrow
\ireachfusion{\csem{\apred}{\asid}}$, that allows to establish the
equivalence of the three points of \autoref{lemma:external-fusion},
requires the introduction of further technical notions. The proof of
the main result of this section relies on an algorithm for the
treewidth boundedness of sets $\reachfusion{\csem{\apred}{\asid}}$,
obtained by external fusion of disjoint canonical $\asid$-models of
$\apred$. By the equivalence of the treewidth boundedness of the sets
$\sidsem{\apred}{\asid}$ and $\reachfusion{\csem{\apred}{\asid}}$
(\autoref{lemma:external-fusion} and
\autoref{lemma:internal-external-fusion}), this is also an algorithm for
the \tbsl\ problem for expandable SIDs.

\subsection{Color schemes}

Intuitively, the \emph{color} of an element from the support of a
structure is the set of relation symbols labeling solely that
element. For the given set $\relations$ of relation symbols, we define
the set of \emph{colors} as $\allcols \isdef \pow{\relations}$. The
elements of a structure are labeled with colors as follows:

\begin{defi}\label{def:color-extraction-function}
  The \emph{coloring} of a structure $\astruc = (\univ,\struc)$ is the
  mapping $\funcol{\astruc} : \univ \rightarrow \allcols$ defined as
  $\funcol{\astruc}(u) \isdef \set{ \arel\in \relations ~|~ \tuple{u,
      \ldots, u} \in \struc(\arel) }$.
\end{defi}
Moreover, we define an abstraction of structures as finite multisets
of colors:

\begin{defi}\label{def:multiset-color-abstraction}
  The \emph{multiset color abstraction} $\mcolabs{\astruc} \in
  \mpow{\allcols}$ of a structure $\astruc = (\univ,\struc)$ is
  $\mcolabs{\astruc} \isdef \mset{ \funcol{\astruc}(u) ~|~ u \in
    \supp{\struc}}$.  For an integer $k\geq0$, the \emph{$k$-multiset
    color abstraction} $\kmcolabs{k}{\astruc} \subseteq
  \mpow{\allcols}$ is $\kmcolabs{k}{\astruc} \isdef \set{ M \subseteq
    \mcolabs{\astruc} ~|~ \cardof{M} \le k}$. These abstractions are
  lifted to sets $\structures$ of structures, yielding the sets of
  multisets $\mcolabs{\structures} \isdef \set{ \mcolabs{\astruc} ~|~
    \astruc \in {\structures}}$ and $\kmcolabs{k}{\structures} \isdef
  \bigcup_{\astruc \in \structures} \kmcolabs{k}{\astruc}$.
\end{defi}
Colors are organized in \emph{RGB color schemes}, defined below:

\begin{defi}\label{def:rgb-color-scheme}
  A partition $(\redcols,\greencols,\bluecols)$ of $\allcols$ is an
  \emph{RGB-color scheme} iff: \begin{enumerate}
  \item $\ucolor_1 \cap \ucolor_2 \not= \emptyset$, for all
    $\ucolor_1,\ucolor_2 \in \bluecols$,
  \item $\ucolor_1 \cap \ucolor_2 \not= \emptyset$, for all $\ucolor_1
    \in \greencols$ and $\ucolor_2 \in \bluecols$,
  \item for all $\ucolor_1 \in \redcols$ there exists $\ucolor_2 \in
    \bluecols$ such that $\ucolor_1 \cap \ucolor_2 = \emptyset$.
  \end{enumerate}
\end{defi}
Note that an RGB-color scheme is fully specified by the set
$\bluecols$. Indeed, any color not in $\bluecols$ is unambiguously
placed within $\redcols$ or $\greencols$, depending on whether it is
disjoint from some color in $\bluecols$.  For example,
\autoref{fig:rgb-color-schemes} shows several RGB-color schemes
for the relational signature $\relations=\set{\aarel,\abrel,\acrel}$.

  \begin{figure}[htbp]
    \begin{center}
    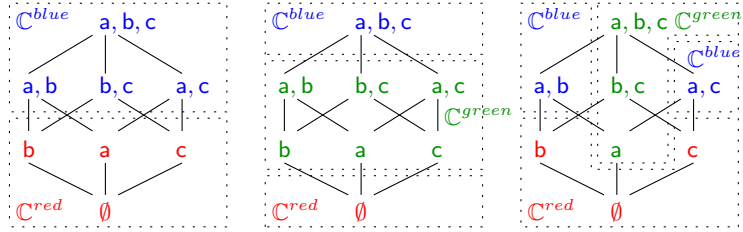
    \caption{\label{fig:rgb-color-schemes} Examples of RGB color schemes}
    \end{center}
  \end{figure}

Because a fusion operation only joins element with disjoint colors,
blue elements can only be joined with red elements, green elements can
be joined with green or red elements, whereas red elements can be
joined with elements of any other color, provided that they are
disjoint subsets of $\relations$. Moreover, a fusion operation can
always join a pair of elements with disjoint colors:

\begin{lem}\label{lemma:single-pair-equivalence-compatible}
  Let $\astruc_1=(\univ_1,\struc_1)$ and
  $\astruc_2=(\univ_2,\struc_2)$ be disjoint structures. Let $u_1 \in
  \supp{\struc_1}$, $u_2 \in \supp{\struc_2}$ be elements such that
  $\funcol{\astruc_1}(u_1) \cap \funcol{\astruc_2}(u_2) = \emptyset$.
  Then, the equivalence relation on $\univ_1 \cup \univ_2$ generated
  by $(u_1,u_2)$ is compatible with $\astruc_1 \comp \astruc_2$.
\end{lem}
\begin{proof}
    We denote by $\approx$ the relation $\set{(u_1, u_2)}^=$ in the
  following.  Let $\arel \in \relations$ be a relation and let
  $\tuple{u_{1,1}, u_{1,2}, \ldots, u_{1,\arityof{\arel}}} \in
  \struc_1(\arel)$, $\tuple{u_{2,2}, u_{2,2}, \ldots,
    u_{2,\arityof{\arel}}} \in \struc_2(\arel)$ be distinct tuples. If
  for some index $i\in \interv{1}{\arityof{\arel}}$ either $u_{1,i}
  \neq u_1$ or $u_{2,i} \neq u_2$ then $u_{1,i} \not\approx u_{2,i}$,
  by the definition of $\approx$. Otherwise, if for all indices $i\in
  \interv{1}{\arityof{\arel}}$ both $u_{1,i}=u_1$ and $u_{2,i}=u_2$
  then $r \in \funcol{\astruc_1}(u_1)$ and $r \in
  \funcol{\astruc_2}(u_2)$.  This implies $\funcol{\astruc_1}(u_1)
  \cap \funcol{\astruc_2}(u_2) \not= \emptyset$ and contradicts the
  hypothesis about the choice of $u_1$, $u_2$.  Therefore, no tuples
  from $\astruc_1$ and $\astruc_2$ respectively are merged by the
  fusion.  Finally, it is also an easy check that no tuples from
  $\astruc_1$ (resp. $\astruc_2$) are merged, because when restricted
  to $\astruc_1$ (resp. $\astruc_2$) the equivalence $\approx$ becomes
  the identity.
\end{proof}

The first ingredient of a decidable condition, equivalent to the
treewidth boundedness of a set $\reachfusion{\csem{\apred}{\asid}}$,
is conformance with an RGB color scheme, defined below:

\begin{defi}\label{def:conforming-structures}
  A set $\structures$ of structures \emph{conforms to
  $(\redcols,\greencols,\bluecols)$} if and only if: \begin{enumerate}
  \item\label{it1:def:conforming-structures} for all structures
    $\astruc=(\univ,\struc)\in\structures$, if
    $\funcol{\astruc}(u)\in\redcols$, for some element $u \in
    \supp{\struc}$, then $\funcol{\astruc}(u')\in\bluecols$, for all
    other elements $u'\in\supp{\struc}\setminus\set{u}$, and
  \item\label{it2:def:conforming-structures} $\mcolabs{\astruc} \cap
    \greencols \subseteq \mset{ \ucolor, \ucolor \mid \ucolor \in
      \greencols}$, for all structures $\astruc \in
    \reachfusion{\structures}$.
  \end{enumerate}
  Moreover, a structure $\astruc\in\structures$ is said to be of
  either type: \begin{itemize}[label=$\triangleright$]
  \item $\redtype$ if~ $\mcolabs{\astruc} \in \mpow{\bluecols \cup \redcols}$
    and $\cardof{\mcolabs{\astruc} \sqcap \redcols} = 1$,
  \item $\greentype$ if~ $\mcolabs{\astruc} \in \mpow{\bluecols \cup \greencols}$ and
    $\cardof{\mcolabs{\astruc} \sqcap \greencols} > 0$, and
  \item $\bluetype$ if~ $\mcolabs{\astruc} \in \mpow{\bluecols}$.
  \end{itemize}
\end{defi}

Conformance to some RGB color scheme is the key to bounding the
treewidth of the sets of structures obtained by external fusion of a
treewidth bounded set of structures.

An equivalence relation $\approx$ is said to be \emph{generated} by a
set of pairs $(u_1,v_1), \ldots, (u_k,v_k)$ if it is the least
equivalence relation, such that $u_i \approx v_i$, for all $i \in
\interv{1}{k}$. Furthermore, we say that $\approx$ is $k$-generated if
$k$ is the minimal cardinality of a set of pairs that generates
$\approx$.

\begin{lem}\label{lemma:twb-conforming-structures}
  Let $\structures$ be a treewidth bounded set of structures
  conforming to an RGB color scheme.  Then, for any structure $\astruc
  \in \reachfusion{\structures}$, the following hold:
  \begin{enumerate}
  \item\label{it1:lemma:twb-conforming-structures} $\astruc$ is
    of type either $\redtype$, $\greentype$ or
    $\bluetype$,
  \item\label{it2:lemma:twb-conforming-structures} if $\astruc =
    (\astruc_1 \comp \astruc_2)_{/\approx}$ for some $\astruc_1,
    \astruc_2 \in \reachfusion{\structures}$ then exactly one of the
    following hold: \begin{enumerate}
    \item\label{it21:lemma:twb-conforming-structures} $\approx$ is
      $1$-generated, or
    \item $\approx$ is $2$-generated and either~ \begin{enumerate*}
    \item\label{it221:lemma:twb-conforming-structures} $\astruc_1$,
      $\astruc_2$ are of type $\redtype$, or
    \item\label{it222:lemma:twb-conforming-structures} $\astruc_1$,
      $\astruc_2$ are of type $\greentype$ and
      $\cardof{\mcolabs{\astruc_1} \sqcap \greencols} =
      \cardof{\mcolabs{\astruc_2} \sqcap \greencols} = 2$
    \end{enumerate*}
    \end{enumerate}

  \item\label{it3:lemma:twb-conforming-structures} $\twof{\astruc} \le \twof{\structures} + 1$.
  \end{enumerate}
\end{lem}
\begin{proof}
  (\ref{it1:lemma:twb-conforming-structures}) By induction on the
  derivation of $\astruc \in \reachfusion{\structures}$ from
  $\structures$. \autoref{table:extfusion-types} summarizes the
  possible types of $\extfusion{\astruc_1}{\astruc_2}$ on
  structures $\astruc_1$ and $\astruc_2$ of types $\redtype$,
  $\greentype$ or $\bluetype$, respectively.
\begin{table}[htbp]
\centering
\begin{tabular}{ c  c  c  c }
\toprule
$\extfusion{\astruc_1}{\astruc_2}$ &
$\astruc_2$ of $\redtype$ type &
$\astruc_2$ of $\greentype$ type &
$\astruc_2$ of $\bluetype$ type \\ \midrule
$\astruc_1$ of $\redtype$ type &
$\redtype, \greentype, \bluetype$ &
$\greentype, \bluetype$ &
$\bluetype$ \\ \midrule
$\astruc_1$ of $\greentype$ type &
$\greentype, \bluetype$ &
$\greentype, \bluetype$ & $\bot$ \\ \midrule
$\astruc_1$ of $\bluetype$ type &
$\bluetype$ & $\bot$ & $\bot$ \\ \bottomrule
\end{tabular}
\caption{The types of structures obtained by external fusion ($\bot$
  means none)}\label{table:extfusion-types}
\end{table}

\skipnoindent (\ref{it2:lemma:twb-conforming-structures})
We distinguish two cases: \begin{itemize}[label=$\triangleright$]
\item $\astruc_1$ is of type $\redtype$: If $\astruc_2$ is of type
  $\bluetype$ or $\greentype$ then $\astruc_1$ and $\astruc_2$ can be
  fused only by equivalences $\approx$ generated by a single pair,
  that contains the element from the support of $\astruc_1$ with color
  in $\redcols$, thus matching the case
  (\ref{it1:lemma:twb-conforming-structures}) from the
  statement. Else, if $\astruc_2$ is of type $\redtype$ then
  $\astruc_1$ and $\astruc_2$ can be fused by equivalences generated
  by at most two pairs, each containing an element with color from
  $\redcols$, from either $\astruc_1$ or $\astruc_2$, thus matching
  the case (\ref{it221:lemma:twb-conforming-structures}) from the
  statement.
\item $\astruc_1$, $\astruc_2$ are both of type $\greentype$: By
  contradiction, assume they can be fused by an equivalence $\approx$
  generated by three pairs of elements $(u_{1i},u_{2i})_{i=1,2,3}$.
  Let $G_{1i}=\funcol{\astruc_1}(u_{1i})$, $G_{2i}
  =\funcol{\astruc_2}(u_{2i})$ be the colors from $\greencols$ of the
  matching elements in the two structures, for $i=1,2,3$.  Then, we
  can construct structures using $\astruc_1$ and $\astruc_2$ where any
  of these colors repeat strictly more than twice, henceforth,
  contradicting the conformance property to the RGB color scheme.  The
  principle of the construction is depicted in \autoref{fig:green}.
  Finally, note that the construction depicted in \autoref{fig:green}
  fuse actually only pairs of colors $(G_{1i},G_{2i})$ for $i=1,2$.
  Henceforth, the conformance property is also contradicted if
  $\astruc_1$ and $\astruc_2$ can be fused by a $2$-generated
  equivalence relation $\approx$, such that the support of either
  $\astruc_1$ or $\astruc_2$ contains more than three elements with
  colors in $\greencols$.
\end{itemize}

  \begin{figure}[htbp]
    \begin{center}
    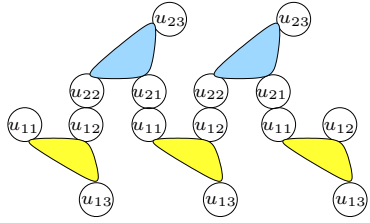
    \caption{\label{fig:green} External fusion of $\greentype$
    structures by $3$-generated matchings}
    \end{center}
  \end{figure}

\skipnoindent (\ref{it3:lemma:twb-conforming-structures})
The previous point shows that, under the hypotheses of the lemma,
every structure in $\reachfusion{\structures}$ is constructed by
external fusion with matchings generated by one or two pairs of
elements. This result can be actually refined, i.e., we can consider
only external fusions where the $1$-generated matchings are applied
before the $2$-generated matchings. That is, assume $\astruc =
((\astruc_1 \comp \astruc_2)_{/\approx_2} \comp
\astruc_3)_{/\approx_1}$ where $\approx_1$, $\approx_2$ are $1$-,
resp. $2$-generated and $\astruc_1$, $\astruc_2$, $\astruc_3 \in
\reachfusion{\structures}$.  Without loss of generality, assume
moreover, $\approx_1$ is matching some element of $\astruc_3$ with an
element of $\astruc_1$ (the other case is symmetric). Then, we can
find $1$-, resp. $2$-generated matchings $\approx_1'$, $\approx_2'$
such that $\astruc = ((\astruc_1 \comp \astruc_3)_{/\approx_1'} \comp
\astruc_2)_{/\approx_2'}$. That is, first fuse $\astruc_1$ and
$\astruc_3$ by a single matching pair, then fuse the result with
$\astruc_2$ by two matching pairs. Therefore, we can w.l.o.g. assume
in the following that $\reachfusion{\structures} = \reachfusiontwo{
  \reachfusionone{ \structures} }$, where $\reachfusionk{k}$ denotes
the external fusion using only $k$-generated matchings.

Given a tree decomposition $T$ for a structure
$\astruc=(\univ,\struc)$ and an equivalence relation $\approx
\ \subseteq\univ\times\univ$, we denote by $T_{/\approx}$ the tree
decomposition of the quotient structure $S_{/\approx}$, obtained by
the relabeling of elements $u$ in the bags of $T$ by their
representatives $[u]_{\approx}$. We prove the following facts:

\begin{fact}\label{fact1:twb-conforming-structures}
  $\twof{ \reachfusionone{\structures} } \le \twof{\structures}$.
\end{fact}
\begin{proof} By induction on the derivation of
  $\astruc\in\reachfusionone{{\structures}}$ from $\structures$.

  \skipnoindent \underline{Base case:} Immediate,
  as for any $\astruc \in \structures$ we have $\twof{\astruc} \le
  \twof{\structures}$.

  \skipnoindent \underline{Induction step:}
  Consider $\astruc = (\astruc_1 \comp \astruc_2)_{/\approx}$ where
  $\approx \ = \set{(u_1, u_2)}^=$.  Let $T_1$, $T_2$ be tree
  decompositions of respectively $\astruc_1$, $\astruc_2$ such that
  $\width{T_1}, \width{T_2} \le \twof{\structures}$. We first build a
  tree decomposition $T_{12}$ of $\astruc_1 \comp \astruc_2$ by
  \begin{itemize}[label=$\triangleright$]
  \item transforming $T_2$ into $T_2'$ by reversing edges such that a
    node $n_2$ containing $u_2$ in $T_2$ becomes the root of $T_2'$ and
  \item linking the root of $T_2'$ to a node $n_1$ of $T_1$ containing
    $u_1$.
  \end{itemize}
  This ensures that $T_{12/\approx}$ is a valid
  tree decomposition for $\astruc$, and moreover $\width{T_{12/\approx}}
  = \width{T_{12}} = \max(\width{T_1},\width{T_2}) \le
  \twof{\structures}$.
\end{proof}

\begin{fact}\label{fact2:twb-conforming-structures}
  Let  $\astruc = (\univ,\struc) \in \reachfusiontwo{ \reachfusionone{\structures}}$
  be a structure. Then, one of the following holds: \begin{enumerate}[(A)]
  \item\label{it1:fact2:twb-conforming-structures}
    $\twof{\astruc} \le \twof{\structures}$,
  \item\label{it2:fact2:twb-conforming-structures} $\astruc
    \mbox{ is of type } \bluetype,~ \twof{\astruc} \le
    \twof{\structures} + 1$,
  \item\label{it3:fact2:twb-conforming-structures} $\astruc
    \mbox{ is of type } \greentype$ and there exists
    $T=(\nodes,\edges,r,\alabel)$ a tree decomposition of $\astruc$
    such that $\width{T} \le \twof{\structures} + 1$ and
    $\funcol{\astruc}(u) \in \greencols$ implies $u \in \alabel(r)$,
    for all $u \in \supp{\struc}$.
  \end{enumerate}
\end{fact}
\begin{proof}
  By induction on the derivation of
  $\astruc\in\reachfusiontwo{\reachfusionone{\structures}}$ from
  $\reachfusionone{\structures}$.

  \skipnoindent \underline{Base case:} Immediate,
  as we already shown $\twof{\astruc} \le \twof{\structures}$, that
  is, \ref{it1:fact2:twb-conforming-structures} for any $\astruc \in
  \reachfusionone{\structures}$.

  \skipnoindent \underline{Induction step:}
  Consider $\astruc = (\astruc_1 \comp \astruc_2)_{/\approx}$ where
  $\approx \ = \set{(u_{11}, u_{21}),(u_{12},u_{22})}^=$.  Since
  $\approx$ is $2$-generated, we know from the previous point
  (\ref{it2:lemma:twb-conforming-structures}) that
  either: \begin{itemize}[label=$\triangleright$]
  \item $\astruc_1$, $\astruc_2$ are of type $\redtype$: From the
    induction hypothesis on $\astruc_1$ and $\astruc_2$ it follows
    that both must satisfy
    \ref{it1:fact2:twb-conforming-structures}, hence
    $\twof{\astruc_1} \le \twof{\structures}$, $\twof{\astruc_2} \le
    \twof{\structures}$ respectively.  We are therefore in the
    situation of composing two structures of type $\redtype$ by a
    $2$-generated matching, hence obtaining the structure $\astruc$ of
    type $\bluetype$. Without loss of generality consider
    $\funcol{\astruc_1}(u_{11}) \in \redcols$,
    $\funcol{\astruc_1}(u_{12}) \in \bluecols$,
    $\funcol{\astruc_2}(u_{21}) \in \bluecols$,
    $\funcol{\astruc_2}(u_{22}) \in \redcols$.  Let $T_1$, $T_2$ be
    tree decompositions of $\astruc_1$, $\astruc_2$ respectively.
    First, we construct a tree decomposition $T_{12}$ for $\astruc_1
    \comp \astruc_2$ as follows: \begin{itemize}
  \item construct $T_1'$ from $T_1$ by propagating the element
    $u_{11}$ to all the nodes,
  \item construct $T_2'$ from $T_2$ by propagating the element
    $u_{22}$ to all the nodes and then reverting the edges such that a
    node $n_2$ containing the element $u_{21}$ becomes the root,
  \item link the root node of $T_2'$ to a node $n_1$ of $T_1'$
    containing the element $u_{12}$.
  \end{itemize}
  This ensures that $T_{12/\approx}$ is a valid tree decomposition for
  $\astruc$, and moreover $\width{T_{12/\approx}} = \width{T_{12}} =
  \max(\width{T_1'},\width{T_2'}) = \max (\width{T_1} + 1, \width{T_2} +
  1) \le \twof{\structures} + 1$.  This completes the proof that
  $\astruc$ satisfies condition \ref{it2:fact2:twb-conforming-structures}.
\item $\astruc_1$, $\astruc_2$ are of type $\greentype$ and
  $\cardof{\mcolabs{\astruc_1} \sqcap \greencols} =
  \cardof{\mcolabs{\astruc_2} \sqcap \greencols} = 2$: First, let us
  observe that the structure $\astruc$ is either of type $\bluetype$
  or $\greentype$.  According to the induction hypothesis, we consider
  the following two cases: \begin{itemize}
  \item both $\astruc_1$ and $\astruc_2$ satisfy the condition
    \ref{it3:fact2:twb-conforming-structures}, namely there exist the
    tree decompositions $T_1$, $T_2$ of width at most
    $\twof{\structures} + 1$ such that moreover all the elements with
    colors in $\greencols$ are located at their root nodes.  Then, we
    can construct a tree decomposition $T_{12}$ for $\astruc_1 \comp
    \astruc_2$ by simply linking the root of $T_2$ as a child to the
    root of $T_1$.  We obtain that $T_{12/\approx}$ is a valid
    decomposition for $\astruc$ and satisfies $\width{T_{12/\approx}}
    = \width{T_{12}} = \max(\width{T_1}, \width{T_2}) \le
    \twof{\structures} + 1$.  Moreover, if $\astruc$ is of type
    $\greentype$ observe that all the elements with colors in
    $\greencols$ are located at the root node of $T_{12/\approx}$.
    Therefore, in any case, the structure $\astruc$ satisfies either
    \ref{it2:fact2:twb-conforming-structures} or
    \ref{it3:fact2:twb-conforming-structures}.
  \item either one or both of $\astruc_1$ or $\astruc_2$ satisfy the
    condition \ref{it1:fact2:twb-conforming-structures}.  Without loss
    of generality consider $\twof{\astruc_1} \le \twof{\structures}$.
    We know however that $\astruc_1$ satisfies
    $\cardof{\mcolabs{\astruc_1} \sqcap \greencols} = 2$, that is, it
    has exactly two elements with colors in $\greencols$.  But then,
    we can show that $\astruc_1$ satisfies the condition
    \ref{it3:fact2:twb-conforming-structures} as well.  That is,
    consider a tree decomposition $T_1$ for $\astruc_1$ such that
    $\width{T_1} \le \twof{\structures}$.  Let $u_{11}$, $u_{12}$ be
    the two elements with colors in $\greencols$.  We can build a tree
    decomposition $T_1'$ fulfilling
    \ref{it3:fact2:twb-conforming-structures} by first propagating the
    element $u_{11}$ to all the nodes of $T_1$ and then reverting the
    edges such that some node containing the element $u_{12}$ becomes
    the root.  Obviously, $\width{T_1'} = \width{T_1} + 1 \le
    \twof{\structures} + 1$ and all the elements with color in
    $\greencols$ are located at the root node.  We can proceed
    similarly with $\astruc_2$.  Then, the proof is completed as in
    the first case. \qedhere
  \end{itemize}
  \end{itemize}
\end{proof}
This completes the proof of the point
(\ref{it3:lemma:twb-conforming-structures}) from the statement.
\end{proof}

Moreover, conformance with RGB schemes allow us to
infer a bound on the set obtained by applying external fusion to a
treewidth bounded set of structures:

\begin{lem}\label{cor:twb-conforming-structures}
  Let $\structures$ be a treewidth bounded set of structures, that
  conforms to an RGB color scheme. Then, we have \(
  \twof{\intfusion{\reachfusion{\structures}}} \le \twof{\structures}
  + 1 \).
\end{lem}
\begin{proof}
  $\reachfusion{\structures}$ is treewidth bounded as a direct
  consequence of \autoref{lemma:twb-conforming-structures}, that
  establishes the bounds for every type of structure from
  $\reachfusion{\structures}$. Moreover,
  $\intfusion{\reachfusion{\structures}}$ is treewidth-bounded
  because, using the tree decompositions $T$ constructed for
  structures $\astruc$ in $\reachfusion{\structures}$ one obtains tree
  decomposition $T'$ and treewidth bounds for any structures $\astruc'
  = \astruc_{/\approx}$ obtained by internal fusion, as follows:
  \begin{itemize}[label=$\triangleright$]
  \item if $\astruc$ is of type $\redtype$ then by internal fusion one
    glues the unique element $u_1$ with color in $\redcols$ to some
    other element in the structure. We know, from the inductive
    property used in the proof of
    \autoref{lemma:twb-conforming-structures}
    (\ref{it3:lemma:twb-conforming-structures}) that if $\astruc$ is
    of type $\redtype$ then it must satisfies condition
    \ref{it1:fact2:twb-conforming-structures} of Fact
    \ref{fact2:twb-conforming-structures} that is, $\twof{\astruc} \le
    \twof{\structures}$.  Therefore, one can construct $T'$ from $T$
    by replicating $u_1$ in all nodes and then $T'_{/\approx}$ is a
    valid tree decomposition for $\astruc'$.  Obviously $\width{T'}
    \le \width{T} + 1 \le \twof{\structures} + 1$.
  \item if $\astruc$ is of type $\greentype$ then by internal fusion
    one glue elements with color in $\greencols$. As before, using the
    same inductive property, we know that $\astruc$ satisfies either
    condition \ref{it1:fact2:twb-conforming-structures} or
    \ref{it3:fact2:twb-conforming-structures} of
    \autoref{fact2:twb-conforming-structures}.  If $\astruc$
    satisfies condition \ref{it3:fact2:twb-conforming-structures}, as
    all elements that could be glued are already present in the root
    node, $T_{/\approx}$ is a valid tree decomposition for $\astruc'$.
    Obviously, the treewidth bound remains unchanged, that is, at most
    $\twof{\structures} + 1$.  If $\astruc$ satisfies condition
    \ref{it1:fact2:twb-conforming-structures} and two of its nodes
    $u_1$, $u_2$ with colors in $\greencols$ can be fused, then two
    copies of $\astruc$, respectively $\astruc'$, $\astruc''$ can also be fused via the two-pair matching $(u_1',u_2'')$, $(u_1'',u_2')$.
    Henceforth, $u_1$ and $u_2$ must be the unique elements with
    colors in $\greencols$, otherwise contradicting the point
    (\ref{it2:lemma:twb-conforming-structures}) of
    \autoref{lemma:twb-conforming-structures}.  But then, from the
    tree decomposition $T$ of $\astruc$ such that $\width{T} \le
    \twof{\structures}$ we can construct the tree decomposition $T'$
    be simply propagating $u_2$ to all other nodes in $T$ and hence,
    preserving the bound of $\twof{\structures} + 1$.
  \item if $\astruc$ is of type $\bluetype$ then no non-trivial
    internal fusion exists, and obviously, the treewidth bound remains
    unchanged. \qedhere
  \end{itemize}
\end{proof}

\subsection{Connected structures}

We shall check conformance with RGB schemes for sets of
\emph{maximally connected} structures, defined below:

\begin{defi}\label{def:connected}
A \emph{path} from $u$ to $v$ in a structure $\astruc=(\univ,\struc)$
is a finite sequence of tuples:
\[\tuple{u_{1,1}, \ldots, u_{1,n_1}} \in \struc(\arel_1), \ldots, \tuple{u_{k,1}, \ldots, u_{k,n_k}} \in \struc(\arel_k) \text{, for some }
\arel_1, \ldots, \arel_k \in \relations\]
where $u \in \set{u_{1,1},
  \ldots, u_{1,n_1}}$, $v \in \set{u_{k,1}, \ldots, u_{k,n_k}}$ and
$\{u_{i,1}, \ldots, u_{i,n_i}\} \cap \{u_{i+1,1}, \ldots,
u_{i+1,n_{i+1}}\} \neq \emptyset$, for all $i \in
\interv{1}{k-1}$. The structure $\astruc$ is \emph{connected} iff
there exists a path from $u$ to $v$, for all $u,v\in\supp{\struc}$.
\end{defi}

\begin{defi}\label{def:split-connected}
  A structure $\astruc_1$ is a \emph{maximal connected substructure}
  of another structure $\astruc_2$, denoted $\astruc_1 \mcsubstruc
  \astruc_2$, iff \begin{enumerate*}[(i)]
  \item $\astruc_1 \substruc \astruc_2$ (see \autoref{def:substructure}),
  \item $\astruc_1$ is connected, and
  \item for any connected substructure $\astruc_1' \substruc
    \astruc_2$, we have $\astruc_1 \substruc \astruc_1'$ only if
    $\astruc_1 = \astruc_1'$.
  \end{enumerate*}
  For a structure $\astruc$ we denote by $\funsplit{\astruc} \isdef
  \set{\astruc' ~|~ \astruc' \mcsubstruc \astruc}$ the set of
  maximally connected substructures, lifted to sets of structures
  $\structures$ as $\funsplit{\structures} \isdef \cup_{\astruc \in
    \structures} \funsplit{\astruc}$.
\end{defi}
Note that $\twof{\structures} = \twof{\funsplit{\structures}}$ for any
set of structures $\structures$. The next lemma shows that both
internal and external fusions preserve maximally connected
substructures:

\begin{lem}\label{lemma:split-fusion}
  For each set $\structures$ of structures, the following
  hold: \begin{enumerate}
  \item\label{it1:lemma:split-fusion}
    $\funsplit{\reachfusion{\structures}} =
    \reachfusion{\funsplit{\structures}}$, and
  \item\label{it2:lemma:split-fusion}
    $\funsplit{\ireachfusion{\structures}} =
    \ireachfusion{\funsplit{\structures}}$.
    \end{enumerate}
\end{lem}
\begin{proof}
  For space reasons, this proof is given in \autoref{app:split-fusion}.
\end{proof}

The core of our algorithm is a decidable equivalent condition for the
treewidth boundedness of a set obtained by applying external fusion to
a set of connected structures. This condition is that, in any of the
structures produced by external fusion, there is no way of connecting
six elements $u_1, v_1, w_1$ and $u_2, v_2, w_2$, labeled with
non-disjoint colors $\ucolor_1$ and $\ucolor_2$, respectively. Assume
that this condition is violated by some structures $\astruc_1$ and
$\astruc_2$ with elements $u_1, v_1, w_1$ and $u_2, v_2, w_2$, such
that $\ucolor_1 \cap \ucolor_2 = \emptyset$.  In this case,
\autoref{fig:grids} depicts the construction of a structure with an
$n\times n$ square grid minor, of treewidth at least $n$, for any
$n\geq 1$. Intuitively, $\ucolor_1 \cap \ucolor_2 = \emptyset$ allows
to glue the elements $u_1$ with $u_2$, $v_1$ with $v_2$ and $w_1$ with
$w_2$, respectively.

\begin{lem}\label{lemma:btw-external-fusion}
  The following are equivalent, for any treewidth-bounded set
  $\structures$ of structures: \begin{enumerate}
  \item\label{it1:lemma:btw-external-fusion}
    $\reachfusion{\structures}$ is treewidth bounded,
  \item\label{it2:lemma:btw-external-fusion}
    $\mset{\ucolor_1,\ucolor_1,\ucolor_1},\mset{\ucolor_2,\ucolor_2,\ucolor_2}
    \in \kmcolabs{3}{(\reachfusion{\funsplit{\structures}})}$ implies
    $\ucolor_1 \cap \ucolor_2 \neq \emptyset$, for all
    $\ucolor_1,\ucolor_2\in\allcols$,
  \item\label{it3:lemma:btw-external-fusion} $\funsplit{\structures}$
    conforms to some RGB color scheme.
  \end{enumerate}
\end{lem}
\begin{proof}
  ``(\ref{it1:lemma:btw-external-fusion}) $\Rightarrow$
  (\ref{it2:lemma:btw-external-fusion})'' If
  $\reachfusion{\structures}$ is treewidth-bounded then
  $\funsplit{\reachfusion{\structures}}$ is treewidth-bounded. Using
  \autoref{lemma:split-fusion} the later set is equal to
  $\reachfusion{\funsplit{\structures}}$ and henceforth treewidth
  bounded as well. By contradiction, assume that
  (\ref{it2:lemma:btw-external-fusion}) does not hold. Then, there
  exist colors $\ucolor_1,\ucolor_2 \in \allcols$, connected
  structures $\astruc_1,\astruc_2 \in
  \reachfusion{\funsplit{\structures}})$ such that
  $\mset{\ucolor_1,\ucolor_1,\ucolor_1} \in \mcolabs{\astruc}_1$,
  $\mset{\ucolor_2,\ucolor_2,\ucolor_2} \in \mcolabs{\astruc}_2$ and
  moreover $\ucolor_1 \cap \ucolor_2 = \emptyset$. We shall use
  $\astruc_1$ and $\astruc_2$ to build infinitely many connected
  structures containing arbitrarily large square grid minors. First,
  construct the connected structure ${\astruc_{12}} \in
  \reachfusion{\funsplit{\structures}}$ by fusing one pair $(u_1,u_2)$
  with colors $\ucolor_1$, $\ucolor_2$.  Let $v_1$, $w_1$ resp. $v_2$,
  $w_2$ be the remaining distinct elements of $\astruc_{12}$ with
  color $\ucolor_1$, $\ucolor_2$ from respectively $\astruc_1$,
  $\astruc_2$.  For arbitrarily positive $n$, consider $n\times n$
  disjoint copies $(\astruc_{12}^{i,j})_{i,j=1,n}$ of $\astruc_{12}$.
  Let $\approx^{1,j}$ be $\set{(v_1^{1,j}, v_2^{1,j-1})}^=$,
  $\approx^{i,1}$ be $\set{(w_2^{i,1},w_1^{i-1,1})}^=$,
  $\approx^{i,j}$ be $\set{(v_1^{i,j}, v_2^{i,j-1}),
    (w_2^{i,j},w_1^{i-1,j})}^=$ for all $i,j=2,n$. Second, construct
  the grid-like connected structure $X^{n,n} \in
  \reachfusion{\funsplit{\structures}}$:
  \[
    X^{n,n} = (...( ... ((\astruc_{12}^{1,1}
    \comp \astruc^{1,2}_{12})_{/\approx^{1,2}} \comp
    \astruc^{2,1}_{12})_{/\approx^{2,1}} \comp...  \comp
    \astruc^{i,j}_{12})_{/\approx^{i,j}} \comp ...  \comp
    \astruc^{n,n}_{12})_{/\approx^{n,n}}
  \]
  where structures $S^{i,j}_{12}$ are added to the fusion in
  increasing order of $i+j$. The construction is illustrated in
  \autoref{fig:grids}. We can show that $X^{n,n}$ contains an $n
  \times n$ square grid minor. Finally, as $n$ can be taken
  arbitrarily large, we conclude that
  $\reachfusion{\funsplit{\structures}}$ contains structures with
  arbitrarily large square grid minors, it is not treewidth-bounded,
  contradicting (\ref{it1:lemma:btw-external-fusion}).

  \begin{figure}[htbp]
    \begin{center}
    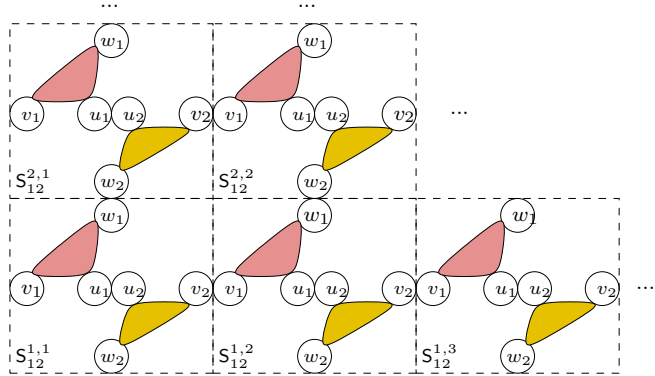
    \caption{\label{fig:grids} The principle of grid construction}
    \end{center}
  \end{figure}

  \skipnoindent ``(\ref{it2:lemma:btw-external-fusion}) $\Rightarrow$
    (\ref{it3:lemma:btw-external-fusion})'' We define a RGB color scheme
  by selecting:
  \[
  \bluecols = \set{ \ucolor \in \allcols ~|~ \mset{\ucolor,\ucolor,\ucolor}
    \in \kmcolabs{3}{(\reachfusion{\funsplit{\structures}})}}
  \]
  Since (\ref{it2:lemma:btw-external-fusion}) holds, this is a valid
  definition for $\bluecols$, which induces a partitioning of the remaining colors into $\greencols$ and
  $\redcols$.  
  We show that $\funsplit{\structures}$ is conforming to
  this RGB partitioning, by checking the two points of
  \autoref{def:conforming-structures}: \begin{itemize}[left=.5\parindent]
  \item[(\ref{it1:def:conforming-structures})] Let $\astruc \in
    \funsplit{\structures}$ and prove that for any two colors
    $\ucolor_1,\ucolor_2\in\allcols$, if $\mset{\ucolor_1, \ucolor_2}
    \subseteq \mcolabs{S}$ and $\ucolor_1 \in \redcols$ then
    $\ucolor_2 \in \bluecols$. Since $\ucolor_1 \in \redcols$, there
    must exists a color $\ucolor_1' \in \bluecols$, such that
    $\ucolor_1 \cap \ucolor_1' = \emptyset$, by
    \autoref{def:rgb-color-scheme}. By the definition of $\bluecols$
    in our RGB-color scheme, this further implies $\mset{\ucolor_1',
      \ucolor_1',\ucolor_1'} \in
    \kmcolabs{3}{(\reachfusion{\funsplit{\structures}})}$.
    Henceforth, there exists a structure $\astruc' \in
    \reachfusion{\funsplit{\structures}}$ such that $\mset{\ucolor_1',
      \ucolor_1',\ucolor_1'} \subseteq \mcolabs{\astruc'}$.  We can
    now use $\astruc'$ and three disjoint copies of $\astruc$ to build
    a new structure $\astruc''$ by gluing progressively, each one of
    the three elements of color $\ucolor_1'$ in $\astruc'$ to the
    element of color $\ucolor_1$ of $\astruc$.  Then, by construction,
    the structure $\astruc''$ will also contain three elements of
    color $\ucolor_2$, one from each disjoint copy of $\astruc$.
    Therefore, $\mset{\ucolor_2, \ucolor_2, \ucolor_2} \in
    \mcolabs{\astruc''}$ and because $\astruc'' \in
    \reachfusion{\funsplit{\structures}}$ this implies
    $\mset{\ucolor_2, \ucolor_2, \ucolor_2} \in
    \kmcolabs{3}{(\reachfusion{\funsplit{\structures}})}$ and
    therefore $\ucolor_2 \in \bluecols$.
  \item[(\ref{it2:def:conforming-structures})] By contradiction, let $\astruc
    \in \reachfusion{\funsplit{\structures}}$ be such that $
    \mcolabs{\astruc} \sqcap \greencols \not\subseteq \mset{ \ucolor,
      \ucolor ~|~ \ucolor \in \greencols}$.  Then there exists
    $\ucolor' \in (\mcolabs{\astruc} \sqcap \greencols) \setminus
    \mset{ \ucolor, \ucolor ~|~ \ucolor \in \greencols}$, i.e.,
    $\ucolor' \in \greencols$ and $\mset{\ucolor',\ucolor',\ucolor'}
    \subseteq \mcolabs{\astruc}$. The latter implies
    $\mset{\ucolor',\ucolor',\ucolor'} \in \kmcolabs{3}{\astruc}
    \subseteq \kmcolabs{3}{(\reachfusion{\funsplit{\structures}})}$.
    But this implies $\ucolor' \in \bluecols$ according to the
    definition of the RGB color scheme, contradicting $\ucolor' \in
    \greencols$.
\end{itemize}

  \skipnoindent ``(\ref{it3:lemma:btw-external-fusion}) $\Rightarrow$
    (\ref{it1:lemma:btw-external-fusion})'' By
  \autoref{cor:twb-conforming-structures},
  $\reachfusion{\funsplit{\structures}}$ is treewidth bounded.  Then, by
  \autoref{lemma:split-fusion}, $\funsplit{\reachfusion{\structures}}$
  is treewidth bounded, thus $\reachfusion{\structures}$ is treewidth
  bounded.
\end{proof}

A first consequence of this result is the equivalence between the
treewidth boundedness of the sets $\sidsem{\apred}{\asid}$ and
$\reachfusion{\csem{\apred}{\asid}}$. The following lemma establishes
this equivalence, by providing the missing direction to
\autoref{lemma:external-fusion}:

\begin{lem}\label{lemma:internal-external-fusion}
  Given a SID $\asid$ and a nullary predicate symbol $\apred$,
  $\reachfusion{\csem{\apred}{\asid}}$ is treewidth bounded only if
  $\ireachfusion{\csem{\apred}{\asid}}$ is treewidth bounded.
\end{lem}
\begin{proof}
  For any set of structures $\structures$, we have
  $\ireachfusion{{\structures}} =
  \intfusion{\reachfusion{\structures}}$, because the operations of
  internal and external fusion commute, namely
  $\extfusion{\intfusion{\astruc_1}}{\astruc_2} \subseteq
  \intfusion{\extfusion{\astruc_1}{\astruc_2}}$, for any structures
  $\astruc_1, \astruc_2$. By \autoref{lemma:btw-external-fusion},
  $\reachfusion{\csem{\apred}{\asid}}$ is treewidth-bounded only if
  $\funsplit{\csem{\apred}{\asid}}$ conforms to an RGB color
  scheme. Then,
  $\intfusion{\reachfusion{\funsplit{\csem{\apred}{\asid}}}} =
  \ireachfusion{\funsplit{\csem{\apred}{\asid}}} =
  \funsplit{\ireachfusion{\csem{\apred}{\asid}}}$ is
  treewidth-bounded, by \autoref{cor:twb-conforming-structures} and
  \autoref{lemma:split-fusion}. Thus,
  $\ireachfusion{\csem{\apred}{\asid}}$ is treewidth bounded.
\end{proof}

Our algorithm that decides the treewidth boundedness of a set
$\reachfusion{\csem{\apred}{\asid}}$ checks whether the set
$\kmcolabs{3}{(\reachfusion{\funsplit{\csem{\apred}{\asid}}})}$ meets
condition (\ref{it2:lemma:btw-external-fusion}) of
\autoref{lemma:btw-external-fusion}. For this check to be effective,
the latter set must be constructed in finite time from the description
of $\asid$ and $\apred$, provided as input. This construction proceeds
in three consecutive stages. First, we show that, for any set
$\structures$ of structures, the $k$-color abstraction
$\kmcolabs{k}(\reachfusion{\structures})$ can be built from
$\kmcolabs{k}{\structures}$ by an effectively computable abstract
operator. Second, we build a SID $\csid$ and a nullary predicate
$\ppred$, such that $\csem{\ppred}{\csid} =
\funsplit{\csem{\apred}{\asid}}$, i.e., it encodes the set of
maximally connected substructures from some canonical $\asid$-model of
$\apred$
(\autoref{subsec:maximally-connected-substructures}). Finally, we
compute the $k$-multiset color abstraction of $\csem{\ppred}{\csid}$
(\autoref{subsec:external-fusion-k-multiset}). We end this section
with a proof of decidability for the treewidth boundedness problem for
expandable SIDs (\autoref{thm:expandable-tb}).

\subsection{Color abstractions of externally fused sets}
\label{subsec:k-multiset-color-abstraction}

We describe now the effective construction of a $k$-multiset
abstraction $\kmcolabs{k}{(\reachfusion{\structures})}$ from the
abstraction $\kmcolabs{k}{\structures}$ of a set $\structures$ of
structures, for a given integer $k\geq1$. First, as we are interested
only in $k$-multisets color abstractions, we can restrict external
fusion to $1$-generated matchings, with no loss of generality.

\begin{defi}\label{def:external-fusion-one}
  The \emph{single-pair external fusion} of disjoint structures
  $\astruc_1=(\univ_1,\struc_1)$ and $\astruc_2=(\univ_2,\struc_2)$ is
  the external fusion (\autoref{def:external-fusion}) induced by
  $1$-generated matchings. We denote by
  $\extfusionone{\astruc_1}{\astruc_2}$ the set of structures obtained
  by single-pair external fusion of $\astruc_1$ and $\astruc_2$.  For
  a set of structures $\structures$, we denote by
  $\reachfusionone{\structures}$ the closure of $\structures$ under
  single-pair external fusions.
\end{defi}
In general, the single-pair external fusion is strictly less
expressive than external fusion, yet it produces the same $k$-multiset
color abstractions:

\begin{lem}\label{lemma:reach-fusion-one}
  $\kmcolabs{k}{(\reachfusion{\structures})} =
  \kmcolabs{k}{(\reachfusionone{\structures})}$ for any set
  $\structures$ of structures and integer $k\geq1$.
\end{lem}
\begin{proof}
  ``$\kmcolabs{k}{(\reachfusionone{\structures})} \subseteq
  \kmcolabs{k}{(\reachfusion{\structures})}$'' This direction follows
  directly from $\reachfusionone{\structures} \subseteq
  \reachfusion{\structures}$. ``$\kmcolabs{k}{(\reachfusion{\structures})}
  \subseteq \kmcolabs{k}{(\reachfusionone{\structures})}$'' We prove
  the stronger property:
\[
  \forall \astruc \in \reachfusion{\structures}.~
  \exists \astruc' \in \reachfusionone{\structures}.~
  \mcolabs{\astruc} \subseteq \mcolabs{\astruc'}
\]
By induction on the derivation of
$\astruc\in\reachfusion{\structures}$ from $\structures$.

\skipnoindent \underline{Base case:} Assume $\astruc \in \structures$.
Then $\astruc' = \astruc$ satisfies the property.

\skipnoindent \underline{Induction step:} Assume $\astruc = (\astruc_1
\comp \astruc_2)_{/\approx}$ for some $\astruc_1, \astruc_2 \in
\reachfusion{\structures}$ and some equivalence relation $\approx$,
defined as $\set{(u_{1i}, u_{2i}) \mid i \in \interv{1}{n}}^=$, that
conforms to the requirements of external fusion for $\astruc_1,
\astruc_2$.  Let $\ucolor_{1i} = \funcol{\astruc_1}(u_{1i})$,
$\ucolor_{2i} = \funcol{\astruc_2}(u_{2i})$, for all
$i\in\interv{1}{n}$. According to the definition of external fusion,
$\astruc = (\astruc_1 \comp \astruc_2)_{/\approx}$ implies
$\ucolor_{1i} \cap \ucolor_{2i} = \emptyset$ and moreover:
\[
  \mcolabs{\astruc} = \mset{(\ucolor_{1i} \cup \ucolor_{2i}) \mid i\in\interv{1}{n}} \cup
  (\mcolabs{\astruc_1} \setminus \mset{(\ucolor_{1i}) \mid i\in\interv{1}{n}}) \cup
  (\mcolabs{\astruc_2} \setminus \mset{(\ucolor_{2i}) \mid i\in\interv{1}{n}})
\]
By induction hypothesis, for $\astruc_1, \astruc_2$ there exists
$\astruc_1', \astruc_2' \in \reachfusionone{\structures}$ such that
$\mcolabs{\astruc_1} \subseteq \mcolabs{\astruc_1'}$,
$\mcolabs{\astruc_2} \subseteq \mcolabs{\astruc_2'}$. We use
$\astruc_1'$ and $n$ disjoint copies $\astruc_{2,1}', ...,
\astruc_{2,n}'$ of $\astruc_2'$ to construct $\astruc'$ with the
required property. The idea is that, for every pair $u_{1i} \approx
u_{2i}$, we fuse some element $u'_{1i}$ with color $\ucolor_{1i}$ from
$\astruc_1'$ with some element $u'_{2i}$ with color $\ucolor_{2i}$
from $\astruc_{2,i}'$. Such elements always exist, because
$\mcolabs{\astruc_1} \subseteq \mcolabs{\astruc_1'}$,
$\mcolabs{\astruc_2} \subseteq \mcolabs{\astruc_2'}$. Therefore,
consider the equivalence relations $\approx_i' = (u_{1i}', u_{2i}')^=$
for some pair of elements as above, for all $i\in\interv{1}{n}$ and
define:
\[
  \astruc' = ( \ldots ((\astruc_1' \comp \astruc_{2,1}')_{/\approx'_1} \comp
  \astruc_{2,2}')_{/\approx'_2} \comp \ldots \comp \astruc_{2,n}')_{\approx'_n}
\]
Then $\astruc' \in \reachfusionone{\structures}$ and, moreover, we have
$\mcolabs{\astruc} \subseteq \mcolabs{\astruc'}$, because:
\[
  \mcolabs{\astruc'} = \mset{(\ucolor_{1i} \cup \ucolor_{2i}) \mid i\in\interv{1}{n}} \cup
  (\mcolabs{\astruc_1'} \setminus \mset{(\ucolor_{1i}) \mid i\in\interv{1}{n}}) \cup
  \bigcup\nolimits_{i\in\interv{1}{n}} (\mcolabs{\astruc_2'} \setminus \mset{\ucolor_{2i}}) \hspace{1cm} \qedhere
\]
\end{proof}

Second, the closure $\kmcolabs{k}{(\reachfusionone{\structures})}$ can
be computed by a least fixpoint iteration of an abstract operation on
the domain of $k$-multiset color abstractions. As the later domain is
finite, this fixpoint computation is guaranteed to terminate.

\begin{defi}\label{def:multiset-color-fusion-abstraction}
  The \emph{single-pair multiset fusion} is defined below, for $M_1,
  M_2 \in \mpow{\allcols}$:
  \begin{align*}
    \absextfusionone{M_1}{M_2} \isdef
    \big\{M \in \mpow{\allcols} ~|&~ \exists \ucolor_1 \in M_1.~ \exists \ucolor_2\in M_2.~
    \ucolor_1 \cap \ucolor_2 = \emptyset, \\
    & ~~~ M = \mset{\ucolor_1 \cup \ucolor_2} \cup \bigcup\nolimits_{i=1,2} (M_i \setminus \mset{\ucolor_i}) \big\}
  \end{align*}
  Given an integer $k\geq1$, the \emph{single-pair $k$-multiset
    fusion} is defined for $M_1$, $M_2 \in \mpow{\allcols}$, such that
  $\cardof{M_1} \le k$ and $\cardof{M_2}\le k$:
  \[
  \kabsextfusionone{k}{M_1}{M_2} \isdef \set{M ~|~ \exists M' \in
    \absextfusionone{M_1}{M_2}.~ M \subseteq M',~ \cardof{M} \le k}
  \]
  For a set $\mathcal{M}$ of multisets (resp. $k$-multisets) of
  colors, let $\absreachfusionone{{\mathcal M}}$
  (resp. $\kabsreachfusionone{k}{\mathcal M}$) be the closure of
  $\mathcal{M}$ under taking single-pair fusion on multisets
  (resp. $k$-multisets).
\end{defi}

\begin{lem}\label{lemma:reach-fusion-one-abstraction}
  $\kmcolabs{k}{(\reachfusionone{\structures})} =
  \kabsreachfusionone{k}{\kmcolabs{k}{\structures}}$, for any set
  $\structures$ of structures and integer $k\geq1$.
\end{lem}
\begin{proof}
    Abusing notation, we write $\kmcolabs{k}{M} \isdef \set{M' ~|~ M'
    \subseteq M,~ \cardof{M'} \le k}$. Then, we have
  $\kmcolabs{k}{(\reachfusionone{\structures})} =
  \kmcolabs{k}{(\mcolabs{(\reachfusionone{\structures})})}$, by
  \autoref{def:multiset-color-abstraction}. Using
  \autoref{def:external-fusion-one} of single pair external fusion
  and \autoref{def:multiset-color-fusion-abstraction} of single pair
  fusion of multisets, we can prove that for all structures
  $\astruc_1, \astruc_2$ it holds
  $\mcolabs{(\extfusionone{\astruc_1}{\astruc_2})} =
  \absextfusionone{\mcolabs{\astruc_1}}{\mcolabs{\astruc_2}}$.  This
  immediately extends to their respective closure, henceforth,
  $\mcolabs{(\reachfusionone{\structures})} =
  \absreachfusionone{\mcolabs{\structures}}$. Henceforth, we are left
  with proving that
  $\kmcolabs{k}{(\absreachfusionone{\mcolabs{\structures}})} =
  \kabsreachfusionone{k}{\kmcolabs{k}{\structures}}$.

  \skipnoindent ``$\kmcolabs{k}{(\absreachfusionone{\mcolabs{\structures}})}
  \subseteq \kabsreachfusionone{k}{\kmcolabs{k}{\structures}}$'' We
  prove that, for all $M \in
  \absreachfusionone{\mcolabs{\structures}}$, we have $\kmcolabs{k}{M}
  \subseteq \kabsreachfusionone{k}{\kmcolabs{k}{\structures}}$. The
  proof goes by induction on the derivation of $M\in
  \absreachfusionone{\mcolabs{\structures}}$ from
  $\mcolabs{\structures}$.

  \skipnoindent \underline{Base case:} Assume $M
  \in \mcolabs{\structures}$.  Then $\kmcolabs{k}{M} \subseteq
  \kmcolabs{k}{(\mcolabs{\structures})} = \kmcolabs{k}{\structures}
  \subseteq \kabsreachfusionone{k}{\kmcolabs{k}{\structures}}$.

  \skipnoindent \underline{Induction step:} Assume $M \in
  \absextfusionone{M_1}{M_2}$ for some multisets of colors $M_1$,
  $M_2$ such that $\kmcolabs{k}{M_1}, \kmcolabs{k}{M_2} \subseteq
  \kabsreachfusionone{k}{\kmcolabs{k}{\structures}}$. Then, there
  exists $\ucolor_1 \in M_1$, $\ucolor_2 \in M_2$ such that $\ucolor_1
  \cap \ucolor_2 = \emptyset$ and $M = (M_1 \setminus
  \mset{\ucolor_1}) \cup (M_2 \setminus \mset{\ucolor_2}) \cup
  \mset{\ucolor_1 \cup \ucolor_2}$, by
  \autoref{def:multiset-color-fusion-abstraction}. Let $M' \in
  \kmcolabs{k}{M}$, that is, $M' \subseteq M$, $\cardof{M'} \le k$.
  We distinguish several cases: \begin{itemize}[label=$\triangleright$]
  \item $M' \subseteq M_1$ (the case $M' \subseteq M_2$ is symmetric):
    $M' \in \kmcolabs{k}{M_1}$, thus $M' \in
    \kabsreachfusionone{k}{\kmcolabs{k}{\structures}}$.
  \item $M' \not\subseteq M_i$, for $i=1,2$ and $\ucolor_1 \cup
    \ucolor_2 \not\in M$': $M'$ can be partitioned in two nonempty
    parts $M_1' \subseteq M_1$, $M_2' \subseteq M_2$ such that $M =
    M_1' \uplus M_2'$. As both parts are not empty, we have $M_1' \in
    \kmcolabs{k-1}{M_1}$, $M_2' \in \kmcolabs{k-1}{M_2}$, thus $(M_1'
    \cup \mset{\ucolor_1}) \in \kmcolabs{k}{M_1}$, $(M_2' \cup
    \mset{\ucolor_2}) \in \kmcolabs{k}{M_2}$. It is an easy check that
    $M' \in \kabsextfusionone{k}{(M_1' \cup \mset{\ucolor_1})}{(M_2'
      \cup \mset{\ucolor_2})}$. This implies $M' \in
    \kabsreachfusionone{k}{\kmcolabs{k}{\structures}}$ as both
    subterms belong to
    $\kabsreachfusionone{k}{\kmcolabs{k}{\structures}}$.
  \item $M' \not\subseteq M_i$, for $i=1,2$ and $\ucolor_1 \cup
    \ucolor_2 \in M'$: we proceed as in the previous case but
    considering a partitioning of $M' \setminus \mset{\ucolor_1 \cup
      \ucolor_2}$.  We obtain $M' \in
    \kabsreachfusionone{k}{\kmcolabs{k}{\structures}}$, as well.
  \end{itemize}

  \skipnoindent ``$\kabsreachfusionone{k}{\kmcolabs{k}{\structures}} \subseteq
    \kmcolabs{k}{(\absreachfusionone{\mcolabs{\structures}})}$'' We
  prove that, for all $k$-multiset $M'\in
  \kabsreachfusionone{k}{\kmcolabs{k}{\structures}}$, there exists $M \in
  \absreachfusionone{\mcolabs{\structures}}$, such that $M' \subseteq M$, by
  induction on the derivation of $M'$ from $\kmcolabs{k}{\structures}$.

  \skipnoindent \underline{Base case:} Assume $M'
  \in \kmcolabs{k}{\structures} =
  \kmcolabs{k}{(\mcolabs{\structures})}$. Then, there exists $M\in
  \mcolabs{S}$ such that $M' \subseteq M$.  Obviously, $M \in
  \absreachfusionone{\mcolabs{\structures}}$.

  \skipnoindent \underline{Induction step:} Assume $M' \in
  \kabsextfusionone{k}{M_1'}{M_2'}$ for some $k$-multisets of colors
  $M_1', M_2' \in \kabsreachfusionone{k}{\kmcolabs{k}{\structures}}$.
  By the inductive hypothesis, there exists multisets $M_1, M_2 \in
  \absreachfusionone{\mcolabs{\structures}}$ such that $M_1' \subseteq
  M_1$, $M_2' \subseteq M_2$.  Since $M_1', M_2'$ can be composed such
  that to obtain (a superset of) the multiset $M'$, one can use
  precisely the same pairs of colors to compose $M_1, M_2$ and
  henceforth to obtain the multiset $M \in
  \absreachfusionone{\mcolabs{\structures}}$, which is the superset of
  $M'$.
\end{proof}

\subsection{Maximally connected substructures}
\label{subsec:maximally-connected-substructures}

Since we consider canonical models, we can assume w.l.o.g. that the
given SID $\asid$ contains no disequalities (such atoms are trivially
unsatisfiable or valid). We represent the set of maximally connected
structures $\funsplit{\csem{\apred}{\asid}}$ as a set of canonical
models $\csem{\ppred}{\csid}$, for a fresh nullary predicate $\ppred$
and a SID $\csid$, whose construction is described next.

Given a qpf formula $\psi$, we define $\connof{\psi} \subseteq
\fv{\psi} \times \fv{\psi}$ to be the least equivalence relation such
that $(y,z) \in \connof{\psi}$ if
$\arel(x_1,\ldots,x_{\arityof{\arel}})$ occurs in $\psi$ and
$y,z\in\set{x_1,\ldots,x_{\arityof{\arel}}}$, for some
$\arel\in\relations$. Intuitively, $\connof{\psi}$ consists of the
pairs of free variables of $\psi$ that are connected by a path (see
\autoref{def:connected}) in each canonical model of $\psi$.

Let $\bpred(y_1, \ldots, y_{\arityof{\bpred}})$ be a predicate atom,
$J = \set{j_1,\ldots,j_p} \subseteq \interv{1}{\arityof{\bpred}}$ be a
set of indices ordered as $j_1 \le \ldots \le j_p$, $\xi \subseteq J
\times J$ be an equivalence relation and $\bpred^\xi$ be a fresh
predicate of arity $p$. In particular, $\arityof{\bpred^\xi}=0$ if
$\xi=\emptyset$ is the empty relation. We define the shorthands:
\begin{align*}
  \fvof{J}{\bpred(y_1, \ldots, y_{\arityof{\bpred}})} \isdef \set{y_j ~|~ j \in J} \hspace*{8mm}
  \xi(\bpred(y_1, \ldots, y_{\arityof{\bpred}})) \isdef & \set{(y_j,y_k) \mid (j,k)\in\xi} \\
  \bpred(y_1, \ldots, y_{\arityof{\bpred}})_{/\xi} \isdef & \bpred^\xi(y_{j_1}, \ldots, y_{j_p})
\end{align*}
We build definitions for the predicate atoms $\bpred^\xi(x_{j_1},
\ldots, x_{j_p})$, by ``narrowing'' the definitions of $\bpred(x_1,
\ldots, x_{\arityof{\bpred}})$, respectively.  More precisely, every
structure $\astruc \in \csem{\exclof{\bpred^\xi(y_{j_1}, \ldots,
    y_{j_p})}}{\asid}$ will correspond to a set of maximally
connected substructures in a structure $\astruc' \in
\csem{\exclof{\bpred(y_1,\ldots,y_{\arityof{\bpred}})}}{\asid}$ such that,
moreover \begin{enumerate*}[(i)]
\item for every such substructure, there exists an element
  associated with $y_j$, for some $j \in J$, and
\item $y_j$ and $y_k$ are mapped to elements from the same connected
  substructure of $\astruc'$ if and only if $(j,k)\in\xi$.
\end{enumerate*}
In other words, the equivalence relation $\xi$ is used to summarize
the information about the maximally connected structures from any
canonical model of $\bpred^\xi$.

We describe the construction of $\csid$ next. Consider a rule of
$\asid$ of the form:
\begin{equation} \label{eq:asidrule}
  \bpred_0(x_1,\ldots,x_{\arityof{\bpred_0}}) \leftarrow \exists y_1
  \ldots \exists y_m ~.~ \psi * \Bigstar\nolimits_{i=1}^\ell
  \bpred_i(z_{i,1},\ldots,z_{i,\arityof{\bpred_i}})
\end{equation}
formul{\ae} $\psi'$, $\psi''$, sets $J_{i}
\uplus\overline{J}_{i}=\interv{1}{\arityof{\bpred_{i}}}$, equivalence
relations $\xi_{i} \subseteq J_{i} \times J_{i}$, for all $i \in
\interv{1}{\ell}$, an equivalence relation
$\Xi\subseteq\big(\set{x_1,\ldots,x_{\arityof{\bpred_0}}}\cup\set{y_1,\ldots,y_m}\big)
\times
\big(\set{x_1,\ldots,x_{\arityof{\bpred_0}}}\cup\set{y_1,\ldots,y_m}\big)$,
such that the following hold: \begin{enumerate}
\item\label{it1:mcsid} $\psi = \psi' * \psi''$ modulo a reordering of
  atoms, such that $\fv{\psi'} \cap \fv{\psi''} = \emptyset$,
\item\label{it2:mcsid} $\fvof{J_{i}}{\bpred_i(z_{i,1},\ldots,z_{i,\arityof{\bpred_{i}}})} \cap
  \fv{\psi''} = \emptyset$ and \\
  $\fvof{\overline{J}_{i}}{\bpred_i(z_{i,1},\ldots,z_{i,\arityof{\bpred_{i}}})}
  \cap \fv{\psi'} = \emptyset$, for all $i \in \interv{1}{\ell}$,
\item\label{it3:mcsid} $\Xi = \big(\connof{\psi'} \cup \bigcup_{i=1}^\ell
  \xi_{i}(\bpred_{i}(z_{i,1},\ldots,z_{i,\arityof{\bpred_{i}}})) \big)^=$.
\end{enumerate}
Intuitively, the conditions (\ref{it1:mcsid})--(\ref{it3:mcsid}) above
guarantee that the models of
$\bpred_{i}(z_{i,1},\ldots,z_{i,\arityof{\bpred_{i}}})_{/\xi_{i}}$
(recall, these are sets of maximally connected structures) compose
with a model of $\psi'$ without losing neither connectivity nor
maximality, in the context of the rule (\ref{eq:asidrule}). At this
point, we distinguish two cases:
\begin{itemize}[label=$\triangleright$]
\item If there exist sets
  $J_0\uplus\overline{J}_0=\interv{1}{\arityof{\bpred_0}}$, $J_0 \neq
  \emptyset$ and an equivalence relation $\xi_0 \subseteq J_0 \times
  J_0$, such that:\begin{enumerate} \setcounter{enumi}{3}
  \item\label{it4:mcsid} $\fvof{J_0}{\bpred_0(x_1, \ldots, x_{\arityof{\bpred_0}})} \cap \fv{\psi''} = \emptyset$ and
    $\fvof{\overline{J}_0}{\bpred_0(x_1, \ldots, x_{\arityof{\bpred_0}})} \cap \fv{\psi'} = \emptyset$,
  \item\label{it5:mcsid} $\fvof{J_0}{\bpred_0(x_1, \ldots, x_{\arityof{\bpred_0}})}
    \cap
    \fvof{\overline{J}_{i}}{\bpred_{i}(z_{i,1},\ldots,z_{i,\arityof{\bpred_{i}}})}
    = \emptyset$ and \\
    $\fvof{\overline{J}_0}{\bpred_0(x_1, \ldots,
      x_{\arityof{\bpred_0}})} \cap
    \fvof{J_{i}}{\bpred_{i}(z_{i,1},\ldots,z_{i,\arityof{\bpred_{i}}})} =
    \emptyset$, for all $i\in\interv{1}{\ell}$,
  \item\label{it6:mcsid} for all $y \in \big(\fv{\psi'} \cup \bigcup_{i=1}^\ell
    \fvof{J_{i}}{\bpred_{i}(z_{i,1},\ldots,z_{i,\arityof{\bpred_{i}}})}\big) \cap
    \set{y_1,\ldots, y_m}$ \\ there exists $x\in\fvof{J_0}{\bpred_0(x_1,
      \ldots, x_{\arityof{\bpred_0}})}$, such that $(x,y) \in \Xi$,
    %
  \item \label{it7:mcsid}$\xi_0(\bpred_0(x_1, \ldots, x_{\arityof{\bpred_0}})) =
    \proj{\Xi}{\set{x_1,\ldots,x_{\arityof{\bpred_0}}}} \cup \set{(x,x)~|~ x\in
      \fvof{J_0}{\bpred_0(x_1, \ldots, x_{\arityof{\bpred_0}})}}$
  \end{enumerate}
  then we add to $\csid$ the following rule:
  \begin{equation}\label{eq:mcsid-middle}
    \bpred_0(x_1, \ldots, x_{\arityof{\bpred_0}})_{/\xi_0} \leftarrow \exists y_1 \ldots
    \exists y_m ~.~ \psi' * \Bigstar\nolimits_{i \in \interv{1}{\ell},J_i \neq \emptyset}
    \bpred_{i}(z_{i,1},\ldots,z_{i,\arityof{\bpred_{i}}})_{/\xi_{i}}
  \end{equation}
  Intuitively, the conditions (\ref{it4:mcsid})--(\ref{it7:mcsid})
  identify the set $J_0$ and the equivalence relation $\xi_0$ for
  which the result of the composition becomes a model of
  $\bpred_0(x_1, \ldots, x_{\arityof{\bpred_0}})_{/\xi_0}$.
  Altogether, these lead to the definition of the rules of the form
  \ref{eq:mcsid-middle} which propagate the construction of maximally
  connected structures in $\csid$.

\item If \begin{enumerate*}
\setcounter{enumi}{7}
\item\label{it8:mcsid} $\Xi$ defines an unique equivalence class, and
\item\label{it9:mcsid} $(x,x) \not\in\Xi$ for all $x \in \set{x_1,\ldots,x_{\arityof{\bpred_0}}}$
\end{enumerate*}
  then we add to $\csid$ the following rule:
  \begin{equation}\label{eq:mcsid-begin}
    \ppred \leftarrow \exists y_1 \ldots
  \exists y_m ~.~ \psi' * \Bigstar\nolimits_{i \in \interv{1}{\ell},J_i \neq \emptyset}
  \bpred_i(z_{i,1},\ldots,z_{i,\arityof{\bpred_i}})_{/\xi_i}
  \end{equation}
  Intuitively, conditions (\ref{it8:mcsid})--(\ref{it9:mcsid}) in
  addition to (\ref{it1:mcsid})--(\ref{it3:mcsid}), ensure that by
  composing the models of
  $\bpred_{i}(z_{i,1},\ldots,z_{i,\arityof{\bpred_{i}}})_{/\xi_{i}}$
  with a model of $\psi'$ we obtain a single maximally connected
  structure in the context of the rule (\ref{eq:asidrule}), which is
  moreover not referred by any of the parameters of $\bpred_0$.
  Henceforth, the result of this composition is actually a model of
  $\funsplit{\csem{\apred}{\asid}}$ and consequently is added as a
  model of $\ppred$ by the rules of the form \ref{eq:mcsid-begin}.
\end{itemize}
Recall that $\asid$ was assumed to be equality-free
(\autoref{def:eq-free}) and all-satisfiable for $\apred$
(\autoref{def:all-sat}). Moreover, we assume that every predicate
defined by a rule of $\asid$ occurs on some complete $\asid$-unfolding
of $\apred$. Obviously, the rules that do not meet this requirement
can be removed from $\asid$ without changing
$\csem{\apred}{\asid}$. The following lemma shows that the set of
canonical $\csid$-models of $\ppred$ is a correct representation of
the set of canonical $\asid$-models of $\apred$:

\begin{lem}\label{lemma:sid-connected}
  For each equality-free SID $\asid$, which is all-satisfiable for a
  nullary predicate symbol $\apred$, one can effectively build a SID
  $\csid$ and a nullary predicate $\ppred$, such that
  $\funsplit{\csem{\apred}{\asid}} = \csem{\ppred}{\csid}$.
\end{lem}
\begin{proof}
  For space reasons, this proof is given in \autoref{app:maximally-connected-substructures}.
\end{proof}

\subsection{Color abstractions of canonical models}
\label{subsec:external-fusion-k-multiset}

We compute the $k$-multiset color abstraction
$\kmcolabs{k}{(\csem{\ppred}{\csid})}$ by a least fixpoint iteration
in a finite abstract domain, defined directly from the rules in the
SID. The elements of the domain are composed of the colors of
parameter values and the $k$-multiset color abstraction of the
elements not referenced by parameters.

A \emph{$k$-bounded color triple} $\tuple{X, c, M}$ consists of a
finite set of variables $X \subseteq \vars$, a mapping $c : X
\rightarrow \allcols$, and a multiset $M \in \mpow{\allcols}$, such
that $\cardof{M} \le k$.  Since $X$ and $\relations$ are finite,
there are finitely many color triples. The following operations on
color triples are lifted to sets, as usual: \begin{description}
\item[\emph{$k$-composition}] $\tuple{X_1, c_1, M_1} \kcomp \tuple{X_2, c_2, M_2} \isdef$
  \begin{align*}
    \{ \tuple{X_1 \cup X_2, c_{12}, M_{12}} ~|~
    & c_{12}(x) = c_1(x) \uplus c_2(x) \mbox{, for all } x \in X_1 \cap X_2,~ \\
    & c_{12}(x) = c_i(x) \mbox{ for all } x\in X_i \setminus X_{3-i} \mbox{, for all } i\in\set{1,2},\\
    & M_{12} \subseteq M_1 \cup M_2,~ \cardof{M_{12}} \le k \}
  \end{align*}
  This operation is undefined, if $c_1(x)\cap c_2(x)\neq\emptyset$,
  for some $x \in X_1 \cap X_2$.
\item[\emph{substitution}] $\tuple{X, c, M}[s] \isdef \tuple{Y, c
  \circ s, M}$, for any bijection $s : Y \rightarrow X$
\item[\emph{$k$-projection}] $\kproj{(X, c, M)}{Y} \isdef
  \set{\tuple{Y, \proj{c}{Y}, M'} ~|~ M' \subseteq M \cup
    \mset{c(x)~|~ x\in X \setminus Y},~ \cardof{M'} \le k}$, for $Y
  \subseteq X$.
\end{description}
For a qpf formula $\psi$, let
$\colorof{\psi}\isdef\tuple{\fv{\psi},\lambda x \in \fv{\psi} ~.~
  \set{\arel \in \relations ~|~ \arel(x,\ldots,x) \mbox{ occurs in }
    \psi},\emptyset}$.
Given a predicate $\bpred$, we denote by $\kabssem{k}{\bpred}{\csid}$
the least sets of $k$-bounded color triples over the variables
$x_1,\ldots,x_{\arityof{\bpred}}$, the satisfies the following constraints:
\begin{align} \label{eq:asidconstraint}
  \kabssem{k}{\bpred_0}{\csid} \supseteq \kproj{\big(\colorof{\psi} \kcomp
    \Kcomp\nolimits_{i\in\interv{1}{\ell}} \kabssem{k}{\bpred_i}{\csid}
         [x_1/z_{i,1},\ldots,x_{\arityof{\bpred_i}}/z_{i,\arityof{\bpred_i}}]\big)}
          {\set{x_1,\ldots,x_{\arityof{\bpred_0}}}}
\end{align}
one for each rule of $\csid$ of the form (\ref{eq:asidrule}). Note
that the operations on sets of color triples are monotonic and the
sets thereof are finite, since the arity of predicates is finite and
$k$ is fixed. Henceforth, the least solution can be computed in finite
time by an ascending Kleene iteration. For a $n$-ary relation $R$, we
denote by $\projrel{R}{k}$ the set of elements that occur on the
$k$-th position in a tuple from $R$.

\begin{lem}\label{lemma:sid-k-multiset-abstraction}
  $\kmcolabs{k}{(\csem{\ppred}{\csid})}=\projrel{\kabssem{k}{\ppred}{\csid}}{3}$,
  for any $k\geq1$, SID $\csid$ and nullary predicate $\ppred$.
\end{lem}
\begin{proof}
  For space reasons, this proof is given in \autoref{app:sid-k-multiset-abstraction}.
\end{proof}

\subsection{The expandable treewidth boundedness problem}

We end this section with a proof of decidability for the treewidth
boundedness problem of the sets of $\asid$-models of a nullary
predicate $\apred$, provided that $\asid$ is an expandable SID (see
\autoref{def:expandable}). In case such a bound exists, we provide
the optimal upper bound, in terms of the input SID $\asid$.  We recall
that $\maxvarinruleof{\asid}$ is the maximum number of variables that occur,
either free or bound by an existential quantifier, in some rule from
$\asid$ (\autoref{lemma:canonical-btw}).

\begin{thm}\label{thm:expandable-tb}
  There exists an algorithm that decides, for each expandable SID
  $\asid$ and nullary predicate $\apred$, whether the set
  $\sidsem{\apred}{\asid}$ has bounded treewidth. If, moreover, this
  is the case then $\twof{\sidsem{\apred}{\asid}} \le
  \maxvarinruleof{\asid}$.
\end{thm}
\begin{proof}
  Because $\asid$ is expandable, by \autoref{lemma:external-fusion}
  and \autoref{lemma:internal-external-fusion}, $\sidsem{\apred}{\asid}$
  is treewidth bounded iff $\reachfusion{\csem{\apred}{\asid}}$ is
  treewidth bounded. By \autoref{lemma:btw-external-fusion},
  $\reachfusion{\csem{\apred}{\asid}}$ is treewidth bounded iff
  $\mset{\ucolor_1,\ucolor_1,\ucolor_1},\mset{\ucolor_2,\ucolor_2,\ucolor_2}
  \in \kmcolabs{3}{(\reachfusion{\funsplit{\csem{\apred}{\asid}}})}
  \Rightarrow \ucolor_1 \cap \ucolor_2 \neq \emptyset$, for all
  $\ucolor_1,\ucolor_2\in\allcols$. The latter condition can be
  effectively checked by computing the finite set
  $\kmcolabs{3}{(\reachfusion{\funsplit{\csem{\apred}{\asid}}})}$. By
  \autoref{lemma:reach-fusion-one} and
  \autoref{lemma:reach-fusion-one-abstraction}, we have
  $\kmcolabs{3}{(\reachfusion{\funsplit{\csem{\apred}{\asid}}})} =
  \kabsreachfusionone{k}{\kmcolabs{k}{\funsplit{\csem{\apred}{\asid}}}}$.
  By \autoref{lemma:sid-connected}, one can effectively build a SID
  $\csid$ and a nullary predicate $\ppred$, such that
  $\funsplit{\csem{\apred}{\asid}} = \csem{\ppred}{\csid}$. Moreover,
  by \autoref{lemma:sid-k-multiset-abstraction}, we obtain
  $\kmcolabs{3}{(\csem{\ppred}{\csid})}=\projrel{\kabssem{3}{\ppred}{\csid}}{3}$,
  hence $\kmcolabs{3}{(\reachfusion{\funsplit{\csem{\apred}{\asid}}})}
  = \projrel{\kabssem{3}{\ppred}{\csid}}{3}$, which is effectively
  computable by a ascending Kleene iteration in the finite domain of
  $3$-bounded color triples.

  For the upper bound, since $\sidsem{\apred}{\asid} \subseteq
  \ireachfusion{\csem{\apred}{\asid}}$
  (\autoref{lemma:strong-internal-fusion}), we have
  $\twof{\sidsem{\apred}{\asid}} \le
  \twof{\ireachfusion{\csem{\apred}{\asid}}}$. By
  \autoref{cor:twb-conforming-structures},
  $\twof{\sidsem{\apred}{\asid}} \le \twof{\csem{\apred}{\asid}}+1$
  and, by \autoref{lemma:canonical-btw}, we obtain
  $\twof{\sidsem{\apred}{\asid}} \le \maxvarinruleof{\asid}$.
\end{proof}

The bound given by \autoref{thm:expandable-tb} is optimal, as shown by
the following example:

\begin{exa}\label{ex:expandable-tb}
  Let us consider the following SID:
  \[\asid = \left\{\begin{array}{rl}
  \apred \leftarrow & \exists y_1 \exists y_2 ~.~ \mathsf{a}(y_1) * \mathsf{e}(y_1,y_2) * \apred \\
  \apred \leftarrow & \emp
  \end{array}\right.\]
  This SID is expandable for $\apred$, because any canonical
  $\asid$-model $\astruc=(\univ,\struc)$ of $\apred$ consists of a set
  of pairs $(u_1,u_2) \in \struc(\mathsf{e})$, such that $u_1 \in
  \struc(\mathsf{a})$ and $u_2 \not\in \struc(\mathsf{a})$. Hence, any
  sequence of canonical $\asid$-models of $\apred$ can be embedded as
  substructures in a canonical $\asid$-model of $\apred$. Moreover,
  $\maxvarinruleof{\asid}=2$ and any cyclic list of $\mathsf{e}$-related
  adjacent elements labeled by $\mathsf{a}$ is a $\asid$-model of
  $\apred$ of treewidth $2$.
\end{exa}


\section{The Reduction to Expandable Sets of Inductive Definitions}
\label{sec:general-tb}

This section completes the proof of decidability of the treewidth
boundedness problem \tbsl, by showing a reduction to the decidable
treewidth boundedness problem for expandable SIDs
(\autoref{thm:expandable-tb}). Moreover, an analysis of this reduction
allows to compute upper bounds on the treewidth of the set of models
of an \slr\ sentence, provided that such a bound exists. The core of
the reduction is the following lemma:

\begin{lem}\label{lemma:expansion}
  Let $\asid$ be a SID and $\apred$ be a nullary predicate. Then, one
  can build finitely many SIDs $\csid_1, \ldots, \csid_n$, that are
  expandable for a nullary predicate $\bpred$, such that
  $\sidsem{\apred}{\asid}$ is treewidth bounded iff each
  $\sidsem{\bpred}{\csid_i}$ is treewdith bounded, for $i \in
  \interv{1}{n}$.
\end{lem}

The rest of this section is concerned with the proof of this lemma.
For technical reasons, the construction of expandable SIDs with an
equivalent treewidth boundedness problem uses a representation of the
SID as a tree automaton (\autoref{sec:automata}). This representation
allows to distinguish the purely structural aspects, related to the
dependencies between rules, from details related to the flow of
parameters. An important class of automata distinguish between the
so-called \emph{$1$-transitions}, that occur exactly once, from the
\emph{$\infty$-transitions}, that may occur any number of times on an
accepting run. These automata are called \emph{choice-free}
(\autoref{def:choice-free}). Each automaton admits a finite
choice-free decomposition that preserves its language
(\autoref{lemma:choice-free}).

Next, we consider tree automata whose alphabets are finite sets of qpf
formul{\ae} (\autoref{sec:formulae-alphabets}). The trees recognized
by these automata are representations of the predicate-free
formul{\ae} produced by the complete $\asid$-unfoldings of
$\apred$. Since we assume that $\asid$ is all-satisfiable for $\apred$
(\autoref{def:all-sat}), each accepting run of an automaton
``produces'' a canonical $\asid$-model of $\apred$
(\autoref{lemma:sid-ta}).

Furthermore, we define \emph{persistent} variables, whose values are
carried along each sequence of $\infty$-transitions of a choice-free
automaton (\autoref{def:profile}). Identifying and removing the
persistent variables from a choice-free automaton with an alphabet of
qpf formul{\ae} constitutes an important ingredient of the
construction, because of point (\ref{it3:def:expandable}) of
\autoref{def:expandable}, that requires the embedded canonical models
of an expandable SID to be placed sufficiently far away one from
another. In particular, this guarantees that the coloring
(\autoref{def:color-extraction-function}) of an element from an
embedded substructure does not change in the larger structure. The
effective transformation of a choice-free automaton over an alphabet
of qpf formul{\ae} into an automaton without persistent variables is
described in \autoref{sec:persistent-variables-elimination}. This
transformation does not preserve the language, nor the set of models
corresponding to the trees recognized by the automaton, but is shown
to preserve the existence of a (computable) bound on the treewidths of
these models.

\subsection{Tree Automata}
\label{sec:automata}

Let $\alphabet$ be a \emph{ranked alphabet}, each symbol $a \in
\alphabet$ having an associated integer \emph{rank}
$\rankof{a}\geq0$. The elements of $\nat_+^*$ are finite sequences of
strictly positive natural numbers, called \emph{positions}. We write
$pq$ for the concatenation of $p,q\in\nat^*$ and $q\cdot
P\isdef\set{qp \mid p \in P}$, for $P \subseteq \nat^*$. A
\emph{ranked tree} is a finite partial function $t : \nat^*
\rightarrow \alphabet$, such that the set $\dom{t}$ is
\emph{prefix-closed}, i.e., for each $p \in \dom{t}$, if $q$ is a
prefix of $p$, then $q \in \dom{t}$, and \emph{sibling-closed}, i.e.,
$\set{i \in \nat \mid pi \in \dom{t}} = \set{1, \ldots,
  \rankof{t(p)}}$, for all $p \in \dom{t}$. The \emph{frontier} of $t$
is the set $\frof{t} \isdef \set{p \in \dom{t} \mid p1 \not\in
  \dom{t}}$. We denote by $\subtree{t}{p}$ the subtree of $t$ at
position $p\in\dom{t}$ i.e., $\subtree{t}{p}$ is the tree such that
$\dom{\subtree{t}{p}}=\set{q\in\nat^* \mid pq\in\dom{t}}$ and
$\subtree{t}{p}(q)=t(pq)$, for each $q\in\dom{\subtree{t}{p}}$. A tree
$u$ is \emph{embedded in $t$ at position $p \in \dom{t}$} iff
$pq\in\dom{t}$ and $u(q)=t(pq)$, for each $q \in \dom{u}$.

\begin{defi}\label{def:tree-automaton}
An \emph{($\alphabet$-labeled tree) automaton} is $\mathcal{A} =
(\alphabet,\states,\initstates,\trans)$, where $\states$ is a finite
set of states, $\initstates \subseteq \states$ is a set of initial
states (if $\initstates$ is a singleton, we denote it by $\initstate
\in \states$), $\trans$ is a finite set of transitions $\tau : q_0
\arrow{a}{} (q_1,\ldots,q_{\rankof{a}})$. For a transition $\tau : q_0
\arrow{a}{} (q_1,\ldots,q_\ell) \in \trans$, let $\pre{\tau}\isdef
q_0$ be the source and $\post{\tau} \isdef \mset{q_1,\ldots,q_\ell}$
the multiset of targets of $\tau$. For a set of transitions
$T\subseteq\trans$, let $\pre{T}\isdef \set{\pre{\tau}\mid\tau\in T}$
and $\post{T}\isdef\bigcup_{\tau\in T} \post{\tau}$. For a set of
states $S\subseteq\states$, let $\pre{S} \isdef \set{\tau \mid
  \pre{\tau} \not\in S,~ \post{\tau} \cap S \neq \emptyset}$,
$\post{S} \isdef \set{\tau \mid \pre{\tau} \in S,~ \post{\tau} \cap S
  = \emptyset}$ and $\prepost{S} \isdef \set{\tau \mid \pre{\tau} \in
  S,~ \post{\tau} \cap S \neq \emptyset}$.
\end{defi}

The following notions concern the structure of automata. The relation
$\reach \ \subseteq \states \times \states$ is defined as $q \reach
q'$ iff there exists $\tau \in \trans$ such that $q = \pre{\tau}$ and
$q' \in \post{\tau}$. A \emph{strongly connected component} (SCC) is a
maximal set $S \subseteq \states$, such that $q \reach^* q'$, for all
$q,q' \in S$. An SCC $S$ is \emph{nonlinear} iff there exists a
transition $\tau\in\prepost{S}$ such that $\cardof{\post{\tau}\cap S}
\geq 2$ and \emph{linear} otherwise. The \emph{SCC graph} of
$\mathcal{A}$ is the directed graph $\graphof{\mathcal{A}} \isdef
(\nodes,\edges)$, where $\nodes$ is the set of SCCs of $\mathcal{A}$
and $(S,S') \in \edges$ iff $S \neq S'$ and there exists $q \in S$ and
$q' \in S'$, such that $q \reach q'$, for all $S,S' \in \nodes$. We
write $\graphof{\mathcal{A}}=(\nodes,\edges,S)$ if
$\graphof{\mathcal{A}}$ is a tree with root $S\in\nodes$.

The execution of automata is defined next. A \emph{run} $\arun$ of
$\mathcal{A}$ over a ranked tree $t$ is a tree $\arun : \dom{t}
\rightarrow Q$ such that $\arun(p) \arrow{t(p)}{} (\arun(p1), \ldots,
\arun(p\ell)) \in \trans$, for all $p \in \dom{t}$, where
$\ell=\rankof{t(p)}$. Note that the frontier of a run is labeled by
states $q$ such that there exists a transition $q \arrow{\alpha}{} ()
\in \trans$, in analogy to the final states of a word automaton. A
weaker notion is that of \emph{partial} runs, where the previous
condition holds for $\dom{t}\setminus\frof{t}$, instead of the entire
$\dom{t}$. A run $\arun$ is \emph{accepting} if
$\arun(\epsilon)\in\initstates$. The \emph{language} of $\mathcal{A}$
is $\langof{}{\mathcal{A}} \isdef \bigcup_{q \in \initstates}
\langof{q}{\mathcal{A}}$, where $\langof{q}{\mathcal{A}} \isdef \{t
\mid \mathcal{A} \text{ has a run } \arun \text{ over } t \text{ and }
\arun(\epsilon)=q\}$.

An automaton is \emph{rooted} iff $\initstates = \set{\initstate}$ and
$\initstate \not\in \post{\trans}$. For an automaton $\mathcal{A}$ one
can build finitely many rooted automata $\mathcal{A}_1, \ldots,
\mathcal{A}_n$ such that $\langof{}{\mathcal{A}} = \bigcup_{i=1}^n
\langof{}{\mathcal{A}_i}$. A rooted automaton $\mathcal{A}$ is
\emph{trim} iff $\initstate \reach^* q$ and $\langof{q}{\mathcal{A}}
\neq\emptyset$, for each state $q\in\states$. Each automaton with
non-empty language can be transformed into a trim one with the same
language, by a simple marking algorithm. We use the following notions
of \emph{simulation} and \emph{refinement} between automata:

\begin{defi}\label{def:refinement}
  Let
  $\mathcal{A}=(\alphabet,\states_{\mathcal{A}},\initstate_{\mathcal{A}},\trans_{\mathcal{A}})$
  and
  $\mathcal{B}=(\alphabet,\states_{\mathcal{B}},\initstate_{\mathcal{B}},\trans_{\mathcal{B}})$
  be automata. A mapping $h : \states_{\mathcal{A}} \rightarrow
  \states_{\mathcal{B}}$ is a \emph{simulation} if and only if the following hold: \begin{enumerate}
  \item\label{it1:def:refinement} $h(\initstate_{\mathcal{A}})=\initstate_{\mathcal{B}}$ and
  \item\label{it2:def:refinement} $q_0\arrow{a}{}
    (q_1,\ldots,q_\ell)\in\trans_{\mathcal{A}}$ only if $h(q_0)
    \arrow{a}{} (h(q_1),\ldots,h(q_\ell))\in\trans_{\mathcal{B}}$, for
    all $q_0, \ldots, q_\ell \in \states_{\mathcal{A}}$.
  \end{enumerate}
  A simulation $h$ is a \emph{refinement} if and only if,
  moreover: \begin{enumerate}\setcounter{enumi}{2}
  \item\label{it3:def:refinement} $q'_0 \arrow{a}{}
    (q'_1,\ldots,q'_\ell)\in\trans_{\mathcal{B}}$ only if there exist
    $q_0\in h^{-1}(q'_0), \ldots, q_\ell\in h^{-1}(q'_\ell)$, such
    that $q_0\arrow{a}{} (q_1,\ldots,q_\ell)\in\trans_{\mathcal{A}}$,
    for all $q_0 \in \states_{\mathcal{A}}$ and $q'_1, \ldots,
    q'_\ell\in\states_{\mathcal{B}}$.
  \end{enumerate}
  If $h : \states_{\mathcal{A}} \rightarrow \states_{\mathcal{B}}$ is
  a simulation then $\mathcal{B}$ \emph{simulates} $\mathcal{A}$. If
  $h$ is a refinement then $\mathcal{A}$ \emph{refines} $\mathcal{B}$.
\end{defi}

The key properties of simulations and refinements are stated and proved below:

\begin{lem}\label{lemma:refinement}
  If $\mathcal{B}$ simulates $\mathcal{A}$ then
  $\langof{}{\mathcal{A}} \subseteq \langof{}{\mathcal{B}}$. If
  $\mathcal{A}$ is a refinement of $\mathcal{B}$ then
  $\langof{}{\mathcal{A}} = \langof{}{\mathcal{B}}$.
\end{lem}
\begin{proof}
  Let
  $\mathcal{A}=(\alphabet,\states_{\mathcal{A}},\initstate_{\mathcal{A}},\trans_{\mathcal{A}})$,
  $\mathcal{B}=(\alphabet,\states_{\mathcal{B}},\initstate_{\mathcal{B}},\trans_{\mathcal{B}})$
  and $h : \states_{\mathcal{A}} \rightarrow \states_{\mathcal{B}}$ be
  a mapping. ``$\subseteq$'' Assume that $h$ is a simulation. Let $t
  \in \langof{}{\mathcal{A}}$ be a tree and $\arun$ be an accepting
  run of $\mathcal{A}$ over $t$. Then one shows that $h \circ \arun$
  is an accepting run of $\mathcal{B}$ over $t$, by induction on $t$,
  using points (\ref{it1:def:refinement}) and
  (\ref{it2:def:refinement}) of \autoref{def:refinement}.
  ``$\supseteq$'' Assume that $h$ is a refinement. Let $t \in
  \langof{}{\mathcal{B}}$ be a tree and $\arun$ be an accepting run of
  $\mathcal{B}$ over $t$. We build an accepting run of $\mathcal{A}$
  over $t$ by induction on $t$, using points
  (\ref{it1:def:refinement}) and (\ref{it3:def:refinement}) of
  \autoref{def:refinement}.
\end{proof}

\begin{figure}[htbp]
  \begin{minipage}{.74\textwidth}
    \[ \begin{array}{ll ll}
      q_0 \arrow{a}{} (q_0,q_1) & (\infty)\hspace{1cm} & q_2 \arrow{a}{} (q_2,q_2) & (\infty) \\
      q_1 \arrow{b}{} (q_1) & (\infty) & q_2 \arrow{b}{} (q_3) & (\infty) \\
      q_1 \arrow{c}{} () & (\infty) & q_2 \arrow{b}{} (q_4) & (\infty) \\
      q_0 \arrow{b}{} (q_2) & (1) & q_3 \arrow{c}{} () & (\infty) \\
      \textcolor{red}{q_0 \arrow{b}{} (q_5)} & & q_4 \arrow{c}{} () & (\infty) \\
      \textcolor{red}{q_5 \arrow{c}{} ()} & & &
    \end{array} \]
  \end{minipage}
  \begin{minipage}{.25\textwidth}
    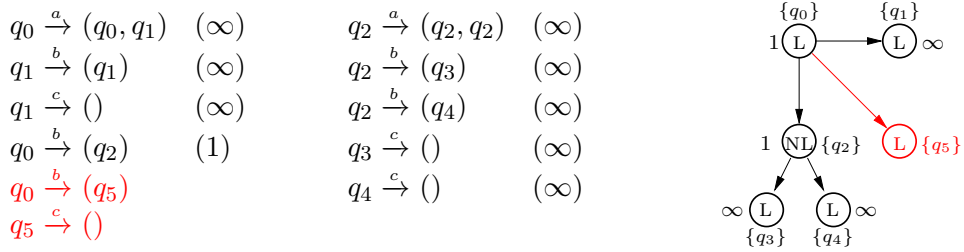
  \end{minipage}
  \caption{\label{fig:choice-free} A choice-free tree automaton and its non-choice-free extension}
\end{figure}

The following structural property of automata is key for building
expandable SIDs:
\begin{defi}\label{def:choice-free}
  An automaton $\mathcal{A} = (\alphabet,\states,\initstate,\trans)$
  is \emph{choice-free} iff the following hold:\begin{enumerate}
  \item\label{it1:def:choice-free} the SCC graph of $\mathcal{A}$ is a
    tree $\graphof{\mathcal{A}} = (\nodes,\edges,S_0)$, where $\pre{S}
    = \set{\tau}$ and $\cardof{\post{\tau}\cap S}=1$, for all
    $S\in\nodes \setminus\set{S_0}$, i.e., any non-root SCC is entered
    by one branch of a single transition,
  \item\label{it2:def:choice-free} there exists a mapping $\Lambda :
    \nodes \cup \trans \rightarrow \set{1,\infty}$ such
    that: \begin{enumerate}
    \item\label{it21:def:choice-free} for all $S \in \nodes$, if $S$
      is linear and $\Lambda(S)=1$ then $\cardof{\post{S}}=1$.
    \item\label{it22:def:choice-free} for all $\tau \in \trans$,
      $\Lambda(\tau)=1$ iff $\tau\in\post{S}$, for some linear $S \in
      \nodes$ such that $\Lambda(S)=1$,
    \item\label{it23:def:choice-free} for all $S \in \nodes$,
      $\Lambda(S)=1$ iff $S=S_0$ or $\pre{S}=\set{\tau}$, for some
      $\tau\in\trans$ such that $\Lambda(\tau)=1$.
    \end{enumerate}
  \end{enumerate}
  Let $\trans=\trans^1\uplus\trans^\infty$, where $\trans^k \isdef
  \set{\tau\in\trans\mid\Lambda(\tau)=k}$ and $k\in\set{1,\infty}$, be
  the partition of the set of transitions induced by the mapping
  $\Lambda$. A state $q \in \post{(\trans^1)} \cap
  \pre{(\trans^\infty)}$ is called a \emph{pivot state}. Let
  $\runsof{\infty}{q}{\mathcal{A}}$ denote the set of partial runs
  $\theta$ of $\mathcal{A}$, such that $\theta(\epsilon)=q$ and for
  all $p \in \dom{\theta}\setminus\frof{\theta}$, there exists $a \in
  \alphabet$ such that $\arun(p) \arrow{a}{} (\arun(p1), \ldots,
  \arun(pn)) \in \trans^\infty$.
\end{defi}
Intuitively, the structure of choice-free automata allows them to
traverse a unique sequence of linear SCCs, before entering a
non-linear SCC. Note that the labeling $\Lambda$ of SCCs and
transitions from Definition \ref{def:choice-free} is unique, because
the SCC graph of a choice-free automaton is a tree whose root is a
$1$-SCC, and $\Lambda$ is determined by the linearity of the SCCs in
this tree. Hence, the partition of the transitions of a choice-free
automaton into $1$- and $\infty$-transitions is unambiguous.

\begin{exa}\label{ex:choice-free}
The automaton from \autoref{fig:choice-free} is choice free.  The
linear (resp. nonlinear) SCCs are labeled by L (resp. NL). The
labeling of transitions and SCCs, resp. the SCC graph by $\Lambda$ are
represented in \autoref{fig:choice-free} (resp. right). The
choice-freeness is violated by adding the transitions $q_0 \arrow{b}{}
(q_5)$ and $q_5 \arrow{c}{} ()$ (in red), because the linear SCC
$\set{q_0}$ is labeled with $1$ and has two outgoing transitions, thus
contradicting point (\ref{it21:def:choice-free}) of
\autoref{def:choice-free} (the additional transitions are not
labeled).
\end{exa}

The transitions from $\trans^1$, called \emph{$1$-transitions}, are
used to move from one linear SCC to another, hence all of them occur
exactly once on each accepting run:

\begin{lem}\label{lemma:one-transitions}
  Let $\mathcal{A}=(\alphabet,\states,\initstate,\trans)$ be a
  choice-free automaton, such that
  $\trans=\trans^1\uplus\trans^\infty$ (\autoref{def:choice-free})
  and let $\arun$ be an accepting run of $\mathcal{A}$ over a tree
  $t$. Then, for each $1$-transition $q_0 \arrow{a}{}
  (q_1,\ldots,q_\ell)\in\trans^1$ there exists exactly one position $p
  \in \dom{\arun}$, such that $\arun(p)=q_0$, $t(p)=a$ and
  $\arun(pi)=q_i$, for all $i \in \interv{1}{\ell}$.
\end{lem}
\begin{proof}
  For space reasons, the proof of this lemma is given in
  \autoref{app:one-transitions}.
\end{proof}
The transitions from $\trans^\infty$, called
\emph{$\infty$-transitions}, can be applied any number of times on
some accepting run. This fact occurs as an easy consequence of the
lemma below:

\begin{lem}\label{lemma:choice-free-property}
  Let $\mathcal{A} = (\alphabet,\states,\initstate,\trans)$ be a
  choice-free automaton, where $\trans=\trans^1\uplus\trans^\infty$
  (\autoref{def:choice-free}). Then, for any state
  $q\in\pre{(\trans^\infty)}$ there exists a pivot state
  $q_0\in\post{(\trans^1)}\cap\pre{(\trans^\infty)}$ and a partial run
  $\theta_0 \in \runsof{\infty}{q_0}{\mathcal{A}}$ consisting only of
  $\infty$-transitions, such that $\theta_0(p)=q$ for some $p \in
  \frof{\theta_0}$ and either: \begin{enumerate}
   \item\label{it1:lemma:choice-free-property} $\mset{q,q_0} \subseteq
     \mset{\theta_0(p) \mid p \in \frof{\theta_0}}$, i.e., if $q=q_0$
     then $q$ occurs twice on $\frof{\theta_0}$, or
   \item\label{it2:lemma:choice-free-property} each partial run
     $\theta \in \runsof{\infty}{q}{\mathcal{A}}$ can be extended to a
     partial run $\theta' \in \runsof{\infty}{q}{\mathcal{A}}$ such
     that $q_0$ occurs on the frontier of $\theta'$.
  \end{enumerate}
\end{lem}
\begin{proof}
  For space reasons, the proof of this lemma is given in
  \autoref{app:choice-free-property}.
\end{proof}
Moreover, any automaton can be decomposed into finitely many
choice-free automata:

\begin{lem}\label{lemma:choice-free}
  Given an automaton $\mathcal{A} =
  (\alphabet,\states,\initstate,\trans)$, one can build finitely many
  choice-free automata $\mathcal{A}_i =
  (\alphabet,\states_i,\initstate_i,\trans_i)$, for $i \in
  \interv{1}{n}$, such that $\langof{}{\mathcal{A}} = \bigcup_{i=1}^n
  \langof{}{\mathcal{A}_n}$ and, moreover, $\cardof{\trans^1_i} \leq
  \max(\cardof{\states},\max\set{\rankof{a}\mid a\in
    \alphabet}^{\cardof{\states}})$, for all $i \in \interv{1}{n}$.
\end{lem}
\begin{proof}
  For space reasons, the proof of this lemma is given in
  \autoref{app:choice-free}.
\end{proof}

\begin{exa}\label{ex:choice-free-decomp}
The choice-free decomposition of the automaton $\mathcal{A}$ from
\autoref{fig:choice-free} consists of the following choice-free
automata, with transitions labeled according to
\autoref{def:choice-free}:
\[
\mathcal{A}_1 = \left\{\begin{array}{ll}
q_0 \arrow{a}{} (q_0,q_1) & (\infty) \\
q_1 \arrow{b}{} (q_1) & (\infty) \\
q_1 \arrow{c}{} () & (\infty) \\
q_0 \arrow{b}{} (q_2) & (1) \\
q_2 \arrow{a}{} (q_2,q_2) & (\infty) \\
q_2 \arrow{b}{} (q_3) & (\infty) \\
q_2 \arrow{b}{} (q_4) & (\infty)  \\
q_3 \arrow{c}{} () & (\infty)  \\
q_4 \arrow{c}{} () & (\infty)
\end{array}\right.
\hspace*{1cm}
\mathcal{A}_2 = \left\{\begin{array}{ll}
q_0 \arrow{a}{} (q_0,q_1) &  (\infty) \\
q_1 \arrow{b}{} (q_1) & (\infty) \\
q_1 \arrow{c}{} () &(\infty)  \\
q_0 \arrow{b}{} (q_5) & (1) \\
q_5 \arrow{c}{} () & (1)
\end{array}\right.
\]
It can be seen that $\langof{q_0}{\mathcal{A}} =
\langof{q_0}{\mathcal{A}_1} \cup \langof{q_0}{\mathcal{A}_2}$ because
$\mathcal{A}$ has the choice in $q_0$ between taking the transition
$q_0 \arrow{b}{} (q_2)$ or $q_0 \arrow{b}{} (q_5)$. Since these
transitions occur at most once in each accepting run of $\mathcal{A}$,
the choice-free decomposition of $\mathcal{A}$ produces two automata
in which each such transition occurs exactly once on each accepting
run.
\end{exa}

\subsection{Automata with Alphabets of Formul{\ae}}
\label{sec:formulae-alphabets}

The construction of the expandable SIDs from
\autoref{lemma:expansion} uses automata that recognize trees labeled
with qpf formul{\ae} taken from a finite set. We recall that every
model of a sentence is defined by a complete unfolding that replaces
the predicate atoms with corresponding definitions, recursively. The
steps of these unfoldings can be placed into a tree labeled with
predicate-free formul{\ae} from an alphabet $\Sigma$, reflecting the
partial order in which the rules from the SID are applied. These
unfolding trees form the language of an automaton defined directly
from the syntax of the SID. Dually, from any $\Sigma$-labeled
automaton one can build a SID whose unfolding trees form the language
of the automaton.

\begin{defi}\label{def:alpha-sid}
  Let $\Sigma$ be the set of qpf formul{\ae} $\alpha$ of rank
  $\rankof{\alpha} = \ell$, such that:~\begin{enumerate}
  \item\label{it1:alpha-sid} $\fv{\alpha} \subseteq
    \set{\atpos{x}{\epsilon}_1, \ldots, \atpos{x}{\epsilon}_{n_0}}
    \cup \set{\atpos{y}{\epsilon}_1, \ldots, \atpos{y}{\epsilon}_{m}}
    \cup \bigcup_{i=1}^{\ell} \set{\atpos{x}{i}_1, \ldots,
      \atpos{x}{i}_{n_i}}$, for some $m, n_0, \ldots, n_{\ell} \in
    \nat$; a variable $\atpos{x}{i}_j$ is called a
    \emph{$i$-variable}, for all $i \in
    \set{\epsilon}\cup\interv{1}{\ell}$,
  \item\label{it2:alpha-sid} $\atpos{x}{i}_j \not\eqof{\alpha}
    \atpos{x}{i}_k$, for all $i \in \interv{1}{\ell}$ and $1 \le j <
    k \le n_i$.
  \end{enumerate}
  The \emph{characteristic formula} of a $\Sigma$-labeled tree $t$ is
  the qpf formula $\charform{t} \isdef \Bigstar_{\!\!p \in \dom{t}}
  ~\atpos{t(p)}{p}$, where the formul{\ae} $\atpos{t(p)}{p}$ are
  obtained from $t(p)\in\Sigma$ by replacing each occurrence of a
  variable $\atpos{x}{q}$ by $\atpos{x}{pq}$, for all $p\in\dom{t}$.
\end{defi}

Given a SID $\asid$, the $\Sigma$-labeled automaton
$\auto{\asid}{\apred} \isdef (\Sigma, \states_\asid, q_{\apred},
\trans_\asid)$ is defined as follows:
\begin{itemize}[label=$\triangleright$]
\item $\states_\asid$ contains states $q_\bpred$, where $\bpred$ is a
  predicate occurring in $\asid$; each state has an associated arity
  $\arityof{q_\bpred}\isdef\arityof{\bpred}$,
\item $\trans_\asid$ contains a transition $q_{\apred_0}
  \arrow{\alpha_\arule}{} (q_{\apred_1}, \ldots, q_{\apred_\ell})$,
  where $\alpha_\arule$ is the symbol:
  \[
  \alpha_\arule \isdef \psi[x_1/\atpos{x}{\epsilon}_1, \ldots,
    x_{n_0}/\atpos{x}{\epsilon}_{n_0},~ y_1/\atpos{y}{\epsilon}_1,
    \ldots, y_m/\atpos{y}{\epsilon}_m] * \Bigstar\nolimits_{i=1}^\ell
  \Bigstar\nolimits_{j=1}^{n_i} ~\atpos{z}{\epsilon}_{i,j} = \atpos{x}{i}_j \label{rule:alpha}
  \]
  of rank $\rankof{\alpha_\arule} = \ell$ that corresponds to the rule
  $\arule\in\asid$, where $\psi$ is a qpf formula:
  \[
    \arule ~:~ \apred_0(x_1, \ldots, x_{n_0}) \leftarrow \exists y_1 \ldots
    \exists y_m ~.~ \psi * \Bigstar\nolimits_{i=1}^\ell
    \apred_i(z_{i,1},\ldots,z_{i,n_i}) \label{rule:sid}
  \]
\end{itemize}

\begin{exa}\label{ex:sid-ta}
  Let us consider the following SID:
  \[\asid = \left\{\begin{array}{rcl}
    \apred & \leftarrow & \exists y_1 \exists y_2 \exists y_3 ~.~ \bpred(y_1,y_2,y_3) \\
    \bpred(x_1,x_2,x_3) & \leftarrow & \exists y_4 ~.~ \mathsf{e}(x_1,x_3) * \mathsf{e}(x_1,y_4) * \bpred(y_4,x_2,x_3) \\
    \bpred(x_1,x_2,x_3) & \leftarrow & \mathsf{e}(x_1,x_3) * \mathsf{e}(x_1,x_2) * \mathsf{e}(x_2,x_3)
  \end{array}\right.\]
  The automaton $\auto{\asid}{\apred}$ has the following transitions:
  \[\auto{\asid}{\apred} = \left\{\begin{array}{ll}
  q_{\apred} & \arrow{\atpos{y}{\epsilon}_1 = \atpos{x}{1}_1 *~ \atpos{y}{\epsilon}_2 = \atpos{x}{1}_2 *~ \atpos{y}{\epsilon}_3 = \atpos{x}{1}_3}{} (q_\bpred) \\
  q_{\bpred} & \arrow{\mathsf{e}(\atpos{x}{\epsilon}_1,~\atpos{x}{\epsilon}_3) ~*~ \mathsf{e}(\atpos{x}{\epsilon}_1,~\atpos{y}{\epsilon}_4) ~*~
  \atpos{y}{\epsilon}_4=\atpos{x}{1}_1 *~ \atpos{x}{\epsilon}_2 = \atpos{x}{1}_2 *~ \atpos{x}{\epsilon}_3 = \atpos{x}{1}_3}{} (q_\bpred) \\
  q_{\bpred} & \arrow{\mathsf{e}(\atpos{x}{\epsilon}_1,~\atpos{x}{\epsilon}_3) ~*~
    \mathsf{e}(\atpos{x}{\epsilon}_1,~\atpos{x}{\epsilon}_2) ~*~
    \mathsf{e}(\atpos{x}{\epsilon}_2,~\atpos{x}{\epsilon}_3)}{} ()
  \end{array}\right.
  \]
\end{exa}



Dually, given an automaton $\mathcal{A} =
(\Sigma,\states,\initstate,\trans)$, the SID $\asid_{\mathcal{A}}$
consists of the following rules, one for each transition $q_0
\arrow{\alpha}{} (q_1,\ldots,q_\ell) \in \trans$:
\begin{align}
  \hspace*{-3mm}\apred_{q_0}(x_1,\ldots,x_{\arityof{q_0}}) \leftarrow & ~\exists y_1
  \ldots \exists y_m ~.~ \alpha[\atpos{x}{\epsilon}_1/x_1, \ldots,
    \atpos{x}{\epsilon}_{\arityof{q_0}}/x_{\arityof{q_0}}] *
  \Bigstar\nolimits_{j=1}^\ell
  \apred_{q_j}(\atpos{x}{j}_1,\ldots,\atpos{x}{j}_{\arityof{q_j}}) \label{rule:ta-sid} \\
  \text{where } & \set{y_1, \ldots, y_m} \isdef \fv{\alpha}
  \setminus \big(\set{\atpos{x}{\epsilon}_1, \ldots,
    \atpos{x}{\epsilon}_{\arityof{q_0}}} \cup \bigcup\nolimits_{j=1}^\ell
  \set{\atpos{x}{j}_1, \ldots, \atpos{x}{j}_{\arityof{q_j}}}\big) \nonumber
\end{align}

The similarity between SIDs and $\Sigma$-labeled automata
(\autoref{lemma:sid-ta}) motivates the use of similar terminology.
For a $\Sigma$-labeled automaton $\mathcal{A}$, we define
$\sem{\mathcal{A}} \isdef \bigcup_{t\in\langof{}{\mathcal{A}}}
\sem{\exclof{\charform{t}}}$, $\csem{\mathcal{A}}{} \isdef
\bigcup_{t\in\langof{}{\mathcal{A}}} \csem{\exclof{\charform{t}}}{}$
and $\rcsem{\mathcal{A}}{} \isdef \bigcup_{t\in\langof{}{\mathcal{A}}}
\rcsem{\exclof{\charform{t}}}{}$. Moreover, a $\Sigma$-labeled
automaton $\mathcal{A}$ is all-satisfiable if the formula
$\charform{t}$ is satisfiable, for all $t \in
\langof{}{\mathcal{A}}$. The relation between SIDs and
$\Sigma$-labeled automata is formally stated below:

\begin{lem}\label{lemma:sid-ta} \hfill
  \begin{enumerate}
  \item\label{it1:lemma:sid-ta} Given a SID $\asid$ and a nullary
    predicate $\apred$, one can build an automaton
    $\auto{\asid}{\apred}$ such that $\sidsem{\apred}{\asid} =
    \sem{\auto{\asid}{\apred}}$. Moreover, if $\asid$ is
    all-satisfiable for $\apred$, then $\auto{\asid}{\apred}$ is
    all-satisfiable.
  \item\label{it2:lemma:sid-ta} Given an automaton
    $\mathcal{A}=(\Sigma,\states,\initstate,\trans)$, one can build a
    SID $\asid_{\mathcal{A}}$, such that $\sem{\mathcal{A}} =
    \sidsem{\apred_\initstate}{\asid_{\mathcal{A}}}$ and
    $\rcsem{\mathcal{A}}{} =
    \rcsem{\apred_\initstate}{\asid_{\mathcal{A}}}$, for a nullary
    predicate $\apred$.
  \end{enumerate}
\end{lem}
\begin{proof}
  For space reasons, the proof of this lemma is given in
  \autoref{app:sid-ta}.
\end{proof}

\subsection{Persistent Variables}
\label{sec:persistent-variables}

The second ingredient of the construction of the expandable SIDs from
\autoref{lemma:expansion} are the \emph{persistent} variables of a
$\Sigma$-labeled choice-free automaton. These are variables introduced
by the $1$-transitions of the automaton, whose values propagate via
equalities throughout each run of the choice-free automaton. We define
persistent variables using the notion of \emph{profile}:

\begin{defi}\label{def:profile}
  Let $\mathcal{A}=(\Sigma,\states,\initstate,\trans)$ be a
  choice-free automaton, where $\trans=\trans^1\uplus\trans^\infty$
  (\autoref{def:choice-free}). A \emph{positional function} $\posfunc
  : \states \rightarrow \pow{\nat}$ associates each state $q$ with a
  set $\posfunc(q)\subseteq\interv{1}{\arityof{q}}$. The \emph{profile
    of $\mathcal{A}$} is the pointwise largest positional function
  $\profile{\mathcal{A}}$ such that, for each transition $q_0
  \arrow{\alpha}{} (q_1, \ldots, q_\ell) \in \trans^\infty$, each $k \in
  \interv{1}{\ell}$ and each $r \in \profile{\mathcal{A}}(q_k)$,
  there exists $s \in \profile{\mathcal{A}}(q_0)$, such that
  $\atpos{x}{\epsilon}_s\eqof{\alpha}\atpos{x}{k}_r$. A variable
  $\atpos{x}{i}_j$ that occurs within the label of a transition
  $q_0 \arrow{\alpha}{} (q_1,\ldots,q_\ell)\in\trans$ is said
  to be \emph{persistent} iff $j \in
  \profile{\mathcal{A}}(q_i)$, for all $i \in \interv{0}{\ell}$.
\end{defi}
Intuitively, $\profile{\mathcal{A}}(q)$ is the set of indices of those
variables, associated with a state, that will be equated, through a
chain of equalities in the characteristic formula $\charform{t}$, to
the same variable associated with the corresponding pivot state (Lemma
\ref{lemma:choice-free-property}) in every run of $\infty$-transitions
of $\mathcal{A}$ over $t$. Note that the profile is computable by a
finite greatest fixpoint Kleene iteration over each SCC of the
automaton (see the proof of \autoref{lemma:reset} for an explicit
statement of that fixed point).

Without loss of generality, we assume that a $\Sigma$-labeled
automaton does not have trivial SCCs, i.e., consisting of a single
state, with no transitions that are both outgoing and incoming to that
state. The profile associated with such state $q$ would be the
interval $\interv{1}{\arityof{q}}$ and any variable in the label of an
incoming or outgoing transition would be unnecessarily considered persistent. The
trivial SCCs of a $\Sigma$-labeled automaton can be eliminated by a
pre-processing step which combines the labels of the incoming and
outgoing transitions and renames the variables according to the
convention (\autoref{def:alpha-sid}).

\begin{exa}\label{ex:profile}(continued from \autoref{ex:sid-ta})
  The profile of the automaton $\auto{\asid}{\apred}$, that
  corresponds to the SID $\asid$ from \autoref{ex:sid-ta}
  associates $q_\apred$ with the empty set and $q_\bpred$ with the set
  $\set{2,3}$. Note that $\arityof{q_\apred}=0$ and
  $\arityof{q_\bpred}=3$. The first and third transitions in
  \autoref{ex:sid-ta} are $1$-transitions, whereas the second
  transition is an $\infty$-transition. The variables $\atpos{x}{1}_2$
  and $\atpos{x}{1}_3$ from the label of the second transition
  ($\infty$) are persistent.
\end{exa}

A \emph{context} $\arun_{p \leftarrow q}$ is a partial run over a tree
$t$ such that $p \in \frof{\arun_{p \leftarrow q}}$, $\arun_{p
  \leftarrow q}(p)=q$ and $\arun_{p \leftarrow q}(r) \arrow{t(r)}{}
()$, for all $r \in \frof{\arun_{p \leftarrow q}}\setminus{p}$, i.e.,
the partial run has exactly one ``open'' frontier position $p$ that is
labeled with a state $q$. A key property of automata is that
equalities between non-persistent variables vanish in contexts
consisting of $\infty$-transitions only (\autoref{lemma:reset}). These contexts, called \emph{resets}, are formally
defined below:

\begin{defi}\label{def:reset}
  A context $\arun_{p \leftarrow q} \in
  \runsof{\infty}{q}{\mathcal{A}}$ over a tree $t$ is a
  $q$-\emph{reset} iff~\begin{enumerate*}
  \item\label{it1:def:reset} $\atpos{x}{\epsilon}_j
    \eqof{\charform{t}} \atpos{x}{p}_j$, for all $j \in
    \profile{\mathcal{A}}(q)$, and
  \item\label{it2:def:reset} $\atpos{x}{\epsilon}_j
    \not\eqof{\charform{t}} \atpos{x}{p}_k$, for all $j,k \in
    \interv{1}{\arityof{q}}$, such that $k \not\in
    \profile{\mathcal{A}}(q)$.
  \end{enumerate*}
  The path between $\epsilon$ and $p$ in $\arun_{p\leftarrow q}$ is
  called a \emph{reset path}.
\end{defi}

\begin{exa}\label{ex:reset}(continued from \autoref{ex:sid-ta})
For instance, in the context $\arun_{1 \leftarrow q_\bpred}$, that
consists of the second transition $q_\bpred \arrow{\alpha}{} q_\bpred$
of the automaton in \autoref{ex:sid-ta}, we have
$\atpos{x}{\epsilon}_1 \not\eqof{\alpha} \atpos{x}{1}_1$. Then, the
value of $\atpos{x}{\epsilon}_1$ is ``forgotten'' along any run that
iterates this transition at least twice.
\end{exa}

\begin{lem}\label{lemma:reset}
  Let $\mathcal{A}=(\Sigma,\states,\initstate,\trans)$ be a trim
  automaton. Then, there exists a $q$-reset for \begin{enumerate*}
  \item\label{it1:lemma:reset} each pivot state $q \in
    \post{(\trans^1)} \cap \pre{(\trans^\infty)}$, and
  \item\label{it2:lemma:reset} each state $q \in \pre{(\trans^1)} \cap
    \pre{(\trans^\infty)}$, i.e., that is the origin of both a
    $1$-transition and a $\infty$-transition of $\mathcal{A}$.
  \end{enumerate*}
\end{lem}
\begin{proof}
  By \autoref{def:profile}, $\profile{\mathcal{A}}$ is the greatest
  fixpoint of the monotone function $\mathcal{F}$ on the domain of
  positional functions  $\posfunc : \states \rightarrow \pow{\nat}$ , defined below:
  \[
    \mathcal{F}(\posfunc) \isdef \lambda q ~.~ \bigcap_{\begin{array}{c}
         \scriptstyle{q_0 \arrow{\alpha}{} (q_1,\ldots,q_\ell) \in \trans^\infty} \\
         \scriptstyle{q = q_k \in \set{q_1,\ldots,q_\ell}}
     \end{array}} \set{r\in\interv{1}{\arityof{q_k}} \mid \exists s \in \posfunc(q_0) ~.~
       \atpos{x_s}{\epsilon} \eqof{\alpha} \atpos{x_r}{k}}
  \]
  Namely, we have $\profile{\mathcal{A}} = \mathcal{F}^i(\top) =
  \mathcal{F}^j(\top)$, for a sufficiently large $i\geq1$ and any $j
  \geq i$, where $\top$ is the positional function $\lambda q ~.~
  \interv{1}{\arityof{q}}$. Now consider the following ``big-step''
  function $\mathcal{G}$ on the domain of positional functions:
  \[
    \mathcal{G}(\posfunc) \isdef \lambda q ~.~ \bigcap_{\begin{array}{c}
        \scriptstyle{\arun \in \runsof{\infty}{q}{\mathcal{A}} \text{ partial run over } t} \\
        \scriptstyle{p \in \frof{\arun} \text{, such that } \arun(p)=q}
    \end{array}} \set{r \in \interv{1}{\arityof{q}} \mid \exists s \in \posfunc(q) ~.~
      \atpos{x_s}{\epsilon} \eqof{\charform{t}} \atpos{x_r}{p}}
  \]
  We prove the following:
  \begin{fact}\label{fact:reset}
    $\gfp(\mathcal{F})(q) = \gfp(\mathcal{G})(q)$, for any pivot state
    $q$ of $\mathcal{A}$.
  \end{fact}
  \begin{proof}``$\subseteq$'' Each partial run
    $\arun\in\runsof{\infty}{q}{\mathcal{A}}$ such that $\arun(p)=q$,
    for some $p \in \frof{\arun}$ corresponds to a finite sequence of
    transitions from $\trans^\infty$. ``$\supseteq$'' Since $q$ is a
    pivot state we have $q\in\pre{(\trans^\infty)}$, thus necessarily
    $q=q_0$, where $q_0\in\post{(\trans^1)}$ is the state whose
    existence is stated by \autoref{lemma:choice-free-property}. Then
    every $\infty$-transition incoming to $q$ belongs to a partial run
    $\arun\in\runsof{\infty}{q}{\mathcal{A}}$, such that $q$ occurs on
    the frontier of $\arun$.
  \end{proof}

  \skipnoindent Back to the proof, we prove the two points of the statement below:

  \skipnoindent (\ref{it1:lemma:reset}) Let $q$ be a pivot state of
  $\mathcal{A}$. By \autoref{fact:reset}, we have
  $\profile{\mathcal{A}}(q)=\mathcal{G}^i(\top)(q)$ for a sufficiently
  large finite integer $i\geq0$. We show that the latter condition is
  equivalent to the existence of a $q$-reset $\arun_{p \leftarrow
    q}\in\runsof{\infty}{q}{\mathcal{A}}$.

  \skipnoindent``$\Leftarrow$'' Assume that there exists a $q$-reset
  $\arun\in\runsof{\infty}{q}{\mathcal{A}}$ over some tree $t$. Then
  $\mathcal{G}(\top)(q) = \set{j \in \interv{1}{\arityof{q}} \mid
    \exists k \in \interv{1}{\arityof{q}} ~.~ \atpos{x_j}{\epsilon}
    \eqof{\charform{t}} \atpos{x_k}{p}} = \profile{\mathcal{A}}(q)$.

  \skipnoindent``$\Rightarrow$'' Assume there exists $i\geq0$, such
  that $\mathcal{G}^i(\top)(q) = \profile{\mathcal{A}}(q)$ and let $i$
  be the smallest such integer.  Then $\top, \mathcal{G}(\top),
  \mathcal{G}^2(\top), \ldots, \mathcal{G}^{i}(\top)$ is a strictly
  decreasing sequence hence, for each $j\in\interv{1}{i}$, there
  exists a partial run $\arun_j \in \runsof{\infty}{q}{\mathcal{A}}$
  over some tree $t_j$ and a position $p_j \in \frof{t_j}$, such that
  $\arun_j(p_j)=q$ and \(\set{r \in \interv{1}{\arityof{q}} \mid
    \exists s \in \mathcal{G}^{j-1}(\top)(q) ~.~ \atpos{x_s}{\epsilon}
    \eqof{\charform{t_j}} \atpos{x_r}{p_j}} \subsetneq
  \mathcal{G}^{j-1}(\top)(q)\). We compose these partial runs
  $\arun_1, \ldots, \arun_i$ by appending each $\arun_j$ to
  $\arun_{j-1}$ at position $p_{j-1} \in \frof{\arun_{j-1}}$, for all
  $j \in \interv{2}{i}$ into a partial $\arun'' \in
  \runsof{\infty}{q}{\mathcal{A}}$. We define a context
  $\overline{\arun}_{p_i\leftarrow q}$ by appending to each position
  $r\in\frof{\arun}\setminus\set{p_i}$ a complete run starting in
  $\arun(r)$. By the fact that $\mathcal{A}$ is trim, such a run
  exists. The context $\overline{\arun}_{p_i\leftarrow q}$ satisfies
  condition (\ref{it2:def:reset}) of \autoref{def:reset}, but not
  necessarily (\ref{it1:def:reset}). Let $\pi :
  \profile{\mathcal{A}}(q) \rightarrow \profile{\mathcal{A}}(q)$ be a
  permutation defined as $\pi(i)=j$ iff $x_i
  \eqof{\charform{\overline{t}}} x_j$, there $\overline{t}$ is the
  tree recognized by the partial run $\overline{\arun}_{p_i\leftarrow
    q}$ of $\mathcal{A}$. Note that the choice of $j$ is not unique,
  but one exists, by \autoref{def:profile}. Then we define the
  $q$-reset $\arun_{p_i \leftarrow q}$ by appending
  $\overline{\arun}_{p_i\leftarrow q}$ to itself at position $p_i$ a
  number of times equal to the order of $\pi$. Then, one can check
  that $\arun_{p_i \leftarrow q}$ satisfies both conditions of
  \autoref{def:reset}.

  \skipnoindent (\ref{it2:lemma:reset}) Let $S$ be the SCC of $q$ in
  $\mathcal{A}$. Since $q = \pre{\tau}$, for some transition $\tau \in
  \trans^1$, it must be the case that $S$ is a linear SCC, by
  \autoref{def:choice-free}.  Also $q \in \pre{(\trans^\infty)}$ thus,
  by \autoref{lemma:choice-free-property}, there exists a pivot state
  $q_0$ in $S$ and let $\arun^0 \in \runsof{\infty}{q}{\mathcal{A}}$
  be a partial run from $q$ to $q_0$ with transitions from
  $\prepost{S} \subseteq \trans^\infty$. From point
  (\ref{it1:lemma:reset}) above we obtain a $q_0$-reset $\arun^1_{p_1
    \leftarrow q_0} \in \runsof{\infty}{q_0}{\mathcal{A}}$ such that
  $\atpos{x}{\epsilon}_j \eqof{\charform{t}} \atpos{y}{p_1}_k$, for
  all $j,k\in\profile{\mathcal{A}}(q_0)$ and $\atpos{x}{\epsilon}_j
  \not\eqof{\charform{t}} \atpos{y}{p_1}_k$, for all $j,k \in
  \interv{1}{\arityof{q_0}}$, $k \not\in
  \profile{\mathcal{A}}(q_0)$. Moreover, there exists another context
  $\arun^2_{p_2 \leftarrow q} \in \runsof{\infty}{q_0}{\mathcal{A}}$.
  By the choice of the pivot state $q_0$, there exists a position $p_0
  \in \frof{\arun^0}$ such that $\arun^0(p_0)=q_0$. Let $p \isdef
  p_0p_1p_2$ and $\overline{\arun}_{p \leftarrow q}$ be the context
  consisting of $\arun^0$ to which we append, in this order:
  \begin{itemize}[label=$\triangleright$]
  \item $\arun^1$ on position $p_0$,
  \item $\arun^2$ on the position $p_0p_2$,
  \item to any other position $r \in (\frof{\arun^0} \setminus
    \set{p_0}) \cup (p_0 \cdot \frof{\arun^1} \setminus \set{p_1})
    \cup (p_0p_1 \cdot \frof{\arun^2} \setminus \set{p_2})$ a complete
    run starting in:
    \begin{itemize}
    \item $\arun^0(r)$ if $r \in \frof{\arun^0} \setminus \set{p_0}$,
    \item $\arun^1(r')$ if $r = p_0 r'$ and $r' \in \frof{\arun^1}
      \setminus \set{p_1}$, and
    \item $\arun^2(r'')$ if $r = p_0 p_1 r''$ and $r'' \in
      \frof{\arun^2} \setminus \set{p_2}$.
    \end{itemize}
    Such runs exist by the assumption that $\mathcal{A}$ is
    trim. Moreover, these runs use only $\infty$-transitions, because
    their states are from $\infty$-SCCs (\autoref{def:choice-free}).
  \end{itemize}
  It is easy to check that $\overline{\arun}$ satisfies condition
  (\ref{it2:def:reset}) of \autoref{def:reset}. In order to satisfy
  condition (\ref{it1:def:reset}), in addition to
  (\ref{it2:def:reset}), we append $\overline{\arun}$ to itself at
  position $p$, using the same idea as in the construction at point
  (\ref{it1:lemma:reset}).
\end{proof}

The purpose of introducing resets is proving that any sequence of
partial runs consisting of $\infty$-transitions can be embedded in a
complete run, such that each two such partial runs are separated by
any number of resets. This is a key ingredient for the proof of the
``embedding'' property of the canonical models of expandables SIDs
(\autoref{def:expandable}).

  \begin{figure}[htbp]
    \begin{center}
    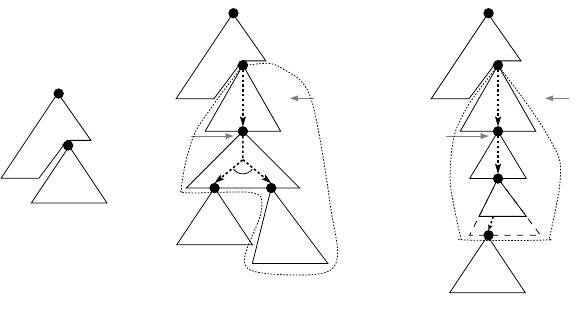
    \caption{\label{fig:runs-embedding} Embedding of a partial run $\arun_i$ in $\arun$}
    \end{center}
  \end{figure}

\begin{lem}\label{lemma:runs-embedding}
  Let $\mathcal{A}$ be a trim choice-free automaton. Given partial
  runs $\arun_1\in\runsof{\infty}{q_1}{\mathcal{A}}$, $\ldots$,
  $\arun_n \in \runsof{\infty}{q_n}{\mathcal{A}}$ and an integer
  $k\geq1$, there exists an accepting run $\arun$ of $\mathcal{A}$
  such that: \begin{enumerate}
  \item\label{it1:runs-embedding} $\arun_i$ is embedded in $\arun$ at
    some position $p_i\in\dom{\arun}$, for each $i\in\interv{1}{n}$,
  \item\label{it2:runs-embedding} $p_i\cdot\dom{\arun_i} \cap
    p_j\cdot\dom{\arun_j} = \emptyset$, for all $1 \le i < j \le n$,
  \item\label{it3:runs-embedding} the path between $p_i$ and $p_j$ in
    $\arun$ traverses $k$ times some reset path disjoint from
    $\bigcup_{\ell=1}^n p_\ell\cdot\dom{\arun_\ell}$, for all $1 \le
    i < j \le n$.
  \end{enumerate}
\end{lem}
\begin{proof}
  Let $\arun$ be an arbitrary accepting run of $\mathcal{A}$. By
  \autoref{lemma:one-transitions}, each $1$-transition occurs exactly
  once on $\arun$, hence $\arun$ visits each pivot state at least
  once. The partial runs $\arun_1, \ldots, \arun_n$ will be inserted
  into $\arun$ one by one, as described next. First, for each $\arun_i
  \in \runsof{\infty}{q_i}{\mathcal{A}}$, we have a pivot state
  $q^0_i$ and a partial run $\arun^0_i \in
  \runsof{\infty}{q^0_i}{\mathcal{A}}$, satisfying condition
  \ref{it1:lemma:choice-free-property} or
  \ref{it2:lemma:choice-free-property} of
  \autoref{lemma:choice-free-property}. Since $q^0_i$ occurs on
  $\arun$, we can insert in $\arun$ a new partial run
  $\arun'_i\in\runsof{\infty}{q^i_0}{\mathcal{A}}$ defined next.  By
  \autoref{lemma:reset} (\ref{it1:lemma:reset}), there exists a
  $q^i_0$-reset sequence $\arun_r^i \in
  \runsof{\infty}{q^i_0}{\mathcal{A}}$. The partial run $\arun'_i$ is
  obtained by composing $\arun_r^i$ with itself $k$ times, followed by
  $\arun^i_0$.  These compositions are possible, because $q^i_0$
  occurs at the root of $\arun^i_r$ and $\arun^i_0$, as well as the
  frontier of $\arun^i_r$. Depending on which condition of
  \autoref{lemma:choice-free-property} is satisfied by $q^i_0$ and
  $\arun^i_0$, we distinguish the following cases (see
  \autoref{fig:runs-embedding} for an illustration):
  \begin{itemize}[label=$\triangleright$]
  \item \underline{condition (\ref{it1:lemma:choice-free-property}) of
    \autoref{lemma:choice-free-property} holds}: in this case $q_i$
    and $q^i_0$ occur on different positions on the frontier of
    $\arun^i_0$, thus we append $\arun_i$ on the position where $q_i$
    occurs and the rest of $\arun$ on the position where $q^i_0$
    occurs.
  \item \underline{condition (\ref{it2:lemma:choice-free-property}) of
    \autoref{lemma:choice-free-property} holds}: in this case only
    $q^i_0$ occurs on the frontier of $\arun^i_0$, thus we continue
    with $\arun_i$, which can be extended to reach $q^i_0$ again, by
    \autoref{lemma:choice-free-property}. From this second
    occurrence of $q^i_0$, we continue with $\arun$.
  \end{itemize}

  \noindent We prove the points from the statement of the Lemma
  below: \begin{itemize}[left=.5\parindent]
  \item[(\ref{it1:runs-embedding})] The runs $\arun_1, \ldots,
    \arun_n$ are inserted into $\arun$ at positions $p_1, \ldots,
    p_n$, respectively.
  \item[(\ref{it2:runs-embedding})] Since $\arun_1, \ldots, \arun_n$
    are inserted one after the other (the order is not important), we
    have $p_i \cdot \dom{\arun_i} \cap p_j \cdot \dom{\arun_j} =
    \emptyset$, for all $1 \le i < j \le n$.
  \item[(\ref{it3:runs-embedding})] By the definition of $\arun'_1,
    \ldots, \arun'_n$, the path between $p_i$ and $p_j$ traverses $k$
    times the $\arun^i_r$ or $\arun^j_r$ reset sequences that are,
    moreover, disjoint from each $p_k \cdot \dom{\arun_k}$, for $k \in
    \interv{1}{n}$. \qedhere
  \end{itemize}
\end{proof}

\subsection{Eliminating Persistent Variables}
\label{sec:persistent-variables-elimination}

For the rest of this section, let $\auto{\asid}{\apred} =
(\Sigma, \states_\asid, q_\apred, \trans_\asid)$ be the automaton built
for the given SID $\asid$ and nullary predicate $\apred$, by the
construction of \autoref{lemma:sid-ta}. Since $\asid$ was assumed to
be all-satisfiable, the same can be assumed about
$\auto{\asid}{\apred}$, by \autoref{lemma:sid-ta}
(\ref{it1:lemma:sid-ta}).

Moreover, we can assume, without loss of generality, that
$\auto{\asid}{\apred}$ is choice-free and let $\trans_\asid =
\trans^1_\asid \uplus \trans^\infty_\asid$ be the partition of the
transitions of $\auto{\asid}{\apred}$
(\autoref{def:choice-free}). If this is not the case, we consider one
of the finitely many automata in the language-preserving choice-free
decomposition of $\auto{\asid}{\apred}$ (\autoref{lemma:choice-free}).

The transformation proceeds in three stages, denoted (I), (II) and
(III) below. The result of each stage is one or more choice-free
automata that are treewidth bounded if and only if the set
$\sem{\auto{\asid}{\apred}}$ is treewidth bounded.

\begin{exa}\label{ex:persistent-variables-elimination}
  We shall illustrate each stage of the construction on the following SID:
  \[\asid \left\{\begin{array}{rcl}
  \apred & \leftarrow & \exists y_1 \exists y_2 \exists y_3 ~.~ \cpred_1(y_1,y_2,y_3) \\
  \cpred_1(x_1,x_2,x_3) & \leftarrow & \exists y_4 ~.~ \erel(x_1,y_4) * \erel(x_3,y_4) * \cpred_1(y_4,x_2,x_3) \\
  \cpred_1(x_1,x_2,x_3) & \leftarrow & \exists y_5 ~.~ \erel(x_1,x_2) * \cpred_2(x_2,y_5,x_3) \\
  \cpred_2(x_1,x_2,x_3) & \leftarrow & \exists y_6 ~.~ \erel(x_1,y_6) * \erel(x_3,y_6) * \cpred_2(y_6,x_2,x_3) \\
  \cpred_2(x_1,x_2,x_3) & \leftarrow & \erel(x_1,x_2)
  \end{array}\right.\]
  For simplicity, the existentially quantified variables are given pairwise distinct
  names.  The automaton $\auto{\asid}{\apred}$ has the following transitions:
  \[\auto{\asid}{\apred} \left\{ \begin{array}{lll}
    \tau_1: & q_\apred \arrow{\atpos{y}{\epsilon}_1 = \atpos{x}{1}_1 ~*~ \atpos{y}{\epsilon}_2 = \atpos{x}{1}_2 ~*~\atpos{y}{\epsilon}_3 = \atpos{x}{1}_3}{} (q_{\cpred_1}) & (1) \\
    \tau_2 & q_{\cpred_1} \arrow{ \erel(\atpos{x}{\epsilon}_1,~ \atpos{y}{\epsilon}_4) ~*~ \erel(\atpos{x}{\epsilon}_3,~ \atpos{y}{\epsilon}_4) ~*~
        \atpos{y}{\epsilon}_4 = \atpos{x}{1}_1 ~*~ \atpos{x}{\epsilon}_2 = \atpos{x}{1}_2 ~*~\atpos{x}{\epsilon}_3 = \atpos{x}{1}_3
    }{} (q_{\cpred_1}) & (\infty) \\
    \tau_3: & q_{\cpred_1} \arrow{ \erel(\atpos{x}{\epsilon}_1,~ \atpos{x}{\epsilon}_2) ~*~
        \atpos{x}{\epsilon}_2 = \atpos{x}{1}_1 ~*~ \atpos{y}{\epsilon}_5 = \atpos{x}{1}_2 ~*~\atpos{x}{\epsilon}_3 = \atpos{x}{1}_3
    }{} (q_{\cpred_2}) & (1) \\
    \tau_4: & q_{\cpred_2} \arrow{ \erel(\atpos{x}{\epsilon}_1,~ \atpos{y}{\epsilon}_6) ~*~ \erel(\atpos{x}{\epsilon}_3,~ \atpos{y}{\epsilon}_6) ~*~
        \atpos{y}{\epsilon}_6 = \atpos{x}{1}_1 ~*~ \atpos{x}{\epsilon}_2 = \atpos{x}{1}_2 ~*~\atpos{x}{\epsilon}_3 = \atpos{x}{1}_3
    }{} (q_{\cpred_2}) & (\infty) \\
    \tau_5: & q_{\cpred_2} \arrow{ \erel(\atpos{x}{\epsilon}_1,~ \atpos{x}{\epsilon}_2) }{} () & (1)
  \end{array} \right. \]
  The $\asid$-models of
  $\apred$ have the structure depicted in \autoref{fig:star-chains}~(a), with elements that are values of
  persistent variables annotated by the name of the first occurrence
  of the persistent variable during the run.
\end{exa}

\begin{figure}[htbp]
  \begin{center}
    \begin{minipage}{.49\textwidth}
      \begin{center}
        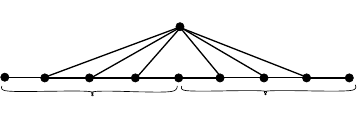

        \vspace*{-.3\baselineskip}
        \footnotesize{(a)}
      \end{center}
    \end{minipage}
    \begin{minipage}{.49\textwidth}
      \begin{center}
        \vspace*{1.5\baselineskip}
        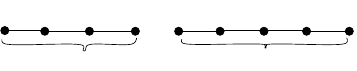

        \vspace*{-.3\baselineskip}
        \footnotesize{(b)}
      \end{center}
    \end{minipage}
  \end{center}
  \caption{\label{fig:star-chains} Models of the Running Example for the Elimination of Persistent Variables}
\end{figure}

\paragraph{I. Removing relation and disequality atoms from $1$-transitions}
This step replaces each symbol $\alpha$ that labels a $1$-transition
$q_0 \arrow{\alpha}{} (q_1,\ldots,q_\ell)$ of $\autsat{\asid}{\apred}$
with the symbol obtained by removing all relation and disequality atoms from
$\alpha$. The outcome of this transformation is denoted
$\auto{\asid}{\apred}^{I} \isdef (\Sigma,\states_\asid,q_\apred,\trans^I_\asid)$.

\begin{exa}\label{ex:removing-relation-atoms}
  (continued from \autoref{ex:persistent-variables-elimination})
  The result of removing the relation and disequality atoms from the
  $1$-transitions of the choice-free automaton given in
  \autoref{ex:persistent-variables-elimination} is shown below:
  \[\auto{\asid}{\apred}^{I} \left\{\begin{array}{rll}
      \tau_1: & q_\apred \arrow{\atpos{y}{\epsilon}_1 = \atpos{x}{1}_1 ~*~ \atpos{y}{\epsilon}_2 = \atpos{x}{1}_2 ~*~\atpos{y}{\epsilon}_3 = \atpos{x}{1}_3}{} (q_{\cpred_1}) & (1) \\
      \tau_2: & q_{\cpred_1} \arrow{ \erel(\atpos{x}{\epsilon}_1,~ \atpos{y}{\epsilon}_4) ~*~ \erel(\atpos{x}{\epsilon}_3,~ \atpos{y}{\epsilon}_4) ~*~
        \atpos{y}{\epsilon}_4 = \atpos{x}{1}_1 ~*~ \atpos{x}{\epsilon}_2 = \atpos{x}{1}_2 ~*~\atpos{x}{\epsilon}_3 = \atpos{x}{1}_3
      }{} (q_{\cpred_1}) & (\infty) \\
      \tau_3: & q_{\cpred_1} \arrow{
        \atpos{x}{\epsilon}_2 = \atpos{x}{1}_1 ~*~ \atpos{y}{\epsilon}_5 = \atpos{x}{1}_2 ~*~\atpos{x}{\epsilon}_3 = \atpos{x}{1}_3
      }{} (q_{\cpred_2}) & (1) \\
      \tau_4: & q_{\cpred_2} \arrow{ \erel(\atpos{x}{\epsilon}_1,~ \atpos{y}{\epsilon}_6) ~*~ \erel(\atpos{x}{\epsilon}_3,~ \atpos{y}{\epsilon}_6) ~*~
        \atpos{y}{\epsilon}_6 = \atpos{x}{1}_1 ~*~ \atpos{x}{\epsilon}_2 = \atpos{x}{1}_2 ~*~\atpos{x}{\epsilon}_3 = \atpos{x}{1}_3
      }{} (q_{\cpred_2}) & (\infty) \\
      \tau_5: & q_{\cpred_2} \arrow{ \emp }{} () & (1)
      \end{array}\right.\]
\end{exa}

The formal properties of $\auto{\asid}{\apred}^{I}$ are stated and
proved below. Note that $\auto{\asid}{\apred}^{I}$ is choice-free,
because $\auto{\asid}{\apred}$ is choice-free, i.e., the re-labeling
of the transitions of $\auto{\asid}{\apred}$ does not change the
structure of its SCC graph. In the statement of this and the upcoming
lemmas, we assume addition and order within the set of natural numbers
with infinity, i.e., $n+\infty=\infty+n=\infty$, $n\leq\infty$ and
$\infty\leq\infty$, for all $n \in \nat$.

\begin{lem}\label{lemma:remove-relations} \hfill
  \begin{enumerate}
  \item\label{it1:lemma:remove-relations} $\auto{\asid}{\apred}^{I}$ is all-satisfiable
  \item\label{it2:lemma:remove-relations} $\twof{\sem{\auto{\asid}{\apred}^{I}}} \le
    \twof{\sem{\auto{\asid}{\apred}}} + \cardof{\trans^1_\asid} \cdot
    (\maxrelatominruleof{\asid} + \maxvarinruleof{\asid})$
  \item\label{it3:lemma:remove-relations} $\twof{\sem{\auto{\asid}{\apred}}} \le
    \twof{\sem{\auto{\asid}{\apred}^{I}}} +  \cardof{\trans^1_\asid} \cdot \maxvarinruleof{\asid}$
  \end{enumerate}
\end{lem}
\begin{proof}
  For space reasons, the proof of the lemma is given in \autoref{app:remove-relations}.
\end{proof}
  
\paragraph{II. Removing equalities involving non-persistent variables}
At this point, the labels of the $1$-transitions of
$\auto{\asid}{\apred}^{I}$ consist of equalities only. We now remove
the equalities that would be lost when adding resets before and after
1-transitions that is, we forget equalities involving non-persistent
variables while keeping equalities between persistent ones. To this
end, we modify the label of each $1$-transition $q_0 \arrow{\alpha}{}
(q_1,\ldots,q_\ell)$ of $\auto{\asid}{\apred}^{I}$ in two
steps: \begin{enumerate}
\item\label{step1:remove-equalities} for each non-persistent
  $\epsilon$-variable $\atpos{x}{\epsilon}_j$, i.e., $j \in
  \interv{1}{\arityof{q_0}} \setminus
  \profile{\auto{\asid}{\apred}^{I}}(q_0)$, occurring in $\alpha$ in
  some equality with a persistent $i$-variable, i.e.,
  $\atpos{x}{\epsilon}_j = \atpos{x}{i}_k$, $k \in
  \profile{\auto{\asid}{\apred}^{I}}(q_i)$, we substitute
  $\atpos{x}{\epsilon}_j$ in $\alpha$ with a fresh variable
  $\atpos{y}{\epsilon} \not\in \fv{\alpha}$,
\item\label{step2:remove-equalities} remove each equality involving a
  non-persistent variable $\atpos{x}{i}_j$, for $i = \epsilon$ and $j
  \in \interv{1}{\arityof{q_0}}$, or $i\in \interv{1}{\ell}$ and $j
  \in \interv{1}{\arityof{q_i}}$.
  %
\end{enumerate}

\begin{exa}\label{ex:remove-equalities}
  (continued from \autoref{ex:removing-relation-atoms}) The profile of
  the automaton $\auto{\asid}{\apred}^I$ from
  \autoref{ex:removing-relation-atoms} is
  \(\profile{\auto{\asid}{\apred}^I}(q_\apred) = \emptyset\) and
  \(\profile{\auto{\asid}{\apred}^I}(q_{\cpred_1}) =
  \profile{\auto{\asid}{\apred}^I}(q_{\cpred_2}) = \set{2,3}\). By
  removing the equalities involving non-persistent variables, we
  obtain $\auto{\asid}{\apred}^{II}$ from $\auto{\asid}{\apred}^I$ by
  transforming only the 1-transition $\tau_3$ into:
    \[\tau_3': ~~ q_{\cpred_1} \arrow{
        \atpos{y}{\epsilon}_5 = \atpos{x}{1}_2
        ~*~\atpos{x}{\epsilon}_3 = \atpos{x}{1}_3 }{}
    (q_{\cpred_2}) \hspace{2cm} (1) \] Since $\atpos{x}{1}_1$ is
    non-persistent, the equality $\atpos{x}{\epsilon}_2 =
    \atpos{x}{1}_1$ is removed from the label of the original
    transition. Note that the equalities $\atpos{y}{\epsilon}_5 =
    \atpos{x}{1}_2$ and $\atpos{x}{\epsilon}_3 = \atpos{x}{1}_3$ are
    kept, because $\atpos{x}{1}_2$, $\atpos{x}{\epsilon}_3$ and
    $\atpos{x}{1}_3$ are persistent variables. Note that
    $\atpos{y}{\epsilon}_5$ is not persistent, because it is not
    associated with a state in $\tau'_3$.
\end{exa}

The result is the choice-free automaton
$\auto{\asid}{\apred}^{II}\isdef(\Sigma,\states_\asid,q_\apred,\trans_\asid^{II})$,
whose properties are stated and proved below:

\begin{lem}\label{lemma:remove-equalities} \hfill
  \begin{enumerate}
  \item\label{it1:lemma:remove-equalities} $\auto{\asid}{\apred}^{II}$ is all-satisfiable,
  \item\label{it2:lemma:remove-equalities}
    $\twof{\sem{\auto{\asid}{\apred}^{II}}} \le
    \twof{\sem{\auto{\asid}{\apred}^{I}}} +
    \cardof{{(\trans^{I}_\asid)}^1}\cdot \maxvarinruleof{\asid}$,
  \item\label{it3:lemma:remove-equalities}
    $\twof{\sem{\auto{\asid}{\apred}^{I}}} \le
    \twof{\sem{\auto{\asid}{\apred}^{II}}}$.
  \end{enumerate}
\end{lem}
\begin{proof}
  For space reasons, the proof of this lemma is given in \autoref{app:remove-equalities}.
\end{proof}
  
\paragraph{III. Removing persistent variables}
We build from the choice-free automaton $\auto{\asid}{\apred}^{II}$ a set of
choice-free automata $\overline{\mathcal{B}}_1, \ldots,
\overline{\mathcal{B}}_m$, \emph{having no persistent variables within
  the transition labels}, such that $\sem{\auto{\asid}{\apred}^{II}}$
is treewidth bounded if and only if $\sem{\overline{\mathcal{B}}_i}$
is treewidth bounded, for each $i \in \interv{1}{m}$.

We recall that each $1$-transition of a choice free automata occurs
exactly once in each accepting run over a $\Sigma$-labeled tree $t$
and each such occurrence corresponds to one subformula $t(p)^p$ of
$\charform{t}$, for a position $p \in \dom{t}$
(\autoref{lemma:one-transitions}). Using renaming, if necessary, we
can assume that the $\epsilon$-variables that are not associated with
the states of the transition have distinct names between the labels of
the $1$-transitions of $\auto{\asid}{\apred}^{II}$ and let
$\mathcal{Y} \isdef \set{\atpos{y}{\epsilon}_1, \ldots,
  \atpos{y}{\epsilon}_{\mathcal{M}}}$ denote their set, in the
following. For instance, we have $\mathcal{Y} =
\set{\atpos{y}{\epsilon}_1, \atpos{y}{\epsilon}_2,
  \atpos{y}{\epsilon}_3, \atpos{y}{\epsilon}_5}$ in Example
\ref{ex:removing-relation-atoms}. The transformation is done in three
steps: \begin{enumerate}[(A)]
\item\label{it1:remove-persistent} We annotate each state $q$ of
  $\auto{\asid}{\apred}^{II}$ with an injective partial function $a :
  \interv{1}{\arityof{q}} \rightarrow \interv{1}{\mathcal{M}}$ that
  maps each persistent variable $\atpos{x}{i}_j$, associated with $q$,
  to a variable $\atpos{y}{\epsilon}_{a(j)} \in \mathcal{Y}$, such
  that $\atpos{x}{i}_j \eqof{\charform{t}} \atpos{y}{\epsilon}_{a(j)}$
  holds for each tree $t\in\langof{}{\auto{\asid}{\apred}^{II}}$.
\item\label{it2:remove-persistent} We split the automaton obtained
  from the annotation of $\auto{\asid}{\apred}^{II}$ into several
  choice-free automata $\widetilde{\mathcal{B}}_1, \ldots,
  \widetilde{\mathcal{B}}_m$ such that
  $\langof{}{\auto{\asid}{\apred}^{II}} = \bigcup_{i=1}^m
  \langof{}{\widetilde{\mathcal{B}}_i}$.
\item\label{it3:remove-persistent} The annotation of the states in
  each $\widetilde{\mathcal{B}}_i$ is used to replace each occurrence
  of a relation atom $\arel(\atpos{z}{\epsilon}_1, \ldots,
  \atpos{z}{\epsilon}_{\arityof{\arel}})$, occurring within
  the label of an annotated transition $(q_0,a_0) \arrow{\alpha}{}
  ((q_1,a_1), \ldots, (q_\ell,a_\ell))$, with a fresh relation atom
  $\arel_g(\atpos{z}{\epsilon}_{i_1}, \ldots,
  \atpos{z}{\epsilon}_{i_k})$, where $g : \interv{1}{\arityof{\arel}}
  \rightarrow \interv{1}{\mathcal{M}}$ maps each persistent variable
  from the set $\set{\atpos{z}{\epsilon}_1, \ldots,
    \atpos{z}{\epsilon}_{\arityof{\arel}}}$ to its
  corresponding variable from $\mathcal{Y}$ and
  $\set{\atpos{z}{\epsilon}_{i_1}, \ldots, \atpos{z}{\epsilon}_{i_k}}$
  are the remaining non-persistent variables. The persistent variables
  are subsequently removed from $\alpha$ and the remaining variables
  are renamed according to the conventions from
  \autoref{def:alpha-sid}. Consequently, the arities of the
  states $(q_i,a_i)$, $i \in \interv{0}{\ell}$ are changed as well.
\end{enumerate}

\begin{exa}\label{ex:remove-persistent-variables}
  (continued from \autoref{ex:remove-equalities}) Let us consider the
  automaton $\auto{\asid}{\apred}^{II}$ from
  \autoref{ex:remove-equalities}. The $\epsilon$-variables from the
  labels of $1$-transitions, that are not associated with states
  thereof, are $\atpos{y}{\epsilon}_1, \atpos{y}{\epsilon}_2,
  \atpos{y}{\epsilon}_3, \atpos{y}{\epsilon}_5$, renamed as
  $\mathcal{Y}=\set{\atpos{y}{\epsilon}_1, \atpos{y}{\epsilon}_2,
    \atpos{y}{\epsilon}_3, \atpos{y}{\epsilon}_4}$, respectively. We
  recall that $\profile{\auto{\asid}{\apred}^{II}}(q_\apred) =
  \emptyset$ and $\profile{\auto{\asid}{\apred}^{II}}(q_{\cpred_1}) =
  \profile{\auto{\asid}{\apred}^{II}}(q_{\cpred_2}) = \set{2,3}$ is
  the profile of $\auto{\asid}{\apred}^{II}$.  The automaton obtained
  by annotating the states of $\auto{\asid}{\apred}^{II}$ with
  assignments is already choice-free and the result of the elimination
  of persistent variables is shown below:
  \[\widetilde{\mathcal{B}} \left\{\begin{array}{lll}
  \tau_1: & (q_\apred,\emptyset) \arrow{\emp}{} ((q_{\cpred_1},a)) & (1) \\
  \tau_2: & (q_{\cpred_1},a) \arrow{ \erel(\atpos{x}{\epsilon}_1,~ \atpos{y}{\epsilon}_4) ~*~
    \textcolor{red}{\erel_g(\atpos{y}{\epsilon}_4)} ~*~
    \atpos{y}{\epsilon}_4 = \atpos{x}{1}_1
  }{} ((q_{\cpred_1},a)) & (\infty) \\
  \tau_3: & (q_{\cpred_1},a) \arrow{
    \emp
  }{} ((q_{\cpred_2},a)) & (1) \\
  \tau_4: & (q_{\cpred_2},a) \arrow{ \erel(\atpos{x}{\epsilon}_1,~ \atpos{y}{\epsilon}_6) ~*~
    \textcolor{red}{\erel_g(\atpos{y}{\epsilon}_6)} ~*~
    \atpos{y}{\epsilon}_6 = \atpos{x}{1}_1
  }{} ((q_{\cpred_2},a)) & (\infty) \\
  \tau_5: & (q_{\cpred_2},a) \arrow{ \emp }{} () & (1)
  \end{array}\right.\]
  where: \begin{itemize}[label=$\triangleright$]
  \item $a : \interv{1}{3} \rightarrow \interv{1}{4}$ is the partial
    mapping defined as $a(2)=2$, $a(3)=3$ and undefined at $1$,
  \item $g : \interv{1}{2} \rightarrow \interv{1}{4}$ is the partial
    mapping defined as $g(1)=3$ and undefined at $2$.
  \end{itemize}
  Note that, because the equality between the persistent variables
  $\atpos{x}{\epsilon}_3$ and $\atpos{x}{1}_3$ has been kept in
  $\auto{\asid}{\apred}^{II}$ (see
  \autoref{ex:remove-equalities}), both variables are mapped by
  $g$ to the same variable $\atpos{y}{\epsilon}_3$, hence the same
  relation symbol $\erel_g$ replaces both
  $\erel(\atpos{x}{\epsilon}_3,\atpos{y}{\epsilon}_4)$ in $\tau_2$ and
  $\erel(\atpos{x}{\epsilon}_3,\atpos{y}{\epsilon}_6)$ in
  $\tau_4$. \autoref{fig:star-chains}~(b) shows the shape of the
  structures from $\sem{\widetilde{\mathcal{B}}}$. Since all but the
  first elements in both the $\cpred_1$ and $\cpred_2$ chains are now
  labeled with the same unary relation symbol $\erel_g$, these
  structures are of treewidth at most two.
\end{exa}

We recall that $\auto{\asid}{\apred}^{II} =
(\Sigma,\states_\asid^I,q_\apred,\trans_\asid^{II})$ and describe the
transformation formally:

\skipnoindent {\textbf{\ref{it1:remove-persistent}}} Let
$\widetilde{\auto{\asid}{\apred}}^{II} \isdef (\Sigma,
\widetilde{\states}^{I}_\asid, (q_\apred,\emptyset),
\widetilde{\trans}^{II}_\asid)$ be the automaton, whose set of states
is: \[\widetilde{\states}_\asid^{I} \isdef \set{(q,a) \mid q \in
  \states_\asid^I,~ a : \interv{1}{\arityof{q}} \rightarrow
  \interv{1}{\mathcal{M}} \text{ is a partial injective mapping}}\]
The initial state of $\widetilde{\auto{\asid}{\apred}}^{II}$ consists
of the initial state $q_\apred$ of $\auto{\asid}{\apred}^{II}$
annotated with the empty mapping, because we have considered
$\arityof{q_\apred}=0$. The set $\widetilde{\trans}^{II}_\asid$
contains a transition \((q_0,a_0) \arrow{\alpha}{} ((q_1,a_1), \ldots,
(q_\ell,a_\ell))\) if and only if either one of the following holds
(by \autoref{def:choice-free}, these conditions are
exclusive): \begin{itemize}[label=$\triangleright$]
\item \underline{$q_0 \arrow{\alpha}{} (q_1,\ldots,q_\ell) \in
  {(\trans_\asid^{II})}^1$}: in this case, for all $k \in
  \interv{1}{\ell}$ and $i \in
  \profile{\auto{\asid}{\apred}^{II}}(q_k)$, we define:
  \[a_k(i) \isdef \left\{\begin{array}{ll}
  a_0(j) & \text{if there exists } j \text{ such that }
  \atpos{x}{k}_i \eqof{\alpha} \atpos{x}{\epsilon}_j, \\
  m & \text{else, if } m \text{ is such that } \atpos{x}{k}_i
  \eqof{\alpha} \atpos{y}{\epsilon}_m
  \end{array}\right.\]
  Note that $a_k$ is well defined, because each $i$-variable is
  equated to a unique $\epsilon$-variable in the definition of
  $\auto{\asid}{\apred}$ and this fact is unchanged by the
  constructions of $\auto{\asid}{\apred}^{I}$ and
  $\auto{\asid}{\apred}^{II}$.
  %
\item \underline{$q_0 \arrow{\alpha}{} (q_1,\ldots,q_\ell) \in
  (\trans_\asid^{II})^\infty$}: in this case, for all $k \in
  \interv{1}{\ell}$ and $i \in \interv{1}{\arityof{q_k}}$, we define
  $a_k(i) \isdef a_0(j)$ if there exists $j$, such that
  $\atpos{x}{k}_i \eqof{\alpha} \atpos{x}{\epsilon}_j$ and undefined,
  otherwise.
\end{itemize}
The property of the first step of persistent variable elimination is
summarized below:
\begin{lem}\label{lemma:remove-persistent:A}
  $\langof{}{\widetilde{\auto{\asid}{\apred}}^{II}} = \langof{}{\auto{\asid}{\apred}^{II}}$.
\end{lem}
\begin{proof}
  Let $h : \widetilde{\states}_\asid^{I} \rightarrow
  \states_\asid^{I}$ be the function defined as $h((q,a))\isdef q$,
  for all $(q,a) \in \widetilde{\states}_\asid^{I}$. We show that $h$
  is a refinement. By \autoref{lemma:refinement}, we obtain
  $\langof{}{\widetilde{\auto{\asid}{\apred}}^{II}} =
  \langof{}{\auto{\asid}{\apred}^{II}}$. We prove the three points of
  \autoref{def:refinement}:
  \begin{itemize}[left=.5\parindent]
  \item[(\ref{it1:def:refinement})] $h((q_\apred,\emptyset)) =
    q_\apred$, by the definition of $h$.
  \item[(\ref{it2:def:refinement})] $(q_0,a_0) \arrow{\alpha}{}
    ((q_1,a_1), \ldots, (q_\ell,a_\ell)) \in
    \widetilde{\trans}^{II}_\asid$ only if $q_0 \arrow{\alpha}{}
    (q_1,\ldots,q_\ell) \in \trans^{II}_\asid$, by the definition of
    $\widetilde{\auto{\asid}{\apred}}^{II}$.
  \item[(\ref{it3:def:refinement})] Let $q_0 \arrow{\alpha}{}
  (q_1,\ldots,q_\ell) \in \trans^{II}_\asid$ be a transition and
  assume w.l.o.g. that $\auto{\asid}{\apred}^{II}$ is trim (if this is
  not the case, a trim automaton with the same language can be
  considered instead). Then, there exists an accepting run $\theta$ of
  $\auto{\asid}{\apred}^{II}$, such that the transition $q_0
  \arrow{\alpha}{} (q_1,\ldots,q_\ell)$ occurs on $\theta$. Each state
  $q$ on $\theta$ can be annotated with an injective partial mapping
  $a : \interv{1}{\arityof{q}} \rightarrow \interv{1}{\mathcal{M}}$
  and the result is an accepting run of
  $\widetilde{\auto{\asid}{\apred}}^{II}$. Hence there exist injective
  partial functions $a_i : \interv{1}{\arityof{q_i}} \rightarrow
  \interv{1}{\mathcal{M}}$, for $i \in \interv{0}{\ell}$ such that
  $(q_0,a_0) \arrow{\alpha}{} ((q_1,a_1), \ldots, (q_\ell,a_\ell)) \in
  \widetilde{\trans}^{II}_\asid$. Moreover, $(q_i,a_i) \in
  h^{-1}(q_i)$, for each $i \in \interv{0}{\ell}$. \qedhere
  \end{itemize}
\end{proof}

\skipnoindent \textbf{\ref{it2:remove-persistent}}
The problem, at this point, is that
$\widetilde{\auto{\asid}{\apred}}^{II}$ is not necessarily
choice-free, because annotating the states of
$\auto{\asid}{\apred}^{II}$ may cause several transitions to occur
between different linear SCCs. These transitions originate from the
same $1$-transition of $\auto{\asid}{\apred}^{II}$ and differ only in
the annotations added at step \ref{it1:remove-persistent}. We
circumvent this problem by decomposing
$\widetilde{\auto{\asid}{\apred}}^{II}$ into choice-free automata
$\widetilde{\mathcal{B}}_1, \ldots, \widetilde{\mathcal{B}}_m$, such
that $\langof{}{\widetilde{\auto{\asid}{\apred}}^{II}} =
\bigcup_{i=1}^m \langof{}{\widetilde{\mathcal{B}}_i}$. To this end, we
choose sets ${\widetilde{\trans}}^1_1, \ldots,
{\widetilde{\trans}}^1_m$, such that:
\begin{itemize}[label=$\triangleright$]
\item $(\widetilde{\trans}^{II})^1 = \bigcup_{i=1}^m
  {\widetilde{\trans}}^1_i$, and
\item for each $i \in \interv{1}{m}$, the set
  ${\widetilde{\trans}}^1_i$ contains exactly one transition
  $(q_0,a_0) \arrow{\alpha}{} ((q_1,a_1), \ldots, (q_\ell,a_\ell))$
  from $\widetilde{\trans}_\asid^{II}$, for each transition $q_0
  \arrow{\alpha}{} (q_1,\ldots,q_\ell)$ from $(\trans_\asid^{II})^1$.
\end{itemize}
Moreover, we define the set:
\[{(\overline{\trans})}^\infty \isdef \set{(q_0,a_0)
  \arrow{\alpha}{} ((q_1,a_1), \ldots, (q_\ell,a_\ell))
  \in \widetilde{\trans}_\asid^{II} \mid q_0 \arrow{\alpha}{}
  (q_1,\ldots,q_\ell) \in {(\trans_\asid^{II})}^\infty}\]
For each $i \in \interv{1}{m}$, let $\widetilde{\mathcal{B}}_i \isdef
(\Sigma, \widetilde{\states}_\asid^{I}, (q_\apred,\emptyset),
{\widetilde{\trans}}^1_i \uplus {\widetilde{\trans}}^\infty)$.
We prove below that $\widetilde{\mathcal{B}}_1, \ldots,
\widetilde{\mathcal{B}}_m$ is indeed a choice-free decomposition of
$\widetilde{\auto{\asid}{\apred}}^{II}$:

\begin{lem}\label{lemma:remove-persistent:B} \hfill
  \begin{enumerate}
  \item\label{it1:lemma:remove-persistent:B} $\widetilde{\mathcal{B}}_i$
    is all-satisfiable and choice-free, for $i \in \interv{1}{m}$.
  \item\label{it2:lemma:remove-persistent:B} $\langof{}{\widetilde{\auto{\asid}{\apred}}^{II}} =
    \bigcup_{i=1}^m \langof{}{\widetilde{\mathcal{B}}_i}$.
  \end{enumerate}
\end{lem}
\begin{proof}
  For space reasons, the proof of this lemma is given in \autoref{app:remove-persistent:B}.
\end{proof}

\skipnoindent \textbf{\ref{it3:remove-persistent}} Let us fix a
choice-free automaton $\widetilde{\mathcal{B}} =
(\Sigma,\widetilde{\states}_\asid^{I},
(q_\apred,\emptyset),{\widetilde{\trans}}^1 \uplus
{\widetilde{\trans}}^\infty)$ among $\widetilde{\mathcal{B}_1},\ldots
\widetilde{\mathcal{B}_m}$. Consider an arbitrary transition $\tau ~:~
(q_0,a_0) \arrow{\alpha}{} ((q_1,a_1), \ldots, (q_\ell,a_\ell))$ $\in$
${\widetilde{\trans}}^1 \uplus {\widetilde{\trans}}^\infty$.  We
denote by $P_0 \isdef \set{\atpos{x}{\epsilon}_j \mid j \in \dom{a_0}}
\cup \set{\atpos{y}{\epsilon}_j \in \fv{\alpha} \mid \tau \in
  {\widetilde{\trans}}^1}$, $P_i \isdef \set{\atpos{x}{i}_j \mid j \in
  \dom{a_i}}$, for $i \in \interv{1}{\ell}$ and $P \isdef
\bigcup_{i=0}^\ell P_i$ the set of persistent variables occurring in
$\alpha$. The goal of this step is to remove from $\alpha$ the
variables from $P$. The outcome of this transformation of $\alpha$
will be denoted as $\overline{\alpha}$. In order to guarantee the
preservation of the naming conventions from \autoref{def:alpha-sid},
the remaining (non-persistent) variables from $\fv{\alpha} \setminus
P$ are renamed using the injective mapping $\zeta_\tau :
\fv{\alpha}\setminus P \rightarrow \fv{\alpha}$ defined as follows:
\begin{itemize}[label=$\triangleright$]
\item $\zeta_\tau(\atpos{x}{i}_k) \isdef \atpos{x}{i}_{k - p}$ where
  $p \isdef \cardof{ \set{\atpos{x}{i}_j \mid j < k} \cap P_i}$ for all $i
  \in \interv{0}{\ell}$, $k \in \interv{1}{n_i}$,
  $\atpos{x}{i}_k \not\in P_i$,
\item $\zeta_\tau(\atpos{y}{\epsilon}_k) = \atpos{y}{\epsilon}_k$, for
  all $k \in \interv{1}{m}$, $\atpos{y}{\epsilon}_k \not\in P_0$
\end{itemize}
where $m, n_0, \ldots, n_\ell$ are the numbers of
$\atpos{y}{\epsilon}_j$, $\atpos{x}{\epsilon}_j, \ldots,
\atpos{x}{i}_j$, respectively (see \autoref{def:alpha-sid}).
Intuitively, the renaming \emph{shifts to the left} the $j$ indices of
the $\atpos{x}{i}_{j}$ variables so that the persistent variables
indexed according to the assignments $a_i$ are ignored. The
transformation of $\alpha$ to $\overline{\alpha}$ is described for
each atom of $\alpha$, as follows:
\begin{itemize}[label=$\triangleright$]
\item every relation atom $\arel(\atpos{z}{\epsilon}_1, \ldots,
  \atpos{z}{\epsilon}_{\arityof{\arel}})$ is replaced by
  $\arel_g(\zeta_\tau(\atpos{z}{\epsilon}_{i_1}), \ldots,
  \zeta_\tau(\atpos{z}{\epsilon}_{i_k}))$, where $\arel_g$ is a fresh
  relation symbol of arity $k$ such that $\set{z_{i_1}, \ldots,
    z_{i_k}} \isdef \set{z_1,\ldots,z_{\arityof{\arel}}} \setminus
  P_0$ and $g : \interv{1}{\arityof{\arel}} \rightarrow
  \interv{1}{\mathcal{M}}$ is the partial function:
  \[ g(j) \isdef \left\{ \begin{array}{l}
    a_0(j) \text{, if } \atpos{z}{\epsilon}_j \in P_0 \\
    \text{undefined, otherwise}
  \end{array} \right.
  \]
  Note that the arity of $\arel_g$ is at least one, for the following
  reason. Since relation atoms occur only in $\infty$-transitions,
  they can repeat arbitrary many times in characteristic formul{\ae}
  over runs. Hence, if such an atom has only persistent variables,
  these characteristic formul{\ae} will become unsatisfiable if the
  atom repeat more than twice. This contradicts, however, the
  hypothesis that $\widetilde{\mathcal{B}}_i$ is all-satisfiable.

\item every (dis-)equality atom $x \sim y$ is replaced by
  $\zeta_\tau(x) \sim \zeta_\tau(y)$ for $\sim \in \set{=,\neq}$ if
  $\set{x,y} \cap P = \emptyset$ and removed otherwise.  In
  particular, note that there is no equality in $\alpha$ between a
  variable in $P$ and another one not in $P$ due to elimination of
  equalities with non-persistent variables in 1-transitions and to the
  rule of propagation through $\infty$-transitions.  Moreover, this
  transformation turns the label of each $1$-transition into $\emp$,
  because the labels of $1$-transitions contain only equalities
  involving persistent variables.
\end{itemize}

The result of this transformation of $\widetilde{\mathcal{B}}_i$ is
denoted $\overline{\mathcal{B}}_i$, for each $i \in \interv{1}{m}$.
Let $\overline{\mathcal{B}}$ be any of $\overline{\mathcal{B}}_1,
\ldots, \overline{\mathcal{B}}_m$ and $\widetilde{\mathcal{B}}$ be the
corresponding automaton before the removal of persistent
variables. The properties of this transformation are stated and proved
below:

\begin{lem}\label{lemma:remove-persistent} \hfill
  \begin{enumerate}
  \item\label{it1:lemma:remove-persistent} $\overline{\mathcal{B}}$
    is all-satisfiable and choice-free,
  \item\label{it2:lemma:remove-persistent} $\twof{\sem{\overline{\mathcal{B}}}} \le
    \twof{\sem{\widetilde{\mathcal{B}}}}$,
  \item\label{it3:lemma:remove-persistent}
    $\twof{\sem{\widetilde{\mathcal{B}}}} \le
    \twof{\sem{\overline{\mathcal{B}}}} +
    \cardof{\widetilde{\trans}^1} \cdot \maxvarinruleof{\asid}$.
\end{enumerate}
\end{lem}
\begin{proof}
  For space reasons, the proof of this lemma is given in \autoref{app:remove-persistent}.
\end{proof}
  
\subsection{Wrapping $1$-transitions}

In \autoref{sec:persistent-variables-elimination} we have transformed
a given automaton $\auto{\asid}{\apred}$ into choice-free automata
$\overline{\mathcal{B}}_1, \ldots, \overline{\mathcal{B}}_m$ without
persistent variables. We can assume w.l.o.g. that the $1$-transitions
of these automata are labeled by $\emp$, because the only remaining
equalities are between variables $\atpos{y}{\epsilon}$ not associated
with states and non-persistent variables $\atpos{x}{i}_j$. In the
following, let $\overline{\mathcal{B}} = (\Sigma, \states, \initstate,
\overline{\trans})$ be any of $\overline{\mathcal{B}}_1, \ldots,
\overline{\mathcal{B}}_m$ and $\overline{\trans}= \overline{\trans}^1
\uplus \overline{\trans}^\infty$ be the partition of its transitions
into $1$- and $\infty$-transitions (\autoref{def:choice-free}). In
order to obtain an expandable SID from $\overline{\mathcal{B}}$, i.e.,
using \autoref{lemma:sid-ta} (\ref{it2:lemma:sid-ta}), we must be able
to embed any sequence of runs of $\overline{\mathcal{B}}$ into a
single run of $\overline{\mathcal{B}}$. However, this is currently not
possible, because the labels of the $1$-transitions cannot be viewed
as labels of $\infty$-transitions. The problem is shown by the
following example:

\begin{exa}\label{ex:embedding}
  Consider two disjoint structures $\astruc_1 = (\univ_1,\struc_1),
  \astruc_2 = (\univ_2,\struc_2) \in \csem{\overline{\mathcal{B}}}{}$,
  having the shape depicted in \autoref{fig:star-chains}~(b). Each
  $\astruc_i$ has all but two elements, call them $e^1_i$ and $e^2_i$,
  labeled with a unary relation symbol $\erel_g$. Then, the
  composition $\astruc_1 \comp \astruc_2$ will have four unlabeled
  elements $e^1_1,e^2_1,e^1_2$ and $e^2_2$. Since any structure
  $\astruc=(\univ,\struc) \in \csem{\overline{\mathcal{B}}}{}$ has
  exactly two unlabeled elements, the structure $\astruc_1 \comp
  \astruc_2$ is not a substructure of $\astruc$, according to
  \autoref{def:substructure}. This is because
  $e^1_1,e^2_1,e^1_2,e^2_2 \in (\supp{\struc_1} \cup \supp{\struc_2})
  \setminus (\struc_1(\erel_g) \uplus \struc_2(\erel_g))$ are pairwise
  distinct, hence $\cardof{(\supp{\struc_1} \cup \supp{\struc_2})
    \setminus (\struc_1(\erel_g) \uplus \struc_2(\erel_g))} = 4$,
  whereas $\cardof{\supp{\struc}\setminus\struc(\erel_g)}=2$,
  independently of the choice of $\astruc$.
\end{exa}

For a $\Sigma$-labeled tree $t$, two positions $p$ and $s$, such that
$p \in \dom{t}$ (nothing is required about $s$), and a sequence of
variables $x_1, \ldots, x_k$, we define the formula:
\[
  \allstof{t}{p}{s}{x_1,\ldots,x_k} \isdef  \Bigstar\set{
    \arel(\atpos{x}{s}_1, \ldots, \atpos{x}{s}_k) \mid \arel(\atpos{x}{p_1}_1, \ldots, \atpos{x}{p_k}_k)
    \text{ occurs in } \charform{t},~ \atpos{x}{p_i}_i \eqof{\charform{t}} \atpos{x}{p}_i,~
    \forall i \in \interv{1}{k}}
\]
For simplicity, assume, for each transition $q_0 \arrow{\emp}{} (q_1,
\ldots, q_\ell) \in \overline{\trans}^1$ of $\overline{\mathcal{B}}$,
that $q_0, \ldots, q_\ell$ belong to non-trivial SCCs. If some $q_i$
belongs to a trivial SCC, i.e., with no transitions from $q_i$ to
itself or other states in the same SCC, then $\arity{q_i}$ must be
zero, because all parameters from its profile have been removed by the
previous transformation. In this case, the construction below is
adapted by replacing the formul{\ae}
$\allstof{t_0}{\epsilon}{\epsilon}{}$ and $\allstof{t_j}{p_j}{j}{}$ by
$\emp$, for the corresponding states that belong to trivial SCCs. The
automaton $\mathcal{B}$ is obtained from $\overline{\mathcal{B}}$ by
replacing the label of each such $1$-transition with the following
formula, for some trees $t_i$ corresponding to resets $\theta^i_{p_i
  \leftarrow q_i}$ of $\overline{\mathcal{B}}$, for $i \in
\interv{0}{\ell}$:
\[
  q_0 \arrow{
    \Bigstar_{i_1,\ldots,i_k \in \interv{1}{\arityof{q_0}}} \allstof{t_0}{\epsilon}{\epsilon}{x_{i_1}, \ldots, x_{i_k}} ~*~
    \Bigstar_{j \in \interv{1}{\ell},~ i_1, \ldots, i_k \in \interv{1}{\arityof{q_j}}} \allstof{t_j}{p_j}{j}{x_{i_1}, \ldots, x_{i_k}}
  }{} (q_1, \ldots, q_\ell)
\]
Note that the existence of such resets is guaranteed by \autoref{lemma:reset} and the previous assumption. This construction is
illustrated by the following example:

\begin{exa}\label{ex:wrapping}
  (continued from \autoref{ex:remove-persistent-variables})
  The outcome of applying the above transformation to the automaton
  $\overline{\mathcal{B}}$ from \autoref{ex:remove-persistent-variables} is the automaton $\mathcal{B}$
  given below:
    \[\mathcal{B} = \left\{\begin{array}{lll}
    \tau_1: & (q_\apred,\emptyset) \arrow{
      \textcolor{red}{\erel_g(\atpos{x}{1}_1)}
    }{} ((q_{\cpred_1},a)) & (1) \\
    \tau_2: & (q_{\cpred_1},a) \arrow{ \erel(\atpos{x}{\epsilon}_1,~ \atpos{y}{\epsilon}_4) ~*~ \erel_g(\atpos{y}{\epsilon}_4) ~*~
      \atpos{y}{\epsilon}_4 = \atpos{x}{1}_1
    }{} ((q_{\cpred_1},a)) & (\infty) \\
    \tau_3: & (q_{\cpred_1},a) \arrow{
      \textcolor{red}{\erel_g(\atpos{x}{1}_1)}
  }{} ((q_{\cpred_2},a)) & (1) \\
    \tau_4: & (q_{\cpred_2},a) \arrow{ \erel(\atpos{x}{\epsilon}_1,~ \atpos{y}{\epsilon}_6) ~*~ \erel_g(\atpos{y}{\epsilon}_6) ~*~
      \atpos{y}{\epsilon}_6 = \atpos{x}{1}_1
    }{} ((q_{\cpred_2},a)) & (\infty) \\
    \tau_5: & (q_{\cpred_2},a) \arrow{ \emp }{} () & (1)
    \end{array}\right.\]
    The transformation adds the relation atoms
    $\erel_g(\atpos{x}{1}_1)$ to the labels of $\tau_1$ and $\tau_3$,
    respectively. Note that $\tau_2$ and $\tau_4$ are
    $(q_{\cpred_1},a)$- and $(q_{\cpred_2},a)$-resets,
    respectively. The added relation atoms correspond to the relation
    atoms that label the $\atpos{x}{1}_1$ variable within the labels
    of these resets, taken backwards. All elements from a structure
    $\astruc \in \csem{\mathcal{B}}{}$ are now labeled by $\erel_g$,
    hence the composition of any sequence of structures from
    $\csem{\mathcal{B}}{}$ can be embedded as a substructure of a
    structure from the same set. In other words, the SID obtained from
    $\mathcal{B}$ using \autoref{lemma:sid-ta}
    (\ref{it2:lemma:sid-ta}) is expandable
    (\autoref{def:expandable}).
\end{exa}

The formal properties of $\mathcal{B}$ are stated and proved below:

\begin{lem}\label{lemma:exp} \hfill
  \begin{enumerate}
  \item\label{it1:lemma:exp} $\mathcal{B}$ is all-satisfiable,
  \item\label{it2:lemma:exp} $\twof{\sem{\mathcal{B}}} \le
    \twof{\sem{\overline{\mathcal{B}}}} + \cardof{\overline{\trans}^1}
    \cdot \maxvarinruleof{\asid}$,
  \item\label{it3:lemma:exp} $\twof{\sem{\overline{\mathcal{B}}}} \le
    \twof{\sem{\mathcal{B}}} + \cardof{\overline{\trans}^1} \cdot (1 +
    \maxrulearityof{\asid}) \cdot \relationsno{\asid} \cdot
    \maxpredarityof{\asid}^{\maxrelarityof{\asid}}$.
  \end{enumerate}
\end{lem}
\begin{proof}
  For space reasons, the proof of this lemma is given in \autoref{app:exp}.
\end{proof}
  
\subsection{The Proof of \autoref{lemma:expansion}}

We have collected all the ingredients needed to prove the decidability
of the treewdith boundedness problem, for \slr\ formul{\ae}
interpreted over general SIDs. The key point is the proof of
\autoref{lemma:expansion}, the main result being an immediate
consequence of this lemma and \autoref{thm:expandable-tb}. To prove
\autoref{lemma:expansion}, we rely on \autoref{lemma:runs-embedding},
which states that any sequence of $\infty$-runs of a choice-free
automaton can be disjointly embedded into an accepting run of the same
automaton, such that any two of the embedded runs are separated by a
given number of reset paths in the enclosing run.

Before proving \autoref{lemma:expansion}, we state a property of the
outcome of the above transformation of automata
(\autoref{fig:automata-expansion}). This property uses the following
notion:

\begin{defi}\label{def:view}
  Let $\mathcal{A}$ be a choice-free automaton. A \emph{view} for
  $\mathcal{A}$ is a tuple $\tuple{\arun,t,\store,\astruc}$, such that
  $\arun \in \runsof{\infty}{q}{\mathcal{A}}$ is a partial
  $\infty$-run over a $\Sigma$-labeled tree $t$, $\store$ is a
  canonical store for $\charform{t}$ and $\astruc$ is a structure,
  such that $\astruc \models^\store \charform{t}$. A structure
  $\astruc'=(\univ',\struc')$ is \emph{encapsulated} by the view
  $\tuple{\arun,t,\store,\astruc}$ if \begin{enumerate*}[(i)]
  \item\label{it1:def:view} $\astruc' \substruc \astruc$,
  \item\label{it2:def:view} $\supp{\struc'} \cap
    \store(\set{\atpos{x}{\epsilon}_1, \ldots,
      \atpos{x}{\epsilon}_{\arityof{\arun(\epsilon)}}}) = \emptyset$,
    and
  \item\label{it3:def:view} $\supp{\struc'} \cap
    \store(\set{\atpos{x}{p}_1, \ldots, \atpos{x}{p}_{\arityof{q}}}) =
    \emptyset$ if, moreover, the partial run $\theta$ is a
    $\theta_{p\leftarrow q}$ context.
  \end{enumerate*}
\end{defi}
Intuitively, a structure is encapsulated by a view if it is a
substructure of the structure in the view and it does not ``touch''
the values of the variables from the root (resp. frontier) point of
the partial run in the view.

Let $\mathcal{B}$ be any choice-free automaton resulting from the
transformation in \autoref{fig:automata-expansion}. The following
lemma shows that each canonical model from $\csem{\mathcal{B}}{}$ can
be decomposed into pairwise disjoint structures, each of which being
encapsulated by a separate view for $\mathcal{B}$:

\begin{lem}\label{lemma:exp-decomposition}
  For each structure $\astruc=(\univ,\struc)\in\csem{\mathcal{B}}{}$,
  there exist pairwise disjoint structures $\astruc_1, \ldots,
  \astruc_n$ and views $\tuple{\arun_1,t_1,\store_1,\astruc'_1},
  \ldots, \tuple{\arun_n,t_n,\store_n,\astruc'_n}$ for $\mathcal{B}$
  such that $\astruc'_1,\ldots,\astruc'_n$ are pairwise disjoint,
  $\astruc=\astruc_1\comp\ldots\comp\astruc_n$ and $\astruc_i$ is
  encapsulated by $\tuple{\arun_i,t_i,\store_i,\astruc'_i}$, for each
  $i \in \interv{1}{n}$.
\end{lem}
\begin{proof}
  Let $\mathcal{B}\isdef(\Sigma,\states,\initstate,\trans)$. Since
  $\astruc = (\univ,\struc)\in\csem{\mathcal{B}}{}$, there exists a
  tree $t_0 \in \langof{}{\mathcal{B}}$ such that $\astruc
  \models^{\store_0} \charform{t_0}$ for a store $\store_0$, that is
  canonical for $\charform{t_0}$.  Let $\arun_0$ be an accepting run
  of $\mathcal{B}$ over $t_0$. Because $\mathcal{B}$ is choice-free,
  $\arun_0$ can be decomposed into \begin{enumerate*}[(i)]
  \item maximal partial runs $\arun_{01} \in
    \runsof{\infty}{r_1}{\mathcal{B}}, \ldots, \arun_{0n} \in
    \runsof{\infty}{r_n}{\mathcal{B}}$ consisting of (arbitrarily
    many) connected $\infty$-transitions, and
  \item partial runs $\tau_1, \ldots, \tau_m$ consisting of
    a single 1-transition each,
  \end{enumerate*}
  such that $r_i$ is the state at the root of $\arun_{0i}$ and $q_i$
  is the left-hand side of a $1$-transition, for each $i \in
  \interv{1}{n}$. We refer to \autoref{fig:run-decomposition}~(a)
  for an illustration of the decomposition.

  For every $i\in\interv{1}{n}$ we define $\astruc_i$ as the
  substructure of $\astruc$ constructed along the maximal partial run
  $\arun_{0i}$.  That is, $\astruc_i$ contains all the relation atoms
  defined on $\infty$-transitions in $\arun_{0i}$ and the relation
  atoms defined on the entering (and possibly exiting) 1-transition(s)
  involving common variables for entering (resp. exiting) state(s).
  Intuitively, all these relation atoms occur in the gray part in
  \autoref{fig:run-decomposition}~(a).  Note that, since the
  $1$-transitions of $\mathcal{B}$ do not equate variables
  $\atpos{x}{i}_j$, for $i \in \nat\cup\set{\epsilon}$, the structures
  $\astruc_i$ are pairwise disjoint, and $\astruc = \astruc_1 \comp
  \ldots \comp \astruc_n$.

  \begin{figure}[htbp]
    \begin{center}
    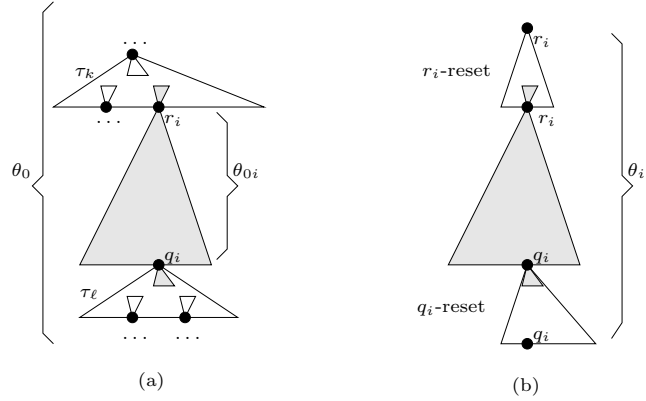
    \caption{\label{fig:run-decomposition} Run decomposition}
    \end{center}
  \end{figure}

  We build the views
  $\set{\tuple{\arun_i,t_i,\store_i,\astruc'_i}}_{i\in\interv{1}{n}}$
  as follows. For every $i\in\interv{1}{n}$, we define the partial run
  $\arun_i$ by extending the partial run $\arun_{0i}$ by the same
  $r_i$- and $q_i$-resets used in the transformation of
  $\overline{\mathcal{B}}$ into $\mathcal{B}$ (note that the
  $q_i$-reset is needed only if $\arun_{0i}$ reaches a
  $1$-transition). This construction is illustrated in
  \autoref{fig:run-decomposition}~(b). Let $t_i$ be the
  $\Sigma$-labeled tree corresponding to the run $\theta_i$.  Let
  $\store_i$ be a canonical store for $\charform{t_i}$ constructed by
  extending $\store_0$ such that $\store_0(\atpos{x}{p}) =
  \store_i(\atpos{x}{p'})$ whenever $p$ and $p'$ correspond to the
  same relative position within $\arun_{0i}$. In other words,
  $\store_i$ preserves the elements occurring in $\astruc_i$ and
  associates new distinct elements to all other variables introduced
  by the resets (in particular, all variables that are equated in
  $\charform{t_i}$ are mapped to the same element). The store
  $\store_i$ defines a structure $\astruc_i' = (\univ,\struc'_i)$,
  such that $\astruc_i' \models^{\store_i} \charform{t_i}$, as
  \(\struc'_i(\arel) \isdef \set{\tuple{\store_i(\atpos{z}{p_1}_1),
      \ldots, \store_i(\atpos{z}{p_k}_k)} \mid \arel(\atpos{z}{p_1}_1,
    \ldots, \atpos{z}{p_k}_k) \text{ occurs in } \charform{t_i}}\),
  for each relation symbol $\arel\in\relations$. We prove below that
  $\astruc_i$ is encapsulated by the view
  $\tuple{\arun_i,t_i,\store_i,\astruc'_i}$:
  \begin{itemize}[label=$\triangleright$]
  \item By the choice of $\store_i$, that extends $\store_0$ as
    explained above, we have $\astruc_i \subseteq \astruc_i'$.
    Moreover, no tuples are added to $\astruc_i$ by the inserted
    resets, because the labeling of $1$-transitions in $\mathcal{B}$
    is constructed precisely from the resets that guarantee this
    property, i.e., the construction of $\mathcal{B}$ guarantees that
    the set of relation atoms occurring in the reset was used to label
    the $1$-transitions. Thus, we have $\astruc_i \substruc
    \astruc'_i$.
  \item By the construction of the partial run $\arun_i$, the set of
    variables $\set{ \atpos{x_1}{\epsilon}, \ldots,
      \atpos{x_{\arityof{r_i}}}{\epsilon}}$ in $\arun_i$ are not
    related by equalities to any of the variables at the root of
    $\arun_{0i}$ in $\charform{t_i}$. In particular, this is ensured
    by the fact that $\mathcal{B}$ has no persistent variables. Since,
    $\store_i$ is canonical for $\charform{t_i}$, we obtain that
    $\supp{\struc_i} \cap \set{\store_i(\set{\atpos{x}{\epsilon}_1,
        \ldots, \atpos{x}{\epsilon}_{\arityof{r_i}}})} = \emptyset$.
    The same argument applies to the set of variables occurring at the
    frontier position of $\arun_i$, i.e., when $\arun_{0i}$ reaches
    the left-hand side $q_i$ of a $1$-transition, as in
    \autoref{fig:run-decomposition}~(b). \qedhere
  \end{itemize}
\end{proof}

Dually, the following lemma gathers pairwise disjoint structures into
a single rich canonical model, that meets the conditions of
expandability (\autoref{def:expandable}):

\begin{lem}\label{lemma:exp-composition}
  Let $\astruc_1, \ldots, \astruc_n$ be pairwise disjoint structures
  encapsulated by the views $\tuple{\arun_1,t_1,\store_1,\astruc'_1}$,
  $\ldots$, $\tuple{\arun_n,t_n,\store_n,\astruc'_n}$ for
  $\mathcal{B}$, where $\astruc'_1, \ldots, \astruc'_n$ are also
  pairwise disjoint. Then, there exists a rich canonical model
  $(\astruc,\diseq)\in\rcsem{\mathcal{B}}{}$, such that the conditions
  (\ref{it1:def:expandable}), (\ref{it2:def:expandable}) and
  (\ref{it3:def:expandable}) from \autoref{def:expandable} hold for
  $\astruc_1, \ldots, \astruc_n$ and $(\astruc,\diseq)$.
\end{lem}
\begin{proof}
  Since $\mathcal{B}$ is a choice-free automaton, by
  \autoref{lemma:runs-embedding}, there exists an accepting run
  $\arun$ of $\mathcal{B}$ over a tree $t$, such that the following
  hold: \begin{enumerate}[(a)]
  \item\label{it1:proof:exp-composition} each partial run
    $\arun_{i}$ is embedded in $\arun$ at some position $r_{i} \in
    \dom{\arun}$, for all $i \in \interv{1}{n}$,
  \item\label{it2:proof:exp-composition} $r_{i} \cdot
    \dom{\arun_{i}} \cap r_{j} \cdot \dom{\arun_{j}} =
    \emptyset$, for all $1 \le i < j \le n$,
  \item\label{it3:proof:exp-composition} each path between the
    positions $r_{i}$ and $r_{j}$ traverses at least once some reset
    path, disjoint from $\bigcup_{k=1}^n r_k \cdot \dom{\arun_k}$, for
    all $1 \le i < j \le n$.
  \end{enumerate}
  For each $i \in \interv{1}{n}$, since
  $\tuple{\arun_i,t_i,\store_i,\astruc'_i}$ is a view, we have
  $\astruc'_i \models^{\store_i} \charform{t_i}$, hence
  $\store_i(\fv{\charform{t_i}}) \subseteq \supp{\struc'_i}$, where we
  assume w.l.o.g. that $\astruc'_i \isdef (\univ,\struc'_i)$.  By
  point \ref{it2:proof:exp-composition} above, the subformul{\ae}
  corresponding to the subtrees of $t$ with domains $r_{i} \cdot
  \dom{\arun_{i}}$, for $i \in \interv{1}{n}$, have disjoint sets of
  free variables. We define a store $\store$ as follows, for each
  variable $\atpos{x}{p} \in \fv{\charform{t}}$:
  \begin{itemize}[label=$\triangleright$]
  \item if $p = r_ip'$ and $p'\in \dom{t_i}$, for some
    $i\in\interv{1}{n}$, then we set $\store(\atpos{x}{p}) \isdef
    \store_i(\atpos{x}{p'})$,
  \item otherwise, we chose a fresh value $\store(\atpos{x}{p})$, such that $\store(\atpos{x}{p}) \not\in
    \bigcup_{i=1}^n \supp{\struc'_i}$ and
    $\store(\atpos{x}{p})\neq\store(\atpos{z}{r})$, for each variable
    $\atpos{z}{r}$, such that $\atpos{x}{p} \not\eqof{\charform{t}}
    \atpos{z}{r}$.
  \end{itemize}
  By the fact that $\supp{\struc'_i} \cap \supp{\struc'_j} =
  \emptyset$, i.e., $\store_i(\fv{\charform{t_i}}) \cap
  \store_j(\fv{\charform{t_j}}) = \emptyset$, for all $1 \le i < j
  \le n$, and the construction of $\store$, we obtain that $\store$
  is canonical for $\charform{t}$.  The store $\store$ defines the
  structure $\astruc=(\univ,\struc)$, as follows:
  \[\struc(\arel)\isdef\set{\tuple{\store(z_1), \ldots, \store(z_{\arityof{\arel}})}
    \mid \arel(z_1, \ldots,z_{\arityof{\arel}}) \text{ occurs in }
    \charform{t}} \text{, for all } \arel\in\relations\] By the
  definition of $\astruc$, we have $\astruc \models^\store
  \charform{t}$. Moreover, we define the relation: \[\diseq \isdef
  \set{(\store(x),\store(y)) \mid x \neq y \text{ or } y \neq x \text{
      occurs in } \charform{t}}\] We have obtained a rich canonical
  model $(\astruc,\diseq)\in\rcsem{\exclof{\charform{t}}}{}$ and, since $t
  \in \langof{}{\mathcal{B}}$, we have $(\astruc,\diseq)\in
  \rcsem{\mathcal{B}}{}$. We prove below the three conditions from
  \autoref{def:expandable}: \begin{itemize}[left=.5\parindent]
  \item[(\ref{it1:def:expandable})] By the construction of $\astruc$,
    we have $\astruc'_i \subseteq \astruc$, for all $i \in
    \interv{1}{n}$. Since $\astruc_1, \ldots, \astruc_n$ are pairwise
    disjoint, their composition is defined, hence $\astruc_1 \comp
    \ldots \comp \astruc_n \subseteq \astruc'_1 \comp \ldots \comp
    \astruc'_n \subseteq \astruc$. W.l.o.g., let $\astruc_i =
    (\univ,\struc_i)$, for all $i\in\interv{1}{n}$. To prove
    $\astruc_1 \comp \ldots \comp \astruc_n \substruc \astruc$, by
    \autoref{def:substructure}, we must prove that:
    \begin{equation}\label{eq1:def:expandable}
      \struc_1(\arel) \uplus \ldots \uplus \struc_n(\arel) =
      \set{\tuple{u_1, \ldots, u_{\arityof{\arel}}} \in \struc(\arel)
        \mid u_1, \ldots, u_{\arityof{\arel}} \in \supp{\struc_1} \cup
        \ldots \cup \supp{\struc_n}}
    \end{equation}
    The ``$\subseteq$'' direction follows from $\astruc_1 \comp \ldots
    \comp \astruc_n \subseteq \astruc$, hence we are left with proving
    the dual ``$\supseteq$'' direction. Let $\tuple{u_1, \ldots,
      u_{\arityof{\arel}}} \in \struc(\arel)$ be a tuple, such that
    $u_1, \ldots, u_{\arityof{\arel}} \in \bigcup_{i=1}^n
    \supp{\struc_i}$. By the definition of $\struc$, there exists a
    relation atom $\arel(\atpos{z}{p}_1, \ldots,
    \atpos{z}{p}_{\arityof{\arel}})$ in $\charform{t}$, such that
    $\store(z_i)=u_i$, for all $i \in \interv{1}{\arityof{\arel}}$. To
    simplify matters, we assume that the position of each variable in
    the relation atom is the same, the case where these positions are
    either $p$ and $pi$, or $pi$ and $pj$, for some $p\in\nat^*$ and
    $i \neq j \in \nat$ is treated in a similar way and left to the
    reader. Moreover, for each $i \in \interv{1}{\arityof{\arel}}$,
    there exists a unique $k_i \in \interv{1}{n}$, such that $u_i \in
    \supp{\struc_{k_i}}$. Suppose, for a contradiction, that $k_i \neq
    k_j$, for some $1 \le i < j \le n$. Then, there exist paths
    between $p$ and some positions $s_i \in r_i \cdot \dom{t_i}$ and
    $s_j \in r_j \cdot \dom{t_j}$, such that $\atpos{z}{p}_i
    \eqof{\charform{t}} \atpos{\xi}{s_i}_i$ and $\atpos{z}{p}_j
    \eqof{\charform{t}} \atpos{\xi}{s_j}_j$. Consider the case where
    $\arun_i$ and $\arun_j$ are runs (the case where one of them is a
    context uses a similar argument and is left to the reader). Since
    $t_i$ and $t_j$ are embedded in $t$ at positions $r_i$ and $r_j$,
    respectively, at least one of these paths, say the one from $p$ to
    $s_i$, contains the position $r_i$. Then, there exists a variable
    $\atpos{x}{r_i}_{\ell_i}$, for $\ell_i \in
    \interv{1}{\arityof{\arun_i(\epsilon)}}$, such that
    $\store(\atpos{x}{r_i}_{\ell_i}) = u_i$. Hence, $\supp{\struc_i}
    \cap \store(\set{\atpos{x}{r_i}_1, \ldots,
      \atpos{x}{r_i}_{\arityof{\arun_i(\epsilon)}}}) = \supp{\struc_i}
    \cap \store_i(\set{\atpos{x}{\epsilon}_1, \ldots,
      \atpos{x}{\epsilon}_{\arityof{\arun_i(\epsilon)}}}) \neq
    \emptyset$, in contradiction with the fact that $\astruc_i$ is
    encapsulated by $\tuple{\arun_i,t_i,\store_i,\astruc'_i}$, by
    condition \ref{it2:def:view} of \autoref{def:view}. We obtained
    that $k_1 = \ldots = k_{\arityof{\arel}}$, hence $u_1, \ldots,
    u_{\arityof{\arel}} \in \supp{\struc_k}$, leading to $\tuple{u_1,
      \ldots, u_{\arityof{\arel}}} \in \struc_k(\arel)$, for some
    index $k \in \interv{1}{n}$. This proves
    (\ref{eq1:def:expandable}). Since the choice of $\arel$ was
    arbitrary, we obtain that $\astruc_1 \comp \ldots \comp \astruc_n
    \substruc \astruc$, by \autoref{def:substructure}.
  \item[(\ref{it2:def:expandable})] Suppose, for a contradiction, that
    there exists a pair $(u,v)\in\diseq$, such that $u \in
    \supp{\struc_i}$ and $v \in \supp{\struc_j}$, for some indices $1
    \le i < j \le n$. Then, there exists a disequality $\atpos{x}{p}
    \neq \atpos{y}{p}$ (or $y \neq x$, this case being symmetric) in
    $\charform{t}$, such that $\store(\atpos{x}{p})=u$ and
    $\store(\atpos{y}{p})=v$. Since $u \in \supp{\struc_i}$, there
    exists a variable $\atpos{\xi}{p_i}_i$ such that $p_i \in r_i
    \cdot \dom{t_i}$ and $\store(\atpos{\xi}{p_i}_i)=u$. Since
    $\store$ is canonical for $\charform{t}$, we have $x
    \eqof{\charform{t}} \atpos{\xi}{p_i}_i$. Suppose, for a
    contradiction, that $p \not\in r_i \cdot \dom{t_i}$. Then, by a
    similar argument as the one used in the proof of point
    (\ref{it1:def:expandable}), we obtain a contradiction with
    condition \ref{it2:def:view} of \autoref{def:view}, hence $p \in
    r_i \cdot \dom{t_i}$. Symmetrically, we obtain $p \in r_j \cdot
    \dom{t_j}$, hence $r_i \cdot \dom{t_i} \cap r_j \cdot \dom{t_j}
    \neq \emptyset$, which contradicts point
    \ref{it2:proof:exp-composition} above.
  \item[(\ref{it3:def:expandable})] Suppose, for a contradiction, that
    there exists a relation symbol $\arel\in\relations$ and tuples
    $\tuple{u_1, \ldots, u_{\arityof{\arel}}}$, $\tuple{v_1, \ldots,
      v_{\arityof{\arel}}} \in \struc(\arel)$, such that
    $\set{u_1,\ldots,u_{\arityof{\arel}}} \cap \supp{\struc_i} \neq
    \emptyset$, $\set{v_1,\ldots,v_{\arityof{\arel}}} \cap
    \supp{\struc_j} \neq \emptyset$ and
    $\set{u_1,\ldots,u_{\arityof{\arel}}} \cap
    \set{v_1,\ldots,v_{\arityof{\arel}}} \neq \emptyset$, for some
    indices $1 \le i < j \le n$. Then, there exists two distinct
    relation atoms $\arel(\atpos{z}{p_1}_1, \ldots,
    \atpos{z}{p_1}_{\arityof{\arel}})$ and $\arel(\atpos{z}{p_2}_1,
    \ldots, \atpos{z}{p_2}_{\arityof{\arel}})$ in $\charform{t}$ and
    variables $\atpos{\xi}{s_i}_i$, $s_i \in r_i \cdot \dom{t_i}$ and
    $\atpos{\xi}{s_j}_j$, $s_j \in r_j \cdot \dom{t_j}$, such that
    $\atpos{z}{p_1}_k \eqof{\charform{t}} \atpos{\xi}{s_i}_i$,
    $\atpos{z}{p_2}_\ell \eqof{\charform{t}} \atpos{\xi}{s_j}_j$, for
    some indices $k,\ell \in \interv{1}{\arityof{\arel}}$. For
    simplicity, we consider that the position of the variables is the
    same in the above relation symbols, i.e., $p_1$ and $p_2$,
    respectively. By an argument similar to the one used in the proof
    of point (\ref{it2:def:expandable}), we obtain that $p_1 \in r_i
    \cdot \dom{t_i}$ and $p_2 \in r_j \cdot \dom{t_j}$. However, since
    $\set{u_1,\ldots,u_{\arityof{\arel}}} \cap
    \set{v_1,\ldots,v_{\arityof{\arel}}} \neq \emptyset$, there exist
    indices $g,h \in \interv{1}{\arityof{\arel}}$ such that
    $\atpos{z}{p_1}_g \eqof{\charform{t}} \atpos{z}{p_2}_h$.  For
    simplicity, we consider the case where $\arun_i$ and $\arun_j$ are
    runs, the case where at least one of them is a context uses a
    similar argument being left to the reader. Then the path between
    $r_i$ and $r_j$ is contained with the path between $p_1$ and
    $p_2$. By point \ref{it3:proof:exp-composition} above, this path
    contains a reset path disjoint from $\bigcup_{k=1}^n r_k \cdot
    \dom{\arun_k}$. Since, moreover, $\mathcal{B}$ has no persistent
    variables, by construction, we obtain $\atpos{z}{p_1}_g
    \not\eqof{\charform{t}} \atpos{z}{p_2}_h$, contradiction. \qedhere
  \end{itemize}
\end{proof}

  \begin{figure}[htbp]
    \begin{center}
    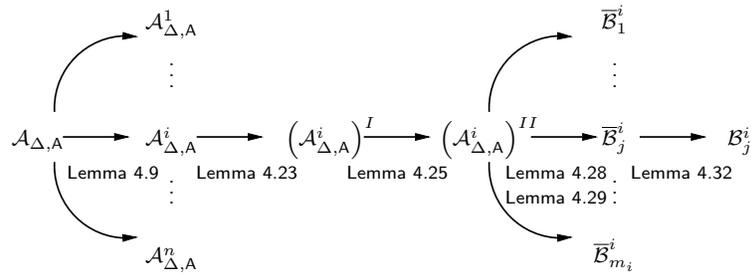
    \caption{\label{fig:automata-expansion} The Chain of Automata Transformations}
    \end{center}
  \end{figure}

We are now ready to prove \autoref{lemma:expansion}:

\begin{proof}Let $\auto{\asid}{\apred}$ be the $\Sigma$-labeled automaton
  corresponding to $\asid$ and $\apred$, such that
  $\sidsem{\apred}{\asid}=\sem{\auto{\asid}{\apred}}$, by
  \autoref{lemma:sid-ta} (\ref{it1:lemma:sid-ta}). The chain of
  transformations depicted in \autoref{fig:automata-expansion}
  produces the set $\set{\mathcal{B}_j^i}_{i\in\interv{1}{n},~
    j\in\interv{1}{m_i}}$ of choice-free and all-satisfiable automata,
  such that $\sem{\auto{\asid}{\apred}}$ is treewidth bounded if and
  only if $\sem{\mathcal{B}^i_j}$ is treewidth bounded, for each $i
  \in \interv{1}{n}$ and $j \in \interv{1}{m_i}$. In particular,
  $\auto{\asid}{\apred}^1, \ldots, \auto{\asid}{\apred}^n$ is the
  language-preserving choice-free decomposition of
  $\auto{\asid}{\apred}$ (\autoref{lemma:choice-free}) and
  $\overline{\mathcal{B}}^i_1, \ldots, \overline{\mathcal{B}}^i_{m_i}$
  are obtained by removing the persistent variables from each
  automaton $\widetilde{\mathcal{B}}^i$ in the language-preserving
  choice-free decomposition of
  $(\widetilde{\auto{\asid}{\apred}}^i)^{II}$
  (\autoref{lemma:remove-persistent:B}). We assume, without loss of
  generality, that the initial state of each $\mathcal{B}^i_j$ is
  $q_\bpred$ and let $\csid^i_j$ be the SID, such that
  $\sem{\mathcal{B}^i_j}=\sidsem{\bpred}{\csid^i_j}$ and
  $\rcsem{\mathcal{B}^i_j}{}=\rcsem{\bpred}{\csid^i_j}$, by
  \autoref{lemma:sid-ta} (\ref{it2:lemma:sid-ta}). We are left with
  proving that $\csid^i_j$ is expandable for $\bpred$, for each $i \in
  \interv{1}{n}$ and $j \in \interv{1}{m_i}$. Let $\csid$ be any of
  the SIDs $\csid^1_1, \ldots, \csid^n_{m_n}$ and $\mathcal{B}$ be the
  automaton such that $\csem{\bpred}{\csid}=\csem{\mathcal{B}}{}$, by
  \autoref{lemma:sid-ta} (\ref{it2:lemma:sid-ta}). Let $\astruc_1,
  \ldots, \astruc_k \in \csem{\bpred}{\csid}$ be pairwise disjoint
  structures. By \autoref{lemma:exp-decomposition}, for each $i \in
  \interv{1}{k}$ there exists a decomposition $\astruc_i = \astruc^i_1
  \comp \ldots \comp \astruc^i_{\ell_i}$ into pairwise disjoint
  structures, such that the structure $\astruc^i_j$ is encapsulated by
  a view $\tuple{\arun^i_j, t^i_j, \store^i_j,
    {\overline{\astruc}}^i_j}$ for $\mathcal{B}$, for each $i \in
  \interv{1}{k}$ and $j \in \interv{1}{\ell_i}$. Without loss of
  generality, we assume that the structures
  $\set{{\overline{\astruc}}^i_j}_{i\in\interv{1}{k},~ j \in
    \interv{1}{\ell_i}}$ are pairwise disjoint, hence the composition
  of $\astruc^1_1, \ldots, \astruc^n_{\ell_n}$ is defined, thus we
  obtain $\astruc_1 \comp \ldots \comp \astruc_k = \astruc^1_1 \comp
  \ldots \comp \astruc^n_{\ell_n}$. By
  \autoref{lemma:exp-composition}, there exists a rich canonical
  model $(\astruc,\diseq)\in\rcsem{\mathcal{B}}{}$, such that
  conditions (\ref{it1:def:expandable}), (\ref{it2:def:expandable})
  and (\ref{it3:def:expandable}) from \autoref{def:expandable}
  hold for $\astruc_1, \ldots, \astruc^k_{\ell_k}$ and
  $(\astruc,\diseq)$. Since
  $\rcsem{\mathcal{B}}{}=\rcsem{\bpred}{\csid}$, by
  \autoref{lemma:sid-ta} (\ref{it2:lemma:sid-ta}), $\csid$ is
  expandable for $\bpred$.
\end{proof}

\subsection{The Decidability of the Treewidth Boundedness Problem for \slr}

The proof of the main result (\autoref{thm:main}) uses a reduction of
an arbitrary SID to a set of expandable SIDs, having an equivalent
treewidth boundedness status (\autoref{lemma:expansion}). Since
treewidth boundedness is decidable for expandable SIDs
(\autoref{thm:expandable-tb}), this proves the decidability of the
problem, in the general case. We conclude this section with the proof
of the main result, that is the decidability of the \tbsl\ problem,
for unrestricted SIDs:

\begin{thm}\label{thm:main}
  There exists an algorithm that decides, for each SID $\asid$ and
  nullary predicate $\apred$, whether the set $\sidsem{\apred}{\asid}$
  has bounded treewidth. If, moreover, this is the case, then
  $\twof{\sidsem{\apred}{\asid}} \leq \maxvarinruleof{\asid} + N \cdot
  M$, where:
  \begin{align*}
    M \isdef & 2\cdot\maxvarinruleof{\asid} + (1 + \maxrulearityof{\asid}) \cdot
    \relationsno{\asid} \cdot \maxpredarityof{\asid}^{\maxrelarityof{\asid}} \\
    N \isdef & \max(K, \maxrulearityof{\asid}^K) \\
    K \isdef & \predsno{\asid} \cdot \relationsno{\asid} \cdot {\maxpredarityof{\asid}}^{\maxpredarityof{\asid}+\maxrelarityof{\asid}}
  \end{align*}
\end{thm}
\begin{proof}
  By \autoref{lemma:expansion}, the treewidth boundedness problem for
  the set $\sidsem{\apred}{\asid}$ can be effectively reduced to
  finitely many treewidth boundedness problems for sets
  $\sidsem{\bpred}{\csid_1}, \ldots, \sidsem{\bpred}{\csid_k}$, where
  $\csid_1, \ldots, \csid_k$ are expandable for $\bpred$. The latter
  problem is decidable, by \autoref{thm:expandable-tb}. The upper
  bound follows from the sequence of transformations given in
  \autoref{fig:automata-expansion}. Let $\csid$ be any of the
  expandable SIDs $\csid_1, \ldots, \csid_k$. By
  \autoref{thm:expandable-tb}, we have $\twof{\sidsem{\bpred}{\csid}}
  \leq \maxvarinruleof{\asid}$. Note that the maximum number of
  variables occurring in a rule is not increased by the construction
  of $\csid$ from $\asid$. By \autoref{lemma:remove-relations}
  (\ref{it3:lemma:remove-relations}),
  \autoref{lemma:remove-equalities}
  (\ref{it3:lemma:remove-equalities}),
  \autoref{lemma:remove-persistent}
  (\ref{it3:lemma:remove-persistent}) and \autoref{lemma:exp}
  (\ref{it3:lemma:exp}), we obtain $\twof{\sidsem{\apred}{\asid}} \leq
  \twof{\sidsem{\bpred}{\csid}} + \cardof{\trans^1} \cdot M$, where
  $\trans^1$ is the set of $1$-transitions of the choice-free
  automaton used to define $\csid$ (point (\ref{it2:lemma:sid-ta}) of
  \autoref{lemma:sid-ta}) and $M = 2\cdot\maxvarinruleof{\asid} + (1 +
  \maxrulearityof{\asid}) \cdot \relationsno{\asid} \cdot
  \maxpredarityof{\asid}^{\maxrelarityof{\asid}}$. Note that
  $\cardof{\trans^1} \cdot M$ is the sum of the increases in the upper
  bounds of the treewidth along the transformation. By
  \autoref{lemma:choice-free}, we obtain $\cardof{\trans^1} \leq
  \max(\cardof{\states_{\asid}},
  \maxrulearityof{\asid}^\cardof{\states_{\asid}})$, where
  $\states_{\asid}$ denotes the set of states of
  $\auto{\asid}{\apred}$ (\autoref{lemma:sid-ta}). Since the entire
  construction was done assuming that $\asid$ is equality-free and
  all-satisfiable, by lifting these assumptions, we obtain
  $\cardof{\states_\asid} \leq \predsno{\asid} \cdot
  \relationsno{\asid} \cdot
              {\maxpredarityof{\asid}}^{\maxpredarityof{\asid}+\maxrelarityof{\asid}}$
              (\autoref{lemma:eq-free} and \autoref{lemma:all-sat}).
\end{proof}
Note that, unlike \autoref{thm:expandable-tb}, that gives an optimal
upper bound for the treewidths of the structures described by an
expandable SID, \autoref{thm:main} does not provide such an optimal
upper bound. In particular, the upper bound of \autoref{thm:main}
grows doubly exponential in the $\maxpredarityof{\asid}$ and
$\maxrelarityof{\asid}$ parameters and simply exponential in the
$\predsno{\asid}$ and $\relationsno{\asid}$ parameters. Finding an
optimal bound for the general case is considered as subject for future
work.


\section{The Treewidth Boundedness Problem for First Order Logic}
\label{sec:undecidability}

This section proves the undecidability of the treewidth boundedness
problem for first-order logic. This result places the frontier of
decidability for this problem between the classical first-order logic
and substructural \slr\ with simple inductive definitions, i.e., using
only existentially quantified separating conjunctions of atoms.

We recall that the \emph{first-order logic} (\fol) is the set of
formul{\ae} consisting of equalities and relation atoms, connected by
boolean conjunction, negation and existential quantification. The
semantics of first order logic is given by the satisfaction relation
$(\univ,\struc) \Models^\store \phi$ between structures and
formul{\ae}, parameterized by a \emph{store} $\store : \vars
\rightarrow \univ$ such that $(\univ,\struc) \Models^\store
\arel(x_1,\ldots,x_{\arityof{\arel}})$ iff $\tuple{\store(x_1),
  \ldots, \store(x_{\arityof{\arel}})} \in \struc(\arel)$. If $\phi$
is a sentence the store is not important, thus we omit the superscript
and write $\astruc \Models \phi$ instead. The set of \emph{models} of
a \fol\ sentence $\phi$ is denoted as $\sem{\phi}\isdef\set{\astruc
  \mid \astruc \Models \phi}$. Although we use the same notation for
the sets of models of \fol\ and \slr\ formul{\ae}, the underlying
logic is clear from the context.

This section is concerned with the proof of the following theorem:

\begin{thm}\label{thm:fol}
  The problem is $\sem{\phi}$ treewidth-bounded, for a given
  \fol\ sentence $\phi$ with at least two binary relation symbols and
  several unary relation symbols, is undecidable.
\end{thm}
\begin{proof}
We will reduce from the undecidability of the \emph{Tiling
Problem}~\cite{berger1966undecidability}.  We first recall its
definition. Given a finite set of tiles $S = \{t_1,\ldots,t_n\}$ is
there a tiling of the plane such that the colors of neighbouring tiles
match?  (We note that rotating or reflecting the tiles is not
allowed.)  In more detail: We assume the plane is given by integer
coordinates $(x,y)$ with $x,y \in \mathbb{Z}$.  We want to put a copy
of a tile at every coordinate.  We will require that neighbouring
tiles match. For this we assume to be given a relation $H \subseteq S
\times S$ -- specifying which tiles match can be placed next to each
other horizontally -- and $V \subseteq S \times S$ -- specifying which
tiles match can be placed next to each other vertically.  We now
require for every tiling that $(t_i,t_j) \in H$, for all tiles $t_i$
and $t_j$ placed at coordinates $(x,y)$ and $(x+1,y)$, and $(t_i,t_j)
\in V$, for all tiles $t_i$ and $t_j$ placed at coordinates $(x,y)$
and $(x,y+1)$. It is well-known that it is undecidable whether such a
tiling exists~\cite{berger1966undecidability}. In fact, it is known
that is already undecidable whether such a tiling exists for the
upper-right quadrant of the plane, i.e., when coordinates $(x,y)$ are
restricted to $x,y \in \mathbb{N}$.

We will now reduce the tiling problem to deciding whether a given
first-order formula has infinitely many non-isomorphic models of
unbounded treewidth. We consider some instance of the tiling problem.
For encoding this problem, we define the signature $\relations =
\{\upf,\rightf,\north,\south,\east,\west,\internal,\tile_1,\dots,\tile_n\}$
to consist of the binary relations $\upf$ and $\rightf$, the unary
relations $\south,\east,\west,\north,\internal$ and the unary
relations $\tile_1,\dots,\tile_n$ (one for each tile in the tiling
instance).  We then consider the following formula:
\[
\phi \isdef \bigwedge\nolimits_{i=1}^{13} \psi_i \land \bigwedge\nolimits_{j=1}^4 \phi_j
\]
where:
\begin{itemize}[label=$\triangleright$]
\item $\psi_1 \isdef \forall x \exists^{\le 1} y ~.~ \rightf(x,y)
  \land \exists^{\le 1} y ~.~ \rightf(y,x)$ states that $\rightf$ and
  $\rightf^{-1}$ are partial functions, $\psi_2$ states that that
  $\upf$ and $\upf^{-1}$ are partial functions,
  %
  %
\item $\psi_3 \isdef \forall x,y,z ~.~ \upf(x,y) \land \rightf(x,z)
  \rightarrow \exists w. \upf(z,w) \land \rightf(y,w) $ states that
  $\rightf$ and $\upf$ commute, $\psi_4$ states that $\rightf^{-1}$
  and $\upf$ commute, $\psi_5$ states that $\rightf$ and $\upf^{-1}$
  commute and $\psi_6$ states that $\rightf^{-1}$ and $\upf^{-1}$
  commute,
\item $\psi_7 \isdef \forall x ~.~ \south(x) \leftrightarrow \neg \exists
  y. \upf(y,x)$ states that south-labelled nodes are exactly the ones
  that do not have incoming $\upf$ edges, and $\psi_8$, $\psi_9$,
  $\psi_{10}$ define the analogous property for the west, east, and
  north labels,
\item $\psi_{11} \isdef \forall x ~.~ \internal(x) \leftrightarrow
  \neg(\south(x) \lor \east(x) \lor \north(x) \lor \west(x))$ states
  that internal nodes are exactly the ones not labelled by south,
  east, north or west,
\item $\psi_{12} \isdef \forall x ~.~ \internal(x) \rightarrow \exists
  y ~.~ \rightf(x,y) \land \exists y ~.~ \rightf(y,x)$ states that
  internal nodes have exactly one outgoing and exactly one incoming
  $\rightf$ edge, $\psi_{13}$ states that internal nodes have exactly
  one outgoing and exactly one incoming $\upf$ edge,
  \item $\phi_1 \isdef \forall x ~.~ \bigvee_{i=1}^n \tile_i(x)$
    states that every coordinate holds at least one tile,
  \item $\phi_2 \isdef \forall x ~.~ \bigwedge_{i\neq j} \neg
    \tile_i(x) \lor \neg \tile_j(x)$ states that every coordinate
    holds at most one tile,
  \item $\phi_3 \isdef \forall x \forall y ~.~ \rightf(x,y) \rightarrow
    \bigvee_{(t_i,t_j) \in H} \tile_i(x) \land \tile_j(y)$ states that tiles,
    that are next to each other horizontally, satisfy the horizontal
    matching constraint, and
  \item $\phi_4 \isdef \forall x \forall y ~.~ \upf(x,y) \rightarrow
    \bigvee_{(t_i,t_j) \in V} \tile_i(x) \land \tile_j(y)$ states that
    tiles, that are next to each other vertically, satisfy the
    vertical matching constraint.
\end{itemize}
As usual, the formul{\ae} $\exists^{\leq1} x ~.~ \varphi$ stand for
$\exists x ~.~ \varphi \rightarrow \forall y ~.~ \varphi[x/y] \rightarrow
y=x$.  The proof will make use of the fact that $\phi$ encodes grids and
non-standard models of grids, which are (disjoint unions of) grid-like
structures.  We will argue the following:

\begin{fact}
  Each model of $\phi$ can be decomposed into (disjoint unions
  of) grids, cylinders, and toruses, where grids have
  $\south,\east,\west,\north$ borders, cylinders have either
  $\south,\north$ or $\east,\west$ borders, and toruses do not
  have any borders and only consist of internal nodes.
\end{fact}
\proof{We recall that we are only interested in finite models of
  first-order formul{\ae}.  We note that $\phi$ specifies that $\upf$
  and $\rightf$ are (partial) functions, and, hence, we will use
  functional notation in the following.  We now fix a model
  $(\univ,\struc)$ of $\phi$ -- as usual, we require that
  $\univ\neq\emptyset$. We decompose $(\univ,\struc)$ into its
  maximally connected components, connected via $\upf$, $\rightf$. We
  choose a representative $u_C$, for each component $C$.  We observe
  that either there are $j \le 0 \le i$ such that $\rightf^j(u_C)$ is
  $\west$-labelled and $\rightf^i(u_C)$ is $\east$-labelled, or
  $\rightf^i(u_C) = u_C$, for some $i \ge 0$ (because the universe is
  finite and the functionality of $\rightf$ and $\rightf^{-1}$ ensures
  that the only possible loop returns to $u_C$).  An analogous
  statement holds for $\upf$ as well as the $\north$ and $\south$
  labels.  We now call a component $C$ a \emph{grid}, if $u_C$ reaches
  $\south,\east,\west,\north$ via $\upf$, $\rightf$ and their
  inverses, a \emph{cylinder} if $u_C$ reaches $\south,\north$ or
  $\east,\west$ via $\upf$ and its inverse resp. $\rightf$ and its
  inverse, or a \emph{torus}, otherwise. \qed}

We now justify the naming of these components: 

\begin{fact}
  Consider a grid component $C$ with representative $u_C$ such
  that $\rightf^j(u_C)$ is $\west$-labelled, $\rightf^i(u_C)$ is
  $\east$-labelled, $\upf^k(u_C)$ is $\south$-labelled and
  $\upf^l(u_C)$ is $\north$-labelled, for some $j \le 0 \le i$ and $k
  \le 0 \le l$.  We claim that: \begin{enumerate}
  \item\label{enum:existence} the elements $\rightf^a(\upf^b(u_C))$
    exist, for all $j \le a \le i$ and all $k \le b \le l$,
  \item\label{enum:border} an element $\rightf^a(\upf^b(u_C))$, for $j
    \le a \le i$ and $k \le b \le l$, is $\east$-labelled iff $a=i$; 
    analogous claims hold for the labels $\south, \west, \north$,
  \item\label{enum:internal} the elements $\rightf^a(\upf^b(u_C))$ are
    internal nodes, for all $j < a < i$ and all $k < b < l$,
  \item\label{enum:pairwise-different} the elements
    $\rightf^a(\upf^b(u_C))$ are pairwise different, for all $j \le a
    \le i$ and all $k \le b \le l$,    
  \item\label{enum:reachable} all elements of the component can be
    represented as $\rightf^a(\upf^b(u_C))$, for some $j \le a \le i$
    and $k \le b \le l$, and
  \item\label{enum:grid} the component is isomorphic to a grid.
  \end{enumerate}
\end{fact}
\proof{ Items (\ref{enum:existence}), (\ref{enum:border}) and
  (\ref{enum:internal}) directly follow from the commutativity
  requirements.  For~(\ref{enum:pairwise-different}), we consider some
  $j \le a_1,a_2 \le i$ and $k \le b_1,b_2 \le l$.  We will show that
  $a_1 \neq a_2$ resp. $b_1 \neq b_2$ imply that
  $\rightf^{a_1}(\upf^{b_1}(u_C)) \neq
  \rightf^{a_2}(\upf^{b_2}(u_C))$.  We will assume that
  $\rightf^{a_1}(\upf^{b_1}(u_C)) = \rightf^{a_2}(\upf^{b_2}(u_C))$
  and derive a contradiction.  Let us assume that $a_1 > a_2$ (the
  other cases are analogous).  Then, we have
  $\rightf^{a_1+s}(\upf^{b_1}(u_C)) =
  \rightf^{a_2+s}(\upf^{b_2}(u_C))$ for $s=i-a_1 \ge 0$.  However,
  $\rightf^{a_1+s}(\upf^{b_1}(u_C))$ is $\east$-labelled, while
  $\rightf^{a_2+s}(\upf^{b_2}(u_C))$ is not, by (\ref{enum:border}),
  contradiction. For~(\ref{enum:reachable}), we observe that every
  node reachable from $u_C$ connected via $\upf$, $\rightf$ and their
  inverses can be represented as $\rightf^a(\upf^b(u_C))$, because of
  the commutativity requirements; further, we must have $j \le a \le
  i$ and $k \le b \le l$ because the $\south,\east,\west,\north$
  borders do not have outgoing edges.  For~(\ref{enum:grid}), we
  observe that the component is isomorphic to the structure with
  domain $\{(x,y) \mid x,y \in [j,i], y \in [k,l]\}$, where $\rightf$
  is interpreted as $\{((x,y),(x+1,y)) \mid x \in [j,i-1], y \in
  [k,l]\}$ and $\upf$ as $\{((x,y),(x,y+1)) \mid x \in [j,i],y \in
  [k,l-1]\}$. \qed}

\begin{fact}
  Consider a cylinder component $C$ with representative $u_C$
  such that $\rightf^j(u_C)$ is $\west$-labelled, $\rightf^i(u_C)$ is
  $\east$-labelled, and $\upf^k(u_C) = u_C$ for some $j \le 0 \le i$
  and $k \le 0$, where $k$ is the smallest number with this property.
  (The properties stated below hold analogously for $\north, \south$
  cylinders).  We claim that: \begin{enumerate}
  \item\label{enum:existence2} the elements $\rightf^a(\upf^b(u_C))$
    exist, for all $j \le a \le i$ and all $0 \le b < k$,
  \item\label{enum:border2} an element $\rightf^a(\upf^b(u_C))$, for
    $j \le a \le i$ and $0 \le b < k$, is $\east$-labelled iff $a=i$;
    an analogous claim hold for the label $\west$,
  \item\label{enum:internal2} the elements $\rightf^a(\upf^b(u_C))$
    are internal nodes, for all $j < a < i$ and all $0 \le b < k$,
  \item\label{enum:pairwise-different2} all the elements
    $\rightf^a(\upf^b(u_C))$ are pairwise different, for all $j \le a
    \le i$ and all $0 \le b < k$,
  \item\label{enum:reachable2} all nodes of the component can be
    represented as $\rightf^a(\upf^b(u_C))$, for some $j \le a \le i$
    and $0 \le b < k$, and
  \item\label{enum:grid2} the component is isomorphic to a cylinder,
    i.e., a grid for which the north-border connects to south-border.
  \end{enumerate}
\end{fact}
\proof{ Items (\ref{enum:existence2}), (\ref{enum:border2}) and
  (\ref{enum:internal2}) directly follow from the commutativity
  requirements. For~(\ref{enum:pairwise-different2}), we consider some
  $j \le a_1,a_2 \le i$ and $0\le b_1,b_2 < k$.  We will show that
  $a_1 \neq a_2$ resp. $b_1 \neq b_2$ imply that
  $\rightf^{a_1}(\upf^{b_1}(u_C)) \neq
  \rightf^{a_2}(\upf^{b_2}(u_C))$.  We will assume that
  $\rightf^{a_1}(\upf^{b_1}(u_C)) = \rightf^{a_2}(\upf^{b_2}(u_C))$
  and derive a contradiction.  Let us first assume that $a_1 > a_2$
  (the case $a_1 < a_2$ is symmetric).  Then, we have
  $\rightf^{a_1+s}(\upf^{b_1}(u_C)) =
  \rightf^{a_2+s}(\upf^{b_2}(u_C))$ for $s=i-a_1$.  However,
  $\rightf^{a_1+s}(\upf^{b_1}(u_C))$ is $\east$-labelled, while
  $\rightf^{a_2+s}(\upf^{b_2}(u_C))$ is not, by (\ref{enum:border}),
  contradiction. Now we assume $a_1 = a_2$ and $b_1 > b_2$ (the case
  $b_1 < b_2$ is symmetric).  Then, $\rightf^{a_1}(\upf^{b_1}(u_C)) =
  \rightf^{a_2}(\upf^{b_2}(u_C))$ implies that $\upf^{b_1-b_2}(u_C)) =
  u_C$ with $0 \le b_1-b_2 < k$.  However, this contradicts that $k$
  is the smallest number with this property.
  For~(\ref{enum:reachable2}), we observe that every node reachable
  from $u_C$ connected via $\upf$, $\rightf$ and their inverses can be
  represented as $\rightf^a(\upf^b(u_C))$, because of the
  commutativity requirements; further, we can in fact choose $0 \le b
  < k$ because of commutativity and the assumption that $\upf^k(u_C) =
  u_C$.  Moreover, we must have $j \le a \le i$ because the
  $\east,\west$ borders do not have outgoing edges.
  For~(\ref{enum:grid2}), we observe that the component is isomorphic
  to the structure with domain $\{(x,y) \mid x,y \in [j,i], y \in
  [k,l]\}$, where $\rightf$ is interpreted as $\{((x,y),(x+1,y)) \mid
  x \in [j,i-1], y \in [0,k-1]\}$ and $\upf$ as $\{((x,y),(x,y+1))
  \mid x \in [j,i],y \in [0,k-2]\} \cup \{((x,k),(x,1)) \mid x \in
       [j,i]\}$. \qed}

\begin{fact}
  Consider a torus component $C$, with
  representative $u_C$ such that $\upf^k(u_C) = u_C$ and
  $\rightf^l(u_C) = u_C$ for some $k \le 0$ and $l \le 0$, where
  $k$ and $l$ are the smallest numbers with this property.  We
  claim that: \begin{enumerate}
  \item\label{enum:existence3} the elements exists
    $\rightf^a(\upf^b(u_C))$, for all $0 \le a < k$ and all $0 \le b <
    l$,
  \item\label{enum:internal3} the elements $\rightf^a(\upf^b(u_C))$
    are internal nodes, for all $0 \le a < k$ and all $0 \le b < l$,
  \item\label{enum:reachable3} all nodes of the component can be
    represented as $\rightf^a(\upf^b(u_C))$, for some $0 \le a < k$
    and $0 \le b < l$.
  \end{enumerate}
\end{fact}
\proof{ Items (\ref{enum:existence3}) and (\ref{enum:internal3})
  directly follow from the commutativity requirements.
  For~(\ref{enum:reachable3}), we observe that every node reachable
  from $u_C$ connected via $\upf$, $\rightf$ and their inverses can be
  represented as some $\rightf^a(\upf^b(u_C))$ because of the
  commutativity requirements; further, we can in fact choose $0 \le a
  < k$ and $0 \le b < l$ because of commutativity and the assumptions
  that $\upf^k(u_C) = u_C$ and $\rightf^k(u_C) = u_C$.  We note that
  the $\rightf^a(\upf^b(u_C))$ are in general not pairwise different
  (e.g., we might have $\rightf^a(u_C) = \upf^b(u_C)$ for some $0 \le
  a < k$ and $0 \le b < l$). However, in our below argument we do not
  need to distinguish whether all the elements
  $\rightf^a(\upf^b(u_C))$ of a torus component are pairwise
  different. \qed}

The following claim reduces the treewidth boundedness problem for
first-order logic to the tiling problem for the first quadrant of the
plane:
\begin{fact}
  $\phi$ has models of unbounded treewidth iff there is a tiling of
  the upper-right quadrant of the plane.
\end{fact}
\proof{ ``$\Leftarrow$'' Let us assume that there is a tiling of the
  upper-right quadrant of the plane.  Then, for every $n \in
  \mathbb{N}$, this tiling induces a square grid $G_n$ of size $n
  \times n$ with $G_n \models \phi$: simply take the tiles at
  positions $(x,y)$, with $x,y \in [1,n]$, from the tiling of the
  upper-right quadrant, and verify that in this way we obtain a model
  of the formula $\phi$.

  \skipnoindent ``$\Rightarrow$'' We now assume that $\phi$
  has models of unbounded treewidth, i.e., for every $i \ge 1$ there
  is a finite model $(\univ,\struc)$ with $\twof{(\univ,\struc)} \ge
  i$.  If any model $(\univ,\struc)$ contains a torus component
  $C$, we immediately obtain a tiling of the upper-right quadrant by
  unrolling the torus: we define the tiling of the upper right
  quadrant by placing the tile of the element $\rightf^i(\upf^j(u_C))$
  at position $(i,j)$. It is then routine to verify that the
  subformula $\bigwedge_{j=1}^4 \phi_j$ of $\phi$ ensures that the
  matching requirements of a tiling are satisfied. Hence, we are left
  with the case that no model of $\phi$ contains a torus
  component.

  We now observe that an $n \times m$ grid has treewidth $\min\{n,m\}$
  and an $n \times m$ cylinder has treewidth $\min\set{2n,m}$
  resp. $\min\set{n,2m}$ for $\east,\west$ resp. $\south,\north$
  cylinders. For the $n \times m$ grid, this follows from the $k$-cops
  and robber game, defined as follows. A position in the game is a
  pair $(\gamma,r)$, where $\gamma \subseteq \interv{1}{n} \times
  \interv{1}{m}$, $\cardof{\gamma}=k$ and $r \in \interv{1}{n} \times
  \interv{1}{m} \setminus \gamma$. The game can move from
  $(\gamma_i,r_i)$ to $(\gamma_{i+1},r_{i+1})$ iff there exists a path
  between $r_i$ and $r_{i+1}$ in the restriction of the grid to
  $\interv{1}{n} \times \interv{1}{m} \setminus (\gamma_i\cap
  \gamma_{i+1})$. We say that $k$ cops catch the robber iff every
  sequence of moves in the game is finite. It is known that, if the
  treewidth of the graph is greater or equal to $k$, then $k+1$ cops
  catch the robber on a graph $\graph$ \cite{SEYMOUR199322}. Since
  $\min\set{n,m}-1$ cops do not catch the robber (which can always
  move to the intersection of a cop-free row and a cop-free column) it
  follows that the treewidth of the grid is greater than
  $\min\set{n,m}-1$. At the same time, there exists a tree
  decomposition of width $\min\set{n,m}$. For the $n\times m$
  $\north$-$\south$ cylinder (the case of the $\east$-$\west$ cylinder
  is analogous), we need extra $n$ cops to prevent the robber escaping
  wrapping around the $\east$-$\west$ axis, thus the treewidth is
  $\min\set{2n,m}$.

  We now consider some $i\ge 0$ and some model $(\univ,\struc)$ with
  $\twof{(\univ,\struc)} \ge 2i$ that does not contain torus
  components.  Then, $N$ decomposes into grid components and cylinder
  components.  Because of our assumption $\twof{(\univ,\struc)} \ge
  2i$ there must be some component $C$ of $N$ with $\twof{C} \ge 2i$.
  Now, we can deduce that $C$ contains some square grid $M$ of size $i
  \times i$ as a substructure (this follows from $2i \le \min\{n,m\}$
  for grids and from $2i \le \min\{n,2m\}$ resp. $2i \le \min\{2n,m\}$
  for cylinders).  Hence, we can restrict our attention to models of
  $\phi$ that are square grids.  Let $M_1, M_2,\ldots$ be a sequence
  of models with $M_n \models \phi$, where each $M_n$ is a square grid
  of size $n \times n$.  We are now going to construct a sequence of
  models $G_1, G_2,\ldots$ such that each $G_n$ is a square grid of
  size $n$ with $G_n \models \phi$, and each $G_n$ is included in
  $G_{n+1}$, where we say a model $I$ of $\phi$ is included in a model
  $J$ of $\phi$ if $I$ resp. $J$ are square grids of size $n \times n$
  resp. $m \times m$, and we have that $n \le m$ and all tiles at
  positions $(x,y)$, with $x,y \in [1,n]$, are the same in both
  models.  We construct the sequence $G_1, G_2,\ldots$ inductively,
  maintaining an infinite sequence of models $M_1^n, M_2^n,\ldots$,
  for each $n\in\mathbb{N}$, such that $G_n$ is included in all
  $M_i^n$: Take $G_1$ to be a model that consists of a single tile,
  which appears infinitely often at position $(1,1)$ in the models
  $M_1, M_2,\ldots$; then we obtain the sequence $M_1^1, M_2^1,\ldots$
  as the restriction of $M_1, M_2,\ldots$ to the models that include
  $G_1$.  Assume we have already defined $G_n$.  Choose some square
  grid $G_{n+1}$ of size $n+1$ that is included infinitely often in
  models of the sequence $M_1^n, M_2^n,\ldots$ (note that such a
  square grid must exist by the pigeonhole principle); then obtain the
  sequence $M_1^{n+1}, M_2^{n+1},\ldots$ by restricting the sequence
  to the $M_1^n, M_2^n,\ldots$ to the models that include $G_{n+1}$.
  With the sequence $G_1, G_2,\ldots$ at hand we now obtain a tiling
  of the plane: For position $(i,j)$, with $i,j\in\mathbb{N}$, simply
  take the tile at this position in $G_{\max\{i,j\}}$.  We now verify
  that the horizontal resp. vertical requirements of a tiling are
  satisfied. We verify only the horizontal requirement (the vertical
  one is symmetric). Consider tiles at positions $(i,j)$ and
  $(i+1,j)$. If $i \neq j$, then both tiles have been defined by
  $G_{\max\{i,j\}}$, and the matching requirement is satisfied because
  $G_{\max\{i,j\}}$ is a model of $\phi$. If $i = j$ then the tile at
  position $(i,i)$ is defined by $G_i$ and the tile at position
  $(i+1,i)$ is defined by $G_{i+1}$.  Now we observe that the tile at
  position $(i,i)$ in $G_i$ is the same as the tile at position
  $(i,i)$ in $G_{i+1}$, because $G_i$ is included in $G_{i+1}$, and
  the matching requirement is satisfied because $G_{i+1}$ is a model
  of $\phi$. \qed}

This concludes the proof of the theorem. 
\end{proof}


\section{Conclusions}

We have presented a decision procedure for the treewidth boundedness
problem in the context of \slr, a generalization of Separation Logic
over relational signatures, interpreted over structures. This
procedure allows to define the precise fragment of \slr\ in which
every formula has a bound on the treewidth of its models. This
fragment is the right candidate for the definition of a fragment of
\slr\ with a decidable entailment problem. Another application is
checking that each graph defined by a treewidth-bounded \slr\ formula
satisfies \mso-definable properties such as, e.g., Hamiltonicity, or
3-Colorability.

\bibliographystyle{alphaurl}
\bibliography{refs}

\appendix


\section{Proofs from \autoref{sec:preliminaries}}
\subsection{Proof of \autoref{lemma:all-sat}}
\label{app:all-sat}

Without loss of generality, we consider that $\asid$ is equality-free
(\autoref{lemma:eq-free}).  We propose a construction using an idea of
Brotherston et al \cite{DBLP:conf/csl/BrotherstonFPG14}, that
characterizes the satisfiability of a predicate by an abstraction
consisting of tuples of parameters occurring in the interpretation of
relation symbols. A similar abstraction has been used to check
satisfiability of \slr\ formul{\ae} \cite{DBLP:conf/cade/BozgaBI22}.

\begin{defi}
  A \emph{base} $\abstruc$ is a mapping $\abstruc : \relations
  \rightarrow \mpow{\vars^+}$ of relation symbols $\arel$ into
  multisets of tuples of variables of length $\arityof{\arel}$ each. A
  base is \emph{satisfiable} iff $\abstruc(\arel)$ is a set, for all
  $\arel \in \relations$. Given a set of variables $X \subseteq
  \vars$, let $\satbasetuplesof{X}$ denote the set of satisfiable
  bases whose images contain only variables from $X$ and let
  $\satbasetuples \isdef \satbasetuplesof{\vars}$.
\end{defi}

We consider three partial operations on $\satbasetuples$. First, the
\emph{composition} is $\abstruc_1 \basecomp \abstruc_2 \isdef
\abstruc_1 \cup \abstruc_2$ if $\abstruc_1 \cup \abstruc_2$ is
satisfiable, and undefined, otherwise. Second, the \emph{substitution}
$\abstruc[x_1/y_1,\ldots,x_n/y_n]$ replaces simultaneously each
occurrence of $x_j$ by $y_j$ in $\abstruc$, for all $j \in
\interv{1}{n}$. Third, given a set $X \subseteq \vars$ of variables,
the \emph{projection} is:
\[
\proj{\abstruc}{X} \isdef \lambda \arel ~.~ \set{\tuple{x_1, \ldots,
    x_n} \in \abstruc(\arel) \mid x_1, \ldots, x_n \in X}
\]
Finally, for a qpf formula $\phi$, we define:
\[
\basepairof{\phi} \isdef \lambda \arel ~.~
\mset{\tuple{x_1,\ldots,x_n} \mid \arel(x_1,\ldots,x_n) \text{ occurs
    in } \phi}
\]
The predicates $\bpred^{\abstruc}$ of the SID $\overline{\asid}$ are
obtained by annotating each predicate $\bpred$ that occurs in $\asid$
with a satisfiable base $\abstruc$. The arity of each predicate
$\bpred^{\abstruc}$ is the arity of $\bpred$. Then $\overline{\asid}$
contains the rules:
\[
  \bpred_0^{\abstruc_0}(x_1,\ldots,x_{\arityof{\bpred_0}}) \leftarrow
  \exists y_1 \ldots \exists y_m ~.~ \psi *
  \Bigstar\nolimits_{i=1}^\ell
  \bpred_i^{\abstruc_i}(z_{i,1},\ldots,z_{i,\arityof{\bpred_i}})
\]
where \(\bpred_0(x_1,\ldots,x_{\arityof{\bpred_0}}) \leftarrow \exists
y_1 \ldots \exists y_m ~.~ \psi * \Bigstar\nolimits_{i=1}^\ell
\bpred_i(z_{i,1},\ldots,z_{i,\arityof{\bpred_i}})\) is a rule of
$\asid$, $\psi$ is the largest qpf formula from its right-hand side
and, moreover, the following condition holds:
\[
\abstruc_0 = \proj{\Big(\basepairof{\psi} \otimes
  \bigotimes\nolimits_{i=1}^\ell \abstruc_i[x_1 / z_{i,1}, \ldots,
    x_{\arityof{\bpred_i}} / z_{i,\arityof{\bpred_i}}]
  \Big)}{\set{x_1,\ldots,x_{\arityof{\bpred_0}}}}
\]
In addition, $\overline{\asid}$ contains a rule $\apred \leftarrow
\apred^{\abstruc}$, for each satisfiable base
$\abstruc\in\satbasetuples$. Note that $\overline{\asid}$ is finite,
because $\asid$ and $\satbasetuples$ are finite. We now prove the
points from the statement of \autoref{lemma:all-sat}:

\skipnoindent ``$\overline{\asid}$ is all-satisfiable for $\apred$''
We prove the following, more general, fact:

\begin{fact}\label{fact:all-sat}
  For each predicate $\bpred^{\abstruc}$ that occurs in
  $\overline{\asid}$, each predicate-free formula $\phi$ such that
  $\bpred^{\abstruc} \step{\overline{\asid}}^* \phi$ and each
  injective store $\store$ over
  $\set{x_1,\ldots,x_{\arityof{\bpred}}}$, there exists a structure
  $\astruc = (\univ,\struc)$ such that $\astruc \models^\store \phi$
  and, for all $\arel\in\relations$ and tuples of variables
  $\tuple{x_{j_1},\ldots,x_{j_{\arityof{\arel}}}}$ with
  $j_k\in\interv{1}{\arityof{\bpred}}$ for all
  $k\in\interv{1}{\arityof{\arel}}$, we have
  $\tuple{x_{j_1},\ldots,x_{j_{\arityof{\arel}}}} \in \abstruc(\arel)$
  if and only if
  $\tuple{\store(x_{j_1}),\ldots,\store(x_{j_{\arityof{\arel}}})} \in
  \struc(\arel)$.
\end{fact}
\begin{proof}
  We proceed by induction on the length of the unfolding
  $\bpred^{\abstruc} \step{\overline{\asid}}^* \phi$.  The first step
  of the unfolding uses a rule:
  \[
    \bpred^{\abstruc}(x_1,\ldots,x_{\arityof{\bpred}}) \leftarrow
    \exists y_1 \ldots \exists y_m ~.~ \psi *
    \Bigstar\nolimits_{i=1}^\ell
    \bpred_i^{\abstruc_i}(z_{i,1},\ldots,z_{i,\arityof{\bpred_i}})
  \]
  and for all $i\in\interv{1}{\ell}$, we have $\bpred_i^{\abstruc_i}
  \step{\overline{\asid}}^* \phi_i$, such that
  \[
  \phi = \exists y_1 \ldots \exists y_m ~.~ \psi *
  \Bigstar\nolimits_{i=1}^\ell \phi_i [x_1 / z_{i,1}, \ldots,
    x_{\arityof{\bpred_i}} / z_{i,\arityof{\bpred_i}}]
  \]
  modulo a reordering of atoms. Let $\store'$ be an injective store
  that extends $\store$ over $\set{y_1,\ldots,y_m}$.  For all
  $i\in\interv{1}{\ell}$, we define $\store_i$ over
  $\set{x_1,\ldots,x_{\arityof{\bpred_i}}}$ by $\store_i(x_j) \isdef
  \store'(z_{i,j})$ for all $j\in\interv{1}{\arityof{\bpred_i}}$.  By
  the inductive hypothesis applied to the unfolding
  $\bpred_i^{\abstruc_i} \step{\overline{\asid}}^* \phi_i$, there
  exists a structure $\astruc_i = (\univ_i,\struc_i)$ such that
  $\astruc_i \models^{\store_i} \phi_i$ and
  $\tuple{x_{j_1},\ldots,x_{j_{\arityof{\arel}}}} \in
  \abstruc_i(\arel)$ if and only if
  $\tuple{\store_i(x_{j_1}),\ldots,\store_i(x_{j_{\arityof{\arel}}})}
  \in \struc_i(\arel)$.  We can furthermore assume for all $i_1 \neq
  i_2 \in \interv{1}{\ell}$ that $\supp{\struc_{i_1}} \cap
  \supp{\struc_{i_2}} =
  \set{\store'(x_1),\ldots,\store'(x_{\arityof{\bpred}}),\store'(y_1),\ldots,\store'(y_m)}$.
  Since $\psi$ is equality-free, there exists a structure
  $\astruc_\psi$ such that $\astruc_\psi \models^{\store'} \psi$, and
  $\basepairof{\psi}$ is a satisfiable base.  We then prove that the
  structure $\astruc = (\univ,\struc) \isdef \astruc_\psi \comp
  \astruc_1 \comp \ldots \comp \astruc_\ell$ is defined and meets the
  requirements from the statement:
  \begin{itemize}[label=$\triangleright$]
  \item Let $\arel\in\relations$, and suppose, for a contradiction,
    that there exists a tuple $\tuple{u_1,\ldots,u_{\arityof{\arel}}}
    \in \struc_{i_1}(\arel) \cap \struc_{i_2}(\arel)$ for $i_1 \neq
    i_2 \in \interv{1}{\ell}$.  From the induction hypothesis and
    because $\store'$ is injective, there exists
    $\tuple{z_1,\ldots,z_{\arityof{\arel}}} \in \abstruc_{i_1} [x_1 /
      z_{i_1,1}, \ldots, x_{\arityof{\bpred_{i_1}}} /
      z_{i_1,\arityof{\bpred_{i_1}}}] (\arel) ~\cap~ \abstruc_{i_2}
    [x_1 / z_{i_2,1}, \ldots, x_{\arityof{\bpred_{i_2}}} /
      z_{i_2,\arityof{\bpred_{i_2}}}] (\arel)$ with $\store'(z_k) =
    u_k$ for all $k\in\interv{1}{\arityof{\arel}}$. This contradicts
    that the following base composition:
    \[
    \abstruc_{i_1} [x_1 / z_{i_1,1}, \ldots,
      x_{\arityof{\bpred_{i_1}}} / z_{i_1,\arityof{\bpred_{i_1}}}]
    ~\basecomp~ \abstruc_{i_2} [x_1 / z_{i_2,1}, \ldots,
      x_{\arityof{\bpred_{i_2}}} / z_{i_2,\arityof{\bpred_{i_2}}}]
    \]
    is defined. However, this must be the case, for the above
    derivation rule to exist in $\overline{\asid}$. Therefore, the
    composition $\astruc_1 \comp \ldots \comp \astruc_\ell$ is
    defined.  The same type of argument can be used if the tuple
    occurs in the intersection between the interpretation of a
    relation symbol in $\astruc_\psi$ and $\astruc_i$, for some $i \in
    \interv{1}{\ell}$, thus the composition $\astruc_\psi \comp
    \astruc_1 \comp \ldots \comp \astruc_\ell$ is defined.
  \item $\astruc_\psi \comp \astruc_1 \comp \ldots \comp \astruc_\ell
    \models^{\store'} \psi * \Bigstar\nolimits_{i=1}^\ell
    \bpred_i^{\abstruc_i}(z_{i,1},\ldots,z_{i,\arityof{\bpred_i}})$ by
    construction, thus $\astruc \models^{\store} \phi$.
  \item Let $\arel\in\relations$ and
    $j_k\in\interv{1}{\arityof{\bpred}}$ for all
    $k\in\interv{1}{\arityof{\arel}}$.  We omit relations that occur
    in $\psi$ since this is a simple case.  Then
    $\tuple{x_{j_1},\ldots,x_{j_{\arityof{\arel}}}} \in
    \abstruc(\arel)$ if and only if
    $\tuple{x_{j'_1},\ldots,x_{j'_{\arityof{\arel}}}} \in
    \abstruc_i(\arel)$ for some $i\in\interv{1}{\ell}$ and $x_{j_k} =
    z_{i,j'_k}$ for all $k\in\interv{1}{\arityof{\arel}}$, if and only
    if
    $\tuple{\store_i(x_{j'_1}),\ldots,\store_i(x_{j'_{\arityof{\arel}}})}
    \in \struc_i(\arel)$ by induction hypothesis, if and only if
    $\tuple{\store(x_{j_1}),\ldots,\store(x_{j_{\arityof{\arel}}})}
    \in \struc(\arel)$. \qedhere
  \end{itemize}
\end{proof}

\noindent By taking any injective store $\store$ over
$\set{x_1,\ldots,x_{\arityof{\apred}}}$, for every predicate-free
formula $\phi$ such that $\apred^{\abstruc} \step{\overline{\asid}}^*
\phi$, we find a structure $\astruc$ such that $\astruc \models^\store
\phi$, by \autoref{fact:all-sat}. Therefore, $\overline{\asid}$ is
all-satisfiable for $\apred$.

\skipnoindent ``$\sidsem{\apred}{\asid} =
\sidsem{\apred}{\overline{\asid}}$'' The inclusion
$\sidsem{\apred}{\overline{\asid}}\subseteq\sidsem{\apred}{\asid}$ is
immediate by simply removing the base annotations from any derivation
run on $\overline{\asid}$.  For the converse, we prove the following
fact:

\begin{fact}\label{fact:all-sat2}
  For every predicate $\bpred$ of $\asid$, for every structure
  $\astruc \in \sidsem{\bpred}{\asid}$ and store $\store$ such that
  $\astruc \models^\store \bpred(x_1,\ldots,x_{\arityof{\bpred}})$,
  there exists a satisfiable base $\abstruc$, such that $\astruc
  \models^\store \bpred^{\abstruc}(x_1,\ldots,x_{\arityof{\bpred}})$,
  and $\tuple{x_{j_1},\ldots,x_{j_{\arityof{\arel}}}} \in
  \abstruc(\arel)$ if and only if
  $\tuple{\store(x_{j_1}),\ldots,\store(x_{j_{\arityof{\arel}}})} \in
  \struc(\arel)$ for all $\arel\in\relations$ and
  $j_k\in\interv{1}{\arityof{\bpred}}$ for all
  $k\in\interv{1}{\arityof{\arel}}$.
\end{fact}
\begin{proof}
  Let $\astruc \in \sidsem{\bpred}{\asid}$ and $\store$ be such that
  $\astruc \models^{\store}
  \bpred(x_1,\ldots,x_{\arityof{\bpred}})$. The proof is by induction
  on the derivation of $\astruc \models^{\store}
  \bpred(x_1,\ldots,x_{\arityof{\bpred}})$. Assume that the first step
  of this derivation uses a rule
  $\bpred(x_1,\ldots,x_{\arityof{\bpred}}) \leftarrow \exists y_1
  \ldots \exists y_m ~.~ \psi * \Bigstar\nolimits_{i=1}^\ell
  \bpred_i(z_{i,1},\ldots,z_{i,\arityof{\bpred_i}})$. Then, one can
  split $\astruc = \astruc_\psi \comp \astruc_1 \comp \ldots \comp
  \astruc_\ell$, such that $\astruc_\psi \models^{\store'} \psi$ and
  $\astruc_i \models^{\store_i}
  \bpred_i(x_1,\ldots,x_{\arityof{\bpred_i}})$, where $\store'$
  extends $\store_0$ and $\store_i(x_j) = \store'(z_{i,j})$ for all
  $i\in\interv{1}{\ell}$ and $j\in\interv{1}{\arityof{\bpred_i}}$.  By
  the induction hypothesis, for each $i\in\interv{1}{\ell}$, there
  exists a satisfiable base $\abstruc_i$, such that $\astruc_i
  \models^{\store_i}
  \bpred_i^{\abstruc_i}(x_1,\ldots,x_{\arityof{\bpred_i}})$, and
  $\tuple{x_{j_1},\ldots,x_{j_{\arityof{\arel}}}} \in
  \abstruc_i(\arel)$ if and only if
  $\tuple{\store(x_{j_1}),\ldots,\store_i(x_{j_{\arityof{\arel}}})}
  \in \struc(\arel)$.  We consider:
  \[
  \abstruc \isdef \proj{\Big(\basepairof{\psi} \otimes
    \bigotimes\nolimits_{i=1}^\ell \abstruc_i[x_1 / z_{i,1}, \ldots,
      x_{\arityof{\bpred_i}} / z_{i,\arityof{\bpred_i}}]
    \Big)}{\set{x_1,\ldots,x_{\arityof{\bpred}}}}
  \]
  and prove that $\abstruc$ is properly defined above:
  \begin{itemize}[label=$\triangleright$]
  \item $\psi$ is satisfiable, thus so is $\basepairof{\psi} =
    \abstruc_\psi$.  Projections over base pairs do not change the
    satisfiability, nor substitutions because $\overline{\asid}$ is
    equality-free.
  \item We check the satisfiability of the composition in the
    definition of $\abstruc$. Suppose, for a contradiction, that there
    exists a tuple $\tuple{z_1,\ldots,z_{\arityof{\arel}}} \in
    \abstruc_\psi (\arel) ~\cap~ \abstruc_{i} [x_1 / z_{i,1}, \ldots,
      x_{\arityof{\bpred_{i}}} / z_{i,\arityof{\bpred_{i}}}] (\arel)$
    for some $i\in\interv{1}{\ell}$ and
    $z_k\in\set{x_1,\ldots,x_{\arityof{\bpred}},y_1,\ldots,y_m}$ for
    all $k\in\interv{1}{\arityof{\arel}}$. Then there exist
    $j_1,\ldots,j_{\arityof{\arel}}\in\interv{1}{\arityof{\bpred_i}}$
    such that $\tuple{x_{j_1},\ldots,x_{j_{\arityof{\arel}}}} \in
    \abstruc_{i}(\arel)$ and $\store_i(x_{j_k}) = \store'(z_k)$ for
    all $k\in\interv{1}{\arityof{\arel}}$.  Thus we obtain
    $\tuple{\store'(z_1),\ldots,\store'(z_{\arityof{\arel}})} \in
    \struc_\psi(\arel) ~\cap~ \struc_i(\arel)$, which contradicts the
    composition $\astruc_\psi ~\comp~ \astruc_i$.  A similar argument
    ensures that no collisions occur between $\astruc_{i_1}$ and
    $\astruc_{i_2}$, for any $i_1 \neq i_2 \in \interv{1}{\ell}$.
    \end{itemize}
    The last part of \autoref{fact:all-sat2} is similar to the proof
    of \autoref{fact:all-sat}.
\end{proof}

\noindent Therefore, any model $\astruc$ of $\sidsem{\apred}{\asid}$
is a model of $\sidsem{\apred^{\abstruc}}{\overline{\asid}}$ with an
appropriate base $\abstruc$. The upper bound on
$\predsno{\overline{\asid}}$ is obtained by noticing that, for each
predicate symbol $\apred$ that occurs in $\asid$, we introduce at most
$\relationsno{\asid} \cdot \arityof{\apred}^{\maxrelarityof{\asid}}
\leq \relationsno{\asid} \cdot
\maxpredarityof{\asid}^{\maxrelarityof{\asid}}$ new predicate
symbols. \qed

\subsection{Proof of \autoref{lemma:qpf-treewidth}}
\label{app:qpf-treewidth}

The proof follows a generic guideline.  First, recall that for any set
of structures $\structures$ we have
\[
\twof{\structures} = \max_{\astruc \in \structures} \twof{\astruc} =
\max_{\astruc \in \structures} \min \{ \width{\tree} \mid \tree \text{
  is a tree decomposition of } \astruc \}.
\]
Therefore, in order to prove an inequality of the form
$\twof{\sem{\exclof{\phi}}} \le \twof{\sem{\exclof{\psi}}} + k$ for
$\phi,\psi$ two qpf formul{\ae}, we make use of the alternating $\max$
and $\min$ by proving the following:

\skipnoindent \emph{\begin{tabbing}xx \= xx \= xx \= xx \= \kill for
  every structure $\astruc$ and store $\store$ such that $\astruc
  \models^\store \phi$ \\ \> there exists a structure $\astruc'$ and a
  store $\store'$ such that $\astruc' \models^{\store'} \psi$ and \\
    \> \> for every tree decomposition $T'$ of $\astruc'$ \\ \> \> \>
    there exists a tree decomposition $T$ of $\astruc$ such that
    $\width{T} \le \width{T'} + k$.
\end{tabbing}}

\skipnoindent (\ref{it1:qpf-treewidth}) The first point is immediate
since $\sem{\exclof{(\phi * \Bigstar_{i=1}^k x_0 = x_i)}} \subseteq
\sem{\exclof{\phi}}$.

\skipnoindent (\ref{it2:qpf-treewidth}) Since $\sem{\exclof{(\phi *
    \Bigstar_{i=1}^k x_0 \neq x_i)}} \subseteq \sem{\exclof{\phi}}$,
we obtain immediately $\twof{\sem{\exclof{(\phi * \Bigstar_{i=1}^k x_0
      \neq x_i)}}} \le \twof{\sem{\exclof{\phi}}}$.  For the reverse
inequality, recall that $\phi * \Bigstar_{i=1}^k x_0 \neq x_i$ is
satisfiable.  Then $\phi$ must also be satisfiable, so let $\astruc =
(\univ,\struc)$ be a model and $\store$ a store such that $\astruc
\models^\store \phi$.  We distinguish two cases:
\begin{itemize}[label=$\triangleright$]
\item if $\store(x_0) \neq \store(x_i)$ for all $i\in\interv{1}{k}$
  then let $\astruc' = \astruc$, $\store' = \store$ hence $\astruc'
  \models^{\store'} \phi * \Bigstar_{i=1}^k x_0 \neq x_i$ and
  $\twof{\astruc} = \twof{\astruc'}$.
\item if $\store(x_0) = \store(x_i)$ for some $i \in \interv{1}{k}$
  then let us consider a new fresh element $e\in\universe$ and define a
  new store $\store'$ by $\store'(y) \isdef e$ if $y \eqof{\phi} x_0$,
  and $\store'(y) \isdef \store(y)$ otherwise. We define the $\astruc'
  = (\univ \cup \set{e},\struc')$ as follows.  For every $\arel \in
  \relations$, for every tuple $\tuple{u_1,\ldots,u_{\arityof{\arel}}}
  \in \struc(\arel)$, there exists a unique relation atom
  $\arel(y_1,\ldots,y_{\arityof{r}})$ occurring in $\phi$ such that
  $\store(y_j) = u_j$ for all $j\in\interv{1}{\arityof{\arel}}$.
  Then, add the tuple
  $(\store'(y_1),\ldots,\store'(y_{\arityof{\arel}}))$ to
  $\struc'(\arel)$.  By construction $\astruc' \models^{\store'} \phi
  * \Bigstar_{i=1}^k x_0 \neq x_i$. Let $T'$ be a tree decomposition
  of $\astruc'$.  We define $T$ by removing the element $e$ from $T'$
  and adding $\store(x_0)$ in every node of $T'$.  $T$ is a tree
  decomposition of $\astruc$ of width at most $\width{T'}+1$.
  Therefore $\width{T} \le \width{T'}+1$, hence the result.
\end{itemize}
In both cases we obtain the expected result
$\twof{\sem{\exclof{\phi}}} \le \twof{\sem{\exclof{(\phi *
      \Bigstar_{i=1}^k x_0 \neq x_i)}}} + 1$.

\skipnoindent (\ref{it3:qpf-treewidth}) We prove
$\twof{\sem{\exclof{(\phi * \arel(x_1,\ldots,x_k))}}} \le
\twof{\sem{\exclof{\phi}}} + k$. Let $\astruc = (\univ,\struc)$ and
$\store$ such that $\astruc \models^\store \phi *
\arel(x_1,\ldots,x_k)$.  We define $\astruc' = (\univ,\struc')$ from
$\astruc$ by removing the tuple
$\tuple{\store(x_1),\ldots,\store(x_k)}$ from $\struc(\arel)$.  Let
$T'$ be a tree decomposition of $\astruc'$.  We define $T$ by adding
the elements $\store(x_1),\ldots,\store(x_k)$ to every node in $T'$.
This construction does not break connectedness of the subtree of $T$
containing any element, $T$ still contains a node with all elements of
any relation in $\struc'$, and moreover (since $T'$ is not empty) $T$
contains a node (in fact all nodes) with elements
$\store(x_1),\ldots,\store(x_k)$ simultaneously. Therefore $T$ is a
tree decomposition of $\astruc$ of width at most $\width{T'}+k$.

We now prove that $\twof{\sem{\exclof{\phi}}} \le
\twof{\sem{\exclof{(\phi * \arel(x_1,\ldots,x_k))}}}+ 1$.  Recall
$\phi * \arel(x_1,\ldots,x_k)$ is satisfiable from the hypothesis.
Let $\astruc = (\univ,\struc) \models^\store \phi$, and we distinguish
two cases:
\begin{itemize}[label=$\triangleright$]
\item If $\tuple{\store(x_1),\ldots,\store(x_k)} \notin
  \struc(\arel)$, then consider $\store' = \store$ and $\astruc'$
  obtained by adding the above tuple to $\struc(\arel)$.  Then
  $\astruc' \models^{\store'} \phi * \arel(x_1,\ldots,x_k)$ and for
  any tree decomposition $T'$ of $\astruc'$ we have $T = T'$ is also a
  tree decomposition for $\astruc$ hence ensuring $\width{T} =
  \width{T'}$.
\item If $\tuple{\store(x_1),\ldots,\store(x_k)} \in \struc(\arel)$
  then, because $\phi * \arel(x_1,\ldots,x_k)$ is satisfiable, there
  must exist variables $x_1',\ldots,x_k'$ and $j\in\interv{1}{k}$
  such that $\arel(x_1',\ldots,x_k')$ occurs in $\phi$, $\store(x_i) =
  \store(x_i')$, for every $i\in\interv{1}{k}$, and moreover $x_j
  \not\eqof{\phi} x_j'$.  Let us consider a new fresh element $e \in
  \universe$ and define a new store $\store'$ by $\store'(y) = e$ if
  $y \eqof{\phi} x'_j$, and $\store'(y) = \store(y)$ otherwise.  Let
  $\astruc' = (\univ \cup \set{e},\struc')$ defined as follows.  For
  every $\arel' \in \relations$, for every tuple
  $\tuple{u_1,\ldots,u_{\arityof{\arel'}}} \in \struc(\arel')$, there
  exists a unique relation atom $\arel'(y_1,\ldots,y_{\arityof{r'}})$
  occurring in $\phi$ such that $\store(y_j) = u_j$ for all
  $j\in\interv{1}{\arityof{\arel'}}$.  Then, add the tuple
  $(\store'(y_1),\ldots,\store'(y_{\arityof{\arel'}}))$ to
  $\struc'(\arel)$.  Finally, add the tuple
  $\tuple{\store'(x_1),\ldots,\store'(x_k)}$ to $\struc'(\arel)$.  By
  construction $\astruc' \models^{\store'} \phi *
  \arel(x_1,\ldots,x_k)$.  Let $T'$ be a tree decomposition of
  $\astruc'$.  We define $T$ by removing the element $e$ from $T'$ and
  adding $\store(x_j')$ in every node of $T'$.  $T$ is a tree
  decomposition of $\astruc$ of width at most $\width{T'}+1$, that is,
  $\width{T} \le \width{T'} + 1$.
\end{itemize}
In both cases we obtain the expected result, that is,
$\twof{\sem{\exclof{\phi}}} \le \twof{\sem{\exclof{(\phi *
      \arel(x_1,\ldots,x_k))}}} + 1$.

\skipnoindent (\ref{it4:qpf-treewidth}) Similar to point
(\ref{it3:qpf-treewidth}), it generalizes the (right) inequality from
formul{\ae} consisting of a single relation atom to formul{\ae}
consisting of arbitrarily many relation atoms. \qed

\subsection{Proof of \autoref{lemma:sep-treewidth}}
\label{app:sep-treewidth}

Note that $\fv{\eta} \subseteq F \subseteq \fv{\phi}$ and henceforth,
$\fv{\phi * \eta} = \fv{\phi}$. We follow a similar strategy as in the
proof of \autoref{lemma:qpf-treewidth}. Consider a structure $\astruc
= (\univ, \struc)$ and a store $\store$, such that $\astruc
\models^{\store} \phi * \eta$. We build a store $\store'$ and a
structure $\astruc'$ such that $\astruc' \models^{\store'} \phi *
\psi$ and $\twof{\astruc} \le \twof{\astruc'}+\cardof{F}$. First,
consider a store $\store''$, that is canonical for $\psi$
(\autoref{def:canonical-model}), and a structure $\astruc'' = (\univ'',
\struc'')$, such that $\astruc'' \models^{\store''} \psi$.  Such store
and structure exist, because $\psi$ is satisfiable.  Assume without
loss of generality that $\astruc''$ and $\astruc$ are disjoint
structures, that is, $\store''(\fv{\psi}) \cap \store(\fv{\phi}) =
\emptyset$.  Second, we define the store $\store'$, as follows:
\[
\store'(y) \isdef \left\{ \begin{array}{rl} \store''(y) & \mbox{if } y
  \in \fv{\psi} \\ \store''(y') & \mbox{if } y \in \fv{\phi} \setminus
  \fv{\psi} \text{ and there exists } y' \in F \text{ such that } y
  \approx_{\phi} y' \\ \store(y) & \mbox{otherwise}
\end{array} \right.
\]
Note that the definition of $\store'$ is consistent.  In particular,
for any $y \in \fv{\phi} \setminus \fv{\psi}$ there exists at most one
variable $y' \in F$, such that $y \approx_{\phi} y'$, because
otherwise, the hypothesis $x \not\approx_{\phi} y$ for all $x, y \in
F$ would not hold. We build now the structure $\astruc' =
(\univ',\sigma')$ where $\univ' \isdef \univ \cup \univ''$ and
$\sigma'(\arel)$ is defined for every relation symbol
$\arel\in\relations$ as follows:
\begin{itemize}[label=$\triangleright$]
\item add each tuple $\tuple{u_1,\ldots,u_{\arityof{r}}} \in
  \struc''(\arel)$ to $\struc'(\arel)$,
\item for every tuple $\tuple{u_1,\ldots,u_{\arityof{r}}} \in
  \struc(\arel)$, there exists a unique relation atom
  $\arel(y_1,\ldots,y_{\arityof{\arel}})$ occurring in $\phi$, such
  that $\store(y_i) = u_i$ for all $i\in\interv{1}{\arityof{r}}$; we
  add the tuple $\tuple{\store'(y_1),\ldots,\store'(y_{\arityof{r}})}$
  to $\struc'(\arel)$.
\item nothing else belongs to $\struc'(\arel)$.
\end{itemize}
This construction guarantees that $\astruc' \models^{\store'} \phi *
\psi$.  Equality and disequality atoms in $\phi * \psi$ are satisfied,
by the definition of $\store'$.  With regard to relation atoms, notice
that no tuple is added twice to $\sigma'(\arel)$ in the definition
above. That is, if some
$\tuple{\store'(y_1),\ldots,\store'(y_{\arityof{r}})}$ obtained from
$\struc(\arel)$ exists also in $\struc''(\arel)$, then $\phi * \psi$
would not be satisfiable. Let $T'$ be a tree decomposition of
$\astruc'$. We define a tree decomposition $T$ by:
\begin{itemize}[label=$\triangleright$]
\item removing $\set{ \store'(y) ~\mid~ y \in \fv{\psi}}$ from every
  bag of $T'$,
\item adding $\set{\store(y) ~\mid~ y \in F}$ to every bag of $T'$.
\end{itemize}
The result is a tree decomposition $T$ of $\astruc$ of width
$\width{T} \le \width{T'} + \cardof{F}$.  Since the choice of $T'$ was
arbitrary, we obtain $\twof{\astruc} \le
\twof{\astruc'}+\cardof{F}$. Since the choice of $\astruc$ was
arbitrary, we obtain $\twof{\sem{\exclof{(\phi * \eta)}}} \le
\twof{\sem{\exclof{(\phi * \psi)}}} + \cardof{F}$. \qed


\section{Proofs from \autoref{sec:expandable-tb}}
\subsection{Proof of \autoref{lemma:split-fusion}}
\label{app:split-fusion}

\skipnoindent (\ref{it1:lemma:split-fusion})
  ``$\reachfusion{\funsplit{\structures}} \subseteq
  \funsplit{\reachfusion{\structures}}$'' By induction on the
derivation of $\astruc \in \reachfusion{\funsplit{\structures}}$ from
$\funsplit{\structures}$.

\skipnoindent {\underline{Base case:} Let $\astruc \in
\funsplit{\structures}$. We have:
\[
  \exists \astruc' \in \structures.~ \astruc \in \funsplit{\astruc'}  \Rightarrow
  \exists \astruc' \in \reachfusion{\structures}.~ \astruc \in \funsplit{\astruc'}  \Rightarrow
  \astruc \in \funsplit{\reachfusion{\structures}}
\]

\skipnoindent \underline{Induction step:} Assume $\astruc = (\astruc_1
\comp \astruc_2)_{/\approx}$ for some $\astruc_1, \astruc_2 \in
\reachfusion{\funsplit{S}}$ and $\approx$ satisfying the conditions of
\autoref{def:external-fusion} for external fusion of $\astruc_1$ and
$\astruc_2$.  Moreover, assume the induction hypothesis $\astruc_1,
\astruc_2 \in \funsplit{\reachfusion{\structures}}$. Then
{\allowdisplaybreaks \begin{align*}
  \astruc_1, \astruc_2 \in \funsplit{\reachfusion{\structures}},~
  \astruc = (\astruc_1 \comp \astruc_2)_{/\approx} & \Rightarrow \\
  (\exists \astruc_1',\astruc_2' \in \reachfusion{\structures}.~
  \astruc_1 \mcsubstruc \astruc_1', ~\astruc_2 \mcsubstruc \astruc_2'),~
  \astruc = (\astruc_1 \comp \astruc_2)_{/\approx} & \Rightarrow \\
  \exists \astruc_1',\astruc_2' \in \reachfusion{\structures}.~
  (\exists \approx'.~\astruc \mcsubstruc (\astruc_1' \comp \astruc_2')_{/\approx'})  & \Rightarrow \\
  \exists \astruc' \in \reachfusion{\structures}.~ \astruc \mcsubstruc \astruc'& \Rightarrow
  \astruc \in \funsplit{\reachfusion{\structures}}
\end{align*}}
In the above, the equivalence $\approx'$ is taken as the extension by
equality of $\approx$ and henceforth it satisfies the conditions of
\autoref{def:external-fusion} for external fusion of $\astruc_1'$,
$\astruc_2'$.

\skipnoindent ``$\funsplit{\reachfusion{\structures}} \subseteq
  \reachfusion{\funsplit{\structures}}$'' By induction on the
derivation of $\astruc' \in \reachfusion{\structures}$ from
$\structures$.

\skipnoindent \underline{Base case:} Let $\astruc \in
\funsplit{\astruc'}$ for some $\astruc' \in \structures$. We have:
\[
  \astruc \in \funsplit{\astruc'} \Rightarrow
  \astruc \in \reachfusion{\funsplit{\astruc'}} \Rightarrow
  \astruc \in \reachfusion{\funsplit{\structures}}
\]

\skipnoindent \underline{Induction step:} Assume $\astruc \in
\funsplit{\astruc'}$ for some $\astruc' = (\astruc_1' \comp
\astruc_2')_{/\approx'}$ for some $\astruc_1', \astruc_2' \in
\reachfusion{\structures}$ and equivalence $\approx'$ satisfying the
conditions of \autoref{def:external-fusion} for external fusion
of $\astruc_1'$, $\astruc_2'$.  Assume the induction
hypothesis, that is, $\funsplit{\astruc_1'}, \funsplit{\astruc_2'}
\subseteq \reachfusion{\funsplit{\structures}}$.  Then:
\begin{align*}
  \funsplit{\astruc_1'}, \funsplit{\astruc_2'} \subseteq \reachfusion{\funsplit{\structures}},~
  \astruc' = (\astruc_1' \comp \astruc_2')_{/\approx'},~ \astruc \in \funsplit{\astruc'} & \Rightarrow \\
  \funsplit{\astruc_1'}, \funsplit{\astruc_2'} \subseteq \reachfusion{\funsplit{\structures}},~
  \astruc' = (\astruc_1' \comp \astruc_2')_{/\approx'},~ \astruc \mcsubstruc \astruc'
\end{align*}
We distinguish two sub-cases:
\begin{itemize}
 \item $\astruc$ is a maximally connected substructure of $\astruc'_1$
   (the case of $\astruc'_2$ is symmetric) not affected by the external
   fusion defined by $\approx'$:
  \begin{align*}
    \funsplit{\astruc_1'} \subseteq \reachfusion{\funsplit{\structures}},~
    (\exists \astruc_1.~ \astruc_1 \mcsubstruc \astruc_1',~ \astruc = \astruc_1)  & \Rightarrow \\
    \funsplit{\astruc_1'} \subseteq \reachfusion{\funsplit{\structures}},~
    \astruc \in \funsplit{\astruc_1'}  & \Rightarrow
    \astruc \in \reachfusion{\funsplit{\structures}}
  \end{align*}
   \item $\astruc$ is a connected structure including several maximally
     connected substructures, at least one from each $\astruc'_i$, for
     $i=1,2$:
   {\allowdisplaybreaks \begin{align*}
   \funsplit{\astruc_1'}, \funsplit{\astruc_2'} \subseteq \reachfusion{\funsplit{\structures}},~\\
   \exists k_1\ge 1.~ \exists \astruc_{1,1} ... \exists \astruc_{1,k_1}.~
   \astruc_{1,i} \mcsubstruc \astruc_1' \mbox{ for all } i \in \interv{1}{k_1},~ \\
   \exists k_2\ge 1.~  \exists \astruc_{2,1} ... \exists \astruc_{2,k_2}.~
   \astruc_{2,j} \mcsubstruc \astruc_2' \mbox{ for all } j \in \interv{1}{k_2},~ \\
   (\exists \approx.~\astruc = (\astruc_{1,1} \comp ... \comp \astruc_{1,k_1} \comp
   \astruc_{2,1} \comp ... \comp \astruc_{2,k_2})_{/\approx}),~ \astruc \mbox{ connected} & \Rightarrow \\
    \funsplit{\astruc_1'}, \funsplit{\astruc_2'} \subseteq \reachfusion{\funsplit{\structures}},~ \\
    \exists k_1\ge 1.~ \exists \astruc_{1,1} ... \exists \astruc_{1,k_1}.~
    \astruc_{1,i} \in \funsplit{\astruc_1'} \mbox{ for all } i \in \interv{1}{k_1},~ \\
    \exists k_2\ge 1.~  \exists \astruc_{2,1} ... \exists \astruc_{2,k_2}.~
    \astruc_{2,j} \in \funsplit{\astruc_2'} \mbox{ for all } j \in \interv{1}{k_2},~ \\
    (\exists \approx.~ \astruc = (\astruc_{1,1} \comp ... \comp \astruc_{1,k_1} \comp
    \astruc_{2,1} \comp ... \comp \astruc_{2,k_2})_{/\approx}),~ \astruc \mbox{ connected} & \Rightarrow \\
    \exists k_1\ge 1.~ \exists \astruc_{1,1} ... \exists \astruc_{1,k_1}.~
    \astruc_{1,i} \in \reachfusion{\funsplit{\structures}} \mbox{ for all } i \in \interv{1}{k_1},~ \\
    \exists k_2\ge 1.~  \exists \astruc_{2,1} ... \exists \astruc_{2,k_2}.~
    \astruc_{2,j} \in \reachfusion{\funsplit{\structures}} \mbox{ for all } j \in \interv{1}{k_2},~ \\
    (\exists \approx.~ \astruc = (\astruc_{1,1} \comp ... \comp \astruc_{1,k_1} \comp
    \astruc_{2,1} \comp ... \comp \astruc_{2,k_2})_{/\approx}),~ \astruc \mbox{ connected} & \Rightarrow
    \astruc \in \reachfusion{\funsplit{\structures}}
  \end{align*}}
\end{itemize}
In the above, the equivalence $\approx$ is the restriction of $\approx'$ to
the substructures included in the composition.  As $\approx'$ is conforming
for external fusion of $\astruc_1'$, $\astruc_2'$ and since the resulting
structure $\astruc$ is connected, it is always possible to obtain $\astruc$
as a sequence of external fusions conforming to
\autoref{def:external-fusion} from the respective substructures.

\skipnoindent (\ref{it2:lemma:split-fusion})
  ``$\ireachfusion{\funsplit{\structures}} \subseteq
  \funsplit{\ireachfusion{\structures}}$'' By induction on the
derivation of $\astruc\in\ireachfusion{\funsplit{\structures}}$ from
$\funsplit{\structures}$. The induction proceeds as for
(\ref{it1:lemma:split-fusion}), with one additional case in the
induction step.

\skipnoindent \underline{Induction step:} Let $\astruc =
(\astruc_1)_{/\approx}$ for some $\astruc_1 \in
\ireachfusion{\funsplit{S}}$ and equivalence relation $\approx$
conforming to internal fusion of $\astruc_1$. Moreover, assume the
induction hypothesis $\astruc_1 \in
\funsplit{\ireachfusion{\structures}}$.  Then
{\allowdisplaybreaks \begin{align*}
  \astruc_1 \in \funsplit{\ireachfusion{\structures}},~ \astruc = (\astruc_1)_{/\approx} & \Rightarrow \\
  (\exists \astruc_1' \in \ireachfusion{\structures}.~\astruc_1 \mcsubstruc \astruc_1'), ~
  \astruc = (\astruc_1)_{/\approx} & \Rightarrow \\
  \exists \astruc_1' \in \ireachfusion{\structures}.~ (\exists \approx'.~
  \astruc \mcsubstruc (\astruc_1')_{/\approx'}) & \Rightarrow \\
  \exists \astruc' \in \ireachfusion{\structures},~ \astruc \mcsubstruc \astruc' & \Rightarrow
  \astruc \in \funsplit{ \ireachfusion{\structures}}
\end{align*}}
In the above, the equivalence $\approx'$ is taken as the extension by equality of
$\approx$ and hence conforming for internal fusion of structure $\astruc_1'$.

\skipnoindent ''$\funsplit{\ireachfusion{\structures}} \subseteq
\ireachfusion{\funsplit{\structures}}$'' By induction on the
derivation of $\astruc'\in\ireachfusion{\structures}$ from
$\structures$. The induction proceeds as for
(\ref{it1:lemma:split-fusion}), with one additional case in the
induction step.

\skipnoindent \underline{Induction step:} Let $\astruc \in
\funsplit{\astruc'}$ for some $\astruc' = (\astruc_1')_{/\approx'}$
for some $\astruc_1' \in \ireachfusion{\structures}$ and equivalence
$\approx'$ conforming for internal fusion of $\astruc_1'$.  Moreover,
assume the induction hypothesis $\funsplit{\astruc_1'} \subseteq
\ireachfusion{\funsplit{\structures}}$. Then
{\allowdisplaybreaks \begin{align*}
  \funsplit{\astruc_1'} \subseteq \ireachfusion{\funsplit{\structures}},~
  \astruc' = (\astruc_1')_{/\approx'},~ \astruc \in \funsplit{\astruc'} & \Rightarrow \\
  \funsplit{\astruc_1'} \subseteq \ireachfusion{\funsplit{\structures}},~
  \astruc' = (\astruc_1')_{/\approx'},~ \astruc \mcsubstruc \astruc' & \Rightarrow \\
  \funsplit{\astruc_1'} \subseteq \ireachfusion{\funsplit{\structures}},~ \\
  \exists k\ge 1.~ \exists \astruc_{1,1}... \exists \astruc_{1,k}.~
  \astruc_{1,i} \mcsubstruc \astruc_1' \mbox{ for all } i\in\interv{1}{k},~ \\
  (\exists \approx.~ \astruc = (\astruc_{1,1} \comp ...  \comp \astruc_{1,k})_{/\approx},~
  \astruc \mbox{ connected}) & \Rightarrow \\
  \funsplit{\astruc_1'} \subseteq \ireachfusion{\funsplit{\structures}},~ \\
  \exists k\ge 1.~ \exists \astruc_{1,1}... \exists \astruc_{1,k}.~
  \astruc_{1,i} \in \funsplit{\astruc_1'} \mbox{ for all } i\in\interv{1}{k},~ \\
  (\exists \approx.~ \astruc = (\astruc_{1,1} \comp ...  \comp \astruc_{1,k})_{/\approx},~
  \astruc \mbox{ connected}) & \Rightarrow \\
  \exists k\ge 1.~ \exists \astruc_{1,1}... \exists \astruc_{1,k}.~
  \astruc_{1,i} \in \ireachfusion{\funsplit{\structures}} \mbox{ for all } i\in\interv{1}{k},~ \\
  (\exists \approx.~ \astruc = (\astruc_{1,1} \comp ...  \comp \astruc_{1,k})_{/\approx},~
  \astruc \mbox{ connected}) & \Rightarrow 
  \astruc \in \ireachfusion{\funsplit{\structures}}
\end{align*}}
In the above, the equivalence $\approx$ is taken as the restriction of
$\approx'$ to the maximal connected substructures included in the
construction of connected $\astruc$.  Henceforth, $\approx$ is conforming
for internal fusion as well.  As the resulting structure $\astruc$ is
connected, it is always possible to construct it in
$\ireachfusion{\funsplit{\structures}}$ i.e., first by using external
fusion conforming to \autoref{def:external-fusion} to connect all the
included substructures and second, by using internal fusion to further
restrict the result if needed. \qed

\subsection{Proof of \autoref{lemma:sid-connected}}
\label{app:maximally-connected-substructures}

``$\funsplit{\csem{\asid}{\apred}} \subseteq \csem{\csid}{\ppred}$''
We prove first the following fact:

\begin{fact}\label{fact:asid-csid}
  Let $\bpred_0(x_1,\ldots,x_{\arityof{\bpred_0}}) \step{\asid}^*
  \exists y_1 \ldots \exists y_n ~.~ \phi$ be a complete
  $\asid$-unfolding, where $\phi$ is a qpf formula, $\store$ a store
  injective over $\set{x_1,\ldots,x_{\arityof{\bpred_0}}} \cup
  \set{y_1,\ldots,y_n}$, $\astruc=(\univ,\struc)$ a structure such
  that $\astruc\models^\store \phi$ and $\astruc'=(\univ',\struc')$ a
  structure, such that $\astruc' \mcsubstruc \astruc$ and
  $\supp{\struc'} \cap \set{\store(x_1), \ldots,
    \store(x_{\arityof{\bpred_0}})} \neq \emptyset$. Then, there exist
  a nonempty set $J_0\subseteq\interv{1}{\arityof{\bpred_0}}$,
  equivalence relation $\xi_0 \subseteq J_0 \times J_0$ and complete
  $\csid$-unfolding $\bpred_0(x_1, \ldots,
  x_{\arityof{\bpred_0}})_{/\xi_0} \step{\csid}^* \exists y_{1} \ldots
  \exists y_{n} ~.~ \phi'$, where $\phi'$ is a qpf formula, such that
  $\astruc' \models^\store \phi'$.
\end{fact}
\begin{proof}
  By induction on the length of the $\asid$-unfolding. Assume the
  first rule in this unfolding to be of the form (\ref{eq:asidrule}),
  for a qpf formula $\psi_0$. Then, there exist:
  \begin{itemize}[label=$\triangleright$]
  \item unfoldings $\bpred_i(x_{1}, \ldots, x_{\arityof{\bpred_i}})
    \step{\asid}^* \exists y_{j_{i,1}} \ldots \exists y_{j_{i,k_i}}
    ~.~ \phi_i$, where $y_{j_{i,1}}, \ldots, y_{j_{i,k_i}} \in
    \set{y_1,\ldots,y_n}$ and $\phi_i$ are qpf formul{\ae}, for all
    $i \in \interv{1}{\ell}$, and
  \item structures $\astruc_0 = (\univ_0,\struc_0), \ldots,
    \astruc_\ell=(\univ_\ell,\struc_\ell)$, such that $\astruc_0
    \comp \ldots \comp \astruc_\ell = \astruc$, $\astruc_0
    \models^\store \psi_0$ and $\astruc_i \models^\store
    \phi_i\theta_i$, where $\theta_i \isdef [x_1/z_{i,1}, \ldots,
      x_{\arityof{\bpred_i}}/z_{i,\arityof{\bpred_i}}]$, for all $i
    \in \interv{1}{\ell}$.
  \end{itemize}
  Since $\astruc' \mcsubstruc \astruc_0 \comp \ldots \comp
  \astruc_\ell$ is a maximally connected structure, there must exist
  structures $\astruc'_0 = (\univ_0, \struc'_0)$,
  $\astruc'_{i_1}=(\univ_{i_1}, \struc'_{i_1}), \ldots,
  \astruc'_{i_k} =(\univ_{i_k},\struc'_{i_k})$, for $i_1, \ldots,
  i_k \in \interv{1}{\ell}$, such that:
  \begin{itemize}[label=$\triangleright$]
  \item $\astruc' = \astruc'_0 \comp \astruc'_{i_1} \comp \ldots
    \comp \astruc'_{i_k}$,
  \item $\supp{\struc'_{i_h}} \cap\set{\store(z_{i_h,1}), \ldots,
    \store(z_{i_m,\arityof{\bpred_{i_h}}})} \neq \emptyset$ and
    $\astruc'_{i_h} \mcsubstruc \astruc_{i_h}$, for all $h \in
    \interv{1}{k}$,
  \item $\supp{\struc'} \cap \set{\store(z_{i,1}), \ldots,
    \store(z_{i,\arityof{\bpred_i}})} = \emptyset$, for all $i \in
    \interv{1}{\ell} \setminus \set{i_1,\ldots,i_k}$.
  \end{itemize}
  Since $\astruc_{i_h} \models^\store \phi_{i_h} \theta_{i_h}$, we
  have $\astruc_{i_h} \models^{\store \circ \theta_{i_h}^{-1}}
  \phi_{i_h}$, for all $h\in\interv{1}{k}$. By the inductive
  hypothesis, there exist nonempty sets $J_{i_1} \subseteq
  \interv{1}{\arityof{\bpred_{i_1}}}, \ldots, J_{i_k} \subseteq
  \interv{1}{\arityof{\bpred_{i_k}}}$, equivalence relations
  $\xi_{i_1} \subseteq J_{i_1} \times J_{i_1}, \ldots, \xi_{i_k}
  \subseteq J_{i_k} \times J_{i_k}$ and complete $\csid$-unfoldings:
  \[
  \bpred_{i_h}(x_1, \ldots,
  x_{\arityof{\bpred_{i_h}}})_{/\xi_{i_h}} \step{\csid}^* \exists
  y_{j_{i_h,1}} \ldots \exists y_{j_{i_h,k_{i_h}}} ~.~ \phi'_{i_h}
  \]
  such that $\astruc'_{i_h} \models^{\store\circ\theta_{i_h}^{-1}}
  \phi'_{i_h}$, for all $h \in \interv{1}{k}$. Then, we define:
  \begin{itemize}[label=$\triangleright$]
  \item sets $\overline{J}_{i_h} \isdef \interv{1}{\arityof{\bpred_i}}
    \setminus J_{i_h}$, for all $h\in\interv{1}{k}$,
  \item qpf formul{\ae} $\psi'_0$ and $\psi''_0$ satisfying points
    (\ref{it1:mcsid}) and (\ref{it2:mcsid}) from the construction of
    $\csid$,
  \item an equivalence relation $\Xi\isdef\big(\connof{\psi'_0} \cup
    \bigcup\nolimits_{h=1}^k
    \xi_{i_h}(\bpred_{i_h}(z_{i_h,1},\ldots,z_{i_h,\arityof{\bpred_{i_h}}}))
    \big)^=$.
  \end{itemize}
  We argue that the construction of the formul{\ae} $\psi'_0$ and
  $\psi''_0$ is effective. There are no (dis-) equalities in $\asid$,
  i.e., $\psi_0$ consists of relation atoms only. Each atom $\alpha$
  of $\psi_0$, such that $\fv{\alpha}\cap
  \fvof{J_{i_h}}{\bpred_{i_h}(z_{i_h,1},\ldots,z_{i_h,\arityof{\bpred_{i_h}}})}
  \neq \emptyset$, for some $h \in \interv{1}{k}$, is added to
  $\psi'_0$. Moreover, each atom $\alpha$ of $\psi_0$, such that
  $\fv{\alpha} \cap
  \fvof{\overline{J}_{i_h}}{\bpred_{i_h}(z_{i_h,1},\ldots,z_{i_h,\arityof{\bpred_{i_h}}})}
  = \emptyset$, for all $h \in \interv{1}{k}$ is added to
  $\psi''_0$. Note that, each atom can only be added either to
  $\psi'_0$ or $\psi''_0$ but not to both, because $\astruc'_{i_h}
  \mcsubstruc \astruc_{i_h}$ and $\astruc'_{i_h} \models^\store
  \phi'_{i_h}$ imply that no further element can be added to
  $\astruc'_{i_h}$, for all $h \in \interv{1}{k}$. The rest of the
  atoms $\alpha$ from $\psi_0$, i.e., such that $\fv{\alpha}\cap
  \fvof{J_{i_h}}{\bpred_{i_h}(z_{i_h,1},\ldots,z_{i_h,\arityof{\bpred_{i_h}}})}
  = \emptyset$ and $\fv{\alpha} \cap
  \fvof{\overline{J}_{i_h}}{\bpred_{i_h}(z_{i_h,1},\ldots,z_{i_h,\arityof{\bpred_{i_h}}})}
  = \emptyset$, for all $h\in\interv{1}{k}$, are split between
  $\psi'_0$ and $\psi''_0$, by repeating the following steps until a
  fixpoint is reached:
  \begin{itemize}[label=$\triangleright$]
  \item if $\fv{\alpha} \cap \fv{\psi'_0} \neq \emptyset$, then update
    $\psi'_0$ as $\psi'_0 * \alpha$,
  \item else, update $\psi''_0$ as $\psi''_0 * \alpha$.
  \end{itemize}
  By construction, we obtain that $\fv{\psi'_0} \cap \fv{\psi''_0} =
  \emptyset$, as required at point (\ref{it1:mcsid}).
  
  Let $J_0 \isdef \set{j \in \interv{1}{\arityof{\bpred_0}} \mid
    \store(x_j) \in \supp{\struc'}}$. Note that $J_0\neq\emptyset$
  because $\supp{\struc'} \cap \set{\store(x_1), \ldots,
    \store(x_{\arityof{\bpred_0}})} \neq \emptyset$. We define the
  equivalence relation $\xi_0 \subseteq J_0 \times J_0$ as follows:
  \[
  (i,j)\in\xi_0 \iffdef (x_i,x_j) \in \Xi
  \]
  Moreover, one can show that $\astruc'_0 \models^\store \psi'_0$, by
  the construction of $\psi'_0$ and $\psi''_0$, hence $\xi_0$
  satisfies the conditions (\ref{it4:mcsid})--(\ref{it7:mcsid}) from the
  definition of $\csid$, hence $\csid$ contains a rule of the form
  (\ref{eq:mcsid-middle}), with qpf formula $\psi'_0$. Since $\astruc'
  \models^\store \psi'_0 * \Bigstar\nolimits_{h=1}^k \phi'_{i_h} \theta_h$ and
  $\phi' = \psi'_0 * \Bigstar\nolimits_{h=1}^k \phi'_{i_h} \theta_h$ modulo a
  reordering of atoms, we obtain that $\astruc' \models^\store \phi'$.
\end{proof}

\skipnoindent The proof is completed as follows. Let $\astruc' \in
\funsplit{\csem{\apred}{\asid}}$ be a maximally connected substructure
of a canonical model $\astruc \in \csem{\apred}{\asid}$. Then, there
exists a complete $\asid$-unfolding $\apred \step{\asid}^* \exists y_1
\ldots \exists y_m ~.~ \overline{\phi}$ and a store $\store$,
injective over $y_1, \ldots, y_m$, such that $\astruc \models^\store
\overline{\phi}$. Because $\astruc'$ is connected, there exists a
unique (i) predicate atom $\bpred_0(z_1, \ldots,
z_{\arityof{\bpred_0}})$, (ii) subformula $\phi$ of $\overline{\phi}$
and (iii) structure $\astruc'' \substruc \astruc$, such that
$\bpred_0(z_1, \ldots, z_{\arityof{\bpred_0}}) \step{\asid}^* \exists
y_{i_1} \ldots \exists y_{i_n} ~.~ \phi$ is a complete unfolding,
$\supp{\astruc''} \cap \set{\store(z_1), \ldots,
  \store(z_{\arityof{\bpred_0}})}=\emptyset$, $\astruc''
\models^\store \phi$ and $\astruc' \mcsubstruc \astruc''$. Without
losing generality, we assume that the above is the smallest
$\asid$-unfolding with these properties and assume that the first rule
of the $\asid$-unfolding is of the form (\ref{eq:asidrule}), with a
qpf formula $\psi$. Then, because $\astruc'$ is connected, the
right-hand side of this rule contains zero or more predicate atoms
$\bpred_{i_h}(z_{i_h,1}, \ldots, z_{i_h,\arityof{\bpred_{i_h}}})$, for
$i_h \in \set{i_1, \ldots, i_k} \subseteq \interv{1}{\ell}$, such that
$\supp{\struc'} \cap \set{\store(z_{i_h,1}), \ldots,
  \store(z_{i_h,\arityof{\bpred_{i_h}}})} \neq
\emptyset$. Accordingly, we decompose $\astruc''=\astruc''_0 \comp
\ldots \comp \astruc''_k$ such that $\astruc''_0 \models^\store
\psi_0$ and $\astruc''_h \models^\store \phi_h$, where
$\bpred_{i_h}(z_{i_h,1}, \ldots, z_{i_h,\arityof{\bpred_{i_h}}})
\step{\asid}^* \exists y_{j_{h,1}} \ldots \exists y_{j_{h,m_h}} ~.~
\phi_h$ are complete unfoldings, for all $h \in \interv{1}{k}$. This
decomposition of $\astruc''$ induces a decomposition of $\astruc'
\mcsubstruc \astruc''$ as $\astruc'_0 \comp \ldots \comp \astruc'_k =
\astruc'$ such that $\astruc'_h \mcsubstruc \astruc''_h$ and
$\supp{\struc'_h} \cap \set{\store(z_{i_h,1}), \ldots,
  \store(z_{i_h,\arityof{\bpred_{i_h}}})} \neq \emptyset$, for all $h
\in \interv{1}{k}$. Applying \autoref{fact:asid-csid}, we find
nonempty subsets $J_h \subseteq \interv{1}{\arityof{\bpred_{i_h}}}$,
equivalence relations $\xi_h \subseteq J_h \times J_h$ and complete
$\csid$-unfoldings $\bpred_{i_h}(z_{i_h,1}, \ldots,
z_{i_h,\arityof{\bpred_{i_h}}})_{/\xi_h} \step{\csid}^* \exists
y_{p_{h,1}} \ldots \exists y_{p_{h,n_h}} ~.~ \phi'_h$, such that
$\astruc'_h \models^\store \phi'_h$, for all $h \in \interv{1}{k}$.
We define the sets $\overline{J}_h =
\interv{1}{\arityof{\bpred_{i_h}}} \setminus J_h$ and the formul{\ae}
$\psi'_0$ and $\psi''_0$ such that conditions (\ref{it1:mcsid}) and
(\ref{it2:mcsid}) are met. Let $\Xi \isdef \big(\connof{\psi'_0} \cup
\bigcup\nolimits_{h=1}^k
\xi_{i_h}(\bpred_{i_h}(z_{i_h,1},\ldots,z_{i_h,\arityof{\bpred_{i_h}}}))
\big)^=$ be an equivalence relation. Since $\astruc'$ is connected, we
argue that $\Xi$ has a single equivalence class. Moreover, $(z,z)
\not\in \Xi$, for all $z \in \set{z_1, \ldots,
  z_{\arityof{\bpred_0}}}$, since $\supp{\astruc''} \cap
\set{\store(z_1), \ldots, \store(z_{\arityof{\bpred_0}})}=\emptyset$
and $\store$ is injective. Then, by definition, $\csid$ contains a
rule of the form (\ref{eq:mcsid-begin}). This rule and the complete
unfoldings $\bpred_{i_h}(z_{i_h,1}, \ldots,
z_{i_h,\arityof{\bpred_{i_h}}}) \step{\csid}^* \exists y_{j_{h,1}}
\ldots \exists y_{j_{h,m_h}} ~.~ \phi'_h$, for all $h \in
\interv{1}{k}$, are composed to make up a complete $\csid$-unfolding
$\ppred \step{\csid} \exists y_{q_1} \ldots \exists y_{q_r} ~.~
\psi'_0 * \Bigstar\nolimits_{h=1}^k \phi'_h$, such that $\astruc'
\models^\store \psi'_0 * \Bigstar\nolimits_{h=1}^k \phi'_h$. Since
$y_{q_1}, \ldots, y_{q_r} \in \set{y_1, \ldots, y_m}$ and $\store$ is
injective over $y_1, \ldots, y_m$, we obtain that $\astruc' \in
\csem{\ppred}{\csid}$.

\skipnoindent ``$\csem{\csid}{\ppred} \subseteq
\funsplit{\csem{\asid}{\apred}}$'' We prove first two related
facts. First, let
$\bpred_0(x_1,\ldots,x_{\arityof{\bpred_0}})_{/\xi_0} \step{\csid}^*
\exists y_1 \ldots \exists y_n ~.~ \phi$ be a complete
$\csid$-unfolding, where $\phi$ is a qpf formula, $J_0 \subseteq
\interv{1}{\arityof{\bpred_0}}$ a nonempty set and $\xi_0 \subseteq
J_0 \times J_0$ an equivalence relation, $\store$ be a store injective
over $\set{x_1,\ldots,x_{\arityof{\bpred_0}}} \cup
\set{y_1,\ldots,y_n}$ and $\astruc=(\univ,\struc)$ be a structure such
that $\astruc\models^\store\phi$. Given an equivalence class $I
\subseteq J_0$ of $\xi_0$, we define the structure:
\[
\reachof{\astruc}{\store}{I} \isdef (\univ, \lambda \arel ~.~
\{\tuple{u_1, \ldots, u_{\arityof{\arel}}} \in \struc(\arel) \mid
\forall j \in \interv{1}{\arityof{\arel}} ~\exists i \in I ~.~
\store(x_i) \text{ connected to } u_j \text{ in } \astruc\})
\]

\begin{fact}\label{fact:csid-mc}
  $\reachof{\astruc}{\store}{I} \mcsubstruc \astruc$.
\end{fact}
\begin{proof}
  By induction on the length of the $\csid$-unfolding. Assume that the
  first rule of the unfolding is of the form (\ref{eq:mcsid-middle}),
  with a qpf formula $\psi_0$. Then, there exist nonempty sets
  $J_{i_h} \subseteq \interv{1}{\arityof{\bpred_{i_h}}}$ and
  equivalence relations $\xi_{i_h} \subseteq J_{i_h} \times J_{i_h}$,
  for some $i_1, \ldots, i_k \in \interv{1}{\ell}$ and all $h \in
  \interv{1}{k}$ and an equivalence relation
  $\Xi\subseteq\big(\set{x_1,\ldots,x_{\arityof{\bpred_0}}}\cup\set{y_1,\ldots,y_m}\big)
  \times
  \big(\set{x_1,\ldots,x_{\arityof{\bpred_0}}}\cup\set{y_1,\ldots,y_m}\big)$,
  that satisfy points (\ref{it1:mcsid})--(\ref{it7:mcsid}) from the
  definition of $\csid$. By point (\ref{it7:mcsid}), $\set{x_i \mid i
    \in I}$ is an equivalence class of
  $\proj{\Xi}{x_1,\ldots,x_{\arityof{\bpred_0}}}$ and let $X \subseteq
  \fv{\psi'_0} \cup \bigcup\nolimits_{h=1}^k \set{z_{i_h,1}, \ldots,
    z_{i_h,\arityof{\bpred_{i_h}}}}$ be the unique equivalence class
  of $\Xi$ that contains it. For each $h \in \interv{1}{k}$, let
  $I_h\subseteq\interv{1}{\arityof{\bpred_{i_h}}}$ be the equivalence
  class of $\xi_{i_h}$ used to define $X$ (\ref{it3:mcsid}).

  Let $\bpred_{i_h}(x_{1}, \ldots, x_{\arityof{\bpred_{i_h}}})
  \step{\csid}^* \exists y_{j_{h,1}} \ldots \exists y_{j_{h,m_h}} ~.~
  \phi_h$ be complete $\csid$-unfoldings, such that $\phi = \psi'_0 *
  \Bigstar\nolimits_{h=1}^k \phi_h\theta_h$, where $\theta_h \isdef [
    x_1/z_{i_h,1}, \ldots, x_{\arityof{\bpred_{i_h}}} /
    z_{i_h,\arityof{\bpred_{i_h}}} ]$, for each $h \in
  \interv{1}{k}$. Since $\astruc \models^\store \phi$, there exist
  structures $\astruc_0 = (\univ_0,\struc_0), \ldots,
  \astruc_k=(\univ_k,\struc_k)$ such that $\astruc_0 \models^\store
  \psi'_0$ and $\astruc_h \models^{\store} \phi_h\theta_h$, or
  equivalently, $\astruc_h \models^{\store\circ\theta^{-1}_h} \phi_h$
  for all $h \in \interv{1}{k}$. By the inductive hypothesis, we have
  $\reachof{\astruc_h}{\store\circ\theta_h^{-1}}{I_h} \mcsubstruc
  \astruc_h$, for all $h \in \interv{1}{k}$. Since $X$ is an
  equivalence class of $\Xi$, by point (\ref{it3:mcsid}) of the
  definition of $\csid$, we obtain that $\reachof{\astruc}{\store}{I}
  \mcsubstruc \astruc$.
\end{proof}

\skipnoindent Second, let $\bpred_0(x_1,\ldots,x_{\arityof{\bpred_0}})
\step{\asid}^* \exists y_1 \ldots \exists y_m ~.~ \phi'$ be the
complete $\asid$-unfolding obtained by replacing each rule of the form
(\ref{eq:mcsid-middle}) with its corresponding rule
(\ref{eq:asidrule}) in the above $\csid$-unfolding, such that $\phi'$
is a qpf formula and $y_1, \ldots, y_n \in \set{y_1, \ldots,
  y_m}$. Note that the latter can be assumed w.l.o.g., if necessary,
by a renaming of the quantified variables.

\begin{fact}\label{fact:csid-asid}
  There exists a store $\store'$, that is injective over $y_1, \ldots,
  y_m$ and agrees with $\store$ over $y_{1}, \ldots, y_{n}$, and a
  structure $\astruc'=(\univ, \struc')$, such that $\astruc'
  \models^{\store'} \phi'$ and $\reachof{\astruc}{\store}{I}
  \mcsubstruc \astruc'$, for each equivalence class $I\subseteq J_0$
  of $\xi_0$.
\end{fact}
\begin{proof} The store $\store'$ is defined as:
\begin{itemize}[label=$\triangleright$]
\item $\store'(y_i)=\store(y_i)$, for each $i \in \interv{1}{n}$,
\item $\store'(y_i)$ is chosen from $\univ\setminus\set{\store(y_1),
  \ldots, \store(y_n)}$ such that, moreover, $\store'(y_i) \neq
  \store'(y_j)$, for all $n+1 \le i < j \le m$.
\end{itemize}
Note that $\store'$ can be build, because $\univ$ is infinite.
Because of the assumption that each predicate defined by a rule from
$\asid$ occurs on some complete $\asid$-unfolding of $\apred$, there
exists a complete $\asid$-unfolding:
\[
\apred \step{\asid} \ldots \step{\asid}
\bpred_0(z_1,\ldots,z_{\arityof{\bpred_0}}) * \phi'' \step{\asid}
(\exists y_1 \ldots \exists y_m ~.~ \phi')[x_1/z_1, \ldots,
  x_{\arityof{\bpred_0}}/z_{\arityof{\bpred_0}}] * \phi''
\]
where $\phi''$ is a predicate-free formula, possibly containing
existential quantifiers. Since every complete $\asid$-unfolding of
$\apred$ yields a satisfiable formula, there exists a store $\store''$
that agrees with $\store'\circ[z_1/x_1, \ldots,
  z_{\arityof{\bpred_0}}/x_{\arityof{\bpred_0}}]$ over $y_1, \ldots,
y_m$ and a structure $\astruc''$, such that:
\[
\astruc'' \models^{\store''} (\exists y_1 \ldots \exists y_m ~.~
\phi')[x_1/z_1, \ldots, x_{\arityof{\bpred_0}}/z_{\arityof{\bpred_0}}]
* \phi''
\]
Note that, in the above construction, we have taken $\store''$ to
agree with $\store'$ over $y_1, \ldots, y_m$. This is possible because
there are no (dis-)equalities in $\asid$ and the set of models of a
qpf formula is closed under isomorphism-preserving renaming of
elements.

Let $\astruc'$ and $\astruc'''$ be structures such that
$\astruc''=\astruc' \comp \astruc'''$, $\astruc' \models^{\store''}
\phi'[ x_1 / z_1, \ldots, x_{\arityof{\bpred_0}} /
  z_{\arityof{\bpred_0}} ]$, or equivalently $\astruc'
\models^{\store'} \phi'$, and $\astruc''' \models^{\store''}
\phi''$. By induction on the length of the $\asid$-unfolding, relying
on by point (\ref{it1:mcsid}) of the definition of $\csid$, one can
prove that $\phi' = \phi * \overline{\phi}$, where $\overline{\phi}$
is a qpf formula, such that $\fv{\phi} \cap \fv{\overline{\phi}} =
\emptyset$. Since $\astruc' \models^{\store'} \phi'$, there exists a
structure $\overline{\astruc}=(\univ,\overline{\struc})$, such that
$\astruc'=\astruc \comp \overline{\astruc}$. Moreover, since $\store'$
is injective over $y_1, \ldots, y_m$, by construction, we obtain
$\supp{\struc} \cap \supp{\overline{\struc}} = \emptyset$. Let $I
\subseteq J_0$ be an equivalence class of $\xi_0$. By
\autoref{fact:csid-mc}, we have $\reachof{\astruc}{\store}{I}
\mcsubstruc \astruc$. Since $\astruc'=\astruc \comp
\overline{\astruc}$ and $\supp{\struc} \cap \supp{\overline{\struc}} =
\emptyset$, we obtain $\reachof{\astruc}{\store}{I} \mcsubstruc
\astruc'$.
\end{proof}

\skipnoindent The proof is completed as follows. Let $\astruc \in
\csem{\ppred}{\csid}$ be a canonical $\csid$-model of $\ppred$, i.e.,
there exists a complete $\csid$-unfolding $\ppred \step{\csid}^*
\exists y_1 \ldots \exists y_n ~.~ \phi$, where $\phi$ is a qpf
formula, and a store $\store$ injective over $y_1, \ldots, y_n$ such
that $\astruc \models^\store \phi$.  By the definition of $\csid$, the
first rule of this unfolding is of the form (\ref{eq:mcsid-begin}),
with a qpf formula $\psi_0$. Then there exist $\csid$-unfoldings
$\bpred_i(x_{1}, \ldots, x_{\arityof{\bpred_i}})_{/\xi_i}
\step{\csid}^* \exists y_{j_{i,1}} \ldots \exists y_{j_{i,m_i}} ~.~
\phi_i$, for some sets $J_i \subseteq \interv{1}{\arityof{\bpred_i}}$
and equivalence relations $\xi_i \subseteq J_i \times J_i$, for $i \in
\interv{1}{\ell}$, such that:
\[
\phi = \psi_0 * \Bigstar\nolimits_{i\in\interv{1}{\ell},J_i\neq\emptyset} \phi_i\theta_i
\]
where $\theta_i \isdef [x_1/z_{i,1}, \ldots,
  x_{\arityof{\bpred_i}}/z_{i,\arityof{\bpred_i}}]$, $i \in
\interv{1}{\ell}$. Let $\set{i_1, \ldots, i_k} \isdef \set{i \in
  \interv{1}{\ell} \mid J_i \neq \emptyset}$.  Then, there exist
structures $\astruc_0 \comp \ldots \comp \astruc_k = \astruc$, such
that $\astruc_0 \models^\store \psi_0$ and $\astruc_j \models^\store
\phi_{i_j}\theta_{i_j}$, for $i \in \interv{1}{\ell}$. By the
definition of $\csid$, there exists a complete $\asid$-unfolding:
\begin{align*}
  \bpred_0(x_1,\ldots,x_{\arityof{\bpred_0}})
  \step{\asid} & ~\exists y_1 \ldots \exists y_m ~.~
  \psi'_0 * \Bigstar\nolimits_{i=1}^\ell \bpred_i(z_{i,1},\ldots,z_{i,\arityof{\bpred_i}}) 
  \step{\asid} \ldots \\
  \step{\asid}^* & ~\exists y_1 \ldots \exists y_p ~.~
  \psi'_0 * \Bigstar\nolimits_{j=1}^k \phi'_{i_j}\theta_{i_j} ~*~
  \eta
\end{align*}
for qpf formul{\ae} $\psi'_0, \phi'_{i_1}, \ldots, \phi'_{i_k}$ and
predicate-free formula $\eta$. Consider the equivalence relation $\Xi$
over $\set{x_1,\ldots,x_{\arityof{\bpred_0}}} \cup \set{y_1, \ldots,
  y_m}$ defined as:
\[
\Xi \isdef \big(\connof{\psi_0} \cup \bigcup\nolimits_{j=1}^k
\xi_{i_j}(\bpred_{i_j}(z_{i_j,1},\ldots,z_{i_j,\arityof{\bpred_{i_j}}}))
\big)^=
\]
By point (\ref{it8:mcsid}), $\Xi$ has a single equivalence
class $X$ such that: \begin{itemize}[label=$\triangleright$]
\item $X \cap \set{x_{1},\ldots,x_{\arityof{\bpred_0}}} = \emptyset$, 
\item the sets $I_j \isdef \set{h \in
  \interv{1}{\arityof{\bpred_{i_j}}} \mid z_{i_j,h} \in X}$ are unions
  of equivalence classes of $\xi_{i_j}$, namely $I_j = I_{j,1} \uplus
  \ldots \uplus I_{j,q_j}$, where $I_{j,h}$ are equivalence classes of
  $\xi_i$, for all $j \in \interv{1}{k}$.
\end{itemize}
By \autoref{fact:csid-asid}, there exist a store $\store'$, that is
injective over $y_1, \ldots, y_p$ and agrees with $\store$ over $y_1,
\ldots, y_n$, and structures $\astruc'_{1}, \ldots, \astruc'_{k}$,
such that $\astruc'_{j} \models^{\store'} \phi_{j}$ and
$\reachof{\astruc_j}{\store'}{I_{j,h}} \mcsubstruc \astruc'_j$, for
all $j \in \interv{1}{k}$ and $h \in \interv{1}{q_j}$. We argue that
$\astruc_j = \Comp_{h=1}^{q_j} \reachof{\astruc_j}{\store'}{I_{j,h}}
\mcsubstruc \astruc'_j$. Moreover, since $\store'$ is injective over
$y_1, \ldots, y_n$, one can build a structure $\astruc'_0$, such that
$\astruc'_0 \models^{\store'} \psi'_0$. We define $\astruc' \isdef
\astruc'_0 \comp \Comp_{j=1}^k \astruc'_j$. Thus, we have $\astruc
\mcsubstruc \astruc'$ and we are left with showing that $\astruc'$ can
be embedded in a canonical $\asid$-model of $\apred$.

By the assumption that each predicate defined by $\asid$ occurs on
some complete $\asid$-unfolding of $\apred$, there exists another
complete $\asid$-unfolding:
\begin{align*}
  \apred \step{\asid} & \ldots \step{\asid}
  \bpred_0(z_{0,1},\ldots,z_{0,\arityof{\bpred_0}}) * \zeta \\
  \step{\asid} & ~(\exists y_1 \ldots \exists y_m ~.~
  \psi_0 * \Bigstar\nolimits_{i=1}^\ell \bpred_i(z_{i,1},\ldots,z_{i,\arityof{\bpred_i}}))
      [x_1/z_{0,1},\ldots,x_{\arityof{\bpred_0}}/z_{0,\arityof{\bpred_0}}] * \zeta \\
  \step{\asid} & \ldots \step{\asid} ~(\exists y_1 \ldots \exists y_n ~.~
  \psi'_0 * \Bigstar\nolimits_{i=1}^k \phi'_i\theta_i ~*~ \eta)
       [x_1/z_{0,1},\ldots,x_{\arityof{\bpred_0}}/z_{0,\arityof{\bpred_0}}] * \zeta
\end{align*}
for some predicate-free formula $\zeta\isdef \exists y_{n+1} \ldots
\exists y_p ~.~ \eta$, for some variables $y_{n+1}, \ldots, y_p$, such
that $\set{y_{n+1}, \ldots, y_p} \cap \set{y_1,\ldots,y_n} =
\emptyset$ and a qpf formula $\eta$. Since this latter
$\asid$-unfolding yields a satisfiable formula, there exists a
structure $\astruc''$ and a store $\store''$, injective over $y_1,
\ldots, y_p$, that agrees with $\store'$ over $y_1,\ldots,y_n$, such
that $\astruc'' \models^{\store''} \eta$. Then, $\astruc' \comp
\astruc'' \in \csem{\apred}{\asid}$ and, since $\astruc \mcsubstruc
\astruc'$, we obtain $\astruc \mcsubstruc \astruc' \comp \astruc''$,
leading to $\astruc \in \funsplit{\csem{\apred}{\asid}}$. \qed

\subsection{Proof of \autoref{lemma:sid-k-multiset-abstraction}}
\label{app:sid-k-multiset-abstraction}

Without loss of generality, we can consider that $\csid$ is
equality-free (\autoref{lemma:eq-free}) and all-satisfiable for
$\ppred$ (\autoref{lemma:all-sat}).

\skipnoindent ''$\kmcolabs{k}{(\csem{\ppred}{\csid})} \subseteq
\projrel{\kabssem{k}{\ppred}{\csid}}{3}$'' We prove the following,
more general, property:

\begin{quote}
  \emph{Let $\bpred_0(x_1,\ldots,x_{\arityof{\bpred_0}})
  \step{\csid}^* \exists y_1 \ldots \exists y_n ~.~ \phi$ be a
  complete $\csid$-unfolding such that $\phi$ is a qpf formula,
  $\store$ be a store injective over
  $\set{x_1,\ldots,x_{\arityof{\bpred_0}}} \cup \set{y_1,\ldots,y_n}$,
  $\astruc=(\univ,\struc)$ be a structure such that $\astruc
  \models^\store \phi$ and $D \subseteq \supp{\struc} \setminus
  \set{\store(x_1), \ldots, \store(x_{\arityof{\bpred_0}})}$ be a set
  such that $\cardof{D}\le k$. Then there exists
  $\tuple{\set{x_1,\ldots,x_{\arityof{\bpred_0}}},c,M} \in
  \kabssem{k}{\bpred_0}{\csid}$ such that
  $\funcol{\astruc}(\store(x_i)) = c(x_i)$, for all $i \in
  \interv{1}{\arityof{\bpred_0}}$ and $M=\mset{\funcol{\astruc}(u)
    \mid u \in D}$.}
\end{quote}

\noindent The proof is by induction on the length of the complete
$\csid$-unfolding. Assume w.l.o.g. that the first rule applied in the
unfolding is of the form (\ref{eq:asidrule}), with a qpf formula
$\psi_0$. Then, there exist structures $\astruc_0=(\univ_0,\struc_0),
\ldots, \astruc_\ell=(\univ_\ell,\struc_\ell)$, such
that: \begin{itemize}[label=$\triangleright$]
  \item $\astruc=\astruc_0 \comp \ldots \comp \astruc_\ell$,
  \item $\astruc_0 \models^\store \psi_0$, 
  \item there exists a complete $\csid$-unfolding $\bpred_i(z_{i,1},
    \ldots, z_{i,\arityof{\bpred_i}}) \step{\csid}^* \exists
    y_{j_{i,1}} \ldots \exists y_{j_{i,k_i}} ~.~ \phi_i$, where
    $j_{i,1}, \ldots, j_{i,k_i} \in \interv{1}{n}$ and $\phi_i$ is a
    subformula of $\phi$, such that $\astruc_i \models^\store \phi_i$,
    and the indices $j_{i,m}$ are pairwise distinct, for all $m \in
    \interv{1}{k_i}$, for all $i \in \interv{1}{\ell}$.
\end{itemize}
Let $\store_i$ be the store such that $\store_i(x_j)=\store(z_{i,j})$
for all $j \in \interv{1}{\arityof{\bpred_i}}$ and $\store_i$ agrees
with $\store$ everywhere else, for all $i \in \interv{1}{\ell}$. Then,
there exists a complete $\csid$-unfolding:
\[
 \bpred_i(x_1,\ldots,x_{\arityof{\bpred_i}}) \step{\csid}^* (\exists
 y_{j_{i,1}} \ldots \exists y_{j_{i,k_i}} ~.~ \phi_i)[z_{i,1}/x_1, \ldots,
   z_{i,\arityof{\bpred_i}}/x_{\arityof{\bpred_i}}] = \exists y_{j_{i,1}}
 \ldots \exists y_{j_{i,k_i}} ~.~ \psi_i
\]
such that $\astruc_i \models^{\store_i} \psi_i$, where $\psi_i$ is a
qpf formula, for all $i \in \interv{1}{\ell}$. We define the sets:
\[ \begin{array}{rcl}
  D_0 & \isdef & D \cap \big(\supp{\struc_0} \cup
  \bigcup\nolimits_{i=1}^\ell \set{\store(z_{i,1}), \ldots, \store(z_{i,\arityof{\bpred_i}})}\big) \\
  D_i & \isdef & (\supp{\struc_i} \cap D) \setminus
  \set{\store(z_{i,1}), \ldots, \store(z_{i,\arityof{\bpred_i}})}
  \text{, for each} i \in\interv{1}{\ell}
\end{array} \]
and prove the following fact:
\begin{fact}
  $D = D_0 \uplus D_1 \uplus \ldots \uplus D_\ell$
\end{fact}
\begin{proof} The sets $D_0, \ldots, D_\ell$ are pairwise disjoint, since:
  \[ \begin{array}{rcl}    
  D_0 & \subseteq & \supp{\struc_0} \cup \bigcup\nolimits_{i=1}^\ell
  \set{\store(z_{i,1}), \ldots, \store(z_{i,\arityof{\bpred_i}})} \\
  D_i & \subseteq & \supp{\struc_i} \setminus \set{\store(z_{i,1}), \ldots,
    \store(z_{i,\arityof{\bpred_i}})} \text{, for all } i \in \interv{1}{\ell}
  \end{array} \]
  and, moreover for all $1 \le i < j \le \ell$:
  \begin{itemize}[label=$\triangleright$]
      \item \(\big(\supp{\struc_0} \cup \bigcup\nolimits_{i=1}^\ell
        \set{\store(z_{i,1}), \ldots,
          \store(z_{i,\arityof{\bpred_i}})}\big) \cap \supp{\struc_i}
        \subseteq \set{\store(z_{i,1}), \ldots,
          \store(z_{i,\arityof{\bpred_i}})} \)
    \item
      \( \supp{\struc_i} \cap \supp{\struc_j} \subseteq
      \set{\store(z_{i,1}), \ldots, \store(z_{i,\arityof{\bpred_i}})} \cap
      \set{\store(z_{j,1}), \ldots, \store(z_{j,\arityof{\bpred_i}})} \)
  \end{itemize}
  because $\store$ is injective over $y_1, \ldots,
  y_n$. ``$\supseteq$'' We have $D_0 \uplus D_1 \uplus \ldots \uplus
  D_\ell \subseteq D$ because $D_i \subseteq D$, for all $i \in
  \interv{0}{\ell}$. ``$\subseteq$'' Let $u \in D$ be an element. By
  the choice of $D$, we have $u \in \supp{\struc} \setminus
  \set{\store(x_1), \ldots, \store(x_{\arityof{\bpred}})}$. Since
  $\astruc=\astruc_0 \comp \ldots \comp \astruc_\ell$, we have
  $\supp{\struc}=\bigcup\nolimits_{i=0}^\ell \supp{\struc_i}$, hence
  $u \in \supp{\struc_i}$, for some $i \in \interv{0}{\ell}$.  If $u
  \in D_0$ we are done. Otherwise, $u \not\in D_0$, hence $u \not\in
  \supp{\struc_0}$ and $u \in \supp{\struc_i}$, for some $i
  \in\interv{1}{\ell}$. Moreover, $u \not\in \set{\store(z_{i,1}),
    \ldots, \store(z_{i,\arityof{\bpred_i}})}$, for all $j \in
  \interv{1}{\ell}$, hence $u \in D_i$.
\end{proof}

\skipnoindent Back to the proof, since $\cardof{D_i} \le \cardof{D}
\le k$, for all $i \in \interv{1}{\ell}$, by the inductive hypothesis,
there exist $\tuple{\set{x_1,\ldots,x_{\arityof{\bpred_i}}},c_i,M_{i}}
\in \kabssem{k}{\bpred_i}{\csid}$, for $i \in \interv{1}{\ell}$, such
that: \begin{itemize}[label=$\triangleright$]
\item $\funcol{\astruc}(\store_i(x_j)) = c(x_j)$, for all $j \in
  \interv{1}{\arityof{\bpred_i}}$,
\item $M_{i} = \mset{\funcol{\astruc_i}(u) \mid u \in D_i}$.
\end{itemize}
Let $\tuple{\fv{\psi_0},c_0,\emptyset}\isdef\colorof{\psi_0}$ be a
color triple. Since $\astruc_0 \models^\store \psi_0$, we have
$\funcol{\astruc_0}(\store(x_j)) = c(x_j)$, for all $j \in
\interv{1}{\arityof{\bpred_0}}$. By definition, there exists a
constraint of the form (\ref{eq:asidconstraint}) for the above rule
(\ref{eq:asidrule}). We prove that the $\kcomp$-composition from the
right-hand side of the constraint is defined. Suppose, for a
contradiction, that $c_i(x) \cap c_j(x) \neq \emptyset$, for some $x
\in \set{z_{i,1},\ldots,z_{i,\arityof{\bpred_i}}} \cap
\set{z_{j,1},\ldots,z_{j,\arityof{\bpred_j}}}$ and $1 \le i < j \le
\ell$. Then $\funcol{\astruc_i}(x) \cap \funcol{\astruc_j}(x) \neq
\emptyset$, contradicting the fact that $\astruc_i \comp \astruc_j$ is
defined. The same reasoning applies if $c_0(x) \cap c_i(x) \neq
\emptyset$, for some $x \in \fv{\psi} \cap
\set{z_{i,1},\ldots,z_{i,\arityof{\bpred_i}}}$ and $i \in
\interv{1}{\ell}$. Then, there exists a color triple:
\[
  \tuple{X',c',M'}\in
  \tuple{\fv{\psi_0},c_0,\emptyset}
  \kcomp \Kcomp\nolimits_{i\in\interv{1}{\ell}}
  \tuple{\set{x_1,\ldots,x_{\arityof{\bpred_i}}},c_i,M_i}[x_1/z_{i,1},\ldots,x_{\arityof{\bpred_i}}/z_{i,\arityof{\bpred_i}}]
\]
W.l.o.g. we can chose the tuple such that $M' = M_1 \cup \ldots \cup
M_\ell$. This choice is possible since $\cardof{M'} = \sum_{i=1}^\ell
\cardof{M_i} = \sum_{i=1}^\ell \cardof{D_i} \le \cardof{D} \le k$. Let
$\tuple{X,c,M}\in\kproj{\tuple{X',c',M'}}{\set{x_1,\ldots,x_{\arityof{\bpred_0}}}}$
be such that $M = M' \cup \mset{\funcol{\astruc_0}(u) \mid u \in
  D_0}$. This choice is possible, since $\cardof{M} \le
\sum_{i=0}^\ell \cardof{D_i} \le k$. We prove the points of the
statement: \begin{itemize}[label=$\triangleright$]
\item Let $i \in \interv{1}{\arityof{\bpred_0}}$ be an index. By the
  definition of the $\kcomp$-composition, we have:
  \[
   \funcol{\astruc}(\store(x_i)) = \funcol{\astruc_0 \comp \ldots
     \comp \astruc_\ell}(\store(x_i)) = \bigcup\nolimits_{j=0}^\ell
   \funcol{\astruc_j}(\store(x_i)) = \biguplus\nolimits_{j=0}^\ell c_j(x_i)
    = c'(x_i) = c(x_i)
  \]
\item $M = \bigcup\nolimits_{i=0}^\ell \mset{\funcol{\astruc_i}(u) \mid u \in
  D_i} = \mset{\funcol{\astruc_0 \comp \ldots \comp \astruc_\ell}(u)
  \mid u \in \biguplus\nolimits_{i=0}^\ell D_i} = \mset{\funcol{\astruc}(u) \mid
  u \in D}$
\end{itemize}

\skipnoindent ''$\projrel{\kabssem{k}{\ppred}{\csid}}{3} \subseteq
\kmcolabs{k}{(\csem{\ppred}{\csid})}$'' We prove the following, more
general, property:

\begin{quote}
  \emph{Let $\tuple{\set{x_1, \ldots, x_{\arityof{\bpred_0}}},c,M} \in
  \kabssem{k}{\bpred_0}{\csid}$ be a color triple. Then there exists a
  complete $\csid$-unfolding
  $\bpred_0(x_1,\ldots,x_{\arityof{\bpred_0}}) \step{\csid}^* \exists
  y_1 \ldots \exists y_n ~.~ \phi$, whose steps belong to a complete
  $\csid$-unfolding of $\ppred$, such that $\phi$ is a qpf formula, a
  store $\store$ injective over $\set{x_1, \ldots,
    x_{\arityof{\bpred_0}}} \cup \set{y_1, \ldots, y_n}$, a structure
  $\astruc = (\univ,\struc)$ such that $\astruc \models^\store \phi$
  and $D \subseteq \supp{\struc} \setminus \set{\store(x_1), \ldots,
    \store(x_{\arityof{\bpred_0}})}$, $\cardof{D} \le k$, such that
  $\funcol{\astruc}(\store(x_i))=c(x_i)$, for all $i \in
  \interv{1}{\arityof{\bpred_0}}$ and $M = \mset{\funcol{\astruc}(u)
    \mid u \in D}$.}
\end{quote}

\noindent The proof is by induction on the length of the finite
fixpoint iteration that produced $\tuple{\set{x_1, \ldots,
    x_{\arityof{\bpred_0}}},c,M}$. Assume that the last step of the
iteration corresponds to a constraint of the form
(\ref{eq:asidconstraint}), with a qpf formula $\psi_0$.  By
definition, there exists a rule of the form (\ref{eq:asidrule}) in
$\csid$, with the same qpf formula $\psi_0$. Then $\psi_0$ is
satisfiable, because each $\csid$-unfolding of $\ppred$ yields a
satisfiable formula. Then there exists a color triple:
\[
  \tuple{X',c',M'} \in \colorof{\psi_0} \kcomp
  \Kcomp\nolimits_{i\in\interv{1}{\ell}} \kabssem{k}{\bpred_i}{\csid}
        [x_1/z_{i,1},\ldots,x_{\arityof{\bpred_i}}/z_{i,\arityof{\bpred_i}}]
\]
such that
\[ \begin{array}{rcl}
  \tuple{\set{x_1, \ldots, x_{\arityof{\bpred_0}}},c,M} & \in &
  \kproj{\tuple{X',c',M'}}{\set{x_1, \ldots, x_{\arityof{\bpred_0}}}} \\
  M & \subseteq & M' \cup \set{c'(x) \mid x \in X'\setminus\set{x_1, \ldots, x_{\arityof{\bpred_0}}}}
  \end{array} \]
Then, there exist
\[ \begin{array}{rcl}
  \tuple{\fv{\psi_0}, c_0, \emptyset} & \isdef & \colorof{\psi_0} \\
  \tuple{\set{z_{i,1},\ldots,z_{i,\arityof{\bpred_i}}},c'_i,M_i} & \in &
\kabssem{k}{\bpred_i}{\csid}[x_1/z_{i,1},\ldots,x_{\arityof{\bpred_i}}/z_{i,\arityof{\bpred_i}}]
\text{, for all } i \in \interv{1}{\ell} \end{array} \]
such that
\[\tuple{X',c',M'} \in
\tuple{\fv{\psi_0}, c_0, \emptyset} \kcomp
\Kcomp\nolimits_{i\in\interv{1}{\ell}}
\tuple{\set{z_{i,1}, \ldots, z_{i,\arityof{\bpred_i}}},c'_i,M_i}\]
Hence, there exist
$\tuple{\set{x_1, \ldots, x_{\arityof{\bpred_i}}},c_i,M_i}\in\kabssem{k}{\bpred_i}{\csid}$
such that
$c'_i=c_i\circ[x_1/z_{i,1},\ldots,x_{\arityof{\bpred_i}}/z_{i,\arityof{\bpred_i}}]$,
for all $i \in \interv{1}{\ell}$. By the inductive hypothesis, for all
$i \in \interv{1}{\ell}$, there
exist: \begin{itemize}[label=$\triangleright$]
\item a complete unfolding $\bpred_i(x_1, \ldots,
  x_{\arityof{\bpred_i}}) \step{\csid}^* \exists y_{j_{i,1}} \ldots
  \exists y_{j_{i,k_i}} ~.~ \psi_i$ such that $\psi_i$ is a qpf
  formula. By applying an $\alpha$-renaming, if necessary, we assume
  w.l.o.g. that the variables $y_{j_{1,1}} \ldots y_{j_{\ell,k_\ell}}$
  are pairwise distinct and, moreover, distinct from $x_1, \ldots,
  x_{\arityof{\bpred_0}}$.
\item a store $\store_i$ that is injective over $\set{x_1, \ldots,
  x_{\arityof{\bpred_i}}} \cup \set{y_{j_{i,1}}, \ldots,
  y_{j_{i,k_i}}}$. We assume w.l.o.g. that
  $\store_i(x_j)=\store_k(x_m)$ iff $z_{i,j}$ and $z_{k,m}$ are the
  same variable in the rule (\ref{eq:asidrule}), for all $1 \le i < k
  \le \ell$, $j \in \interv{1}{\arityof{\bpred_i}}$ and $m \in
  \interv{1}{\arityof{\bpred_k}}$. Note that this assumption does not
  contradict the fact that $\store_i$ is injective over $\set{x_1,
    \ldots, x_{\arityof{\bpred_i}}} \cup \set{y_{j_{i,1}}, \ldots,
    y_{j_{i,k_i}}}$.
\item a structure $\astruc_i = (\univ_i,\struc_i)$ such that
  $\astruc_i \models^{\store_i} \psi_i$. We assume w.l.o.g. that
  $\supp{\struc_i} \cap \supp{\struc_j} \subseteq \set{\store_i(x_1),
  \ldots, \store_i(x_{\arityof{\bpred_i}})} \cap \set{\store_j(x_1),
  \ldots, \store_j(x_{\arityof{\bpred_j}})}$. Note that this is
  possible by the assumption that $\csid$ is equality-free.
\item a set $D_i \subseteq \supp{\struc_i} \setminus
  \set{\store_i(x_1), \ldots, \store_i(x_{\arityof{\bpred_i}})}$, such
  that $\cardof{D_i} \le k$,
  $\funcol{\astruc_i}(\store_i(x_j))=c_i(x_j)$, for all $j \in
  \interv{1}{\arityof{\bpred_i}}$ and $M_i =
  \mset{\funcol{\astruc_i}(u) \mid u \in D_i}$.
\end{itemize}
We prove the points of the statement. Let $\theta_i$ be the
substitution $[x_1/z_{i,1}, \ldots,
  x_{\arityof{\bpred_i}}/z_{i,\arityof{\bpred_i}}]$, for each $i \in
\interv{1}{\ell}$, where $\bpred_i(z_{i,1}, \ldots,
z_{i,\arityof{\bpred_i}})$ is a predicate atom that occurs on the
right-hand side of the rule (\ref{eq:asidrule}). A complete
$\csid$-unfolding $\bpred_0(x_1,\ldots,x_{\arityof{\bpred_0}})
\step{\csid}^* \exists y_1 \ldots \exists y_n ~.~ \psi$ is built from
the rule (\ref{eq:asidrule}) above, with qpf formula $\psi_0$,
followed by $\bpred_i(x_1, \ldots, x_{\arityof{\bpred_i}})\theta_i
\step{\csid}^* \exists y_{j_{i,1}} \ldots \exists y_{j_{i,k_i}} ~.~
\psi_i\theta_i$, for all $i \in \interv{1}{\ell}$. Hence $\psi =
\psi_0 * \Bigstar\nolimits_{i=1}^\ell \psi_i\theta_i$ modulo a
reordering of atoms. Let $\store'_i\isdef\store_i \circ \theta_i$ and
define the store $\store$ as
follows: \begin{itemize}[label=$\triangleright$]
\item $\store(z)\isdef\store'_i(z)$, for each each $z \in
  \fv{\psi_i\theta_i}$,
\item $\store(z) \not\in \bigcup\nolimits_{i=1}^\ell
  \store'_i(\fv{\psi_i\theta_i})$, for each variable $z \in
  \fv{\psi_0} \setminus \bigcup\nolimits_{i=1}^\ell
  \fv{\psi_i\theta_i}$ such that, moreover, $\store$ is injective over
  $\fv{\psi_0}$. Note that this is possible because we assumed $\csid$
  to be equality-free.
\end{itemize}
Then, we consider a structure $\astruc_0=(\univ_0,\struc_0)$ such
that: \begin{itemize}[label=$\triangleright$]
\item $\astruc_0 \models^\store \psi_0$, and
\item $\supp{\struc_0}\cap\supp{\struc_i} \subseteq
  \store(\fv{\psi_0}) \cap \store(\fv{\psi_i\theta_i})$, for all $i
  \in \interv{1}{\ell}$.
\end{itemize}
Since $\psi_0$ is satisfiable, such a structure exists and we can
consider w.l.o.g. that it satisfies the above conditions, because
$\csid$ is equality-free. It is easy to check that the structures
$\astruc_0, \ldots, \astruc_\ell$ are pairwise locally disjoint, hence
$\astruc=(\univ,\struc)\isdef\astruc_0 \comp \ldots \comp
\astruc_\ell$ is defined. Moreover, we have $\astruc\models^\store
\psi$, because $\psi=\psi_0 * \Bigstar\nolimits_{i=1}^\ell
\psi_i\theta_i$, $\astruc_0 \models^\store \psi_0$ and $\astruc_i
\models^\store \psi_i\theta_i$, for all $i \in
\interv{1}{\ell}$. Further, for all $j \in
\interv{1}{\arityof{\bpred_0}}$, we have:
\[
    \funcol{\astruc}(\store(x_j)) =
    \funcol{\astruc_0 \comp \ldots \comp \astruc_\ell}(\store(x_i)) =
    \biguplus\nolimits_{i=0}^\ell c_i(x_j) = c'(x_j) = c(x_j)
\]
We consider the set $D \isdef \set{u \in \supp{\struc} \mid
  \funcol{\astruc}(u) \in M}$. Suppose, for a contradiction, that
$\store(x_i)\in D$, for some $i \in \interv{1}{\arityof{\bpred_0}}$.
Then $\funcol{\astruc}(\store(x_i))\in M$, hence $c(x_i)\in M$. Since
$M \subseteq M' \cup \set{c'(x) \mid x \in X' \setminus
  \set{x_1,\ldots,x_{\arityof{\bpred_0}}}}$, we must have $c(x_i) \in
M' \subseteq \bigcup\nolimits_{j=1}^\ell M_j$ and let $j \in
\interv{1}{\ell}$ be such that $c(x_i) \in
M_j=\mset{\funcol{\astruc_j} \mid u \in D_j}$, by the inductive
hypothesis. Then there exists $k \in \interv{1}{\arityof{\bpred_j}}$
such that $c(x_i)=c_j(x_k)=\funcol{\astruc_i}(\store_j(x_k))$, thus
$\store_j(x_k) \in D_j \subseteq \supp{\struc_j} \setminus
\set{\store_j(x_1), \ldots, \store_j(x_{\arityof{\bpred_j}})}$,
contradiction. We obtained $D \subseteq \supp{\struc} \setminus
\set{\store(x_1), \ldots, \store(x_{\arityof{\bpred_0}})}$ and are
left with proving that $M = \mset{\funcol{\astruc}(u) \mid u \in
  D}$. ``$\supseteq$'' Immediate, by the definition of
$D$. ``$\subseteq$'' Let $C \in M$ be a color. Then either one of the
following holds: \begin{itemize}[label=$\triangleright$]
\item $C = \set{\arel \in \relations \mid \arel(z,\ldots,z) \text{
    occurs in } \psi_0} \cup \bigcup\nolimits_{i=1}^\ell c_i(z)$, for
  some $z \in \fv{\psi_0}$: in this case, $C =
  \funcol{\astruc}(\store(z))$ and $\store(z) \in \supp{\struc}$,
  hence $\store(z)\in D$.
\item $C \in M_i$, for some $i \in \interv{1}{\ell}$: in this case,
  $C=\funcol{\astruc}(u)=\funcol{\astruc_i}(u)$, for some $u \in D_i$,
  by the inductive hypothesis. Then $u \in \supp{\struc_i} \subseteq
  \supp{\struc}$, hence $u \in D$. \qed
\end{itemize}


\section{Proofs from \autoref{sec:general-tb}}
\subsection{Proof of \autoref{lemma:one-transitions}}
\label{app:one-transitions}

Assume w.l.o.g. that $\mathcal{A}$ is rooted and let
$\graphof{\mathcal{A}} = (\nodes,\edges,S_0)$ be the SCC graph of
$\mathcal{A}$. By \autoref{def:choice-free},
$\graphof{\mathcal{A}}$ is a tree and, moreover, $S_0 =
\set{\initstate}$, because $\mathcal{A}$ is rooted. Let $\Lambda :
\nodes\cup\trans\rightarrow\set{1,\infty}$ be the labeling from
\autoref{def:choice-free}. For every SCC $S \in
\nodes\setminus\set{S_0}$, let $\entryof{S}$ be the unique state $q$
such that $\set{q} = \post{\tau}\cap S$, where $\set{\tau}=\pre{S}$,
by point (\ref{it1:def:choice-free}) of
\autoref{def:choice-free}, and $\entryof{S_0}\isdef
\initstate$.  Moreover, each linear SCC $S\in\nodes$ such that
$\Lambda(S)=1$ has a unique transition $\tau$, such that
$\post{S}=\set{\tau}$, by point (\ref{it21:def:choice-free}) of
\autoref{def:choice-free}. We prove first an invariant of
$1$-labeled linear SCCs:

\begin{fact}\label{fact:one-inv}
  Let $p \in \dom{\arun}$ be a position, such that $\arun(p)\in S$,
  for a linear SCC $S \in \nodes$, such that $\Lambda(S)=1$. Then
  there exists a descendant $p'\in\dom{\arun}$ of $p$, such that
  $\arun(p')=s_0$, $t(p')=\beta$ and $\arun(p'i)=s_i$, for all
  $i\in\interv{1}{k}$, where $\post{S}=\set{s_0 \arrow{\beta}{}
    (s_1,\ldots,s_k)}$.
\end{fact}
\begin{proof}
  Suppose, for a contradiction, that $s_0 \arrow{\beta}{}
  (s_1,\ldots,s_k)$ never occurs below $p$ in $\arun$. Then every
  transition that occurs at some position below $p$ in $\arun$ must be
  from $\prepost{S}$. This, however, cannot be the case for a
  transition $\arun(p') \arrow{t(p')}{} ()$, such that
  $p'\in\frof{\arun}$. Since, moreover, $\arun$ is an accepting run,
  such a transition must occur on the frontier of $\arun$.
\end{proof}

The following facts prove the existence and uniqueness of a position
labeled with the entry state of each $1$-labeled linear SCC:

\begin{fact}\label{fact:one-exists}
  For each SCC $S\in\nodes$, such that $\Lambda(S)=1$, there exists a
  position $p\in\dom{\arun}$, such that $\arun(p)=\entryof{S}$.
\end{fact}
\begin{proof}
  Because $\graphof{\mathcal{A}}$ is a tree with root $S_0$, we have
  that $S$ is reachable from $S_0$ in $\graphof{\mathcal{A}}$ by a
  path of pairs from $\edges$. The proof goes by induction on the
  length $n\geq0$ of this path. For the base case $n=0$ (i.e.,
  $S=S_0$) we take $p=\epsilon$. For the inductive step, let $S'$ be
  the parent of $S$ in $\graphof{\mathcal{A}}$. By points
  (\ref{it1:def:choice-free}) and (\ref{it23:def:choice-free}) of
  \autoref{def:choice-free}, $\pre{S}=\set{\tau}$ for some
  $\tau\in\post{S'}\cap\trans^1$, such that
  $\set{\entryof{S}}=\post{\tau}\cap S$. By the inductive hypothesis,
  there exists a position $p'\in\dom{\arun}$, such that
  $\arun(p')=\entryof{S'}$. By \autoref{fact:one-inv}, there exists a
  descendant $p$ of $p'$, such that $\theta(p)=\entryof{S}$.
\end{proof}

\begin{fact}\label{fact:one-unique}
  For each SCC $S\in\nodes$, such that $\Lambda(S)=1$, there exists at
  most one position $p\in\dom{\arun}$, such that
  $\arun(p)=\entryof{S}$.
\end{fact}
\begin{proof}
  Suppose, for a contradiction, that there exist two positions $p_1,
  p_2 \in \dom{\arun}$, such that
  $\arun(p_1)=\arun(p_2)=\entryof{S}$. By induction of the length of
  $p_i$, we prove the existence of a sequence
  $S_{i,k_i},\tau_{i,k_i},\ldots,S_{i,1},\tau_{i,1},S_{i,0}=S_0$ such
  that $\arun(p_i)\in S_{i,k_i}$, $\pre{S_{i,j}} = \tau_{i,j}$ and
  $\set{\tau_{i,j}} = \post{S_{i,j-1}}$, for all $j\in\interv{1}{k_i}$
  and $i=1,2$. Since $p_1 \neq p_2$, there exists an SCC $S_{1,j_1} =
  S_{2,j_2}$ that violates condition (\ref{it1:def:choice-free}) of
  \autoref{def:choice-free}.
\end{proof}

Let $\tau : q_0 \arrow{\alpha}{} (q_1,\ldots,q_\ell)$ be a transition,
such that $\Lambda(\tau)=1$. By point (\ref{it22:def:choice-free}) of
\autoref{def:choice-free}, we have $\tau\in\post{S}$ for some linear
SCC $S\in\nodes$, such that $\Lambda(S)=1$. By
\autoref{fact:one-exists} and \autoref{fact:one-unique}, there exists
a unique position $p\in\dom{\arun}$, such that
$\arun(p)=\entryof{S}$. By \autoref{fact:one-inv}, there exists a
position $p' \in \dom{\arun}$, such that $\arun(p')=q_0$,
$t(p')=\alpha$ and $\arun(p'i)=q_i$, for all $i \in
\interv{1}{\ell}$. Suppose, for a contradiction, that this position is
not unique, hence there exists another position $p'' \in \dom{\arun}$,
such that $\arun(p'')=q_0$, $t(p'')=\alpha$ and $\arun(p''i)=q_i$, for
all $i \in \interv{1}{\ell}$. Since $\arun(p')=\arun(p'')=q_0 \in S$,
there exists a transition $\tau'$ with $\cardof{\post{\tau'}\cap
  S}\geq2$, in contradiction with the fact that $S$ is linear. This
concludes the proof \qed.

\subsection{Proof of \autoref{lemma:choice-free-property}}
\label{app:choice-free-property}

Let us consider the SCC graph $\graphof{\mathcal{A}} =
(\nodes,\edges,S_0)$ and the mapping $\Lambda : \nodes \cup \trans
\rightarrow \set{1,\infty}$ with the properties stated in
\autoref{def:choice-free} and let $q\in\pre{(\trans^\infty)}$
be a state. W.l.o.g., we consider that
$\langof{}{\mathcal{A}}\neq\emptyset$ and that $\mathcal{A}$ is
trim. Then $q$ is reachable from $S_0 = \set{\initstate}$, i.e., there
exists a partial run $\arun_1$ on $\mathcal{A}$ and a position $p_1$
such that $\arun_1(\epsilon) = \initstate$ and $\arun_1(p_1) = q$. Let
$p_2$ be the longest strict prefix of $p_1$ such that the transition
$\tau: \arun_1(p_2) \arrow{a}{} \tuple{\arun_1(p_2 1), \ldots,
  \arun_1(p_2 \ell)}$ is in $\trans^1$ for some $a\in\alphabet$ and
index $\ell$. This position $p_2$ exists, by
\autoref{def:choice-free}, because $S_0 = \set{\initstate}$ is
linear, $\Lambda(S_0) = 1$ by condition (\ref{it23:def:choice-free}),
$\cardof{\post{S_0}}=1$ by condition (\ref{it21:def:choice-free}), and
the only transition $\tau_0 \in \post{\set{\initstate}}$ is in
$\trans^1$ by condition (\ref{it22:def:choice-free}). This shows that
$\initstate \notin \pre{\trans^\infty}$ hence $q \neq \initstate$ and
$\tau_0$ is a transition in $\trans^1$ on the path from $\initstate$
to $q$ in $\arun_1$, with $\tau$ being the last one.

We decompose $p_1 = p_2 r p_3$ for some index $r\in\interv{1}{\ell}$
and position $p_3$ and define the partial run $\arun_2$ as
$\arun_2(\epsilon) \isdef \arun_1(p_2 r)$ and, for each $u\in\nat^*$
and $i\in\nat$ such that $p_2 rui \in \dom{\arun_1}$ and $p_2 ru$ is a
strict prefix of $p_1$, by $\arun_2(ui) \isdef \arun_1(p_2 r u)$.
Then $\arun_2$ starts from the state $q_0 \isdef \arun_1(p_2 r) \in
\post{\tau}$ and $p_3 \in \frof{\arun_2}$ gives the state
$\arun_2(p_3) = q$. Let $S\in\nodes$ be the SCC in
$\graphof{\mathcal{A}}$ such that $q_0\in S$.  Then $\Lambda(\tau)=1$
and $\tau \in \pre{S}$ (hence $\pre{S} = \set{\tau}$ by condition
(\ref{it1:def:choice-free}) of \autoref{def:choice-free}), thus
$\Lambda(S)=1$ by condition (\ref{it23:def:choice-free}) of
\autoref{def:choice-free}.

We distinguish three cases (see \autoref{fig:choice-free-property}
for an illustration):
\begin{itemize}[label=$\triangleright$]
\item If $S$ is not linear, there exists a transition $\tau'
  \in\prepost{S}$ such that $\cardof{\post{\tau'} \cap S} \geq 2$.
  Let $q' \isdef \pre{\tau'}$ and $q''$, $q'''$ be the states such
  that $\mset{q'', q'''} \subseteq \post{\tau'} \cap S$. Since $q_0,
  q', q'', q'''\in S$, we can construct a partial run $\arun_3 \in
  \runsof{\infty}{q_0}{\mathcal{A}}$ with transitions taken from
  $\prepost{S}$, which reaches $q'$ from $q_0$, then applies $\tau'$
  and reaches $q_0$ from both $q''$ and $q'''$. This gives
  $\arun_3(p_4) = \arun_3(p_5) = q_0$, for two distinct positions
  $p_4, p_5\in \frof{\arun_3}$. We define $\arun_0 \in
  \runsof{\infty}{q_0}{\mathcal{A}}$ as the partial run with domain
  $\dom{\arun_3} \cup \set{p_5 u \mid u\in\dom{\arun_2}}$, that
  extends $\arun_3$ by $\arun_0(p_5 u) \isdef \arun_2(u)$, for all
  $u\in\dom{\arun_2}$. Then $\arun_0$ satisfies point
  \ref{it1:lemma:choice-free-property} of the lemma because
  $\arun_0(p_4) = q_0$ and $\arun_0(p_5 p_3) = q$, with $p_4 \neq p_5
  p_3 \in \frof{\arun_0}$.
\item If $S$ is linear and $q\notin S$, there exists a unique position
  $p_6$ and a transition: \[\tau' : \arun_2(p_6) \arrow{\alpha}{}
  \tuple{\arun_2(p_6 1), \ldots, \arun_2(p_6 k)} \in \trans^\infty\]
  for some alphabet symbol $\alpha\in\alphabet$ and some index
  $k\in\nat$, such that $\arun_2(p_6) \in S$.  Moreover, there exists
  an index $r \in \interv{1}{k}$ such that $\arun_2(p_6 r) \notin S$
  and $p_6 r$ is a prefix of $p_3$. Then $\Lambda(\tau') = \infty$
  and, by condition (\ref{it22:def:choice-free}) of
  \autoref{def:choice-free}, we have $\tau' \notin \post{S}$,
  hence $\tau' \in \prepost{S}$ and $q'\isdef \arun_2(p_6 r') \in S$
  for another index $r' \in \interv{1}{k} \setminus \set{r}$. Then
  there exists a partial run $\arun_4 \in
  \runsof{\infty}{q'}{\mathcal{A}}$ such that $\arun_4(p_7) = q_0$ for
  some position $p_7\in\frof{\arun_4}$. We define the partial run
  $\arun_0 \in \runsof{\infty}{q_0}{\mathcal{A}}$ with domain
  $\dom{\arun_2} \cup \set{p_6 r' u \mid u\in\dom{\arun_4}}$, by
  extending $\arun_2$ with $\arun_0(p_6 r' u) \isdef \arun_4(u)$, for
  all $u\in\dom{\arun_4}$. Then $\arun_0$ satisfies point
  \ref{it1:lemma:choice-free-property} of the lemma because
  $\arun_0(p_6 r' p_7) = q_0$ and $\arun_0(p_3) = q$, with $p_6 r' p_7
  \neq p_3 \in \frof{\arun_0}$.
\item If $S$ is linear and $q\in S$, let $\arun \in
  \runsof{\infty}{q}{\mathcal{A}}$ be a partial run. Then $\post{S}$
  contains only one transition in $\trans^1$, thus for every position
  $u\in\dom{\arun} \setminus \frof{\arun}$, such that $\arun(u)\in S$,
  the transition $\arun(u) \arrow{\alpha}{} \tuple{\arun(u 1), \ldots,
    \arun(u k)}$ belongs to $\prepost{S}$. Then, there exists an index
  $i \in \interv{1}{k}$ such that $\arun(ui) \in S$, and we can find a
  path in $\arun$ which stays in $S$ and reaches the frontier, that is
  $q'\isdef\arun(p_8) \in S$, for some $p_8 \in \frof{\arun}$. Hence,
  there exists a partial run $\arun_5 \in
  \runsof{\infty}{q'}{\mathcal{A}}$ such that $\arun_5(p_9) = q_0$,
  for some position $p_9\in\frof{\arun_5}$. We now can extend $\arun$
  to some partial run $\arun' \in \runsof{\infty}{q}{\mathcal{A}}$
  with domain $\dom{\arun} \cup \set{p_8 u \mid u\in\dom{\arun_5}}$,
  as $\arun'(p_8 u) \isdef \arun_5(u)$ for all
  $u\in\dom{\arun_5}$. The partial run $\arun'$ satisfies point
  \ref{it2:lemma:choice-free-property} of the lemma, because
  $\arun'(p_8 p_9) = q_0$, with $p_8 p_9 \in \frof{\arun'}$. \qed
\end{itemize}

  \begin{figure}[htbp]
    \begin{center}
    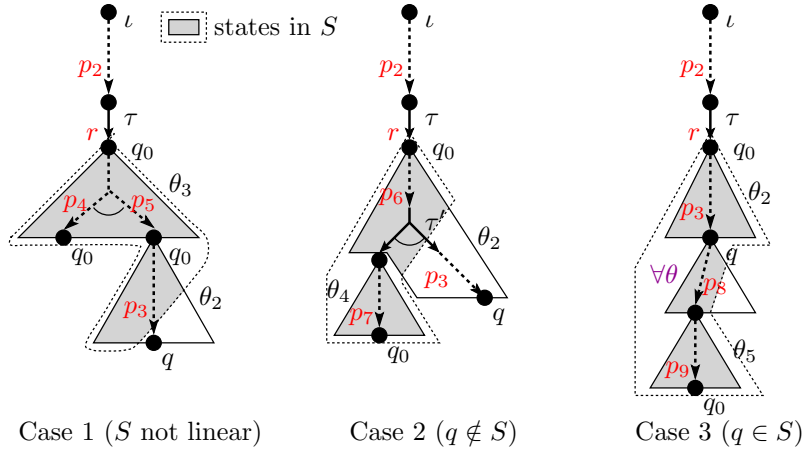
    \caption{\label{fig:choice-free-property} The cases from the proof of
     \autoref{lemma:choice-free-property}}
    \end{center}
  \end{figure}

\subsection{Proof of \autoref{lemma:choice-free}}
\label{app:choice-free}

We assume w.l.o.g. that $\mathcal{A} =
(\alphabet,\states,\initstate,\trans)$ is rooted. Let
$\graphof{\mathcal{A}} = (\nodes,\edges)$ be the SCC graph of
$\mathcal{A}$, where $\nodes = \set{S_1,\ldots,S_M}$ is a topological
ordering of the SCCs i.e., if $(S_i,S_j)\in\edges$ then $i < j$, for
all $i,j \in \interv{1}{M}$. For each $i = 1,\ldots,M$, we iterate the
following transformation of $\mathcal{A}$:
\begin{itemize}[label=$\triangleright$]
\item let $\mathcal{S}_i \isdef \bigcup_{(S_i,S_j)\in\edges^*} S_j$ be
  the set of states from any SCC reachable from $S_i$ in
  $\graphof{\mathcal{A}}$,
\item let $k_i \isdef \sum_{\tau\in\pre{S_i}} \cardof{\post{\tau}\cap
  S_i}$ be the number of edges of $\graphof{\mathcal{A}}$ incoming to
  $S_i$,
\item create $k_i$ copies of the transitions $q_0 \arrow{a}{}
  (q_1,\ldots,q_\ell) \in \trans$ such that
  $\set{q_0,q_{i_1},\ldots,q_{i_j}} = \set{q_0,\ldots,q_\ell} \cap
  \mathcal{S}_i$ i.e., add a transition $(q_0,h) \arrow{a}{} (q_1,
  \ldots, (q_{i_1},h), \ldots, (q_{i_j},h), \ldots, q_\ell)$ for each
  $h \in \interv{1}{k_i}$,
\item connect these new transitions to the rest of the automaton by
  adequately changing the states $q\in\post{\tau}\cap \mathcal{S}_i$
  for $\tau\in\pre{\mathcal{S}_i}$ to their corresponding copies
  $(q,h)$, for all $h \in \interv{1}{k_i}$.
\end{itemize}
It is easy to check that the resulting automaton fulfills condition
(\ref{it1:def:choice-free}) of \autoref{def:choice-free} and
has the same language as $\mathcal{A}$, using
\autoref{lemma:refinement}. We can thus assume w.l.o.g. in the
following that $\graphof{\mathcal{A}}=(\nodes,\edges,S_1)$ is a tree
and let $S_1,\ldots,S_N$ be a topological ordering of its nodes. We
associate a variable $x_i$ (resp. $y_\tau$) ranging over
$\set{0,1,\infty}$ with each SCC $S_i\in\nodes$ (resp. transition
$\tau\in\trans$). Initially, the values of these variables are all
zero. We iterate over the finite sequence $S_1, \ldots, S_N$ as
follows. For each $i \in \interv{1}{N}$, we perform the following
assignments in this order:
\begin{enumerate}[(i)]
\item\label{iter1} let $x_{i} \isdef \left\{\begin{array}{ll}
  1, & \text{if } i=1 \\
  \sum_{\tau \in \pre{S_i}} y_\tau \cdot \cardof{\post{\tau} \cap S_i}, & \text{otherwise}
\end{array}\right.$
\item\label{iter2} for each $\tau \in \prepost{S_i}$,
  let $y_\tau\isdef\left\{\begin{array}{ll}
  \infty, \text{if } x_{i}>0 \\
  0, \text{otherwise}
  \end{array}\right.$
\item\label{iter3} if $x_{i} \in \set{0,\infty}$ or $S_i$ is nonlinear, for each
  $\tau \in \post{S_i}$, let $y_\tau \isdef \left\{\begin{array}{ll}
  \infty, \text{if } x_{i} > 0 \\ 0, \text{otherwise}
\end{array}\right.$
\item\label{iter4} else (i.e., $x_{i} = 1$ and $S_i$ is linear)
  chose for all $\set{y_\tau}_{\tau \in \post{S_i}}$ some values
  from $\set{0,1}$, such that $x_{i} = \sum_{\tau\in\post{S_i}}
  y_\tau$.
\end{enumerate}
Since, for each SCC $S_i \in\nodes$, there is at most one transition
$\tau\in\pre{S_i}$ and $\cardof{\post{\tau}\cap S_i} \le 1$, each
variable $x_i$ is assigned either $0$, $1$ or $\infty$ at
\ref{iter1}. Note that no variable is assigned twice in the above
iteration sequence, because every $x_i$ is assigned exactly once,
every $y_\tau$, for $\tau \in \pre{S_i}$ is assigned before $x_i$ and
every $y_\tau$, for $\tau \in \prepost{S_i} \cup \post{S_i}$ is
assigned after $x_i$. Furthermore, we have $\prepost{S_i} \cap
\prepost{S_j} = \emptyset$ and $\post{S_i} \cap \post{S_j} =
\emptyset$, for all $1 \le i < j \le N$, so that each $y_\tau$, for
$\tau \in \prepost{S_i} \cup \post{S_i}$, is assigned exactly
once. Moreover, since the choice at \ref{iter4} is finite, there are
finitely many outcomes of the above nondeterministic iteration, say
$(\vec{x}_1,\vec{y}_1), \ldots, (\vec{x}_\ell,\vec{y}_\ell)$, where
$\vec{x}_i = \tuple{\overline{x}_{i,j}}_{j\in\interv{1}{N}}$ and
$\vec{y}_i = \tuple{\overline{y}_{i,\tau}}_{\tau\in\trans}$. For each
$i \in \interv{1}{\ell}$, we define the automaton $\mathcal{A}_i =
(\alphabet,\states_i,\initstate,\trans_i)$, where $\states_i \isdef
\bigcup \set{S_j \mid \overline{x}_{i,j} > 0,~ j \in \interv{1}{N}}$
and $\trans_i \isdef \set{\tau \in \trans \mid \overline{y}_{i,\tau} >
  0}$.
\noindent We are left with proving the following facts:

\begin{fact}
  Each automaton $\mathcal{A}_i$ is choice-free, for $i \in
  \interv{1}{N}$.
\end{fact}
\begin{proof}
  We prove below the points of \autoref{def:choice-free}:

  \skipnoindent (\ref{it1:def:choice-free}) Let $S_{j_0}$ be an SCC of
  $\mathcal{A}$, such that $\overline{x}_{i,j_0} > 0$ i.e., $S_{j_0}$
  is a vertex in the SCC graph $\graphof{\mathcal{A}_i}$. Since the
  variable $x_{j_0}$ received its value $\overline{x}_{i,j_0}$ at
  \ref{iter1}, either $j_0=1$ (in which case $\overline{x}_{i,j_0}=1$)
  or there exists an incoming transition $\tau\in\pre{S_{j_0}}$ such
  that $\overline{y}_{i,\tau} > 0$. Let $S_{j_1}$, for some $j_1 <
  j_0$ be the SCC such that $\pre{\tau}\in S_{j_1}$. Then
  $\overline{x}_{i,j_1} > 0$. Repeating the same argument for $j_1$,
  we discover a maximal finite sequence $j_0, \ldots, j_k$ such that
  $(S_{j_{i+1}},S_{j_{i}}) \in \edges$, for all $i \in
  \interv{0}{k-1}$. Moreover, it must be the case that $j_k=1$, or
  else the sequence could be extended, contradicting its
  maximality. Since $\graphof{\mathcal{A}}$ is a tree, the path from
  $S_1$ to $S_{j_0}$ must be unique and, since the choice of $S_{j_0}$
  was arbitrary, $\graphof{\mathcal{A}_i} = (\nodes_i, \edges_i, S_1)$
  is a tree as well. The second point from condition
  (\ref{it1:def:choice-free}) holds already for
  $\graphof{\mathcal{A}}$, hence it must hold for
  $\graphof{\mathcal{A}_i}$.

  \skipnoindent (\ref{it2:def:choice-free}) The mapping $\Lambda_i :
  \nodes_i \cup \trans_i \rightarrow \set{1,\infty}$ is defined as
  $\Lambda_i(S_j)=\overline{x}_{i,j}$ for each $S_j \in \nodes_i$ and
  $\Lambda_i(\tau)=\overline{y}_{i,\tau}$ for each
  $\tau\in\trans_i$. We check that $\Lambda_i$ verifies the conditions
  (\ref{it2:def:choice-free}) from
  \autoref{def:choice-free}: \begin{itemize}[left=.6\parindent]
  \item[(\ref{it21:def:choice-free})] if $S_j\in\nodes_i$ is linear
    and $\overline{x}_{i,j}=1$ then the choice at step \ref{iter4}
    was $\overline{y}_{i,\tau}=1$, for exactly one transition $\tau
    \in \post{S_j}$.
  \item[(\ref{it22:def:choice-free})] $\overline{y}_{i,\tau}=1$ iff
    the value of $y_\tau$ was set at step \ref{iter4} and
    $\tau\in\post{S_j}$ is the unique outgoing transition for which a
    nonzero value was assigned to $y_\tau$, for a linear SCC $S_j$
    with $\overline{x}_{i,j}=1$.
  \item[(\ref{it23:def:choice-free})] $\overline{x}_{i,j}=1$ iff the
    value of $x_j$ was set at step \ref{iter1} and either $j=1$ or
    for all but one transitions $\tau\in\pre{S_j}$ we have
    $\overline{y}_{i,\tau}=1$. \qedhere
  \end{itemize}
\end{proof}

\begin{fact}
  $\langof{}{\mathcal{A}} = \bigcup_{i=1}^\ell \langof{}{\mathcal{A}_i}$.
\end{fact}
\begin{proof}
  ``$\supseteq$'' Since $\trans_i \subseteq \trans$, we have
  $\langof{}{\mathcal{A}_i} \subseteq \langof{}{\mathcal{A}}$, for
  all $i \in \interv{1}{\ell}$.

  ``$\subseteq$'' Let $t \in \langof{}{\mathcal{A}}$ and $\arun$ be an
  accepting run of $\mathcal{A}$ over $t$. We show that there exists
  an iteration \ref{iter1}--\ref{iter4} leading to the values
  $(\vec{x}_i,\vec{y}_i)$ such that, for each transition $\tau$
  occurring on $\arun$ at some position $p\in\dom{\arun}$ i.e.,
  $\arun(p) = \pre{\tau}$, we have: \begin{itemize}[label=$\triangleright$]
  \item $\overline{x}_{i,j} > 0$, where $S_j$ is the unique SCC of
    $\mathcal{A}$ such that $\pre{\tau}\in S_j$, and
  \item $\overline{y}_{i,\tau} > 0$.
  \end{itemize}
  By the second point above we obtain that $\arun$ is an accepting
  run of $\mathcal{A}_i$. The proof is by reverse induction on the
  size of the subtree of $\arun$ rooted at $p$.

  \skipnoindent \underline{Base case:} If $p = \epsilon$, the variable
  $x_1$ is always assigned the value $1$ at step \ref{iter1}. We
  chose the values for all $\set{y_{\tau'}}_{\tau' \in \post{S_1}}$,
  such that $y_\tau$ is assigned $1$ and $y_{\tau'}$ is assigned $0$,
  for all $\tau' \in \post{S_1} \setminus \set{\tau}$ at step
  \ref{iter4}.

  \skipnoindent \underline{Induction step:} If $p \in \dom{\arun}
  \setminus \set{\epsilon}$, since $j\neq1$, by the inductive
  hypothesis, the variable $y_{\tau'}$ is assigned non-zero values,
  for at least one $\tau' \in \pre{S_j}$, thus we assign $x_j$ the
  value $\sum_{\tau' \in \pre{S_j}} y_{\tau'} \cdot
  \cardof{\post{\tau} \cap S_j} > 0$ at step \ref{iter1}. If $\tau \in
  \prepost{S_j}$, then $y_\tau$ is assigned $\infty$ at step
  \ref{iter2}. Otherwise, it must be the case that
  $\tau\in\post{S_j}$ and we distinguish two cases. If $S_j$ is
  nonlinear, then $y_\tau$ is assigned $\infty$ at
  \ref{iter3}. Else, $S_j$ is linear and we can chose the value $1$
  for $y_\tau$ at step \ref{iter4}, because $x_j$ has been already
  assigned to $1$.
\end{proof}

\noindent Let $i \in \interv{1}{\ell}$. We prove the upper bound on
$\cardof{\trans^1_i}$, as follows. Since the SCC graph of
$\mathcal{A}_i$ is a tree, the number of $1$-transitions in
$\mathcal{A}_i$ equals the number of SCCs in $\mathcal{A}_i$. Due to
the expansion of the first step of the proof, we have
$\cardof{\trans^1_i} \leq \max\set{\rankof{a} \mid a \in \alphabet}^s
\leq \max\set{\rankof{a} \mid a \in \alphabet}^{\cardof{\states}}$,
where $s \leq \cardof{\states}$ denotes the number of SCCs in
$\mathcal{A}$. \qed

\subsection{Proof of \autoref{lemma:sid-ta}}
\label{app:sid-ta}

\skipnoindent (\ref{it1:lemma:sid-ta}) For the first part, we prove
the two directions of the following equivalence: for all structures
$\astruc$ and predicates $\bpred$ of arity $n$, there exists a store
$\store$ such that $\astruc \models^\store_\asid \bpred(x_1, \ldots,
x_n)$ iff there exists a tree $t \in
\langof{q_{\bpred}}{\auto{\asid}{\apred}}$ and a store $\bar{\store}$
such that $\astruc \models^{\bar{\store}} \charform{t}$ and
$\bar{\store}(\atpos{x}{\epsilon}_j) = \store(x_j)$ for all
$j\in\interv{1}{n}$.

\skipnoindent ``$\Rightarrow$'' We proceed by induction on the
definition of $\astruc \models^\store_\asid \bpred(x_1, \ldots, x_n)$.
Then $\asid$ contains a rule $\arule$ of the form \(\bpred(x_1,
\ldots, x_n) \leftarrow \exists y_1 \ldots \exists y_m ~.~ \psi *
\Bigstar_{i=1}^\ell \bpred_i(z_{i,1}, \ldots, z_{i,n_i})\) and we can
decompose the structure $\astruc = \astruc_0 \comp \ldots \comp
\astruc_\ell$, such that $\astruc_0 \models^{\store'} \psi$ and
$\astruc_i \models^{\store'}_\asid \bpred_i(z_{i,1},\ldots,z_{i,n_i})$
for all $i\in\interv{1}{\ell}$, for a store $\store'$ that agrees with
$\store$ over $\set{x_1,\ldots,x_{n}}$. For all
$i\in\interv{1}{\ell}$, we consider a store $\store_i$ such that
$\store_i(x_j) = \store'(z_{i,j})$, for all $j \in \interv{1}{n_i}$.
We have $\astruc_i \models^{\store_i}_\asid
\bpred_i(x_1,\ldots,x_{n_i})$ and, by induction hypothesis, there
exists a tree $t_i \in \langof{q_{\bpred_i}}{\auto{\asid}{\apred}}$
and a store $\bar{\store_i}$ such that $\astruc_i
\models^{\bar{\store_i}} \charform{t_i}$ and
$\bar{\store_i}(\atpos{x}{\epsilon}_j) = \store_i(x_j)$ for all
$j\in\interv{1}{n_i}$. Let $\bar{\store}$ be a store such
that: \begin{itemize}[label=$\triangleright$]
\item $\bar{\store}(\atpos{x}{\epsilon}_j) = \store(x_j)$, for all
  $j\in\interv{1}{n}$,
\item $\bar{\store}(\atpos{y}{\epsilon}_j) = \store'(y_j)$ for all
  $j\in\interv{1}{m}$,
\item $\bar{\store}(\atpos{z}{ip}) \isdef
  \bar{\store_i}(\atpos{z}{p})$, for all $i\in\interv{1}{\ell}$ and
  $\atpos{z}{p}\in\fv{\charform{t_i}}$.
\end{itemize}
Note that $\bar{\store}$ is well defined because
$\bar{\store_i}(\atpos{z}{p}) = \bar{\store_j}(\atpos{z}{p}) =
\store'(\atpos{z}{p})$, for all $\atpos{z}{p} \in \fv{\charform{t_i}}
\cap \fv{\charform{t_j}}$. We have $\astruc_i \models^{\bar{\store}}
\atpos{\charform{t_i}}{i}$ for all $i\in\interv{1}{\ell}$ and
$\astruc_0 \models^{\bar{\store}} \alpha_\arule$, thus $\astruc
\models^{\bar{\store}} \charform{t}$ where $t$ is the tree consisting
of a root labelled by $\alpha_\arule$ and $\ell$ children $t_i$ for
$i\in\interv{1}{\ell}$.  Since $t \in
\langof{q_{\bpred}}{\auto{\asid}{\apred}}$ and
$\bar{\store}(\atpos{x}{\epsilon}_j) = \store(x_j)$, for all
$j\in\interv{1}{n_0}$, by definition, we obtain the result.

\skipnoindent ``$\Leftarrow$'' The reverse implication is proven by
induction on the structure of the tree $t\in
\langof{q_{\bpred}}{\auto{\asid}{\apred}}$ such that $\astruc
\models^{\bar{\store}} \charform{t}$.
Since $t \in \langof{q_{\bpred}}{\auto{\asid}{\apred}}$, there is a
transition $q_{\bpred} \arrow{t(\epsilon)}{}
(q_{\bpred_1},\ldots,q_{\bpred_\ell}) \in \trans_\asid$ such that
$t(\epsilon)=\alpha_\arule$ for some rule $\arule$ of the form above
and $\subtree{t}{i} \in \langof{q_{\bpred_i}}{\auto{\asid}{\bpred}}$
for all $i\in\interv{1}{\ell}$.
Meanwhile $\astruc \models^{\bar{\store}} t(\epsilon) *
\Bigstar_{i=1}^\ell ~\atpos{\charform{\subtree{t}{i}}}{i}$, thus we
can decompose the structure as $\astruc = \astruc_0 \comp \ldots \comp
\astruc_\ell$, such that $\astruc_0 \models^{\bar{\store}}
\alpha_\arule$ and $\astruc_i \models^{\bar{\store}}
\atpos{\charform{\subtree{t}{i}}}{i} * \Bigstar_{j=1}^{n_i}
~\atpos{z}{\epsilon}_{i,j} = \atpos{x}{i}_j$ for all
$i\in\interv{1}{\ell}$. Note that the additional equalities from
$\alpha_\arule$ are necessary to remember the links between the
variables from $\arule$. Let $\bar{\store_i}$ be a store, such that
$\bar{\store_i}(\atpos{z}{p}) = \bar{\store}(\atpos{z}{ip})$, for all
$\atpos{z}{p}\in\fv{\charform{\subtree{t}{i}}}$ and all $i \in
\interv{1}{\ell}$. By the inductive hypothesis on $\subtree{t}{i}$,
there exists a store $\store_i$ such that $\astruc_i
\models^{\store_i}_\asid \bpred_i(x_1,\ldots,x_{n_i})$ and
$\store_i(x_j) = \bar{\store_i}(\atpos{x}{\epsilon}_j)$, for all
$j\in\interv{1}{n_i}$.  We consider a store $\store'$ such that
$\store'(x_j) = \bar{\store}(\atpos{x}{\epsilon}_j)$, for all
$j\in\interv{1}{n}$, and $\store'(y_j) \isdef
\bar{\store}(\atpos{y}{\epsilon}_j)$, for all $j\in\interv{1}{m}$.
For all $i\in\interv{1}{\ell}$ and $j\in\interv{1}{n_i}$ we have
$\store_i(x_j) = \bar{\store_i}(\atpos{x}{\epsilon}_j) =
\bar{\store}(\atpos{x}{i}_j) = \bar{\store}(\atpos{z}{\epsilon}_{i,j})
= \store'(z_{i,j})$, because $\atpos{z}{\epsilon}_{i,j} =
\atpos{x}{i}_j$ holds for $\bar{\store}$ in the empty structure.
Therefore $\astruc_i \models^{\store'}_\asid
\bpred_i(z_{i,1},\ldots,z_{i,n_i})$, for all $i \in \interv{1}{\ell}$.
Moreover $\astruc_0 \models^{\store'} \psi$, and by composing the
structures and using $\arule$, we obtain $\astruc
\models^{\store'}_\asid \bpred(x_1, \ldots, x_{n})$.

\skipnoindent (\ref{it2:lemma:sid-ta}) To show
$\sem{\mathcal{A}} = \sidsem{\apred_\initstate}{\asid_{\mathcal{A}}}$
we prove the following equivalence: for all structures $\astruc$ and
states $q_0\in Q$, there exists a tree $t \in
\langof{q_0}{\mathcal{A}}$ and a store $\bar{\store}$ such that
$\astruc \models^{\bar{\store}} \charform{t}$ iff there exists a store
$\store$ such that $\astruc \models^\store_{\asid_{\mathcal{A}}}
\apred_{q_0}(x_1, \ldots, x_{\arityof{q_0}})$ and $\store(x_j) =
\bar{\store}(\atpos{x}{\epsilon}_j)$, for all
$j\in\interv{1}{\arityof{q_0}}$.

\skipnoindent ``$\Rightarrow$'' We reason by induction on the
structure of the tree $t\in \langof{q_0}{\mathcal{A}}$, such that
$\astruc \models^{\bar{\store}} \charform{t}$.
Since $t \in \langof{q_0}{\mathcal{A}}$, there is a transition $q_0
\arrow{t(\epsilon)}{} (q_1,\ldots,q_\ell) \in \trans$ such that
$\subtree{t}{i} \in \langof{q_i}{\mathcal{A}}$, for all
$i\in\interv{1}{\ell}$.
Meanwhile $\astruc \models^{\bar{\store}} t(\epsilon) *
\Bigstar_{i=1}^\ell ~\atpos{\charform{\subtree{t}{i}}}{i}$ thus, we
can decompose the structure $\astruc = \astruc_0 \comp \ldots \comp
\astruc_\ell$, such that $\astruc_0 \models^{\bar{\store}}
t(\epsilon)$ and $\astruc_i \models^{\bar{\store}}
\atpos{\charform{\subtree{t}{i}}}{i}$.  Let $\bar{\store_i}$ be a
store, such that $\bar{\store_i}(\atpos{z}{p}) =
\bar{\store}(\atpos{z}{ip})$, for all
$\atpos{z}{p}\in\fv{\charform{\subtree{t}{i}}}$ and all $i \in
\interv{1}{\ell}$.  By the inductive hypothesis on $\subtree{t}{i}$,
there exists a store $\store_i$ such that $\astruc_i
\models^{\store_i}_{\asid_{\mathcal{A}}}
\apred_{q_i}(x_1,\ldots,x_{\arityof{q_i}})$ and $\store_i(x_j) =
\bar{\store_i}(\atpos{x}{\epsilon}_j)$ for all
$j\in\interv{1}{\arityof{q_i}}$.
We consider a store $\store'$, such that $\store'(x_j) =
\bar{\store}(\atpos{x}{\epsilon}_j)$, for all
$j\in\interv{1}{\arityof{q_0}}$, $\store'(\atpos{x}{i}_j) =
\store_i(x_j)$, for all $i\in\interv{1}{\ell}$ and all
$j\in\interv{1}{\arityof{q_i}}$, and $\store'(z) = \bar{\store}(z)$,
for all other variables $z\in\fv{t(\epsilon)}$.  Then $\astruc_i
\models^{\store'}_{\asid_{\mathcal{A}}}
\apred_{q_i}(\atpos{x}{i}_1,\ldots,\atpos{x}{i}_{\arityof{q_i}})$ and
$\astruc_0 \models^{\store'} t(\epsilon)[\atpos{x}{\epsilon}_1/x_1,
  \ldots, \atpos{x}{\epsilon}_{\arityof{q_0}}/x_{\arityof{q_0}}]$
thus, by composing the structures, we obtain $\astruc
\models^{\store'}_{\asid_{\mathcal{A}}} \apred_{q_0}(x_1, \ldots,
x_{\arityof{q_0}})$.

\skipnoindent ``$\Leftarrow$'' The reverse is shown by induction on
the definition of $\astruc \models^\store_{\asid_{\mathcal{A}}}
\apred_{q_0}(x_1, \ldots, x_{\arityof{q_0}})$. Then there exists a
rule in $\asid_{\mathcal{A}}$ of the form (\ref{rule:ta-sid}) and we
can decompose the structure $\astruc = \astruc_0 \comp \ldots \comp
\astruc_\ell$ such that $\astruc_0 \models^{\store'}
\alpha[\atpos{x}{\epsilon}_1/x_1, \ldots,
  \atpos{x}{\epsilon}_{\arityof{q_0}}/x_{\arityof{q_0}}]$ and
$\astruc_i \models^{\store'}_{\asid_{\mathcal{A}}}
\apred_{q_i}(\atpos{x}{i}_1,\ldots,\atpos{x}{i}_{\arityof{q_i}})$, for
all $i\in\interv{1}{\ell}$, where $\store'$ is a store that agrees
with $\store$ over $\set{x_1,\ldots,x_{\arityof{q_i}}}$.  For all
$i\in\interv{1}{\ell}$, we consider a store $\store_i$ such that
$\store_i(x_j) = \store'(\atpos{x}{i}_j)$.  We have $\astruc_i
\models^{\store_i}_{\asid_{\mathcal{A}}}
\apred_{q_i}(x_1,\ldots,x_{\arityof{q_0}})$ and, by the inductive
hypothesis, there exists a tree $t_i \in \langof{q_i}{\mathcal{A}}$
and a store $\bar{\store_i}$ such that $\struc_i
\models^{\bar{\store_i}} \charform{t_i}$ and
$\bar{\store_i}(\atpos{x}{\epsilon}_j) = \store_i(x_j)$, for all
$j\in\interv{1}{\arityof{q_i}}$. Let $\bar{\store}$ be a store such
that: \begin{itemize}[label=$\triangleright$]
  \item $\bar{\store}(\atpos{x}{\epsilon}_j) \isdef \store(x_j)$, for
    all $j\in\interv{1}{\arityof{q_0}}$,
  \item $\bar{\store}(\atpos{y}{\epsilon}_j) \isdef
    \store'(\atpos{y}{\epsilon}_j)$, for all $j\in\interv{1}{m}$,
  \item $\bar{\store}(\atpos{z}{ip}) \isdef
    \bar{\store_i}(\atpos{z}{p})$, for all $i\in\interv{1}{\ell}$ and
    all $\atpos{z}{p}\in\fv{\charform{t_i}}$.
  \end{itemize}
Note that $\bar{\store}$ is well defined because
$\bar{\store_i}(\atpos{z}{p}) = \bar{\store_j}(\atpos{z}{p}) =
\store'(\atpos{z}{p})$, for all $\atpos{z}{p} \in \fv{\charform{t_i}}
\cap \fv{\charform{t_j}}$.
We have $\astruc_i \models^{\bar{\store}} \atpos{\charform{t_i}}{i}$,
for all $i\in\interv{1}{\ell}$ and $\astruc_0 \models^{\bar{\store}}
\alpha$, thus $\astruc \models^{\bar{\store}} \charform{t}$, where $t$
is the tree consisting of a root labelled by $\alpha$ and children
$t_i$, for $i\in\interv{1}{\ell}$.  Since $t \in
\langof{q_0}{\mathcal{A}}$ and $\bar{\store}(\atpos{x}{\epsilon}_j) =
\store(x_j)$ for all $j\in\interv{1}{\arityof{q_0}}$ by definition, we
obtain the result.

\skipnoindent To show $\rcsem{\mathcal{A}}{}
\subseteq \rcsem{\apred_\iota}{\asid_{\mathcal{A}}}$, let
$(\astruc,\diseq) \in \csem{\mathcal{A}}{}$, where
$\astruc=(\univ,\struc)$ is a structure and
$\diseq\subseteq\univ\times\univ$ is a symmetric relation. Then there
exists a tree $t \in \langof{}{\mathcal{A}}$ and a store $\store$
canonical for $\charform{t}$, such that $\astruc \models^\store
\charform{t}$ and, for each $(u,v)\in\diseq$, there exist variables $x
\in \store^{-1}(u)$ and $y \in \store^{-1}(v)$ such that the
disequality $x\neq y$ occurs in $\charform{t}$.  Let $\arun$ be the
accepting run of $\mathcal{A}$ over $t$. By a depth-first traversal of
$\arun$, we build a complete unfolding $\apred_\iota
\step{\asid_{\mathcal{A}}}^* \exclof{\charform{t}}$. Since $\astruc
\models^\store \charform{t}$, we obtain $(\astruc,\diseq) \in
\rcsem{\apred_\iota}{\asid_{\mathcal{A}}}$, by
\autoref{def:canonical-model}. Conversely, to show
$\rcsem{\mathcal{A}}{} \supseteq
\rcsem{\apred_\iota}{\asid_{\mathcal{A}}}$, let $(\astruc,\diseq) \in
\rcsem{\apred_\iota}{\asid_{\mathcal{A}}}$, where
$\astruc=(\univ,\struc)$ is a structure and
$\diseq\subseteq\univ\times\univ$ is a symmetric relation. Then, there
exists a complete unfolding $\apred_\iota \step{\asid_{\mathcal{A}}}^*
\exists x_1 \ldots \exists x_n ~.~ \psi$, where $\psi$ is a qpf
formula, and a store $\store$ canonical for $\psi$, such that $\astruc
\models^\store \psi$ and, for all $(u,v)\in\diseq$ there exist
variables $x\in\store^{-1}(u)$ and $y\in\store^{-1}(v)$, such that the
disequality $x\neq y$ occurs in $\psi$. By induction on the length of
the unfolding, one can build an accepting run $\arun$ of
$\mathcal{A}$, that recognizes a tree $t \in \langof{}{\mathcal{A}}$,
such that $\charform{t}$ differs from $\psi$ by an $\alpha$-renaming
and permutation of atoms via commutativity and associativity of the
separating conjunction. Hence $(\astruc,\diseq) \in
\rcsem{\exclof{\charform{t}}}{}$, thus $(\astruc,\diseq) \in
\rcsem{\mathcal{A}}{}$. \qed

\subsection{Proof of \autoref{lemma:remove-relations}}
\label{app:remove-relations}

\skipnoindent (\ref{it1:lemma:remove-relations}) Let $t' \in
\langof{}{\auto{\asid}{\apred}^{I}}$ be a tree. Since
$\auto{\asid}{\apred}^{I}$ was obtained from $\auto{\asid}{\apred}$ by
removing relation and disequality atoms from the labels of its
$1$-transitions, there exists a tree $t \in
\langof{}{\auto{\asid}{\apred}}$, such that $\dom{t} = \dom{t'}$ and
$\charform{t} = \charform{t'} * \Bigstar\nolimits_{i=1}^n
\arel_i(z_{i,1}, \ldots, z_{i,k_i}) * \Bigstar\nolimits_{j=1}^m
y_{j,1} \neq y_{j,2}$ modulo reordering of atoms, for some relation
symbols $\arel_i$ and variables $z_{i,1}, \ldots, z_{i,k_i}$,
$y_{j,1}, y_{j,2}$. By \autoref{lemma:sid-ta}, $\charform{t}$ is
satisfiable, hence there exists a structure $(\univ,\struc)$ and a
store $\store$, such that $(\univ,\struc) \models^\store
\charform{t}$. We define the interpretation $\struc'(\arel) =
\struc(\arel) \setminus \{\tuple{\store(z_{i,1}), \ldots,
  \store(z_{i,k_i})} \mid i \in \interv{1}{n},~ \arel_i = \arel\}$,
for all $\arel\in\relations$. It is easy to check that
$(\univ,\struc') \models^\store \charform{t'}$, hence $\charform{t'}$
is satisfiable. Since the choice of $t'$ was arbitrary, we obtain that
$\auto{\asid}{\apred}^{I}$ is all-satisfiable.

\skipnoindent (\ref{it2:lemma:remove-relations}) Let $t' \in
\langof{}{\auto{\asid}{\apred}^{I}}$ be a tree. By the construction of
$\auto{\asid}{\apred}^{I}$ from $\auto{\asid}{\apred}$, there exists a
tree $t \in \langof{}{\auto{\asid}{\apred}}$, such that
$\dom{t}=\dom{t'}$ and $\charform{t} = \charform{t'} *
\Bigstar\nolimits_{i=1}^n \arel_i(z_{i,1}, \ldots, z_{i,k_i}) *
\Bigstar\nolimits_{j=1}^m y_{j,1} \neq y_{j,2}$ modulo reordering of
atoms, for some relation symbols $\arel_i$ and variables $z_{i,1},
\ldots, z_{i,k_i}$, $y_{j,1}, y_{j,2}$.  Let $Y \isdef \bigcup_{j=1}^m
\set{y_{j,1}, y_{j,2}}$ be the set of variables occurring in
disequality atoms.  By \autoref{lemma:one-transitions}, each
$1$-transition of $\auto{\asid}{\apred}$ occurs exactly once in each
accepting run, hence $n \le \cardof{\trans^1_\asid} \cdot
\maxrelatominruleof{\asid}$, $ \cardof{Y} \le \cardof{\trans^1_\asid}
\cdot \maxvarinruleof{\asid}$.  By \autoref{lemma:qpf-treewidth}
(\ref{it2:qpf-treewidth}, \ref{it3:qpf-treewidth}),
$\twof{\exclof{\charform{t'}}} \le \twof{\exclof{\charform{t}}} + n +
\cardof{Y} \le \twof{\exclof{\charform{t}}} + \cardof{\trans^1_\asid}
\cdot (\maxrelatominruleof{\asid} + \maxvarinruleof{\asid})$. Since
the choice of $t'$ was arbitrary, we obtain that
$\sem{\auto{\asid}{\apred}^{I}}$ is treewidth bounded, more precisely
$\twof{\sem{\auto{\asid}{\apred}^{I}}} \le
\twof{\sem{\auto{\asid}{\apred}}} + \cardof{\trans^1_\asid} \cdot
(\maxrelatominruleof{\asid} + \maxvarinruleof{\asid})$.

\skipnoindent (\ref{it3:lemma:remove-relations}) Let $t \in
\langof{}{\auto{\asid}{\apred}}$ be a tree.  By the construction of
$\auto{\asid}{\apred}^{I}$ from $\auto{\asid}{\apred}$, there exists a
tree $t' \in \langof{}{\auto{\asid}{\apred}^{I}}$, such that
$\dom{t'}=\dom{t}$ and $\charform{t} = \charform{t'} *
\Bigstar\nolimits_{i=1}^n \arel_i(z_{i,1}, \ldots, z_{i,k_i}) *
\Bigstar\nolimits_{j=1}^m y_{j,1} \neq y_{j,2}$ modulo reordering of
atoms, for some relation symbols $\arel_i$ and variables $z_{i,1},
\ldots, z_{i,k_i}$, $y_{j,1}, y_{j,2}$.  Let $Y \isdef \bigcup_{i=1}^n
\set{z_{i,1},\ldots z_{i,k_i}}$ be the set of variables occurring in
relation atoms.  By \autoref{lemma:one-transitions}, each
$1$-transition of $\auto{\asid}{\apred}$ occurs exactly once in each
accepting run, hence $\cardof{Y} \le \cardof{\trans^1_\asid} \cdot
\maxvarinruleof{\asid}$.  By \autoref{lemma:qpf-treewidth}
(\ref{it4:qpf-treewidth}) $\twof{\exclof{\charform{t}}} \le
\twof{\exclof{\charform{t'}}} + \cardof{Y} \le
\twof{\exclof{\charform{t'}}} + \cardof{\trans^1_\asid} \cdot
\maxvarinruleof{\asid}$.  Since the choice of $t$ was arbitrary, we
obtain that $\sem{\auto{\asid}{\apred}}$ is treewidth bounded, more
precisely $\twof{\sem{\auto{\asid}{\apred}}} \le
\twof{\sem{\auto{\asid}{\apred}^{I}}} + \cardof{\trans^1_\asid} \cdot
\maxvarinruleof{\asid}$. \qed

\subsection{Proof of \autoref{lemma:remove-equalities}}
\label{app:remove-equalities}

\skipnoindent (\ref{it1:lemma:remove-equalities}) Let $t' \in
\langof{}{\auto{\asid}{\apred}^{II}}$ be a tree. Since
$\auto{\asid}{\apred}^{II}$ was obtained from
$\auto{\asid}{\apred}^{I}$ by removing equality atoms from the labels
of its $1$-transitions (\ref{step2:remove-equalities}), there exists a
tree $t \in \langof{}{\auto{\asid}{\apred}^{I}}$, such that $\dom{t} =
\dom{t'}$ and $\charform{t} = \charform{t'} * \psi$ where $\psi$ is a
conjunction of equality atoms.  By \autoref{lemma:remove-relations}
(\ref{it1:lemma:remove-relations}), $\charform{t}$ is satisfiable,
hence there exists a structure $\astruc$ and a store $\store$ such
that $\astruc \models^\store \charform{t}$.  We immediately obtain
$\astruc \models^\store \charform{t'}$, hence $\charform{t'}$ is
satisfiable. Since the choice of $t'$ was arbitrary, we obtain that
$\auto{\asid}{\apred}^{II}$ is all-satisfiable.

\skipnoindent (\ref{it2:lemma:remove-equalities}) Let
$t'\in\langof{}{\auto{\asid}{\apred}^{II}}$ be a tree and $\arun'$ be
an accepting run over $t'$.  We shall build a tree $t \in
\langof{}{\auto{\asid}{\apred}^{I}}$ related to $t'$ and $\arun'$ and
show that $\twof{\sem{\exclof{\charform{t'}}}} \le
\twof{\sem{\exclof{\charform{t}}}} + K$, where $K$ is constant that
does not depend on the choice of $t'$.  The idea of the construction
of $t$ from $t'$ is to add resets before and after each $1$-transition
in the run $\arun'$, so that the equalities removed by the
transformation from $\trans^{I}_\asid$ to $\trans^{II}_\asid$ can be
added back, without changing the set of models of
$\exclof{\charform{t'}}$. To avoid unnecessary complications, we
consider each $1$-transition separately (recall that there are
finitely many 1-transitions in $\auto{\asid}{\apred}^{II}$).

We refer to \autoref{fig:remove-equalities} for an illustration of
this construction. For a given position $p\in\dom{t'}$ such that
$\arun'(p) \arrow{t'(p)}{} (\arun'(p1),\ldots,\arun'(p\ell)) \in
{(\trans^{II}_\asid)}^1 $, we separate the run $\arun'$ into a context
$\arun_{p \leftarrow \arun(p)}^\text{init}$, before the
$1$-transition, and $\ell\isdef\rankof{t'(p)}$ runs
$\arun_1,\ldots,\arun_\ell$, after the $1$-transition, i.e., $\arun_{p
  \leftarrow \arun(p)}^\text{init}(r) \isdef \arun'(r)$ for every
$r\in\dom{t'}$ that is not a suffix of $p$, and $\arun_i(r) \isdef
\arun'(pir)$ for every $i\in\interv{1}{\ell}$ and position $r$ with
$pir\in\dom{t'}$.

  \begin{figure}[htbp]
    \begin{center}
    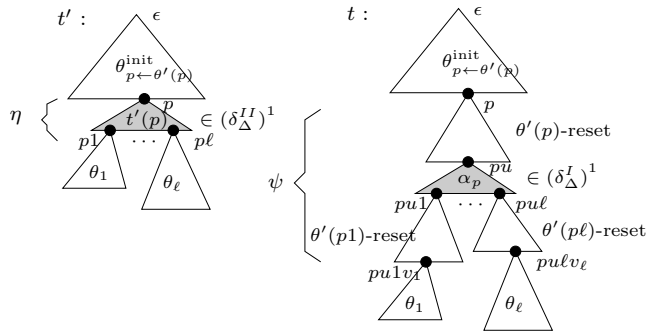
    \caption{\label{fig:remove-equalities} Construction
  of $t$ from $t'$ (\autoref{lemma:remove-equalities})}
    \end{center}
  \end{figure}

Then, we build $t$ and the associated run $\arun$ by combining
successive partial runs: \begin{itemize}[label=$\triangleright$]
\item start $\arun$ with $\arun_{p \leftarrow \arun(p)}^\text{init}$,
  i.e., no change from $\arun'$ above position $p$,
\item add a $\arun'(p)$-reset $\arun_{u \leftarrow \arun'(p)} \in
  \runsof{\infty}{\arun'(p)}{\auto{\asid}{\apred}^{I}}$ at position
  $p$; such a reset exists by \autoref{lemma:reset} because $\arun'(p)
  \in \pre{({(\trans^{II}_\asid)}^1)} \cap
  \pre{({(\trans^{II}_\asid)}^\infty)}$, since $\arun'(p)$ belong to a
  non-trivial SCC,
\item pursue at position $pu$ with the transition $\arun'(p)
  \arrow{\alpha_p}{} (\arun'(p1),\ldots,\arun'(p\ell)) \in
        {(\trans^{I}_\asid)}^1$ corresponding to the original
        $1$-transition from $\auto{\asid}{\apred}^{I}$, after adding
        back the equalities removed by the transformation,
  \item for every $i\in\interv{1}{\ell}$, introduce at position $pui$
    a $\arun'(pi)$-reset $\arun_{v_i \leftarrow \arun'(pi)} \in
    \runsof{\infty}{\arun'(pi)}{\auto{\asid}{\apred}^{II}}$; such a
    reset exists by \autoref{lemma:reset} since $\arun'(pi) \in
    \post{({(\trans^{II}_\asid)}^1)} \cap
    \pre{({(\trans^{II}_\asid)}^\infty)}$ is a pivot state,
  \item continue with $\arun_i$ at position $puiv_i$, for every
    $i\in\interv{1}{\ell}$.
\end{itemize}

\begin{fact}\label{fact:tw_remove_equalities}
  $\twof{\sem{\exclof{\charform{t'}}}} \le
  \twof{\sem{\exclof{\charform{t}}}} + K_1$, where $K_1 \le
  \maxvarinruleof{\asid}$ is the maximal number of $i$-variables,
  $i\in\nat\cup\set{\epsilon}$, in any 1-transition of
  $\auto{\asid}{\apred}^{II}$.
\end{fact}
\proof{Let $\phi \isdef \Bigstar_{\!\!p \text{ not a prefix of } r}
  ~\atpos{t'(r)}{r} * \Bigstar_{\!\!pir \in\dom{t'}}
    ~\atpos{t'(pir)}{puiv_ir}$, which corresponds to the
    characteristic formula $\charform{t'}$ without the 1-transition at
    position $p$, and with the new position labels in $t$.  Now
    $\charform{t} = \phi * \psi$ where $\psi$ is the separating
    conjunction of all $\atpos{t(pr)}{pr}$ with $pr\in\dom{t}$ and
    $uiv_i$ not the prefix of $r$ for any $i\in\interv{1}{\ell}$
    (position $pr$ is either the start of the 1-transition or part of
    one of the reset contexts).  Let $F \isdef \fv{\phi} \cap
    \fv{\psi}$ and $\eta \isdef \Bigstar \set { x = y \mid x, y \in
      V_f, x \eqof{\psi} y}$. Then, $F$ contains the variables at the
    extremity of the added part from $t'$ to $t$, i.e., $F =
    \set{\atpos{x}{p}_k \mid k\in\interv{1}{\arityof{t(p)}}} \cup
    \set{\atpos{x}{puiv_i}_j \mid i\in\interv{1}{\ell},
      j\in\interv{1}{\arityof{t(puiv_i)}}}$. These variables exactly
    correspond to the parameters appearing in the 1-transition at
    position $p$ of $t'$, i.e., $F=\set{\atpos{x}{p}_k \mid
      k\in\interv{1}{\arityof{t'(p)}}} \cup \set{\atpos{x}{pi}_j \mid
      i\in\interv{1}{\ell}, j\in\interv{1}{\arityof{t'(pi)}}}$, hence
    $\cardof{F} \le K_1 \le \maxvarinruleof{\asid}$. By \autoref{def:alpha-sid}
    (\ref{it2:alpha-sid}), $\phi$ does not induce equalities between
    variables of $F$. On the other hand, $\psi$ only induces
    equalities between persistent variables of $F$, thanks to the
    introduced reset paths. These equalities $\eta$ correspond exactly
    to those occurring in $t'(p)$, thus $\phi * \eta$ is equal to
    $\charform{t'}$, modulo a renaming of the variables.
    Now $\phi * \psi$ is satisfiable (since $\auto{\asid}{\apred}^{I}$ is all-satisfiable),
    thus by \autoref{lemma:sep-treewidth}, we obtain
    $\twof{\sem{\exclof{\charform{t'}}}} = \twof{\sem{\exclof{(\phi * \eta)}}}
    \le \twof{\sem{\exclof{(\phi * \psi)}}} + K_1 = \twof{\sem{\exclof{\charform{t}}}} + K_1$.
    \qed}

After doing this transformation ($t'$ to $t$) for all 1-transitions,
the final tree $t$ satisfies $t\in\langof{}{\auto{\asid}{\apred}^{I}}$
since all the added transitions (1-transition or reset) appear in
$\trans^{I}_\asid$.  With the inequality at each step
(\autoref{fact:tw_remove_equalities}) and since each 1-transition of
$\trans^{II}_\asid$ occurs exactly once in the initial tree $t'$
(\autoref{lemma:one-transitions}), we get
$\twof{\sem{\exclof{\charform{t'}}}} \le
\twof{\sem{\exclof{\charform{t}}}} + K$, where $K =
\cardof{{(\trans^{II}_\asid)}^1} \cdot K_1 \le
\cardof{{(\trans^{II}_\asid)}^1} \cdot \maxvarinruleof{\asid}$.  As
$t'$ has been chosen arbitrary and $\cardof{{(\trans^{II}_\asid)}^1} =
\cardof{{(\trans^{I}_\asid)}^1}$, we conclude that
$\twof{\sem{\auto{\asid}{\apred}^{II}}} \le
\twof{\sem{\auto{\asid}{\apred}^{I}}} +
\cardof{{(\trans^{I}_\asid)}^1}\cdot \maxvarinruleof{\asid}$.

\skipnoindent (\ref{it3:lemma:remove-equalities}) Let
$t\in\langof{}{\auto{\asid}{\apred}^{I}}$.  By the construction of
$\auto{\asid}{\apred}^{II}$ from $\auto{\asid}{\apred}^{I}$, there
exists a tree $t'\in\langof{}{\auto{\asid}{\apred}^{II}}$ such that
$\dom{t}=\dom{t'}$ and $\charform{t} = \charform{t'} * \psi$ where
$\psi$ is a conjunction of equalities.  $\charform{t'}$ is a qpf
formula thus by \autoref{lemma:qpf-treewidth}
(\ref{it1:qpf-treewidth}), we get $\twof{\sem{\exclof{\charform{t}}}}
= \twof{\sem{\exclof{(\charform{t'} * \psi)}}} \le
\twof{\sem{\exclof{\charform{t'}}}}$.  As $t$ has been chosen
arbitrary, we obtain that $\twof{\sem{\auto{\asid}{\apred}^{I}}} \le
\twof{\sem{\auto{\asid}{\apred}^{II}}}$.

\subsection{Proof of \autoref{lemma:remove-persistent:B}}
\label{app:remove-persistent:B}

\skipnoindent (\ref{it1:lemma:remove-persistent:B}) By
\autoref{lemma:remove-persistent:A},
$\langof{}{\widetilde{\auto{\asid}{\apred}}^{II}} =
\langof{}{\auto{\asid}{\apred}^{II}}$, hence
$\widetilde{\auto{\asid}{\apred}}^{II}$ is all-satisfiable, because
$\auto{\asid}{\apred}^{II}$ is all-satisfiable, by
\autoref{lemma:remove-equalities} (\ref{it1:lemma:remove-equalities}).
Since $\widetilde{\mathcal{B}}_i$ is obtained by removing zero or more
transitions from $\widetilde{\auto{\asid}{\apred}}^{II}$, we have
$\langof{}{\widetilde{\mathcal{B}}_i} \subseteq
\langof{}{\widetilde{\auto{\asid}{\apred}}^{II}}$, hence
$\mathcal{B}_i$ is all-satisfiable.

Let ${\widetilde{\trans}}^1$ be any of the sets
${\widetilde{\trans}}^1_1, \ldots, {\widetilde{\trans}}^1_m$. We prove
that $\widetilde{\mathcal{B}} \isdef (\Sigma,
\widetilde{\states}_\asid^{I}, (q_\apred,\emptyset),
          {\widetilde{\trans}}^1 \uplus {\widetilde{\trans}}^\infty)$
          is choice-free. To this end, we define the following
          labeling of transitions:
  \[\Lambda(\tau) \isdef \left\{\begin{array}{ll}
  1 & \text{if } \tau \in {\widetilde{\trans}}^1 \\
  \infty & \text{if } \tau \in {\widetilde{\trans}}^\infty
  \end{array}\right.\]
  Let $h : \widetilde{\states}_\asid^{I} \rightarrow
  \states_\asid^{I}$ be the function defined as $h((q,a))\isdef q$,
  for all $(q,a) \in \widetilde{\states}_\asid^{I}$. For each SCC
  $\widetilde{S}$ of $\widetilde{\mathcal{B}}$, we have
  $h(\widetilde{S}) \subseteq S$, for an SCC $S$ of
  $\auto{\asid}{\apred}^{II}$, by the definitions of
  $\widetilde{\auto{\asid}{\apred}}^{II}$ and
  $\widetilde{\mathcal{B}}$. Since $\auto{\asid}{\apred}^{II}$ is
  choice-free, by \autoref{lemma:remove-equalities}
  (\ref{it1:lemma:remove-equalities}), we can extend the labeling
  $\Lambda$ to the SCCs of $\widetilde{\mathcal{B}}$, as follows:
  \[\Lambda(\widetilde{S}) \isdef \left\{\begin{array}{ll}
  1 & \text{if } h(\widetilde{S}) \subseteq S \text{, for some $1$-SCC $S$ of $\auto{\asid}{\apred}^{II}$} \\
  \infty & \text{if } h(\widetilde{S}) \subseteq S \text{, for some $\infty$-SCC $S$ of $\auto{\asid}{\apred}^{II}$}
  \end{array}\right.\]
  We prove first the following fact:

  \begin{fact}\label{fact:remove-persistent:B}
    Let $\widetilde{S}$ be an SCC of $\widetilde{\mathcal{B}}$, such
    that $\Lambda(\widetilde{S})=1$. Then $h(\widetilde{S})$ is a
    $1$-SCC of $\auto{\asid}{\apred}^{II}$. Moreover,
    $h(\widetilde{S})$ is linear if $\widetilde{S}$ is linear.
  \end{fact}
  \proof{ By the definition of $\Lambda$, we have that
    $h(\widetilde{S}) \subseteq S$, for a $1$-SCC of
    $\auto{\asid}{\apred}^{II}$, hence, for the first point, it
    suffices to prove that $h(\widetilde{S})=S$. Suppose, for a
    contradiction, that there exists a state $q \in S \setminus
    h(\widetilde{S})$. Let $q' \in h(\widetilde{S})$ be a state. Since
    $S$ is an SCC of $\auto{\asid}{\apred}^{II}$, we have $q \reach^*
    q' \reach^* q$ in $\auto{\asid}{\apred}^{II}$. Since $q' \in
    h(\widetilde{S})$, there exists an injective partial mapping $a' :
    \interv{1}{\arityof{q'}} \rightarrow \interv{1}{\mathcal{M}}$ such
    that $(q',a') \in \widetilde{S}$. By the construction of
    $\widetilde{\auto{\asid}{\apred}}^{II}$, we obtain $(q',a')
    \reach^* (q,a)$ in $\widetilde{\auto{\asid}{\apred}}^{II}$, for
    some injective partial mapping $a : \interv{1}{\arityof{q}}
    \rightarrow \interv{1}{\mathcal{M}}$. To show a contradiction, we
    prove that $(q,a) \reach^* (q',a')$ in
    $\widetilde{\auto{\asid}{\apred}}^{II}$, hence $q' \in
    h(\widetilde{S})$. By the construction of
    $\widetilde{\auto{\asid}{\apred}}^{II}$, there exist injective
    partial mappings $a_1, a_2, \ldots : \interv{1}{\arityof{q'}}
    \rightarrow \interv{1}{\mathcal{M}}$, such that $(q,a) \reach^*
    (q',a_1) \reach^* (q',a_2) \reach^* \ldots$ Because the underlying
    transitions of $\auto{\asid}{\apred}^{II}$ along this path are
    $\infty$-transitions, the mappings $a_1,a_2,\ldots$ are undefined
    everywhere except for $\profile{\auto{\asid}{\apred}^{II}}(q')$,
    where they are defined. Because these transitions are from
    $\prepost{S}$, each of these mappings can be obtained from the
    previous one by composition with a permutation over
    $\profile{\auto{\asid}{\apred}^{II}}(q')$. Since all such
    permutations can be enumerated in this was, this leads to
    $(q',a'_1) \reach^* (q',a')$ in
    $\widetilde{\auto{\asid}{\apred}}^{II}$, i.e., contradiction.

    For the second point, suppose, for a contradiction, that
    $h(\widetilde{S})$ is not linear. Then, there exists a transition
    $q_0 \arrow{\alpha}{} (q_1, \ldots, q_\ell)$ such that
    $q_0,q_i,q_j \in h(\widetilde{S})$, for some indices $1 \le i < j
    \le \ell$. Since $h(\widetilde{S})$ is a connected component of
    $\auto{\asid}{\apred}^{II}$, we have $q_i \reach^* q_0$ and $q_j
    \reach^* q_0$ in $\auto{\asid}{\apred}^{II}$. Let $a_0, \ldots,
    a_\ell$ be injective partial mappings, such that $(q_0,a_0)
    \arrow{\alpha}{} ((q_1,a_1), \ldots, (q_\ell,a_\ell))$ is a
    transition of $\widetilde{\mathcal{B}}$. By the above argument, we
    obtain $(q_i,a_i) \reach^* (q_0,a_0)$ and $(q_j,a_j) \reach^*
    (q_0,a_0)$, hence $\widetilde{S}$ is non-linear,
    contradiction. \qed}

  \skipnoindent Back to the main proof, we check
  the points of \autoref{def:choice-free} below:

  \skipnoindent (\ref{it1:def:choice-free}) The
  SCC graph of $\widetilde{\mathcal{B}}$ is a tree, by the choice of
  the set $\widetilde{\trans}^1$ corresponding to
  $\widetilde{\mathcal{B}}$. Moreover, each non-root SCC of
  $\widetilde{\mathcal{B}}$ is entered by exactly one branch, since
  this is already the case for the choice-free automaton
  $\auto{\asid}{\apred}^{II}$.

  \skipnoindent (\ref{it21:def:choice-free}) Let
  $\widetilde{S}$ be a linear SCC of $\widetilde{\mathcal{B}}$, such
  that $\Lambda(\widetilde{S})=1$. By
  \autoref{fact:remove-persistent:B}, $h(\widetilde{S})$ is a linear $1$-SCC
  of $\auto{\asid}{\apred}^{II}$, hence
  $\cardof{\post{h(\widetilde{S})}}=1$. Then, we obtain
  $\cardof{\post{\widetilde{S}}}=1$.

  \skipnoindent (\ref{it22:def:choice-free}) Let $\widetilde{\tau} =
  (q_0,a_0) \arrow{\alpha}{} ((q_1,a_1), \ldots, (q_\ell,a_\ell))$ be
  a transition of $\widetilde{\mathcal{B}}$. Then
  $\Lambda(\widetilde{\tau})=1$ iff $\widetilde{\tau} \in
  {\widetilde{\trans}}^1$, by the definition of
  $\Lambda$. ``$\Rightarrow$'' Assume that $\widetilde{\tau} \in
  {\widetilde{\trans}}^1$. Then, the underlying transition $\tau = q_0
  \arrow{\alpha}{} (q_1,\ldots,q_\ell)$ is a $1$-transition of
  $\auto{\asid}{\apred}^{II}$. Since $\auto{\asid}{\apred}^{II}$ is
  choice-free, we have that $\tau \in \post{S}$, for some linear
  $1$-SCC $S$ of $\auto{\asid}{\apred}^{II}$. Let $\widetilde{S}$ be
  an SCC of $\widetilde{\mathcal{B}}$ such that $h(\widetilde{S})
  \subseteq S$. Then $\Lambda(\widetilde{S})=1$, by the definition of
  $\Lambda$. Suppose, for a contradiction that $\widetilde{S}$ is
  non-linear. Then, $h(\widetilde{S})$ is non-linear, hence $S$ is
  non-linear, contradiction. Finally, $\widetilde{\tau} \in
  \post{\widetilde{S}}$ follows from $\tau \in
  \post{S}$. ``$\Leftarrow$'' Assume that
  $\widetilde{\tau}\in\post{\widetilde{S}}$, for a linear SCC
  $\widetilde{S}$ of $\widetilde{\mathcal{B}}$, such that
  $\Lambda(\widetilde{S})=1$. Then $h(\widetilde{S})$ is a linear
  $1$-SCC of $\auto{\asid}{\apred}^{II}$, by
  \autoref{fact:remove-persistent:B}. Moreover,
  $\tau\in\post{h(\widetilde{S})}$, where $\tau$ is the underlying
  transition of $\widetilde{\tau}$ from
  $\auto{\asid}{\apred}^{II}$. Hence $\tau$ is a $1$-transition of
  $\auto{\asid}{\apred}^{II}$ and $\widetilde{\tau} \in
  {\widetilde{\trans}}^1$, by the choice of ${\widetilde{\trans}}^1$.

  \skipnoindent (\ref{it23:def:choice-free}) By
  \autoref{fact:remove-persistent:B}, $\Lambda(\widetilde{S})=1$ iff
  $h(\widetilde{S})$ is a $1$-SCC of $\auto{\asid}{\apred}^{II}$, for
  each SCC $\widetilde{S}$ of
  $\widetilde{\mathcal{B}}$. ``$\Rightarrow$'' Assume that
  $h(\widetilde{S})$ is a $1$-SCC of
  $\auto{\asid}{\apred}^{II}$. Since $\auto{\asid}{\apred}^{II}$ is
  choice-free, then $h(\widetilde{S})$ is the root of the SCC tree of
  $\auto{\asid}{\apred}^{II}$ or
  $\pre{\widetilde{S}}=\set{\widetilde{\tau}}$, for some transition
  $\widetilde{\tau}$ of $\widetilde{\mathcal{B}}$, such that
  $\Lambda(\widetilde{\tau})=1$. In the first case, $h(\widetilde{S})$
  has no incoming transition, hence $\widetilde{S}$ is the root of the
  SCC tree of $\widetilde{\mathcal{B}}$. In the second case,
  $\pre{h(\widetilde{S})}=\set{\tau}$, for a $1$-transition $\tau$ of
  $\auto{\asid}{\apred}^{II}$. Then,
  $\pre{\widetilde{S}}=\set{\widetilde{\tau}}$, for a transition
  $\widetilde{\tau}$, such that $\Lambda(\widetilde{\tau})=1$, by the
  definition of $\Lambda$. ``$\Leftarrow$'' If $h(\widetilde{S})$ is
  the root of the SCC tree of $\auto{\asid}{\apred}^{II}$, then
  $\widetilde{S}$ has no incoming transitions, hence it is the root of
  the SCC tree of $\widetilde{\mathcal{B}}$. If
  $\pre{h(\widetilde{S})}=\set{\tau}$, for a $1$-transition $\tau$ of
  $\auto{\asid}{\apred}^{II}$, we obtain
  $\pre{\widetilde{S}}=\set{\widetilde{\tau}}$, for a transition
  $\widetilde{\tau}$ of $\widetilde{\mathcal{B}}$, such that
  $\Lambda(\widetilde{\tau})=1$, by the definition of $\Lambda$.

  \skipnoindent (\ref{it2:lemma:remove-persistent:B})
  Since $\langof{}{\widetilde{\mathcal{B}}_i} \subseteq
  \langof{}{\auto{\asid}{\apred}^{II}}$, for each $i \in
  \interv{1}{m}$, we have $\bigcup_{i=1}^m
  \langof{}{\widetilde{\mathcal{B}}_i} \subseteq
  \langof{}{\auto{\asid}{\apred}^{II}}$. By the definition of
  $\widetilde{\mathcal{B}}_1, \ldots, \widetilde{\mathcal{B}}_m$, for
  each accepting run $\arun$ of
  $\widetilde{\auto{\asid}{\apred}}^{II}$, there exists $i \in
  \interv{1}{m}$, such that $\arun$ is an accepting run of
  $\widetilde{\mathcal{B}}_i$, thus
  $\langof{}{\auto{\asid}{\apred}^{II}} \subseteq \bigcup_{i=1}^m
  \langof{}{\widetilde{\mathcal{B}}_i}$. \qed

\subsection{Proof of \autoref{lemma:remove-persistent}}
\label{app:remove-persistent}

First, we prove the following fact:
\begin{fact}\label{fact:remove-persistent}
  Let $t \in \langof{}{\widetilde{\mathcal{B}}}$ be a tree, $\arun$ be
  an accepting run of $\widetilde{\mathcal{B}}$ over $t$ and $p \in
  \dom{t}$ be a position such that $\arun(p)=(q,a)$. Then,
  $\atpos{x}{p}_i \eqof{\charform{t}} \atpos{y}{r_i}_{a(i)}$, for each
  $i \in \dom{a}$, such that $r_i$ is the unique position where a
  variable $\atpos{y}{\epsilon}_{a(i)} \in \mathcal{Y}$ occurs in
  $\charform{t}$.
\end{fact}
\begin{proof} By induction on the structure of $t$. \end{proof}

\skipnoindent (\ref{it1:lemma:remove-persistent}) Let $\overline{t}
\in \langof{}{\overline{\mathcal{B}}}$ be a tree. By the construction
of $\overline{\mathcal{B}}$, the accepting run $\overline{\arun}$ of
$\overline{\mathcal{B}}$ over $\overline{t}$ can be transformed into
an accepting run $\arun$ of $\widetilde{\mathcal{B}}$ over a tree $t$,
such that $\dom{t} = \dom{\overline{t}}$, by changing the labels
$\overline{\alpha}$ back to the original labels $\alpha$. Since
$\widetilde{\mathcal{B}}$ is all-satisfiable, by
\autoref{lemma:remove-persistent:B}
(\ref{it1:lemma:remove-persistent:B}), the formula $\charform{t}$ is
satisfiable, and let $(\univ,\struc)$ be a structure and $\store$ be a
store such that $(\univ,\struc) \models^{\store}
\charform{t}$. However, $\charform{\overline{t}}$ is obtained from
$\charform{t}$ by removing several (dis-)equalities from the labels of
$1$-transitions and by changing each relation atom
$\arel(\atpos{z}{p}_{1}, \ldots, \atpos{z}{p}_{\arityof{\arel}})$ into
a relation atom $\arel_g(\zeta_\tau(\atpos{z}{p}_{i_1}), \ldots,
\zeta_\tau(\atpos{z}{p}_{i_k}))$, according to the construction of
$\overline{\mathcal{B}}$. Suppose, for a contradiction, that there
exist two distinct positions $p_1,p_2\in\dom{\overline{t}}$ such that
the relation atoms $\arel_g(\zeta_{\tau_1}(\atpos{z}{p_1}_{i_1}),
\ldots, \zeta_{\tau_1}(\atpos{z}{p_1}_{i_k}))$ and
$\arel_g(\zeta_{\tau_2}(\atpos{z}{p_2}_{i_1}), \ldots,
\zeta_{\tau_2}(\atpos{z}{p_2}_{i_k}))$ occur in
$\charform{\overline{t}}$ and $\zeta_{\tau_1}(\atpos{z}{p_1}_{i_j})
\eqof{\charform{\overline{t}}} \zeta_{\tau_2}(\atpos{z}{p_2}_{i_j})$,
for all $j \in \interv{1}{k}$. Then, there exist relation atoms
$\arel(\atpos{z}{p_1}_1, \ldots, \atpos{z}{p_1}_{\arityof{\arel}})$
and $\arel(\atpos{z}{p_2}_1, \ldots,
\atpos{z}{p_2}_{\arityof{\arel}})$ that occur in $\charform{t}$, such
that $i_1, \ldots, i_k \in \interv{1}{\arityof{\arel}}$ and, by
\autoref{fact:remove-persistent} and properties of the renaming, we
obtain that $\atpos{z}{p_1}_j \eqof{\charform{t}} \atpos{z}{p_2}_j$,
for all $j \in \interv{1}{\arityof{\arel}}$. This however contradicts
the satisfiability of $\charform{t}$, hence such relation atoms cannot
occur in $\charform{\overline{t}}$.  Similarly, we obtain a
contradiction if we suppose a disequality atom $\zeta_{\tau}(x) \neq
\zeta_{\tau}(y)$ occurs in $\charform{\overline{t}}$ and
$\zeta_{\tau}(x) \eqof{\charform{\overline{t}}} \zeta_{\tau}(y)$.
Then, $\charform{\overline{t}}$ is satisfiable and, since the choice
of $\overline{t}$ was arbitrary, $\overline{\mathcal{B}}$ is
all-satisfiable.

\skipnoindent (\ref{it2:lemma:remove-persistent}) Let
$\overline{\astruc} = (\univ, \overline{\struc}) \in
\sem{\overline{\mathcal{B}}}$ be a structure. Then, there exists a
tree $\overline{t} \in \langof{}{\overline{\mathcal{B}}}$ and a store
$\overline{\store}$ such that $\overline{\astruc}
\models^{\overline\store} \charform{\overline{t}}$. Let
$\widetilde{t}$ be the tree obtained from $\overline{t}$ by replacing
back each label $\overline{\alpha}$ with the original label $\alpha$,
according to the construction of $\overline{\mathcal{B}}$ from
$\widetilde{\mathcal{B}}$.  Note that $\dom{\widetilde{t}} =
\dom{\overline{t}}$ and let $\overline{\arun}$, $\widetilde{\arun}$ be
two related accepting runs of $\overline{\mathcal{B}}$,
$\widetilde{\mathcal{B}}$ over $\overline{t}$, $\widetilde{t}$
respectively.  We define a store $\widetilde{\store}$ over
$\fv{\charform{\widetilde{t}}}$ as follows.  Let $p$ be an arbitrary
position in $\dom{\widetilde{t}}$. Let $(q_0,a_0)$ be the state
assigned to $p$ and $\tau ~:~ (q_0,a_0) \arrow{\alpha}{} ((q_1,a_1),
\ldots, (q_\ell,a_\ell))$ be the transition taken at $p$ in the run
$\widetilde{\arun}$.  Let $\zeta_\tau$ be the renaming of variables as
defined in the transformation.  The store $\widetilde{\store}$ is then
defined by: \begin{itemize}[label=$\triangleright$]
\item $\widetilde{\store}(\atpos{y}{p}_i) \isdef
  \left\{ \begin{array}{rl}
    \overline{\store}(\atpos{y}{p}_i) & \mbox{if } \tau \not\in \widetilde{\trans}^1 \\
    \mathit{u_i} & \mbox{if } \tau \in \widetilde{\trans}^1
    \mbox{ for some distinct values } u_i \not\in \overline{\store}(\fv{\charform{\overline{t}}}),~ i\in\interv{1}{\mathcal{M}}
  \end{array}\right.$
\item $\widetilde{\store}(\atpos{x}{p}_i) \isdef
  \left\{ \begin{array}{rl}
    \overline{\store}(\atpos{x}{p}_j) &
    \mbox{if } i \not\in \dom{a_0} \mbox{ and } \zeta_\tau(\atpos{x}{\epsilon}_i) = \atpos{x}{\epsilon}_j \\
    u_{a_0(i)} & \mbox{if } i \in\dom{a_0}
  \end{array} \right.$
\end{itemize}
That is, the store $\widetilde{\store}$ allocates fresh distinct
values for all the persistent variables and re-uses the already given
values for the non-persistent ones in $\overline{\store}$, while
taking into account their renaming by $\zeta_\tau$.  We now consider
the structure $\widetilde{\astruc} = (\univ, \widetilde{\struc})$,
where $\widetilde{\struc}$ is the interpretation that assigns to each
relation symbol $\arel$ the set of tuples
$\tuple{\widetilde{\store}(\xi_1), \ldots,
  \widetilde{\store}(\xi_{\arityof{\arel}})}$ for every relation atom
$\arel(\xi_1, \ldots, \xi_\arityof{\arel})$ in
$\charform{\widetilde{t}}$. We can now check that $(\univ,
\widetilde{\struc}) \models^{\widetilde{\store}}
\charform{\widetilde{t}}$.  Any two tuples from
$\widetilde{\struc}(\arel)$ defined as above are necessarily distinct,
unless the corresponding tuples restricted to non-persistent variables
from $\overline{\struc}(\arel_g)$ are not distinct.  Similarly,
equalities on disequality atoms on non-persistent variables only hold
in $\charform{\widetilde{t}}$ as they already hold in the renamed form
in $\charform{\overline{t}}$.  Finally, equalities and disequalities
in $\charform{\widetilde{t}}$ involving persistent variables hold by
the construction of $\widetilde{\store}$.  But now, every tree
decomposition of $\widetilde{S}$ is a tree decomposition for
$\overline{S}$.  Hence, $\twof{\overline{S}} \le \twof{\widetilde{S}}$
and consequently, since the choice of $\overline{S}$ was arbitrary we
have $\twof{\sem{\overline{\mathcal{B}}}} \le
\twof{\sem{\widetilde{\mathcal{B}}}}$.

\skipnoindent (\ref{it3:lemma:remove-persistent}) Let
$\widetilde{\astruc} = (\univ, \widetilde{\struc}) \in
\sem{\widetilde{\mathcal{B}}}$ be a structure. Then, there exists a
tree $\widetilde{t} \in \langof{}{\widetilde{\mathcal{B}}}$ and a
store $\widetilde{\store}$, such that $\widetilde{\astruc}
\models^{\widetilde\store} \charform{\widetilde{t}}$.  Let
$\overline{t} \in \langof{}{\overline{\mathcal{B}}}$ be the tree
obtained by changing each label $\alpha$ of $\widetilde{t}$ into
$\overline{\alpha}$, according to the construction of
$\overline{\mathcal{B}}$. We consider the structure
$\overline{\astruc} \isdef (\univ, \overline{\struc})$, where
$\overline{\struc}$ interprets each relation symbol $\arel_g$ by the
set of tuples
$\tuple{\widetilde{\store}(\zeta_\tau(\atpos{z}{p}_{i_1})), \ldots,
  \widetilde{\store}(\zeta_\tau(\atpos{z}{p}_{i_k}))}$, such that
$\arel_g(\zeta_\tau(\atpos{z}{p}_{i_1}), \ldots,
\zeta_\tau(\atpos{z}{p}_{i_k}))$ occurs in
$\charform{\overline{t}}$. Let $\overline{T}$ be a tree decomposition
of $\overline{\astruc}$. We consider the tree decomposition
$\widetilde{T}$ obtained by adding the values
$\widetilde{\store}(\atpos{y}{r_i}_i)$, where $r_i$ is the unique
position where a $\epsilon$-variable $\atpos{y}{\epsilon}_i \in
\mathcal{Y}$ occurs in $\charform{\overline{t}}$, to the label of each
node in $\overline{T}$. Then, $\widetilde{T}$ is a tree decomposition
of $\widetilde{\astruc}$ and moreover $\width{\widetilde{T}} \le
\width{\overline{T}}+\mathcal{M}$. Since the choice of $\overline{T}$
was arbitrary we obtain $\twof{\widetilde{\astruc}} \le
\twof{\overline{\astruc}} + \mathcal{M}$. Consequently, since the
choice of $\widetilde{\astruc}$ was arbitrary and $\mathcal{M} \le
\cardof{\widetilde{\trans}^1} \cdot \maxvarinruleof{\asid}$, we obtain
that $\twof{\sem{\widetilde{\mathcal{B}}}} \le
\twof{\sem{\overline{\mathcal{B}}}} + \cardof{\widetilde{\trans}^1}
\cdot \maxvarinruleof{\asid}$. \qed

\subsection{Proof of \autoref{lemma:exp}}
\label{app:exp}

By hypothesis, $\overline{\mathcal{B}}$ has no persistent variables
and since the relabeling of $1$-transitions introduces no existential
quantifiers, then $\mathcal{B}$ has no persistent variables either.

\skipnoindent (\ref{it1:lemma:exp}) Let $t \in \langof{}{\mathcal{B}}$
be a tree and let $\arun$ be an accepting run of $\mathcal{B}$ over
$t$.  Let us consider the run $\arun'$ of $\overline{\mathcal{B}}$
obtained by replacing with $\emp$ the labels of the $1$-transitions
$\expof{\tau} : q_0 \arrow{\alpha}{} (q_1,\ldots,q_\ell) \in
\trans^1$.  We then define the run $\arun''$ of
$\overline{\mathcal{B}}$ by replacing each occurrence of a
1-transition $\tau: q_0 \arrow{\emp}{} (q_1,\ldots,q_\ell)$ in
$\arun'$ by the partial run $\theta_\tau$ constructed by extending
$\tau$ with the resets used in the definition of the transformation,
namely: \begin{itemize}[label=$\triangleright$]
\item a pre- $q_0$-reset
  $\theta^0_{p_0 \leftarrow q_0} \in \runsof{\infty}{q_0}{{\overline{\mathcal{B}}}}$
  if $q_0 \in \pre{(\trans^\infty)}$ and,
\item a post- $q_i$-reset
  $\theta^i_{p_i \leftarrow q_i} \in \runsof{\infty}{q_i}{{\overline{\mathcal{B}}}}$
  for each $i \in \interv{1}{\ell}$ such that $q_i \in \pre{(\trans^\infty)}$.
\end{itemize}

  \begin{figure}[htbp]
    \begin{center}
    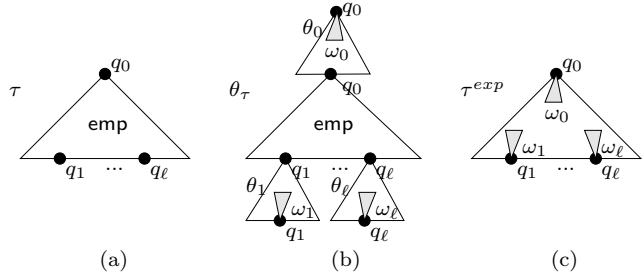
    \caption{\label{fig:exp-1-transitions} Expansion of 1-transitions}
    \end{center}
  \end{figure}

The relation between $\tau$ from $\theta_{\tau}$ is illustrated in
\autoref{fig:exp-1-transitions}~(a,b).  Note that the resets exist
according to \autoref{lemma:reset}.  Then, it is easy to check that
$\arun''$ is indeed a run of $\overline{\mathcal{B}}$. Let $t'' \in
\langof{}{\overline{\mathcal{B}}}$ be the tree accepted by
$\arun''$. Since $\overline{\mathcal{B}}$ is all-satisfiable, there
exists a store $\store$ and a structure $\astruc''$ such that
$\astruc'' \models^\store \charform{t''}$.  Now, by the definition of
$\mathcal{B}$, the label $\alpha$ of each transition $\expof{\tau}:
q_0 \arrow{\alpha}{} (q_1,\ldots,q_\ell) \in \trans^1$ is actually
$\omega_0 * \Bigstar_{j=1}^{\ell} \omega_j$, where:
\[
\begin{array}{rcl}
  \omega_0 & = & \Bigstar_{i_1,\ldots,i_k \in \interv{1}{\arityof{q_0}}} \allstof{t_0}{\epsilon}{\epsilon}{x_{i_1}, \ldots, x_{i_k}} \\
  \omega_j & = & \Bigstar_{i_1, \ldots, i_k \in \interv{1}{\arityof{q_j}}} \allstof{t_j}{p_j}{j}{x_{i_1}, \ldots, x_{i_k}} \text{, for } j\in\interv{1}{\ell}
\end{array}
\]
and $t_0$, $t_1$, $\ldots$, $t_\ell$ are the $\Sigma$-labelled trees
of the corresponding resets, respectively $\theta_0$, $\theta_1$,
$\dots$, $\theta_\ell$.  This construction is illustrated in
\autoref{fig:exp-1-transitions}(c).  Hence, there exists a
substructure $\astruc \substruc \astruc''$ such that $\astruc
\models^\store \charform{t}$.  Actually, $\astruc$ is simply obtained
from $\astruc''$ by removing all but the tuples introduced by the
relation atoms of $\omega_0$, $\ldots$, $\omega_\ell$, defined along
the partial runs $\arun_\tau$ of $\arun''$. Since the choice of $t$
was arbitrary, $\mathcal{B}$ is all-satisfiable.

\skipnoindent (\ref{it2:lemma:exp}) Let $t' \in
\langof{}{\mathcal{B}}$ be a tree.  Let $t \in
\langof{}{\overline{\mathcal{B}}}$ be the tree obtained by $t'$ be
removing relational atoms from the labels of 1-transitions, that is,
reversing the transformation.  As both $\mathcal{B}$ and
$\overline{\mathcal{B}}$ are all-satisfiable, the characteristic
formul{\ae} $\charform{t'}$ and $\charform{t}$ are satisfiable.  They
differ, moreover, only by finitely many relational atoms.  Henceforth,
by using \autoref{lemma:qpf-treewidth} (\ref{it4:qpf-treewidth}), we
obtain that the difference between the treewidth of their models is
bounded by the number of free variables occurring on these atoms, that
is, at most $\cardof{\overline{\trans}^1} \cdot
\maxvarinruleof{\asid}$.  Then $\twof{\sem{\exclof{\charform{t'}}}}
\le \twof{\sem{\exclof{\charform{t}}}} + \cardof{\overline{\trans}^1}
\cdot \maxvarinruleof{\asid}$ and as the choice of $t'$ was arbitrary,
we obtain $\twof{\sem{\mathcal{B}}} \le
\twof{\sem{\overline{\mathcal{B}}}} + \cardof{\overline{\trans}^1}
\cdot \maxvarinruleof{\asid}$.

\skipnoindent (\ref{it3:lemma:exp}) Let $t \in
\langof{}{\overline{\mathcal{B}}}$ be a tree. Let $t' \in
\langof{}{\mathcal{B}}$ be the tree obtained from $t$ by replacing the
$\emp$-labels of 1-transitions by their corresponding formul{\ae}
according to the transformation.  As both $\overline{\mathcal{B}}$ and
$\mathcal{B}$ are all-satisfiable, the characteristic formul{\ae}
$\charform{t}$ and $\charform{t'}$ are satisfiable. These formul{\ae}
differ only by a finite number of relation atoms, that is, the ones
inserted in the labels of the modified $1$-transitions. Henceforth, by
\autoref{lemma:qpf-treewidth} (\ref{it3:qpf-treewidth}), we obtain
that the difference between the treewidth of their models is bounded
by the number $K$ of these relation atoms.  As the choice of $t$ was
arbitrary, we obtain that $\twof{\overline{\mathcal{B}}} \le
\twof{\mathcal{B}} + K$.

We can further obtain an upper bound on $K$ as follows.  First, note
that for a relation symbol $\arel \in \relations$ and $m$ distinct
variables there exists at most distinct $m^{\arityof{\arel}}$ relation
atoms.  Note that, in our context $m \le \maxvarinruleof{\asid}$ and
$\arityof{\arel} \le \maxrelarityof{\asid}$.  Second, such atoms could
be added for every 1-transition, for every one of its states $q_i$.
That is, overall we obtain $K \le \cardof{\overline{\trans}^1} \cdot
(1 + \maxrulearityof{\asid}) \cdot \relationsno{\asid} \cdot
\maxpredarityof{\asid}^{\maxrelarityof{\asid}}$. \qed

\end{document}